\documentclass[12pt, a4paper]{article}
\usepackage[a4paper, margin=1in,inner=1in, outer=1in]{geometry}
\usepackage[authoryear, round]{natbib}
\RequirePackage{amsthm,amsmath,amsfonts,amssymb}

\RequirePackage[authoryear]{natbib}
\RequirePackage[colorlinks,citecolor=blue,urlcolor=blue]{hyperref}
\RequirePackage{graphicx}
\usepackage{amsthm}
\usepackage{adjustbox}

\newcommand{\notocsection}[1]{
\refstepcounter{section}
\section*{\thesection \quad #1}}

\newcommand{\notocsubsection}[1]{
\refstepcounter{subsection}
\subsection*{\thesubsection \quad #1}}

\makeatletter

\renewcommand{\l@section}{%
  \@dottedtocline{1}{1em}{2em}}%

\renewcommand{\l@subsection}{%
  \@dottedtocline{2}{4em}{3em}}%

\makeatother

\usepackage{mathtools}
\mathtoolsset{showonlyrefs=true}
\usepackage{tikz}
\usepackage{bbm}
\usepackage{enumitem}

\newcommand{\iid}{\overset{\text{i.i.d.}}{\sim}}
\newcommand{\tauhat}{\widehat{\tau}}

\newcommand{\NN}{\mathbb{N}}
\newcommand{\PP}{\mathbb{P}}

\newcommand{\RR}{\mathbb{R}}
\newcommand{\EE}{\mathbb{E}}

\newcommand{\N}{\text{N}}
\newcommand\ind{\mathbbm{1}}
\newcommand\normm[1]{\left \lVert #1 \right \rVert}
\newcommand\normmo[1]{\left \lVert #1 \right \rVert_{\Psi_2}}

\newcommand\normmop[1]{\left \lVert #1 \right \rVert_{\mathrm{op}}}

\newcommand{\Tr}{\mathrm{Tr}}

\newcommand{\PSD}{\mathrm{PSD}}

\newcommand{\dyn}{}

\theoremstyle{plain}
\newtheorem{theorem}{Theorem}
\newtheorem{proposition}{Proposition}
\newtheorem{lemma}{Lemma}

\theoremstyle{plain}
\newtheorem{assumption}{Assumption}
\newtheorem{example}{Example}
 
\begin{document}

\title{A grid-based methodology for fast online changepoint detection}

\author{Per August Jarval Moen$^{1}$}

\date{$^1$Department of Mathematics, University of Oslo}

\maketitle

\begin{abstract}
We propose a grid-based methodology that turns suitable offline (non-sequential) changepoint tests into online (sequential) procedures by testing for changes over a dynamically updating grid of candidate locations. For a broad class of test statistics, including those based on empirical averages and certain likelihood ratios, we show that the resulting online procedure has update and storage costs that grow at most logarithmically with the sample size.
We further show that the resulting procedure has explicit control of the false alarm probability, and that finite-sample power guarantees for the test translate directly into non-asymptotic upper bounds on the detection delay, provided the test is robust to constant-factor misspecification of the changepoint location. 
We illustrate the methodology by constructing online detectors for changes in the mean and in the covariance matrix of multivariate data, and prove that their detection delays are near-optimal.  The effectiveness of the methodology is supported by a simulation study, where we compare its performance in detecting mean changes with that of state-of-the-art online methods. To illustrate its practical applicability, we use the methodology to detect structural changes in currency exchange rates in real time.
\end{abstract}

\notocsection{Introduction}\label{sec:intro}
The technological advancements of recent decades have resulted in an unprecedented explosion of data collection and availability, presenting novel challenges and opportunities. Among these is the problem of determining whether the distribution of a data sequence is constant, or if it is changing. As distributional changes may signal the onset of new regimes or anomalous periods, changepoint detection has attracted considerable interest. A considerable volume of research has been devoted to \textit{offline} changepoint detection, where data sets of fixed size are scanned for changepoints retrospectively. For examples of recent works, see for instance \cite{pelt}, \citet{wbs}, \citet{inspect}, \citet{covariancecusum}, \citet{kovacs2022}, and \citet{pilliat}.

Due to the ubiquitous adoption of sensor-based technologies in scientific, industrial and residential settings, data nowadays are often collected in the form of streams, in which the data arrive sequentially over time. Detecting changepoints in such data necessitates online algorithms, that is, methods capable of scanning the data for changepoints in real time. Online methods must operate under strict computational constraints, and the update cost (the computational cost of processing a new data point) and the storage cost (the amount of data stored in memory) should be minimal. Moreover, they should guarantee control over false alarms and detect changepoints with as small a delay as possible. In time-critical applications such as condition monitoring \citep{wind}, health monitoring \citep{health} and finance \citep{finance}, timely and accurate identification of changepoints can lead to significant improvements in efficiency, safety and decision making. Still, online changepoint detection is arguably less explored than its offline counterpart, and is receiving increasing attention.

The study of online changepoint detection dates back at least to \citet{page1954}, who proposed a likelihood‑ratio based method for detecting a change in distribution in univariate data. Page’s method can be adapted to a variety of parametric models, enjoys constant‑order update and storage costs, and is even optimal in a certain asymptotic sense \citep{pageoptimal}. Its main limitation is the requirement that both the pre‑ and post‑change distributions be known in advance. Several more recent contributions also utilise the likelihood ratio, similar to Page's method, although under less stringent assumptions. For example, Generalised Likelihood‑Ratio (GLR) procedures \citep[e.g.][]{lai_2010} maximise the likelihood ratio over the parameters of the model to allow for unknown pre‑ and post‑change parameters, but they typically incur update and storage costs that grow linearly (or faster) in the number of observations. The ocd method of \citet{chen_high-dimensional_2022} detects changes in the mean of high-dimensional data by applying Page's CUSUM to each individual series over a dyadic grid of possible change magnitudes, yielding constant-order update and storage costs. The method of \cite{mei2010} is of a similar flavour. Other contributions, such as those of \cite{xs2013}, \cite{chan2017} and \cite{anote}, do not build upon the likelihood ratio, but instead utilise test statistics tailored for the mean-change model. Most comparable to the methodology to be introduced in this paper is the recent work of \cite{computationalgeometry}, whose method MdFOCuS leverages computational geometry to compute the GLR exactly, with poly-logarithmic update and storage costs in the sample size for exponential family models (see also \citealt{focus}). %

As the above works highlight, designing online methods that are both computationally efficient and endowed with strong statistical guarantees is a delicate task. A seemingly natural alternative is to re‑apply an offline changepoint detection method each time new data arrive. However, this strategy quickly becomes impractical: since the offline method must be re-run whenever new data arrive, the per-update processing time may quickly exceed the data's rate of arrival, and memory resources may be exhausted (see also the discussion in \citealt{chen_high-dimensional_2022}). Still, the numerous offline changepoint tests in the statistical literature should ideally be adaptable for online purposes.

The main contribution of this paper is a grid-based methodology for turning suitable offline changepoint tests into online detection procedures. The key novelty of the methodology is a dynamically updating geometric grid of candidate changepoint locations, over which the offline test is applied. The methodology comes with both computational and statistical performance guarantees. On the computational side, we present specific conditions on the test statistic under which the resulting online procedure has update and storage costs that grow at most logarithmically in the sample size. On the statistical side, we show that the resulting procedure has explicit control of the false alarm probability. Moreover, by exploiting the geometric spacing of the grid, we show that finite-sample power guarantees for the test translate to non-asymptotic upper bounds on the detection delay, as long as the test is robust to constant-factor misspecification of the changepoint location. The computational efficiency of the methodology comes at a statistical price: since changepoints are only tested for over a grid of candidate locations, some detection power may inevitably be lost compared to an exhaustive search. We discuss this trade-off in Sections \ref{sec:unimean} and \ref{sec:statproperties}, and investigate it numerically in the supplementary material. 
This loss appears to be modest, at least in the settings we study. For the special cases of a change in the mean or in the covariance matrix of multivariate data, the resulting online detectors have provably near-optimal detection delay rates, and in a simulation study for detecting a change in multivariate means, the proposed method performs comparably to state-of-the-art online methods. We remark that an open-source Python implementation of the methodology, covering a range of change types, is available in the package \texttt{gridcp}, detailed in \cite{gridcp}.

The paper is organised as follows. In Section \ref{sec:unimean} we study a univariate change-in-mean problem to illustrate the main idea underlying the proposed methodology, and we present a fast online changepoint detection method with minimax rate optimal performance. In Section \ref{sec:general}, we outline the general methodology and state conditions under which logarithmic update and storage costs are achieved, analyse the statistical properties of the resulting method, and provide examples. 
In Section \ref{sec:theoryspec} we rigorously analyse the performance of the methodology in two special cases, namely for detecting (i) a multivariate change in the mean, and (ii) a multivariate change in the covariance matrix.  
In Section \ref{sec:simulations} we investigate the performance of the methodology in detecting changes in the mean vector in a simulation study, comparing it with state-of-the-art methods. In Section \ref{sec:realdata} we apply the methodology to detect covariance changes in currency exchange rate data in real time. In Section \ref{sec:discussion} we conclude with a discussion of the methodology and directions for future work. Further theoretical results, discussions, simulation studies, as well as proofs of the main results, are provided in the supplementary material. 

\notocsubsection{Problem formulation and notation}
We now formalise the online changepoint problem. Focusing on the \textit{detection} of changepoints, we restrict our attention to models with a single changepoint, since methods for such models can simply be restarted after a changepoint is detected. We remark that the online changepoint problem can naturally be extended to include post-detection inference \citep[see, e.g.][]{onlineinference}, although this is beyond the scope of this paper. 

Let $(Y_i)_{i\in \NN}$ be an infinite sequence of (possibly multivariate) independent random variables. The sequence of $Y_i$ is assumed to have a change in distribution at some unknown time index $\tau \in \NN \cup \{\infty\}$, so that $Y_i \sim P_1$ for $i \leq \tau$ and $Y_i \sim P_2$ for $i > \tau$, where $P_1, P_2 \in \mathcal{P}$ denote (possibly) unknown pre- and post-change distributions, respectively, and $\mathcal{P}$ is some class of distributions. The case where $\tau = \infty$ corresponds to there being no changepoint, and otherwise $\tau$ is the changepoint location. We remark that regression models with random covariates are captured by this setup, while regression models with fixed covariate vectors $x_i$ can also be included by replacing $P_1$ and $P_2$ above by $P_1(x_i)$ and $P_2(x_i)$, respectively.

For any value of $\tau$, we let $\PP_{\tau}$ denote the joint distribution of $(Y_i)_{i \in \NN}$, and we let $\EE_{\tau}$ denote the expectation under this distribution. An online changepoint detector is defined as an extended stopping time $\tauhat$ with respect to the natural filtration $( \mathcal{F}_t)_{t \in \NN}$ of the data. That is, $\tauhat$ takes values in $\NN\cup \{\infty\}$, and for any $t \in \NN$, the event $\{\tauhat = t\}$ is $\mathcal{F}_t$-measurable, where $\mathcal{F}_t$ is the $\sigma$-algebra generated by $Y_1, \ldots, Y_t$. Similarly to \citet{anote}, we define the \textit{false alarm probability} of $\tauhat$ to be
\begin{align}
\text{FA}(\tauhat) &= \PP_{\infty}(\tauhat <\infty)  ,
\end{align}
which is our chosen measure of false alarm frequency. Note that this measure  is stricter than the average run length, which is used by e.g.~\citet{chen_high-dimensional_2022} and \citet{li_online_2023}. 
Whenever $\tau < \infty$, we define the \textit{detection delay} of $\tauhat$ to be the random variable $\tauhat - \tau$, which should ideally be as close to $1$ as possible (since $\tau$ is the index of the last pre-change observation). 

We let the \textit{update cost} $\mathrm{UC}(\tauhat, t)$ of $\tauhat$ denote the number of unit‑cost operations (scalar arithmetic or comparisons) required to decide whether the event $\{\tauhat = t\}$ has occurred, given that the event $\{\tauhat =  t-1\}$ has already been evaluated, with all computations performed at machine precision. We also define the \textit{storage cost} $\mathrm{SC}(\tauhat, t)$ of $\tauhat$ to be the number of scalars stored at time $t$ for the online changepoint detector to continue running indefinitely.
We emphasise that a low storage cost is important not only to prevent exhausting memory, but also because hierarchical memory structures (e.g., cache, RAM, disk) become much slower at larger scales \citep[see, e.g.,][]{memory}. Therefore, large storage costs can substantially affect processing times in practice.

We use the following notation throughout the paper. For any $n \in \NN$ we let $[n] = \{1, \ldots, n\}$.
For any pair of integers $i,j$, we let $i \bmod j = i - j \lfloor i/j \rfloor$ denote the remainder when dividing $i$ by $j$.
For any pair $x,y$ of real numbers, we let $x \vee y = \max(x,y)$ and $x \wedge y = \min(x,y)$, and we let $\lfloor x \rfloor$ denote the largest integer no larger than $x$, and $\lceil x \rceil$ the smallest integer no smaller than $x$. Given a set $\mathcal{X}$ of real numbers, we define $x - \mathcal{X} = \left\{ x-s  \  : \ s \in \mathcal{X}\right\}$, and we let $\left| \mathcal{X}\right|$ denote the cardinality of $\mathcal{X}$.
For any vector $v \in \RR^p$ and $q \geq 1$, we let $v(j)$ denote the $j$-th entry of $v$, we let $\normm{v}_q = \{ \sum_{i=1}^p |v(i)|^{q}\}^{1/q}$ denote the $\ell_q$ norm of $v$, we let $\normm{v}_0$ denote the number of non-zero entries in $v$, and we let $\normm{v}_{\infty} = \max_{1\leq j \leq p} |v(j)|$ denote the $\ell_{\infty}$ norm of $v$. 
For any matrix $A$ we define the operator norm of $A$ as $\normmop{A} =\sup_{\|v\|_2 = 1} \normm{Av}_2$, i.e., its largest singular value.  
For any real-valued random variable $X$ with mean zero, we let $\normmo{X} = \inf \left\{ t>0: \ \EE \exp(X^2/t^2)\leq 2 \right\}$ denote the Orlicz-$\Psi_2$ norm (also known as the sub-Gaussian norm) of $X$, and we say that $X$ is sub-Gaussian if $\normmo{X}< \infty$.
For any $p$-dimensional random vector $X$ with mean zero, we define $\normmo{X} = \sup_{v \in \mathbb{S}^{p-1}} \normmo{v^\top X}$, where $\mathbb{S}^{p-1}$ denotes the unit sphere in $\RR^p$ with the Euclidean metric, and we also say that $X$ is sub-Gaussian if $\normmo{X}< \infty$.
Finally, we use the convention that $\inf \emptyset = \infty$.
\notocsection{Univariate change in mean}\label{sec:unimean}
To motivate the proposed methodology, we first consider a univariate change-in-mean problem. Let the pre- and post-change distributions $P_1$ and $P_2$ be the distributions of $Z_1 + \mu_1, Z_2+\mu_2$, respectively, where $\mu_1, \mu_2 \in \RR$ are unknown, $\mu_1\neq \mu_2$, and the $Z_i$ are mean-zero sub-Gaussian random variables satisfying $\normmo{Z_1}, \normmo{Z_2} \leq \sigma$ for some known $\sigma>0$.

Suppose we have observed $Y_1, \ldots, Y_t$ for some $t\geq 2$, suspecting a change to have occurred $g \in \{1, \ldots, t-1\}$ time steps before the last observation (i.e., suspecting that $\tau = t-g$). To measure the discrepancy between the means of the data before and after the suspected changepoint, a widely used statistic in the offline changepoint literature is the CUSUM statistic \citep[see, e.g.,][]{cusumandoptimality}, given by
\begin{align}
    C_g^{(t)} &= \left\{  \frac{g}{t(t-g)}   \right\}^{1/2}  \sum_{i=1}^{t-g} Y_i - \   \left\{ \frac{t-g}{tg}\right\}^{1/2} \sum_{i=t-g+1}^t Y_i,  \label{cusum}
\end{align}
which measures the (weighted) difference between the empirical averages before and after candidate changepoint at location $t-g$.  A natural test statistic for a change in mean having occurred $g$ time steps before the last observation is therefore given by 
\begin{align}
        T^{(t)}_{g} &=  \ind \left \{ \left(C_g^{(t)} \right)^2 > \xi^{(t)} \right \}\label{cusumtest},
\end{align}
which rejects the null hypothesis of no change whenever the squared CUSUM statistic exceeds some time-dependent critical value $\xi^{(t)}>0$.

Of course, the true value of $\tau$ is not known. If $\tau<t$, so that a change has occurred within the current sample, the test statistic in \eqref{cusumtest} will only have high power whenever $g$ is close to $t-\tau$. To have power over all possible changepoint locations, we therefore apply the test $T_g^{(t)}$ over all $g$'s in some set $G^{(t)}$ of steps backward in time. Then, an overall test for a changepoint  before time $t$ is given by 
\begin{align}
T^{(t)} = \underset{g \in G^{(t)}}{\max} \ T^{(t)}_g, \label{ttdef}
\end{align} resulting in the online changepoint detector
\begin{align}
    \widehat{\tau} = {\inf} \ \left \{   t \in \NN   \ : \ t\geq 2, \ T^{(t)} = 1  \right\} \label{stoppingtime}.
\end{align}

If we choose $G^{(t)} = [t-1]$, the test $T^{(t)}$ is equivalent to the Generalised Likelihood-Ratio test for a Gaussian change in mean. However, in the offline changepoint literature, a common choice of $G^{(t)}$ is of the form
\begin{align}
    G^{(t)}_{\mathrm{stat}} = \{1, 2, \ldots, 2^{\left\lfloor \log_2(t-1) \right \rfloor} \}, \label{naivegrid}
\end{align} 
which we call the \textit{static geometric grid}. For example,  \citet{liu_minimax_2021}, \citet{li2023robust} and \citet{moen2024minimax} use grids akin to \eqref{naivegrid} to construct minimax rate optimal testing procedures in the offline changepoint setting. The motivation is as follows. Firstly, the logarithmic cardinality of $G^{(t)}_{\mathrm{stat}}$ makes it easier to control the null rejection probability of $T^{(t)}$, since fewer tests are aggregated in the maximum.
Secondly, under the alternative, the CUSUM's signal $\EE_{\tau}\{(C_g^{(t)})^2\}$ is maximised at $g=t-\tau$, but is robust to constant-factor misspecification of $g$. For instance, if $(t-\tau)/2 \leq g\leq (t-\tau)$ and $t\leq 2\tau$, then the signal $\EE_{\tau}\{(C_g^{(t)})^2\}$ is within a constant fraction of its maximum, thus preserving most of the CUSUM test's power. Due to its geometric spacing, one may always choose such a $g$ from $G^{(t)}_{\mathrm{stat}}$. 

For the online setting, the static geometric grid $G^{(t)}_{\mathrm{stat}}$  enables fast computation of the test statistic $T^{(t)}$ in \eqref{ttdef}, or equivalently, fast evaluation of the event $\{\widehat{\tau} = t \}$. Indeed, letting $S_j = \sum_{i=1}^j Y_i$, an inspection of the CUSUM in \eqref{cusum} reveals that only the cumulative sums $S_t$ and $S_{t-g}$ are needed to compute $T^{(t)}_g$ in \eqref{cusumtest}. Thus, if $S_t$ and $\{S_{t-g} \, : \, g \in G^{(t)}_{\mathrm{stat}}\}$ are stored in memory, computing $T^{(t)}$ requires only
$\mathcal{O}(|G^{(t)}_{\mathrm{stat}}|) = \mathcal{O}(\log t)$ unit‑cost operations.
However, the static geometric grid in \eqref{naivegrid} induces a storage cost that is linear in $t$,  making it inadequate for long sequences of data. 
Indeed, to compute $T^{(t)}$, we must have $S_t$ and $\{S_{t-g} \, : \, g \in G^{(t)}_{\mathrm{stat}}\} = \{S_j \, : \, j \in t - G^{(t)}_{\mathrm{stat}}\}$ available in memory. The reversed grid $t - G^{(t)}_{\mathrm{stat}}$, which contains the indices of the necessary cumulative sums, shifts one step to the right as $t$ increments. Moreover, when $t-1$ reaches a power of two, the grid $G^{(t)}_{\mathrm{stat}}$ acquires a new largest element $t-1$, so that the reversed grid $t - G^{(t)}_{\mathrm{stat}}$ acquires a new smallest element $1$. This is illustrated in Figure \ref{fig:changelocationsstatic}, which shows the elements of $t - G^{(t)}_{\mathrm{stat}}$ indicated by numbered ticks for $t=9,10,11,12$. As a result, at any time $t$, all $S_j$ for $j\leq t$ will be required in memory for evaluating either $T^{(t)}$ or $T^{(t')}$ at some time $t'>t$ in the future. 

\begin{figure}
\centering
\begin{tikzpicture}
\draw[solid,thick] (1,-7) -- (12,-7);
\node at (0,-7) (t9) {$t=9$:};
\draw[thick] (1,-6.8) -- ++ (0,-0.2) node[below] {$1$};
\draw[thick] (2,-6.8) -- ++ (0,-0.2) node[below] {};
\draw[thick] (3,-6.8) -- ++ (0,-0.2) node[below] {};
\draw[thick] (4,-6.8) -- ++ (0,-0.2) node[below] {};
\draw[thick] (5,-6.8) -- ++ (0,-0.2) node[below] {$5$};
\draw[thick] (6,-6.8) -- ++ (0,-0.2) node[below] {};
\draw[thick] (7,-6.8) -- ++ (0,-0.2) node[below] {$7$};
\draw[thick] (8,-6.8) -- ++ (0,-0.2) node[below] {$8$};
\draw[thick] (9,-6.8) -- ++ (0,-0.2) node[below] {$t$};
\draw[thick] (10,-6.8) -- ++ (0,-0.2) node[below] {};
\draw[thick] (11,-6.8) -- ++ (0,-0.2) node[below] {};
\draw[thick] (12,-6.8) -- ++ (0,-0.2) node[below] {};

\draw[solid,thick] (1,-8) -- (12,-8);
\node at (0,-8) (t10) {$t=10$:};
\draw[thick] (1,-7.8) -- ++ (0,-0.2) node[below] {};
\draw[thick] (2,-7.8) -- ++ (0,-0.2) node[below] {$2$};
\draw[thick] (3,-7.8) -- ++ (0,-0.2) node[below] {};
\draw[thick] (4,-7.8) -- ++ (0,-0.2) node[below] {};
\draw[thick] (5,-7.8) -- ++ (0,-0.2) node[below] {};
\draw[thick] (6,-7.8) -- ++ (0,-0.2) node[below] {$6$};
\draw[thick] (7,-7.8) -- ++ (0,-0.2) node[below] {};
\draw[thick] (8,-7.8) -- ++ (0,-0.2) node[below] {$8$};
\draw[thick] (9,-7.8) -- ++ (0,-0.2) node[below] {$9$};
\draw[thick] (10,-7.8) -- ++ (0,-0.2) node[below] {$t$};
\draw[thick] (11,-7.8) -- ++ (0,-0.2) node[below] {};
\draw[thick] (12,-7.8) -- ++ (0,-0.2) node[below] {};

\draw[solid,thick] (1,-9) -- (12,-9);
\node at (0,-9) (t11) {$t=11$:};
\draw[thick] (1,-8.8) -- ++ (0,-0.2) node[below] {};
\draw[thick] (2,-8.8) -- ++ (0,-0.2) node[below] {};
\draw[thick] (3,-8.8) -- ++ (0,-0.2) node[below] {$3$};
\draw[thick] (4,-8.8) -- ++ (0,-0.2) node[below] {};
\draw[thick] (5,-8.8) -- ++ (0,-0.2) node[below] {};
\draw[thick] (6,-8.8) -- ++ (0,-0.2) node[below] {};
\draw[thick] (7,-8.8) -- ++ (0,-0.2) node[below] {$7$};
\draw[thick] (8,-8.8) -- ++ (0,-0.2) node[below] {};
\draw[thick] (9,-8.8) -- ++ (0,-0.2) node[below] {$9$};
\draw[thick] (10,-8.8) -- ++ (0,-0.2) node[below] {$10$};
\draw[thick] (11,-8.8) -- ++ (0,-0.2) node[below] {$t$};
\draw[thick] (12,-8.8) -- ++ (0,-0.2) node[below] {};

\draw[solid,thick] (1,-10) -- (12,-10);
\node at (0,-10) (t12) {$t=12$:};
\draw[thick] (1,-9.8) -- ++ (0,-0.2) node[below] {};
\draw[thick] (2,-9.8) -- ++ (0,-0.2) node[below] {};
\draw[thick] (3,-9.8) -- ++ (0,-0.2) node[below] {};
\draw[thick] (4,-9.8) -- ++ (0,-0.2) node[below] {$4$};
\draw[thick] (5,-9.8) -- ++ (0,-0.2) node[below] {};
\draw[thick] (6,-9.8) -- ++ (0,-0.2) node[below] {};
\draw[thick] (7,-9.8) -- ++ (0,-0.2) node[below] {};
\draw[thick] (8,-9.8) -- ++ (0,-0.2) node[below] {$8$};
\draw[thick] (9,-9.8) -- ++ (0,-0.2) node[below] {};
\draw[thick] (10,-9.8) -- ++ (0,-0.2) node[below] {$10$};
\draw[thick] (11,-9.8) -- ++ (0,-0.2) node[below] {$11$};
\draw[thick] (12,-9.8) -- ++ (0,-0.2) node[below] {$t$};

\end{tikzpicture}
\caption{Plot of the elements of the reversed static geometric grid $t - G^{(t)}_{\mathrm{stat}}$ for $t = 9, \ldots, 12$.}\label{fig:changelocationsstatic}
\end{figure}

The large storage cost induced by the static geometric grid motivates the construction of a new grid that recycles the cumulative sums $S_j$ while retaining geometric growth for low update costs and high statistical power. To this end, we propose a novel \textit{dynamic geometric grid}. The main idea behind this grid is to partition the set $\left\{1, 2,3, \ldots, t-1\right\}$ into successive intervals of width $1,2,4,8,\ldots$, where each interval contributes at most two elements to $G^{(t)}$. Once a new data point arrives, the elements of the intervals are either shifted cyclically one to the right, or deleted, so that the indices $t-G^{(t)}$ of the relevant cumulative sums stay in place as $t$ increments. 
Formally, for $t\geq 2$, we define the dynamic geometric grid as
\begin{align}
G^{(t)}\dyn &= \{1\} \cup \bigcup_{j=1}^{\left \lfloor \log_2\left\{(t-1)/3\right\}\right\rfloor + 1}\left\{g^{(t)}_{\text{L}, j}\right\} \cup \bigcup_{j=1}^{\left \lfloor \log_2(t-1) \right\rfloor-1}\left\{ g_{\text{R},j}^{(t)} \right\},\label{thegrid}
\end{align}
where $g^{(t)}_{\text{L},j} = 2^j + \left\{(t-1)\bmod 2^{j-1} \right\}$ is the contribution to $G^{(t)}\dyn$ from the left half of the interval $[2^j, 2^{j+1}-1]$, and $g^{(t)}_{\text{R},j} =g^{(t)}_{\text{L},j} + 2^{j-1}$ is the contribution from the right half. The evolution of $G^{(t)}\dyn$ as $t$ increases from $17$ to $20$ is illustrated in Figure \ref{fig:gridplot}, in which the intervals $[2^j, 2^{j+1}-1]$ are drawn for $j=1,2,3$ at the bottom of the figure. Here, arrows indicate elements that are shifted when $t$ increments, and elements with no outgoing arrows are discarded when $t$ increments. 

\begin{figure}
\begin{adjustbox}{width=\textwidth}
\begin{tikzpicture}
\draw[solid,thick] (1,-9) -- (15,-9);
\node at (0,-9) (t17) {$t=17$};
\draw[thick] (1,-8.8) -- ++ (0,-0.2) node[below] {$1$};
\draw[thick] (2,-8.8) -- ++ (0,-0.2) node[below] {$2$};
\draw[thick] (3,-8.8) -- ++ (0,-0.2) node[below] {$3$};
\draw[thick] (4,-8.8) -- ++ (0,-0.2) node[below] {$4$};
\draw[thick] (5,-8.8) -- ++ (0,-0.2) node[below] {};
\draw[thick] (6,-8.8) -- ++ (0,-0.2) node[below] {$6$};
\draw[thick] (7,-8.8) -- ++ (0,-0.2) node[below] {};
\draw[thick] (8,-8.8) -- ++ (0,-0.2) node[below] {$8$};
\draw[thick] (9,-8.8) -- ++ (0,-0.2) node[below] {};
\draw[thick] (10,-8.8) -- ++ (0,-0.2) node[below] {};
\draw[thick] (11,-8.8) -- ++ (0,-0.2) node[below] {};
\draw[thick] (12,-8.8) -- ++ (0,-0.2) node[below] {$12$};
\draw[thick] (13,-8.8) -- ++ (0,-0.2) node[below] {};
\draw[thick] (14,-8.8) -- ++ (0,-0.2) node[below] {};
\draw[thick] (15,-8.8) -- ++ (0,-0.2) node[below] {};

\draw[solid,thick] (1,-10) -- (15,-10);
\node at (0,-10) (t18) {$t=18$};
\draw[thick] (1,-9.8) -- ++ (0,-0.2) node[below] {$1$};
\draw[thick] (2,-9.8) -- ++ (0,-0.2) node[below] {$2$};
\draw[thick] (3,-9.8) -- ++ (0,-0.2) node[below] {$3$};
\draw[thick] (4,-9.8) -- ++ (0,-0.2) node[below] {};
\draw[thick] (5,-9.8) -- ++ (0,-0.2) node[below] {$5$};
\draw[thick] (6,-9.8) -- ++ (0,-0.2) node[below] {};
\draw[thick] (7,-9.8) -- ++ (0,-0.2) node[below] {$7$};
\draw[thick] (8,-9.8) -- ++ (0,-0.2) node[below] {};
\draw[thick] (9,-9.8) -- ++ (0,-0.2) node[below] {$9$};
\draw[thick] (10,-9.8) -- ++ (0,-0.2) node[below] {};
\draw[thick] (11,-9.8) -- ++ (0,-0.2) node[below] {};
\draw[thick] (12,-9.8) -- ++ (0,-0.2) node[below] {};
\draw[thick] (13,-9.8) -- ++ (0,-0.2) node[below] {$13$};
\draw[thick] (14,-9.8) -- ++ (0,-0.2) node[below] {};
\draw[thick] (15,-9.8) -- ++ (0,-0.2) node[below] {};

\draw[solid,thick] (1,-11) -- (15,-11);
\node at (0,-11) (t19) {$t=19$};
\draw[thick] (1,-10.8) -- ++ (0,-0.2) node[below] {$1$};
\draw[thick] (2,-10.8) -- ++ (0,-0.2) node[below] {$2$};
\draw[thick] (3,-10.8) -- ++ (0,-0.2) node[below] {$3$};
\draw[thick] (4,-10.8) -- ++ (0,-0.2) node[below] {$4$};
\draw[thick] (5,-10.8) -- ++ (0,-0.2) node[below] {};
\draw[thick] (6,-10.8) -- ++ (0,-0.2) node[below] {$6$};
\draw[thick] (7,-10.8) -- ++ (0,-0.2) node[below] {};
\draw[thick] (8,-10.8) -- ++ (0,-0.2) node[below] {};
\draw[thick] (9,-10.8) -- ++ (0,-0.2) node[below] {};
\draw[thick] (10,-10.8) -- ++ (0,-0.2) node[below] {$10$};
\draw[thick] (11,-10.8) -- ++ (0,-0.2) node[below] {};
\draw[thick] (12,-10.8) -- ++ (0,-0.2) node[below] {};
\draw[thick] (13,-10.8) -- ++ (0,-0.2) node[below] {};
\draw[thick] (14,-10.8) -- ++ (0,-0.2) node[below] {$14$};
\draw[thick] (15,-10.8) -- ++ (0,-0.2) node[below] {};

\draw[solid,thick] (1,-12) -- (15,-12);
\node at (0,-12) (t18) {$t=20$};
\draw[thick] (1,-11.8) -- ++ (0,-0.2) node[below] {$1$};
\draw[thick] (2,-11.8) -- ++ (0,-0.2) node[below] {$2$};
\draw[thick] (3,-11.8) -- ++ (0,-0.2) node[below] {$3$};
\draw[thick] (4,-11.8) -- ++ (0,-0.2) node[below] {};
\draw[thick] (5,-11.8) -- ++ (0,-0.2) node[below] {$5$};
\draw[thick] (6,-11.8) -- ++ (0,-0.2) node[below] {};
\draw[thick] (7,-11.8) -- ++ (0,-0.2) node[below] {$7$};
\draw[thick] (8,-11.8) -- ++ (0,-0.2) node[below] {};
\draw[thick] (9,-11.8) -- ++ (0,-0.2) node[below] {};
\draw[thick] (10,-11.8) -- ++ (0,-0.2) node[below] {};
\draw[thick] (11,-11.8) -- ++ (0,-0.2) node[below] {$11$};
\draw[thick] (12,-11.8) -- ++ (0,-0.2) node[below] {};
\draw[thick] (13,-11.8) -- ++ (0,-0.2) node[below] {};
\draw[thick] (14,-11.8) -- ++ (0,-0.2) node[below] {};
\draw[thick] (15,-11.8) -- ++ (0,-0.2) node[below] {$15$};

\draw[solid,thick] (2,-13) -- (3,-13);
\node at (0,-13) (t102) {};
\draw[thick] (2,-12.8) -- ++ (0,-0.2) node[below] {};
\draw[thick] (3,-12.8) -- ++ (0,-0.2) node[below] {};

\draw[solid,thick] (4,-13) -- (7,-13);
\node at (0,-13) (t103) {};
\draw[thick] (4,-12.8) -- ++ (0,-0.2) node[below] {};
\draw[thick] (7,-12.8) -- ++ (0,-0.2) node[below] {};

\draw[solid,thick] (8,-13) -- (15,-13);
\node at (0,-13) (t104) {};
\draw[thick] (8,-12.8) -- ++ (0,-0.2) node[below] {};
\draw[thick] (15,-12.8) -- ++ (0,-0.2) node[below] {};

\draw[->]        (1.2,-9.3)   -- (1.7,-9.8);
\draw[->]        (2.2,-9.3)   -- (2.7,-9.8);
\draw[->]        (4.2,-9.3)   -- (4.7,-9.8);
\draw[->]        (6.2,-9.3)   -- (6.7,-9.8);
\draw[->]        (8.2,-9.3)   -- (8.7,-9.8);
\draw[->]        (12.3,-9.3)   -- (12.8,-9.8);

\draw[->]        (1.2,-10.3)   -- (1.7,-10.8);
\draw[->]        (2.2,-10.3)   -- (2.7,-10.8);
\draw[->]        (3.2,-10.3)   -- (3.7,-10.8);
\draw[->]        (5.2,-10.3)   -- (5.7,-10.8);
\draw[->]        (9.2,-10.3)   -- (9.7,-10.8);
\draw[->]        (13.3,-10.3)   -- (13.8,-10.8);

\draw[->]        (1.2,-11.3)   -- (1.7,-11.8);
\draw[->]        (2.2,-11.3)   -- (2.7,-11.8);
\draw[->]        (4.2,-11.3)   -- (4.7,-11.8);
\draw[->]        (6.2,-11.3)   -- (6.7,-11.8);
\draw[->]        (10.3,-11.3)   -- (10.8,-11.8);
\draw[->]        (14.3,-11.3)   -- (14.8,-11.8);

\end{tikzpicture}
\end{adjustbox}
\caption{Evolution of the grid $G^{(t)}\dyn$ in \eqref{thegrid} for $t = 17, \ldots, 20$.}\label{fig:gridplot}
\end{figure}

The properties of the dynamic geometric grid are stated in the following lemma. 

\begin{lemma}\label{gridlemma}
    For all $t \geq 2$, the grid $G^{(t)}$ in \eqref{thegrid} satisfies
    \begin{enumerate}
        \item Geometric spacing: For any integer $d \leq t/2$, there exists some $g \in G^{(t)}$ such that $d/2 \leq g \leq d$; \label{gridlemmaclaim1}
        \item Logarithmic cardinality: $| G^{(t)}| < 3\log t$; \label{gridlemmalclaim2}
        \item Recycling property: $G^{(t+1)}\setminus \{1\} -1 \subseteq G^{(t)}$, or equivalently, $(t+1) - G^{(t+1)} \subseteq ( t - G^{(t)} ) \cup \{t\}$. \label{gridlemmaclaim3}
    \end{enumerate}
\end{lemma}

Just like the static grid $G^{(t)}_{\mathrm{stat}}$ in \eqref{naivegrid}, the dynamic geometric grid is geometrically spaced and logarithmic in size.
Additionally, the dynamic geometric grid has the recycling property that $\{ S_j \, : \, j \in (t+1)-G^{(t+1)}\} \subseteq \{ S_j \, : \, j \in t-G^{(t)}\} \cup\{S_t\}$. Thus, at time step $t+1$, we may re-use the cumulative sums from the previous time step $t$. The recycling property of the dynamic geometric grid is illustrated in Figure \ref{fig:changelocations}, in which $t - G^{(t)}\dyn$ is plotted, demonstrating a stark contrast to $t-G^{(t)}_{\mathrm{stat}}$  in Figure \ref{fig:changelocationsstatic}. For example, the cumulative sum $S_3$ only needs to be stored in memory until time step $t=10$, and can be discarded from then on. Combining recycling with a logarithmic cardinality, the dynamic geometric grid $G^{(t)}$ in \eqref{thegrid} achieves a storage cost of order $\mathcal{O}(\log t)$. 

The following theorem formally states the computational and statistical guarantees of the resulting online changepoint detector. 

\begin{figure}
\centering
\begin{tikzpicture}

\draw[solid,thick] (1,-12) -- (12,-12);
\node at (0,-12) (t9) {$t=9$:};
\draw[thick] (1,-11.8) -- ++ (0,-0.2) node[below] {};
\draw[thick] (2,-11.8) -- ++ (0,-0.2) node[below] {};
\draw[thick] (3,-11.8) -- ++ (0,-0.2) node[below] {$3$};
\draw[thick] (4,-11.8) -- ++ (0,-0.2) node[below] {};
\draw[thick] (5,-11.8) -- ++ (0,-0.2) node[below] {$5$};
\draw[thick] (6,-11.8) -- ++ (0,-0.2) node[below] {$6$};
\draw[thick] (7,-11.8) -- ++ (0,-0.2) node[below] {$7$};
\draw[thick] (8,-11.8) -- ++ (0,-0.2) node[below] {$8$};
\draw[thick] (9,-11.8) -- ++ (0,-0.2) node[below] {$t$};
\draw[thick] (10,-11.8) -- ++ (0,-0.2) node[below] {};
\draw[thick] (11,-11.8) -- ++ (0,-0.2) node[below] {};
\draw[thick] (12,-11.8) -- ++ (0,-0.2) node[below] {};

\draw[solid,thick] (1,-13) -- (12,-13);
\node at (0,-13) (t10) {$t=10$:};
\draw[thick] (1,-12.8) -- ++ (0,-0.2) node[below] {};
\draw[thick] (2,-12.8) -- ++ (0,-0.2) node[below] {};
\draw[thick] (3,-12.8) -- ++ (0,-0.2) node[below] {$3$};
\draw[thick] (4,-12.8) -- ++ (0,-0.2) node[below] {};
\draw[thick] (5,-12.8) -- ++ (0,-0.2) node[below] {$5$};
\draw[thick] (6,-12.8) -- ++ (0,-0.2) node[below] {};
\draw[thick] (7,-12.8) -- ++ (0,-0.2) node[below] {$7$};
\draw[thick] (8,-12.8) -- ++ (0,-0.2) node[below] {$8$};
\draw[thick] (9,-12.8) -- ++ (0,-0.2) node[below] {$9$};
\draw[thick] (10,-12.8) -- ++ (0,-0.2) node[below] {$t$};
\draw[thick] (11,-12.8) -- ++ (0,-0.2) node[below] {};
\draw[thick] (12,-12.8) -- ++ (0,-0.2) node[below] {};

\draw[solid,thick] (1,-14) -- (12,-14);
\node at (0,-14) (t11) {$t=11$:};
\draw[thick] (1,-13.8) -- ++ (0,-0.2) node[below] {};
\draw[thick] (2,-13.8) -- ++ (0,-0.2) node[below] {};
\draw[thick] (3,-13.8) -- ++ (0,-0.2) node[below] {};
\draw[thick] (4,-13.8) -- ++ (0,-0.2) node[below] {};
\draw[thick] (5,-13.8) -- ++ (0,-0.2) node[below] {$5$};
\draw[thick] (6,-13.8) -- ++ (0,-0.2) node[below] {};
\draw[thick] (7,-13.8) -- ++ (0,-0.2) node[below] {$7$};
\draw[thick] (8,-13.8) -- ++ (0,-0.2) node[below] {$8$};
\draw[thick] (9,-13.8) -- ++ (0,-0.2) node[below] {$9$};
\draw[thick] (10,-13.8) -- ++ (0,-0.2) node[below] {$10$};
\draw[thick] (11,-13.8) -- ++ (0,-0.2) node[below] {$t$};
\draw[thick] (12,-13.8) -- ++ (0,-0.2) node[below] {};

\draw[solid,thick] (1,-15) -- (12,-15);
\node at (0,-15) (t12) {$t=12$:};
\draw[thick] (1,-14.8) -- ++ (0,-0.2) node[below] {};
\draw[thick] (2,-14.8) -- ++ (0,-0.2) node[below] {};
\draw[thick] (3,-14.8) -- ++ (0,-0.2) node[below] {};
\draw[thick] (4,-14.8) -- ++ (0,-0.2) node[below] {};
\draw[thick] (5,-14.8) -- ++ (0,-0.2) node[below] {$5$};
\draw[thick] (6,-14.8) -- ++ (0,-0.2) node[below] {};
\draw[thick] (7,-14.8) -- ++ (0,-0.2) node[below] {$7$};
\draw[thick] (8,-14.8) -- ++ (0,-0.2) node[below] {};
\draw[thick] (9,-14.8) -- ++ (0,-0.2) node[below] {$9$};
\draw[thick] (10,-14.8) -- ++ (0,-0.2) node[below] {$10$};
\draw[thick] (11,-14.8) -- ++ (0,-0.2) node[below] {$11$};
\draw[thick] (12,-14.8) -- ++ (0,-0.2) node[below] {$t$};

\end{tikzpicture}
\caption{Plot of the elements of the reversed dynamic geometric grid  $t - G^{(t)}\dyn$ for $t = 9, \ldots, 12$.}\label{fig:changelocations}
\end{figure}

\begin{theorem}\label{theorem1}
    Let $\phi = |\mu_1- \mu_2|$. Fix any $\delta \in (0,1)$, and let $T^{(t)}_g$ be defined as in %
    \eqref{cusumtest} with critical value $\xi^{(t)} = \lambda \sigma^2 \log(t/\delta)$ for some $\lambda>0$. Let $T^{(t)}$ be defined as in %
    \eqref{ttdef} using the grid $G^{(t)}$ given in %
    \eqref{thegrid}, and let $\widehat{\tau}$ be defined as in %
    \eqref{stoppingtime}. It then holds that $\mathrm{UC}(\widehat{\tau}, t) = \mathcal{O}(\log t)$, and $\mathrm{SC}(\widehat{\tau},t) = \mathcal{O}(\log t)$.
    
    Moreover, there exists an absolute constant $C_1>0$, and a constant $C_2>0$ depending only on $\lambda$, such that if $\lambda \geq C_1$, then $\mathrm{FA}(\widehat{\tau}) \leq \delta$, and if $\tau < \infty$ and $\tau \phi^2 / \sigma^2 \geq 2 C_2  \log(2\tau / \delta) $, then 
$$
\PP_{\tau} \left(  \tauhat  \leq  \tau+ \left \lceil C_2  \frac{\sigma^2 \log(2\tau / \delta) }{\phi^2} \right \rceil \right) \geq 1-\delta.
$$
\end{theorem}

Theorem \ref{theorem1} guarantees that the changepoint detector $\tauhat$ in \eqref{stoppingtime} has update and storage costs that grow only logarithmically with the sample size $t$, as long as $G^{(t)}$ is chosen as in \eqref{thegrid}.  Furthermore, with a suitable time-varying critical value $\xi^{(t)}$, the false alarm probability is controlled, and the detection delay is of order $(\sigma^2 / \phi^2) \log (\tau/\delta)$. This detection delay rate matches precisely the minimax lower bound in \citet[][Proposition 4.1]{anote}, and is thus minimax rate optimal. Notably, the detector proposed by \citet{anote} achieves the same detection delay rate, although with a storage cost of order $\mathcal{O}(t)$. Thus, $\tauhat$ has the same theoretical guarantees with strictly smaller data storage requirements. 

We remark that Theorem \ref{theorem1} implies that $\tauhat$ is fully multiscale, in the sense that the guaranteed detection-delay rate is minimax rate optimal for all admissible values of $\phi$ and $\tau$. In particular, the rate remains optimal for a small change occurring after a long stretch under the null, since the geometric spacing of $G^{(t)}$ in \eqref{thegrid} preserves high power despite large spacings among its largest elements (in absolute terms). Of course, since the detector evaluates the CUSUM over the geometric grid $G^{(t)}$ in \eqref{thegrid} rather than over all $g\in[t-1]$, in practice we may incur some loss of detection power compared with a full scan. A univariate simulation study in the supplementary material (Section \ref{sec:univariatesimulation}) indicates that this loss is moderate and most visible for weak signals. 

\notocsection{General methodology}\label{sec:general}
We now extend the ideas from the univariate mean-change problem to a general methodology applicable to a wider range of models.  
Assume that we have observed $p$-dimensional data $Y_1, \ldots, Y_t$ for some sample size $t\geq 2$. For any $g = 1, \ldots, t-1$, let $T_g^{(t)} =T_g^{(t)}(Y_1, \ldots, Y_t)$ be some test, taking values in $\{0,1\}$, for testing if a change in distribution occurred $g$ time steps before $t$. Letting $G^{(t)}$ denote the dynamic grid from \eqref{thegrid}, define
\begin{align}
    \widehat{\tau} &= \inf \ \left\{ t \in \NN  \ : \ t\geq 2, \ \underset{g \in G^{(t)}}{\max} \ T^{(t)}_g=1\right\}\label{taudef}.
    \end{align}
We analyse the computational and statistical properties of $\widehat{\tau}$ separately. 
\notocsubsection{Computational properties}

We make the following assumption on the test statistic $T^{(t)}_g$.

\begin{assumption}\label{asscomp}
There exists a collection $\{S_g^{(t)} = S_g^{(t)}(Y_1, \ldots, Y_t) \, : \, g \in [t]\}_{t\in \NN}$ of sets of real numbers,  and non-negative functions $C_T(p), M_S(p), C_S(p)$, such that for each $t\geq 2$ and $g \in G^{(t)}$: 
\begin{enumerate}[label={{\Alph*: }},
  ref={\theassumption.\Alph*}]
    \item \label{asscomp-a}
        The test $T^{(t)}_g$ can be computed from $S^{(t)}_g$ using at most $C_T(p)$ unit‑cost operations;
    \item \label{asscomp-b}
        The set $S^{(t)}_g$ 
has cardinality at most $M_S(p)$ and can be computed from %
$S^{(t-1)}_{(g-1)\vee 1}$ and $Y_t$ using at most $C_S(p)$ unit‑cost operations. 
\end{enumerate}
\end{assumption}
In Assumption~\ref{asscomp}, the quantities $C_T(p)$, $M_S(p)$ and $C_S(p)$ are required to be independent of both $g$ and $t$. In particular, \ref{asscomp-a} requires that $T^{(t)}_g$ can be computed from a set $S^{(t)}_g$ of summary statistics in time that does not grow with $t$. Assumption~\ref{asscomp-b} requires that $S^{(t)}_g$ itself can be updated from a summary set from the previous time step, together with the new observation $Y_t$, again in time independent of $t$, and with cardinality that does not grow with $t$. When these conditions hold, the combination of (i) the bounded size of each $S^{(t)}_g$ and (ii) the recycling property and logarithmic size of the dynamic geometric grid $G^{(t)}$ in \eqref{thegrid} (from Lemma~\ref{gridlemma}) yields logarithmic update and storage costs in $t$. These guarantees, including their dependence on the dimension $p$, are formalised in the following proposition.

\begin{proposition}\label{theorem2}
    Let $\widehat{\tau}$ be as in \eqref{taudef}, with $G^{(t)}$ as in \eqref{thegrid}. If $T^{(t)}_g$ satisfies Assumption \ref{asscomp}, then $\mathrm{UC}(\widehat{\tau}, t) = \mathcal{O}[\{C_T(p) + C_S(p)\}\log t ]$ and $\mathrm{SC}(\widehat{\tau}, t) = \mathcal{O}\{p + M_S(p) \log t\}$. 
\end{proposition}

We remark that the $p$ term in $\mathrm{SC}(\tauhat,t)$ accounts for storing the current observation $Y_t$, while the $\log t$ term reflects the storage required for the summary sets $\{S_g^{(t)}\}_{g \in G^{(t)}}$.

In practice, Assumption~\ref{asscomp} is often satisfied when the test statistic $T^{(t)}_g$ depends on sums or averages, such as estimated model parameters. As seen in Section~\ref{sec:unimean}, for example, the CUSUM statistic $C_g^{(t)}$ can be computed in $\mathcal{O}(1)$ time from two cumulative sums, which play the role of the summary statistics in Assumption~\ref{asscomp}. The next example shows that this mechanism extends to any parameter estimated via empirical averages.

\begin{example}\label{compexample1}
    Suppose we wish to detect a change in a parameter $\theta = \EE\{h(Y_i)\}$, where $h \ : \ \RR^p \rightarrow \RR^v$ and $v = v(p)$ may depend on the dimension $p$. Given $t$ and $g$, we may estimate the pre- and post-change values of $\theta$ by  
\begin{align} &\widehat{\theta}_{1,g}^{(t)} = (t-g)^{-1}\sum_{i=1}^{t-g} h(Y_i), &  \widehat{\theta}_{2,g}^{(t)}= g^{-1}\sum_{i=t-g+1}^{t} h(Y_i),\end{align} and define the test
$$
T^{(t)}_g = \ind \left \{  \left \lVert \widehat{\theta}_{1,g}^{(t)} - \widehat{\theta}_{2,g}^{(t)}\right \rVert \geq \xi_{g}^{(t)}\right\},\label{Texample1}
$$
where $\xi_g^{(t)}$ is some critical value computable using at most $\mathcal{O}(1)$ unit-cost operations, and $\normm{\cdot}$ is some norm on $\RR^v$ that can be computed using at most $\mathcal{O}\{v(p)\}$ unit-cost operations. Define the set $S_g^{(t)}$ of summary statistics as $S_1^{(1)} = \{h(Y_1)\}$ and $S_g^{(t)} = \{\sum_{i=1}^{t-g}h(Y_i), \, \sum_{i=1}^t h(Y_i)\}$. 
Then $T^{(t)}_g$ can be computed from $S_g^{(t)}$ using at most $\mathcal{O}\{v(p)\} + C_h(p)$ unit-cost operations, where $C_h(p)$ denotes the maximum number of unit‑cost operations required to compute $h(\cdot)$. Thus, Assumption \ref{asscomp-a} holds with $C_T(p) = \mathcal{O}\{v(p)\} + C_h(p)$. 

Since $h(Y_i) \in \RR^{v}$, storing each $S_g^{(t)}$ entails storing $\mathcal{O}\{v(p)\}$ scalars. Moreover, for $t \geq 2$, we have that $\sum_{i=1}^{t-g}h(Y_i) \in S^{(t-1)}_{(g-1)\vee 1}$ and $\sum_{i=1}^t h(Y_i) = h(Y_t)+\sum_{i=1}^{t-1}h(Y_i) $, where $\sum_{i=1}^{t-1}h(Y_i) \in S_{(g-1)\vee 1}^{(t-1)}$, and $h(Y_t)$ can be computed from $Y_t$ using at most $C_h(p)$ unit-cost operations. Therefore, Assumption \ref{asscomp-b} holds with $M_S(p) = \mathcal{O}\{v(p)\}$ and $C_S(p) = \mathcal{O}\{v(p)\} + C_h(p)$. Proposition \ref{theorem2} then yields that $\mathrm{UC}(\tauhat, t) = \mathcal{O}[\{v(p) + C_h(p)\}\log t ]$ and $\mathrm{SC}(\tauhat,t) = \mathcal{O}\{ p + v(p) \log t\}$.

\end{example}
Example \ref{compexample1} covers tests based on the distance between sums (averages) of $h(Y_i)$. The next result shows that essentially the same argument applies to any test that can be written in terms of two such sums.

\begin{proposition}\label{enkelprop}
 Assume that the test $T^{(t)}_g$ can be written as
\begin{align}
    T^{(t)}_g = f^{(t)}_g \left( \sum_{i=1}^{t-g} h(Y_i), \sum_{i=t-g+1}^t h(Y_i) \right),\label{enkelpropt}
\end{align}
for all $t \geq 2$ and $g \in G^{(t)}$, where $h \, : \, \RR^p \rightarrow \RR^v$ can be computed using at most $C_h(p)$ unit‑cost operations, $v = v(p)$ may depend on the dimension $p$, and $f_g^{(t)} \, : \,  \RR^v \times \RR^v\rightarrow \{0,1\}$ can be computed using at most $C_f(p)$ unit‑cost operations. Then 
$\mathrm{UC}(\widehat{\tau}, t) = \mathcal{O}[\{v(p) + C_h(p) +  C_f(p) \}\log t]$ and $\mathrm{SC}(\widehat{\tau}, t) = \mathcal{O}\{p + v(p)\log t\}$.
\end{proposition}

The following example shows that Proposition \ref{enkelprop} applies to the likelihood-ratio test when the pre- and post-change distributions are members of an exponential family of distributions. In the supplementary material (Section \ref{sec:regression}), we also show that Proposition \ref{enkelprop} may be applied in regression settings.

\begin{example}\label{examplecomp2}
    Assume that the pre-change distribution $P_1$ and post-change distribution $P_2$ are members of an exponential family of distributions in canonical form \citep[see, e.g.,][]{LehmannCasella1998}, so that $P_1$ and $P_2$ have densities
    $$
    f(y\, ;\, \theta_j) = \eta(y)\exp \left\{ \theta_j^\top h(y) - A(\theta_j)      \right\},
    $$
    respectively for $j=1,2$, where $\theta_j \in \Theta \subseteq\RR^v$ are unknown parameters, $h \, : \, \RR^p \rightarrow \RR^v$ is a sufficient statistic, and $v = v(p)$. Then the likelihood ratio of no change versus the alternative of a change at $t - g$ is given by 
    \begin{align}
    \mathrm{LR}_g^{(t)} &= (\widehat{\theta}_{1,g}^{(t)} - \widehat{\theta}_{0}^{(t)})^\top \sum_{i=1}^{t-g} h(Y_i) + (\widehat{\theta}_{2,g}^{(t)} - \widehat{\theta}_0^{(t)})^\top \sum_{i=t-g+1}^{t} h(Y_i) \\
    &\quad - (t-g)\{  A(\widehat{\theta}_{1,g}^{(t)}) - A(\widehat{\theta}_0^{(t)})\} - g \{A(\widehat{\theta}_{2,g}^{(t)}) - A(\widehat{\theta}_0^{(t)})\},\end{align}
    where $\widehat{\theta}_{1,g}^{(t)}, \ \widehat{\theta}_{2,g}^{(t)}$  are ML estimates given by 
    \begin{align}
        \widehat{\theta}_{1,g}^{(t)} &= \underset{\theta\in\Theta}{\arg \max } \ \left\{ - (t-g) A(\theta) + \theta^\top \sum_{i=1}^{t-g} h(Y_i) \right\}, \\
        \widehat{\theta}_{2,g}^{(t)} &= \underset{\theta\in\Theta}{\arg \max } \ \left\{ -g A(\theta) + \theta^\top \sum_{i=t-g+1}^{t} h(Y_i)\right\}, 
    \end{align}
    and $\widehat{\theta}_{0}^{(t)} = \widehat{\theta}_{1,0}^{(t)}$. 
    Define the test $T^{(t)}_g = \ind \{ \mathrm{LR}^{(t)}_g > \xi^{(t)}\}$, where $\xi^{(t)}$ is a time-dependent critical value computable within at most $\mathcal{O}(1)$ unit-cost operations. Then, since $\widehat{\theta}_{1,g}^{(t)}, \ \widehat{\theta}_{2,g}^{(t)}$ and $\widehat{\theta}_{0}^{(t)}$ depend only on $\sum_{i=1}^{t-g}h(Y_i)$ and $\sum_{i=t-g+1}^t h(Y_i)$, the test $T^{(t)}_g$ is of the form \eqref{enkelpropt}. Proposition \ref{enkelprop} therefore implies that $\text{UC}(\tauhat, t) = \mathcal{O}[ \{v(p) + C_h(p) + C_M(p)\}\log t ]$ and $\text{SC}(\tauhat, t)  = \mathcal{O}\{p + v(p)\log t\}$, where $C_M(p)$ denotes the maximum number of unit-cost operations required to compute $\arg\max_{\theta\in\Theta} \theta^\top \Lambda - A(\theta)$ for any $\Lambda \in \RR^v$, and $C_h(p)$ denotes the maximum number of unit-cost operations required to compute $h(\cdot)$.
\end{example}

Interestingly, the class of distributions considered in Example \ref{examplecomp2} coincides with that assumed by \cite{computationalgeometry}. For this class, their detector MdFOCuS computes the GLR $\max_{g \in [t-1]}\mathrm{LR}_g^{(t)}$ over all $g \in [t-1]$. By contrast, the grid-based methodology proposed here only computes the maximum over the grid $G^{(t)}$, and is thus an approximation of a full GLR scan. This approximation leads to a different computational trade-off. For exponential-family models, MdFOCuS has per‑iteration cost that is poly‑logarithmic in $t$, with the exponent increasing in the dimension. In particular, they show that the expected number of candidate changepoints is of order $\log^p t$, and expected total computational cost $\mathcal{O}(t\log^{p+1} t)$.\footnote{See \citet{computationalgeometry}, Theorem 6 and Lemma 1.} Due to the exponential scaling in the dimension, they propose an approximate algorithm that projects onto at most $\tilde{p}$ coordinates when $p\geq 5$, storing  $\mathcal{O}(p \log^{\tilde p} t)$ candidate changepoint locations in total. By contrast, the proposed grid-based methodology has update and storage costs scaling as $\log t$, up to factors depending on $p$ and $v(p)$. Thus, in high-dimensional settings both approaches rely on approximations in practice.
We remark that the simulation study in Section \ref{sec:simulations} compares the proposed method with the approximate variant of MdFOCuS.

As in \citet{computationalgeometry}, a key ingredient in achieving logarithmic update and storage costs is that the test can be expressed in terms of a small number of easily updated summary statistics. This also clarifies which classes of test statistics are unlikely to satisfy Assumption~\ref{asscomp} or the conditions of Proposition~\ref{enkelprop}, and hence are unlikely to admit logarithmic update or storage costs. Immediate examples include tests based on order statistics, such as the Mann–Whitney U test \citep{rankbased} or the Kolmogorov–Smirnov test \citep{padillanonparametric}, since maintaining the full order statistics requires storing all past observations, leading to linear storage costs. Another example is the permutation test \citep{permutationtest}, where each permutation step is linear in the sample size and again requires storing all past observations. A third example comprises tests based on pairwise distances, such as certain U‑statistics \citep{ustatistics1}, which depend on all pairwise differences and therefore induce linear update and storage costs.

\notocsubsection{Statistical properties}\label{sec:statproperties}
Next, we analyse the statistical properties of the changepoint detector $\tauhat$ in \eqref{taudef}. To this end, we make the following assumption about the Type I and Type II errors of the test $T^{(t)}_g$, letting $\delta \in (0,1)$ denote some target false alarm probability.

\begin{assumption}\label{assstat} \phantom{linebreak} 
\begin{enumerate}[label={{\Alph*: }},
  ref={\theassumption.\Alph*}]
    \item \label{assstat-a}
        For any $t,g$, we have $\PP_{\infty}( T_g^{(t)} = 1) \leq \delta /(3 t^2 \log t)$.
    \item \label{assstat-b}
        For some non-negative functional $\rho \, : \, \mathcal{P}\times \mathcal{P} \rightarrow [0,\infty)$ such that $\rho(P_1, P_2)>0$ whenever $P_1\neq P_2$, and some function $r(t,\delta)>0$ non-decreasing in $t$, the following holds. 
        If $\tau<\infty$,  $\tau < t \leq 2\tau$, and $(t-\tau) \rho(P_1, P_2) \geq r(t, \delta)$,   we have that
        $$
        \PP_{\tau}(T^{(t)}_g = 1) \geq 1-\delta
        $$
        for any $g$ such that $(t - \tau)/2 \leq g \leq t-\tau$.
\end{enumerate}
\end{assumption}

Assumption \ref{assstat-a} is simply a requirement of uniform Type I error control, so that a union bound yields $ \mathrm{FA}(\tauhat)\leq \sum_{t\geq2}|G^{(t)}|\PP_{\infty}( T_g^{(t)} = 1) \leq \delta$. Thus, the specific bound $\delta / (3t^2\log t)$ in Assumption \ref{assstat-a} may be replaced by any sequence $\epsilon_{t,g}$ such that\newline $\sum_{t\geq 2} \sum_{g\in G^{(t)}} \epsilon_{t,g}\leq \delta$. Assumption \ref{assstat-b} is a local power condition, asserting that Type II error control can be obtained whenever the signal-strength condition $(t-\tau)\rho(P_1, P_2)\geq r(t,\delta)$ holds, for changes occurring in the latter half of the sample. Heuristically, $\rho(P_1, P_2)$ measures the signal contained in a single post-change observation, and thus $(t-\tau)\rho(P_1, P_2)$ represents the accumulated signal after the changepoint. Such conditions are standard in offline changepoint testing, for instance for mean-change models \citep{cusumandoptimality}, covariance-change models \citep{cusumandoptimality}, regression models \citep{cho2024detectioninferencechangeshighdimensional} and even nonparametric settings \citep{padillanonparametric}. 

Assumption~\ref{assstat-b} additionally requires a somewhat unusual robustness to misspecification of the number of post-change samples $g$, or equivalently, robustness to misspecification of the changepoint location $\tau$. Specifically, the requirement is that the test $T^{(t)}_g$ has power whenever $(t - \tau)/2 \leq g \leq t-\tau$, i.e., when $g$ is within a factor of 2 of the true lag $t-\tau$. This condition is tailored to the grid $G^{(t)}$ in \eqref{thegrid}: in general, we cannot guarantee that $t-\tau \in G^{(t)}$, even as $t$ grows, which is a direct consequence of the recycling property of $G^{(t)}$ (see Figure \ref{fig:changelocations}). However, due to its geometric spacing, the grid $G^{(t)}$ must contain an element $g$ such that $(t - \tau)/2 \leq g \leq t-\tau$ whenever $t\leq 2\tau$ (Lemma \ref{gridlemma}), and the test $T^{(t)}_g$ is therefore required to have power under this misspecification.  This robustness requirement therefore appears to be inherent to the dynamic geometric grid in \eqref{thegrid}. Interestingly, the requirement may be removed entirely if the dynamic geometric grid is replaced by the static geometric grid in \eqref{naivegrid} (as shown in the supplementary material, Section \ref{sec:staticgrid}), although at the cost of a linear storage cost. 

Before discussing Assumption \ref{assstat-b} further, we will present the main result in this section, which bounds the false alarm probability and moreover shows that the power condition in Assumption \ref{assstat-b} translates to an upper bound on the detection delay of $\widehat{\tau}$. 

\begin{proposition}\label{generalddprop}
    Let $\tauhat$ be chosen as in \eqref{taudef}, with $T_g^{(t)}$  satisfying Assumption \ref{assstat} and $G^{(t)}$ chosen as in \eqref{thegrid}.
    Then $\mathrm{FA}(\tauhat)\leq \delta$ and, for any $\tau \in \NN$, if $\tau \rho(P_1, P_2)\geq r(2\tau,\delta)$, then
    $$
    \PP_{\tau} \left\{ \widehat{\tau}  \leq \tau + \left \lceil \frac{r(2\tau, \delta)}{\rho(P_1, P_2)} \right\rceil \right\}\geq 1- \delta.
    $$
\end{proposition}

Proposition \ref{generalddprop} shows that, as long as the global signal strength condition $\tau \rho(P_1, P_2)\geq r(2\tau,\delta)$ is met, the detector raises an alarm within at most $\mathcal{O}\{r(2\tau,\delta) / \rho(P_1, P_2)\}$ post-change samples, with probability at least $1-\delta$. 

Returning to the robustness requirement in Assumption \ref{assstat-b}, note that it can equivalently be seen as requiring $T_g^{(t)}$ to retain power even when the pre-change segment $\{1,\ldots,t-g\}$ is contaminated by post-change observations. For $g$ in the range $(t-\tau)/2\le g\leq t-\tau$ (as in Assumption \ref{assstat-b}), this segment comprises $t-g$ observations, of which $t-g-\tau$ observations are post-change. Since $g\geq (t-\tau)/2$ and $t\leq2\tau$, the contaminated fraction $(t-g-\tau)/(t-g)$ is at most $1/3$, and thus at least $2/3$ of the segment still consists of genuine pre-change observations. Heuristically, a test that aggregates information across the whole segment, such as one comparing averages, is therefore well suited to this requirement, since a contaminating fraction of at most $1/3$ perturbs such a statistic by only a correspondingly bounded amount. This was noted for the CUSUM statistic in Section \ref{sec:unimean}, whose signal remains within a constant fraction of its maximum over this range of $g$. In the example below, we show that the same mechanism applies to average-type tests in general, by applying Proposition \ref{generalddprop} under an assumption of concentration.

\addtocounter{example}{-2}
\begin{example}[continued]\label{example1cont}
Continuing Example \ref{compexample1}, assume that the transformed observations $h(Y_i)$ concentrate around the true parameter under both the pre- and post-change distributions, in the sense that 
\begin{align}
P\left\{   \normm{n^{-1} \sum_{i=1}^n h(Y_i) - \theta(P)} \geq \frac{\xi(t,\delta)}{\sqrt{n}}     \right\} \leq \delta / (6t^2\log t)\label{concentrationass},
\end{align}
for all integers $t\geq 2$, all $n\in [t]$ and all $P\in \mathcal{P}$, where $\theta(P) = \EE_P\{h(Y_i)\}$ is such that $\theta(P_1)\neq \theta(P_2)$ whenever $P_1 \neq P_2$, and $\xi(t,\delta)>0$ is some non-decreasing function in $t$. The concentration inequality in \eqref{concentrationass} may for instance be a result of a sub-Gaussian or more generally a sub-Weibull class of distributions \citep{subweibull}. In the former case, if $v(p)=1$, $h(\cdot)$ is the identity and $\normm{\cdot} = |\cdot|$, we may take $\xi(t,\delta) = C \normm{Y_i}_{\Psi_2}\sqrt{\log(t/\delta)}$ for some absolute constant $C>0$. Due to \eqref{concentrationass} and a union bound, it follows immediately that the test given by 
\begin{align}T^{(t)}_g = \ind\left \{ \normm{\widehat{\theta}_{1,g}^{(t)} - \widehat{\theta}_{2,g}^{(t)}}> \frac{2 \xi(t,\delta)}{\sqrt{g\wedge( t-g)}} \right\}\label{testexample1}\end{align}
satisfies Assumption \ref{assstat-a}. Moreover, one may show that Assumption \ref{assstat-b} is also satisfied with $\rho(P_1, P_2) = \normm{\theta(P_1) - \theta(P_2)}^2$ and $r(t,\delta) = 88 \xi^2(t,\delta)$. Indeed, as anticipated before, this robustness holds because the post-change observations contaminating the pre-change segment perturb $\widehat{\theta}_{1,g}^{(t)}$ only mildly: for any $g$ as in Assumption \ref{assstat-b}, the norm $\lVert \widehat{\theta}_{1,g}^{(t)} - \widehat{\theta}_{2,g}^{(t)}\rVert$ concentrates around $\{\tau/(t-g)\}\normm{\theta(P_1) - \theta(P_2)}$, and hence retains at least $2/3$ of the signal $\normm{\theta(P_1) - \theta(P_2)}$, since $\tau/(t-g) \geq 2/3$. This yields the following proposition, whose proof formalises these arguments. 
\begin{proposition}\label{example1statprop}
    Let $\tauhat$ be chosen as in \eqref{taudef}, with $T_g^{(t)}$ as in \eqref{testexample1} and $G^{(t)}$ as in \eqref{thegrid}. If \eqref{concentrationass} holds, then $\mathrm{FA}(\tauhat) \leq \delta$, and for any $\tau \in \NN$ such that $\tau \normm{\theta(P_1) - \theta(P_2)}^2 \geq 88 \xi^2 (2\tau, \delta)$, we have
$$
\PP_{\tau} \left\{ \widehat{\tau}  \leq \tau + \left \lceil \frac{ 88 \xi^2(2\tau, \delta)}{\normm{\theta(P_1) - \theta(P_2)}^2} \right\rceil \right\}\geq 1- \delta.
    $$
\end{proposition}
\end{example}
\addtocounter{example}{2}

We remark that the example above further highlights the statistical price of testing over the grid relative to a full scan. Because the test in the example might be applied at a lag within a factor of 2 of the true lag, rather than at the true lag itself, in the worst case it loses $1/3$ of the signal $\normm{\theta(P_1) - \theta(P_2)}$ available to a full scan. More generally, the robustness requirement in Assumption \ref{assstat-b} is precisely what guards against too large a loss of power relative to a full scan. Heuristically, whether a test meets this requirement depends on how its signal is distributed across the sample. When the signal is spread across the whole segment the test is robust, as in the example above. When it is concentrated in a small part of the sample, the requirement may be more demanding, since a constant-factor misspecification of $g$ may then have a large effect on the test's power. For example, a fixed-bandwidth moving-window statistic \citep{mosum} compares the observations in a window of fixed width around the candidate changepoint location $t-g$. Once $g$ is misspecified, this window may no longer overlap the true changepoint, in which case the test loses all power. As an additional example, a test comparing empirical quantiles of the two segments may also fail to be robust, or at least lose power, if the change occurs in the tail-end of the distribution, since the few potential post-change observations contaminating the pre-change segment may suffice to shift such quantiles in the direction of their post-change values.

Establishing Assumption \ref{assstat} for a broad class of tests therefore appears challenging, and will likely have to be verified case by case. This is not only due to the robustness requirement in Assumption \ref{assstat-b} discussed earlier, but also due to the finite-sample probability requirements in both Assumptions \ref{assstat-a} and \ref{assstat-b}. For the exponential-family class in Example \ref{examplecomp2}, for instance, the finite-sample distribution of the likelihood ratio $\mathrm{LR}_g^{(t)}$ is not generally available in closed form, so even Assumption \ref{assstat-a} is difficult to meet with a closed-form critical value without additional assumptions. Similar finite-sample issues arise for other  sequential procedures, such as MdFOCuS, where the main statistical guarantees are developed within the Gaussian mean-change model. 

We emphasise that Assumption \ref{assstat} is only needed to obtain theoretical guarantees on the false alarm probability and the detection delay. In practice, Monte Carlo calibration can be used to achieve either finite-horizon false alarm control (as in the simulation study in Section~\ref{sec:simulations}) or average run length control \citep[as in][]{chen_high-dimensional_2022}. Moreover, Proposition \ref{generalddprop} indicates that tests with high power generally yield small detection delays, at least when they are robust to constant-factor misspecification of $g$. In the supplementary material (Section \ref{simulationpoisson}), we illustrate this for the likelihood-ratio test for a change in a Poisson rate. There, a simulation study shows that the detection delay of the grid-based methodology is slightly larger at weak signals, but still close to that of a full scan, mirroring the univariate change-in-mean simulation study discussed briefly in Section \ref{sec:unimean}. We also remark that the methods of \citet{computationalgeometry} and \citet{focus} can compute the Poisson likelihood-ratio over a full scan exactly with logarithmic update and storage costs.

\notocsection{Statistical theory for special-case models}\label{sec:theoryspec}
We now turn to two fundamental problems in multivariate changepoint detection, namely, a change in the mean, and a change in the covariance. For each, we will apply the grid-based methodology to a suitable offline test, and derive computational and statistical guarantees from the general results of the previous section.  

\notocsubsection{Multivariate change in mean}\label{sec:multimean}
Let us return to the problem of detecting a change in the mean, now assuming that the $Y_i$ are $p$-dimensional with independent Gaussian entries. Possible relaxations of these assumptions, including to temporal dependence or sub-Weibull noise terms, are discussed in the supplementary material (Section \ref{sec:meandiscussion}). Let $P_1 = \N_p(\mu_1, \sigma^2I)$ and $P_2 = \N_p(\mu_2, \sigma^2I)$ respectively denote the pre- and post-change distributions of the $Y_i$, with unknown respective mean vectors $\mu_1, \mu_2 \in \RR^p$ and known per-entry variance $\sigma^2>0$. Whenever $\tau < \infty$, let $k = \normm{\mu_2 - \mu_1}_0$ denote the \textit{sparsity} of the change, i.e., the number of affected entries in the mean vector, and let $\phi = \normm{\mu_2 - \mu_1}_2$ denote the magnitude of the change, both taken to be unknown.

The sparsity $k$ is known to have a substantial impact on the detectability of a changepoint \citep[see, e.g.,][]{liu_minimax_2021,enikeeva_high-dimensional_2019}. In the offline setting, \citet{liu_minimax_2021} constructed a test that is minimax rate optimal simultaneously over all values of $k$, through sophisticated thresholding and aggregation of CUSUM-like quantities over a grid of candidate sparsity levels. The generality of the methodology from Section \ref{sec:general} allows us to apply this optimal offline test online, after minor modifications to its critical values and thresholds for stronger Type I error control, as follows.  
Upon observing data $Y_1, \ldots, Y_t$ for $t\geq 2$, with a suspected changepoint occurring $g$ time steps before $t$,  define the test
    \begin{align}
        T^{(t)}_{g} &=  \ind \left \{ \underset{s \in \mathcal{S}}{\max}\ \frac{A_{s,g}}{\xi_s} > 1 \right\} \label{tdef},
    \end{align}
    which rejects the null of no changepoint whenever the statistic $A_{s,g}=A^{(t)}_{s,g}$ (defined shortly) exceeds some critical value $\xi_s= \xi_s^{(t)}$ at some sparsity level $s$ in a grid $\mathcal{S} = \mathcal{S}^{(t)} = \{ 1,2,4, \ldots, 2^{\log_2\left( \sqrt{p\log t} \ \wedge \ p \right)}\} \cup\{ p\}$ of sparsities. We note that $\mathcal{S}^{(t)}$ is slightly larger than the corresponding grid in \citet{liu_minimax_2021}. The statistic $A^{(t)}_{s,g}$ is the result of variance-rescaling and thresholding the CUSUM statistics of each coordinate of the observed data, tailored for a specific sparsity level $s$, and is given by 
\begin{align}
       A_{s,g}^{(t)} &= \sum_{j=1}^p \left\{   C_g(j)^2/\sigma^2 - \nu_{a(s,t)}   \right\} \ind\{ \left|C_g(j) \right|/\sigma > a(s,t)\}\label{adef},
        \intertext{where $C_g = C^{(t)}_g$ is the CUSUM vector given by}
      C_g^{(t)} &= \left\{  \frac{g}{t(t-g)}   \right\}^{1/2}  \sum_{i=1}^{t-g} Y_i- \   \left\{ \frac{t-g}{tg}\right\}^{1/2} \sum_{i=t-g+1}^t Y_i.
       \label{ydef}
\end{align}
In \eqref{adef}, the term $a(s,t)$ is a threshold value which depends on the candidate sparsity $s$ and the sample size $t$, given by $a^2(s,t) = 4 \log (ep s^{-2} \log t) \ind\{ s \leq \sqrt{p \log t}\}$, and is slightly larger than the threshold value in \citet{liu_minimax_2021}. The threshold $a(s,t)$ decreases with the candidate sparsity $s$, and when $s > \sqrt{p\log t}$ (corresponding to a dense change), $a(s,t)=0$ and no thresholding takes place at all. After thresholding, each entry of $C_g^{(t)}$ in \eqref{adef} is mean-centred by the conditional expectation 
$\nu_{a(s,t)} = \EE \{Z^2 \ | \ |Z| > a(s,t)\}$, where $Z\sim \mathrm{N}(0,1)$. We remark that the original offline test in \citet{liu_minimax_2021} uses a CUSUM-like quantity that is slightly different from the CUSUM in \eqref{ydef}. This could be used in place of \eqref{ydef}, but we opt for the CUSUM in \eqref{ydef} for convenience.

Let the critical value $\xi^{(t)}_s$ be given by $\xi^{(t)}_s= \lambda z(s,p,t)$, where $\lambda>0$ is a tuning parameter, and 
\begin{align}
    z(s,p,t) = \begin{cases}
        \sqrt{p\log t}, & \text{if } s > \sqrt{p\log t},\\
        s\log \left( \frac{ep \log t}{s^2} \right) \vee \log t, & \text{otherwise.}
    \end{cases} \label{rdef}
\end{align}

Let $\tauhat$ be as given in \eqref{taudef} with $G^{(t)}$ chosen as in \eqref{thegrid}, and let $T^{(t)}_g$ be chosen as in \eqref{tdef}. The resulting changepoint detector has the following theoretical performance.

\begin{theorem}\label{theorem3}
    Let $\tauhat$ be defined as above. It then holds that $\mathrm{UC}(\widehat{\tau},t) = \mathcal{O}(p \log p \log t )$ and  $\mathrm{SC}(\widehat{\tau},t) = \mathcal{O}(p \log t)$.
    Moreover, for any $\delta \in (0,1)$, there exists a constant $C_1>0$ depending only on $\delta$, and constant $C_2>0$ depending only on $\lambda$ and $\delta$, such that if $\lambda \geq C_1$, then $\mathrm{FA}(\widehat{\tau}) \leq \delta$, and if $\tau<\infty$ and $\tau \phi^2 / \sigma^2 \geq C_2  z(k,p,2\tau) $, then
\begin{align}
\PP_{\tau} \left\{  \tauhat - \tau  \leq  \left \lceil C_2  \frac{\sigma^2  }{\phi^2} z(k,p,2\tau) \right \rceil \right\} \geq 1-\delta. \label{ddmean}
\end{align}

\end{theorem}

Theorem \ref{theorem3} implies that the update and storage costs of $\tauhat$ grow logarithmically with the sample size $t$, and respectively log-linearly and linearly with the dimension $p$, with a detection delay of order
\begin{align}
    \frac{\sigma^2}{\phi^2} 
    \begin{cases}
        \sqrt{p\log (2\tau)}, & \text{if } k > \sqrt{p\log (2\tau)},\\
        k\log \left\{ \frac{ep \log (2\tau)}{k^2} \right\} \vee \log (2\tau), & \text{otherwise,}
    \end{cases} \label{detectiondelaymean}
\end{align}
whenever the signal strength condition in Theorem \ref{theorem3} is satisfied. 
Importantly, the sparsity $k$ is taken to be unknown, and thus the detector is adaptive to this quantity. When $p=1$, the detection delay is of order $\sigma^2 \log(2\tau) /\phi^2$, matching the rate of the univariate detector in Section \ref{sec:unimean} (ignoring $\delta$). When $p>1$, the detection delay depends considerably on the sparsity $k$. When $k=1$, corresponding to a ``needle in a haystack'' problem, the detection delay is of order $\sigma^2 \log( 2\tau \vee ep) /\phi^2$, which is larger than in the univariate case by a factor of only $\log(2\tau \vee  ep) / \log (2\tau)$. In the other extreme when $k = p$, the order of the detection delay is as large as $\sigma^2 \sqrt{p\log (2\tau)} /\phi^2$. Notably, both the detection delay and the signal strength condition are minimax rate optimal, for any fixed changepoint location $\tau \in \NN$ and any sparsity $k\in [p]$, up to a factor of at most $\log (2\tau)$. We establish this in the supplementary material (Section \ref{sec:optimalitymean}) by proving a non-asymptotic minimax lower bound, which, as far as we are aware, is the first such bound for online detection of a multivariate change in mean. There, we also discuss the remaining gap between the lower and upper bounds.   

Let us now compare the above detector to the Online Changepoint Detection (ocd) method of \citet{chen_high-dimensional_2022} and a sparsity-adaptive variant of MdFOCuS \citep[][Section 4.5]{computationalgeometry}. Both of these changepoint detectors are developed under the same Gaussian model as above (although ocd requires the pre-change mean $\mu_1$ to be zero). The detection delay guarantees of ocd and MdFOCuS are stated under the assumption that their average run lengths satisfy $\EE_{\infty}(\tauhat_{\text{ocd}}),\, \EE_{\infty}( \tauhat_{\text{MdFOCuS}}) \geq \gamma$ for some user-specified $\gamma>0$, where $\tauhat_{\text{ocd}}$ and $\tauhat_{\text{MdFOCuS}}$ respectively denote their outputs.
For a meaningful comparison with the changepoint detector defined above, we will in the following set $\gamma = \tau$, as ocd and MdFOCuS would otherwise be expected to raise a false alarm even before a changepoint occurs. For brevity, we will in the following refer to the changepoint detector $\tauhat$ in \eqref{taudef} (with $T^{(t)}_g$ defined in \eqref{tdef}) as CHAD (\textbf{CHA}ngepoint \textbf{D}etector), which is the name of its corresponding R package \citep{CHAD}.

The detection delay rates of the three detectors are given in Table \ref{tab:ratetable}, together with their update and storage costs.\footnote{For MdFOCuS, the reported update cost is a per-iteration arithmetic average of the expected total run time from Lemma 1 in \cite{computationalgeometry}. The reported storage cost is derived from Theorem 6, which provides an upper bound on the expected number of candidate change locations. For the two other methods, the stated update and storage costs are worst-case. The stated detection delay for ocd is a lower bound on the theoretic rate provided in Theorem 6 in \cite{chen_high-dimensional_2022}. The stated detection delay of MdFOCuS is a lower bound on the theoretic rate provided in Section 4.6 in \cite{computationalgeometry}. } In Table \ref{tab:ratetable}, the variance is assumed to be $\sigma^2=1$, and the parameter $\beta>0$ is a user-provided lower bound on the magnitude of the change $\phi$, which the ocd method requires. Moreover, $k_0$ denotes \citeauthor{chen_high-dimensional_2022}'s notion of \textit{effective sparsity}, which is slightly different from the definition of $k = \normm{\mu_1 - \mu_2}_0$. When, for instance, all non-zero coordinates of $\mu_2 - \mu_1$ have the same magnitude, it holds that $k = k_0$, in which case all three rates may be compared. 

\begin{table}[h]
\centering
\caption{Comparison of detection delay, update cost, and storage cost for $\tauhat$ in \eqref{taudef} (denoted CHAD), ocd \citep{chen_high-dimensional_2022}, and MdFOCuS \citep{computationalgeometry}.}
\label{tab:ratetable}
\resizebox{\textwidth}{!}{%
\begin{tabular}{l c c c}
\hline
Method & Detection delay & Update cost & Storage cost \\
\hline
CHAD &
$\displaystyle
\frac{1}{\phi^{2}}
\begin{cases}
\sqrt{p \log (2\tau)}, & \text{if } k > \sqrt{p \log (2\tau)},\\[4pt]
k \log \!\left\{ \dfrac{e p \log (2\tau)}{k^{2}} \right\} \vee \log (2\tau), & \text{otherwise.}
\end{cases}
$ &
$\mathcal{O}\big(p \log p \log t\big)$ &
$\mathcal{O}\big(p \log t\big)$ \\
\\[-2pt]
ocd &
$\displaystyle
\quad \quad  \quad \hspace{-0.2em} \begin{cases}
\dfrac{\sqrt{p}\,\log(e p \tau)}{\phi^{2}}, \quad\quad\quad\quad\quad \quad\quad \hspace{-0.3em}  \  & \text{if } k_{0} \ge \sqrt{p}\,\log^{-1}(e p),\\[6pt]
\dfrac{k_{0}\,\log(e p \tau)\,\log(e p)}{\beta^{2}}, & \text{otherwise.}
\end{cases}
$ &
$\mathcal{O}\big\{p^{2} \log (e p)\big\}$ &
$\mathcal{O}\big\{p^{2} \log (e p)\big\}$ \\
\\[-2pt]
MdFOCuS &
$\hspace{-6.3em}\displaystyle
\frac{1}{\phi^2} \left( \log \tau + \sqrt{k\log\tau} + \ind\{k<p\} k\log p\right)
$ &
$\mathcal{O}\big\{\log^{p+1} (t)\big\}$ &
$\mathcal{O}\big\{\log^{p} (t)\big\}$ \\
\hline
\end{tabular}
}
\end{table}

Some remarks are in order. Firstly, we emphasise that Table \ref{tab:ratetable} compares theoretical guarantees, namely upper bounds on the detection delays. These guarantees agree in their dominant behaviour: each decays like $1/\phi^2$ in the change magnitude (the exception being ocd's sparse-regime bound, which scales as $1/\beta^2$ with $\beta \leq \phi$) and grows at most logarithmically in $\tau$. Beyond this shared behaviour the guarantees differ slightly. As shown in Section~\ref{sec:ratecomp} of the supplementary material, the guaranteed delay of CHAD is no larger than that of ocd or MdFOCuS, and is strictly smaller in certain regimes, with the gap widening as the dimension $p$ grows. Where the guarantee for CHAD is stronger, this reflects the embedded offline test of \citet{liu_minimax_2021}, whose thresholding adapts to the sparsity $k$ optimally, and which the grid-based methodology allows us to use online. Secondly, in terms of computational guarantees, the ocd method is the only changepoint detector with update and storage costs constant in the sample size. In comparison, the update and storage costs of CHAD scale with the sample size $t$ by a factor of $\log t$, while those of MdFOCuS scale by a factor of $\log^p(t)$. When it comes to the dimension $p$, however, the update and storage cost grows the slowest for CHAD, whose dependence is log-linear, compared to log-quadratic for ocd and exponential for MdFOCuS. Finally, we emphasise that the theoretical detection delay comparisons above need not reflect the ordering of the true delays in practice. Indeed, in the simulation study of Section~\ref{sec:simulations}, where MdFOCuS is replaced by a computationally feasible approximate algorithm, the three methods perform remarkably alike. 

\notocsubsection{Multivariate change in covariance}\label{sec:covariance}
Let us now consider the problem of detecting a change in the covariance matrix of $p$-dimensional sub-Gaussian vectors. 
Let $P_1$ and $P_2$ denote the pre- and post-change distributions of the $Y_i$, with positive definite covariance matrices $\Sigma_1$ and $\Sigma_2$, so that $\EE Y_i Y_i^\top = \Sigma_1$ for $i\leq \tau$ and $\EE Y_i Y_i^\top = \Sigma_2$ for $i > \tau$. We impose the following assumption on the distribution of the $Y_i$.%
\begin{assumption}\label{assmultivariate2} \phantom{linebreak} 
\begin{enumerate}[label={{\Alph*: }},
  ref={\theassumption.\Alph*}]
    \item \label{assmultivariate2-a}
        The $Y_i$ are independent and mean-zero for all $i \in \NN$.
    \item \label{assmultivariate2-b}
        For some $w>0$, all $i \in \NN$ and all $v \in \mathbb{S}^{p-1}$, the random variable \\$v^\top Y_i /  \{ \EE (v^\top Y_i Y_i^\top v)\}^{1/2}$ has a continuous density bounded from above by $w$.
    \item \label{assmultivariate2-c}
        For some $u>0$, all $i \in \NN$, and all $v \in \mathbb{S}^{p-1}$, we have\\$ \lVert v^\top Y_i \rVert_{\Psi_2}^2 \leq u \EE  \{ (v^\top Y_i)^2 \}$.
\end{enumerate}
\end{assumption}
Here, Assumption \ref{assmultivariate2-b} ensures that the data are bounded away from zero with high probability along any axis of variation (needed for variance estimation), while Assumption \ref{assmultivariate2-c} ensures that the sub-Gaussian norm of the data is of the same order as the variance along any axis of variation. We remark that Assumptions \ref{assmultivariate2-b} and \ref{assmultivariate2-c} are satisfied for Gaussian data. 

To test for a change in covariance online, we use a variant of the offline test proposed by \citet{moen2024minimax}, which is a slightly modified variant of the test in \citet{covariancecusum}. We remark that the latter test could also have been used, yielding similar theoretical performance as below, but would require $\normmop{\Sigma_1} \vee \normmop{\Sigma_2}$ to be known. Upon observing $Y_1, \ldots, Y_t$ for some $t\geq 2$, with a suspected changepoint occurring $g$ time steps before $t$, define
    \begin{align}
        T^{(t)}_{g} &=  \ind \left \{ \lVert \widehat{\Sigma}_{1,g}^{(t)} - \widehat{\Sigma}_{2,g}^{(t)} \rVert_{\mathrm{op}} \  (\widehat{\sigma}^{(t)}_g)^{-2} > \xi_g^{(t)}  \right\} \label{tvardef2},
    \end{align}
which rejects the null of no changepoint whenever the operator norm of the difference between the empirical covariances
\begingroup
\mathtoolsset{showonlyrefs=false}
\begin{align}
   & \widehat{\Sigma}_{1,g}^{(t)} = (t-g)^{-1}\sum_{i=1}^{t-g} Y_i Y_i^\top, &\widehat{\Sigma}_{2,g}^{(t)} = g^{-1} \sum_{i=t-g+1}^t Y_{i} Y_i^\top,\label{sigmahatsdef} 
\end{align}
\endgroup
after being normalised by the estimated pre-change noise level $\widehat{\sigma}^{(t)}_g= \lVert \widehat{\Sigma}_{1,g}^{(t)}\rVert_{\mathrm{op}}^{1/2}$, exceeds a critical value given by 
\begin{align}\xi_g^{(t)} &= \lambda \left\{ \frac{p \vee \log t}{g\wedge(t-g)} \vee \sqrt{\frac{p \vee \log t}{g\wedge (t-g)}} \right\},\label{xicovar}
\end{align}
where $\lambda >0$ is a tuning parameter. 

For the purpose of theoretical analysis, define the signal strength parameter
\begin{align}
    \omega = \frac{\normmop{\Sigma_1 - \Sigma_2}}{\normmop{\Sigma_1}\vee \normmop{\Sigma_2}}.\label{kappadefcovar}
\end{align}
Let $\tauhat$ be as given in \eqref{taudef} with $G^{(t)}$ chosen as in \eqref{thegrid}, and let $T^{(t)}_g$ be chosen as in \eqref{tvardef2}. The resulting changepoint detector has the following theoretical performance.

\begin{theorem}\label{theorem5}
    Let $\tauhat$ be defined as above. It then holds that $\mathrm{UC}(\widehat{\tau},t) = \mathcal{O}\left( p^3 \log t\right)$ and $\mathrm{SC}(\widehat{\tau},t) = \mathcal{O}\left( p^2 \log t\right)$. Moreover, if Assumption \ref{assmultivariate2} is satisfied for some $w,u>0$, 
    then for any $\delta \in (0,1)$, there exists a constant $C_1>0$ depending only on $\delta,w,u$, and a constant $C_2>0$ depending only on $\delta,w,u$ and $\lambda$, such that if $\lambda \geq C_1$, then $\mathrm{FA}(\widehat{\tau}) \leq \delta$, and if $\tau< \infty$ and $\tau \omega ^2 \geq C_2 \{ p \vee \log (2\tau) \}$, then
$$
\PP_{\tau} \left\{  \tauhat \leq \tau + \left \lceil C_2  \frac{ p \vee \log (2\tau )}{\omega^2} \right \rceil \right\} \geq 1-\delta.
$$ 
\end{theorem}

Theorem \ref{theorem5} implies that the update and storage costs of $\tauhat$ grow logarithmically with the sample size $t$, and respectively cubically and quadratically with the dimension $p$. 
Moreover, the detection delay of $\tauhat$ is of order $( \normmop{\Sigma_1}^2\vee \normmop{\Sigma_2}^2)\normmop{\Sigma_1 - \Sigma_2}^{-2} \{p \vee \log (2\tau) \}$ whenever the signal strength condition in Theorem \ref{theorem5} is satisfied. This bound, together with the signal strength condition, is minimax rate optimal up to a factor of at most $\log (2\tau)$, for any fixed changepoint location $\tau$, provided the change is dense and the relative magnitude of the covariance change is small to moderate. As in the mean-change case, the lower bound we prove in the supplementary material (Section \ref{sec:optimalitycovariance}) is, as far as we are aware, the first non-asymptotic bound of its kind for online detection of a change in covariance. However, even when the change is sparse, the detection delay of $\tauhat$ grows as much as linearly with $p$. As such, the theoretically guaranteed detection delay of $\widehat{\tau}$ may be unacceptably large for high-dimensional problems where few entries of the data are affected. In the supplementary material (Section \ref{sec:covariancesupp}), we propose an alternative detector using sparse eigenvalues with a smaller detection delay rate for sparse changes. 

Among existing methods, the one most similar to that above is the detector introduced by \citet{li_online_2023}, which uses the Frobenius norm to measure the distance between covariance matrices, a rolling window approach to control computational cost, and even allows for temporal dependence. This detector's detection delay was asymptotically derived, containing implicit constants depending on the desired average run length, making direct comparisons to Theorem \ref{theorem5} challenging. However, there is a qualitative similarity between the two detectors, in that both their detection delays depend on a variance ratio, although for the detector of \citet{li_online_2023}, this ratio is measured in terms of the Frobenius norm. 

\notocsection{Simulation study}\label{sec:simulations}
We now empirically investigate the performance of the grid-based methodology in a simulation study. We consider the model from Section \ref{sec:multimean}, where the goal is to detect a change in the mean vector of multivariate Gaussian variables with covariance matrix $\sigma^2 I$. We compare the performance of the changepoint detector from Section \ref{sec:multimean} to
ocd \citep{chen_high-dimensional_2022}, a sparsity-adaptive variant of MdFOCuS \citep[][Section 4.5]{computationalgeometry} using projections onto $\widetilde{p}=2$ variates (which the authors recommend for $p\geq 5$), and the detectors proposed by \cite{mei2010}, \cite{xs2013} and \cite{chan2017}. The detector proposed in Section \ref{sec:multimean}, called CHAD, is implemented in the R package \texttt{CHAD} (\citeauthor{CHAD}, \citeyear{CHAD}), available on GitHub\footnote{The source code for the simulation study is found in the subdirectory \textit{inst} of the R package CHAD.}. MdFOCuS is implemented both in the publicly available code accompanying \cite{computationalgeometry}\footnote{\href{https://github.com/guillemr/Focused}{github.com/guillemr/Focused}} and the R package \texttt{focus} \citep{focuspackage}. Here, we only use the former implementation since this includes a sparsity-adaptive variant. The remaining methods are implemented in the R package \texttt{ocd} \citep{ocdpackage}, available on CRAN. 

\notocsubsection{Statistical performance}\label{sec:simulations_statistics}
We first investigate the detectors' statistical performance, i.e., their ability to quickly detect a changepoint once it has occurred. Throughout we take the pre-change mean $\mu_1$ to be zero and known, as all detectors except for CHAD and MdFOCuS are designed under this assumption. A similar simulation study involving an unknown pre-change mean is covered in the supplementary material (Section \ref{simnonzeromean}). To adapt CHAD to the assumption that $\mu_1 =0$, the detector was modified by replacing the CUSUM in \eqref{ydef} by $C_g^{(t)} = g^{-1/2} \sum_{i=t-g+1}^t Y_i$, while the code for MdFOCuS includes a dedicated routine for $\mu_1=0$.

While CHAD is designed to control the false alarm probability over an infinite sequence of data points, the implementations of most of the remaining detectors control the average run length, i.e., the expected number of observations under the null until a change is declared. To balance these two distinct measures of false alarm control, we chose a compromise approach by calibrating all detectors via Monte Carlo simulation to achieve a false alarm probability at approximately $5\%$ within a horizon of $N= 2000$ observations. Specifically, the detectors' critical values were chosen as the upper $95\%$ empirical quantiles of their respective test statistics' maximum value over a data stream of length $N=2000$, computed using $1000$ Monte Carlo samples with no changepoint. For the detectors using multiple test statistics, Bonferroni corrections were applied. Details on the critical value calibration can be found in the supplementary material (Section \ref{sec:tuningdetails}).

Fixing dimension $p = 100$ and variance $\sigma^2=1$ (assumed to be known), we set the true changepoint location to $\tau= \lceil N/3\rceil = 667$. The calibrated detectors were then applied to $1000$ independent data sets, each of length $N=2000$, with post-change mean vector $\mu_2 = \phi k^{-1/2} (1_k^\top, 0_{p-k}^\top)^\top$, where $1_k$ is a $k$-dimensional vector of ones and $0_{p-k}$ is
a $(p-k)$-dimensional vector of zeros, for $\phi \in \{0, 0.4, 0.8, \ldots, 8\}$ and $k  \in \{1,5,10,100\}$. %
All detectors except CHAD and MdFOCuS require tuning parameters, and these were all taken to be the default values provided in the package \texttt{ocd}, justified in \cite{chen_high-dimensional_2022}. In particular, the window size of the detectors using sliding windows was set to $200$.

\begin{figure}
    \centering
    \includegraphics[width=\textwidth]{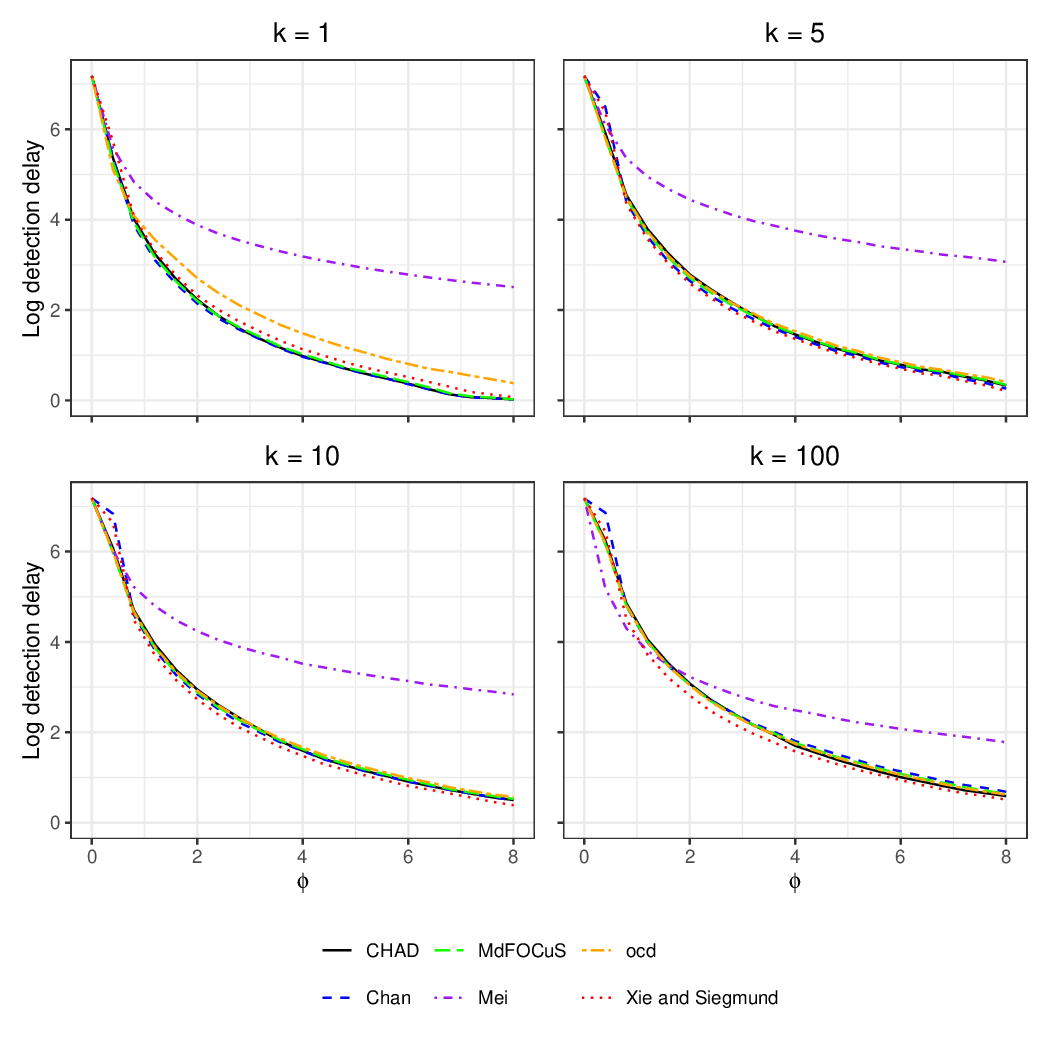}
    \caption{Average detection delay of the detectors (on log scale) for varying change magnitudes ($\phi$) and sparsities $k = 1,5,10,100$.}
    \label{fig:combined_p=100_log}
\end{figure}

Figure \ref{fig:combined_p=100_log} displays the natural logarithm of the average detection delays of the detectors, as functions of the change magnitude $\phi$, for the four different values of the sparsity $k$. Note that premature changepoint declarations ($\tauhat \leq \tau$) were excluded from the average, so that the logarithm of the estimate of $\EE \{\ (\tauhat-\tau) \wedge (N - \tau) \ | \ \tauhat > \tau\}$ is plotted for each detector. All methods display remarkably comparable performance, with the exception of \cite{mei2010} (purple curve), which has a much larger detection delay for large signal strengths than the remaining detectors. Interestingly, the detection delays of CHAD (black curve) and MdFOCuS (green curve) behave particularly similarly across all sparsity regimes and all signal strengths, in spite of relying on two distinct approximations. These two detectors also behave very similarly to ocd, except for the ultra-sparse case where $k=1$, where ocd has a slightly larger detection delay. For $\phi=0$, corresponding to no changepoint, the rate of false alarms (i.e., the frequency of $\tauhat \leq N$) was $5.7\%$ for CHAD, $4.9\%$ for MdFOCuS, $3.9\%$ for ocd, $4.2\%$ for the detector of \cite{mei2010}, $5\%$ for the detector of \cite{xs2013}, and $4.5\%$ for the detector of \cite{chan2017}. These deviations from the nominal $5\%$ are consistent with Monte Carlo error, arising both from calibrating the critical values on a finite number of null sequences and from estimating the rates on a finite number of runs, the latter alone contributing a standard error of approximately $(0.05\cdot 0.95 /1000)^{1/2} \approx 0.7\%$.

A plot of the detection delays on the original (linear) scale can be found in the supplementary material (Section \ref{sec:mainsim_nonlog}). Simulation studies were also run for the pairs $(p,N) = (10, 5000), \, (1000,200)$, which are covered in the supplementary material (Sections \ref{sec:smallplargeN}, \ref{sec:largepsmallN}). These yielded results qualitatively similar to those above, although the detectors of \cite{xs2013} and \cite{chan2017} appear to have a slight edge for large dimensions and short sequences of data. The supplementary material further covers a similar simulation study for the setting with an unknown pre-change mean (Section \ref{simnonzeromean}), in which CHAD, MdFOCuS and the detector of \cite{xs2013} appear to outperform the competitors. The supplementary material also includes a simulation study for the univariate mean-change setting, which compares the univariate detector from Section \ref{sec:unimean} to a full grid scan, for sequences of length $N = 20\, 000$. 

\notocsubsection{Computational performance}\label{sec:simulationcomp}
To evaluate the computational performance of the detectors, we measured their update times (the time to process a new observation) and their memory consumptions as functions of the sample size $t$ and the data dimension $p$. Throughout, all detectors were run sequentially and independently on identical Gaussian data sequences containing no change, and their detection thresholds were set sufficiently high to avoid spurious changepoint declarations. The experiments were performed on an Apple MacBook Pro with an M5 Pro CPU and 24 gigabytes of memory. 

To measure the dependence on $t$, we fixed $p=100$ and let $t$ range over a geometrically spaced grid from $100$ to $10^5$. To measure the dependence on $p$, we fixed $t=500$ and let $p$ range over a geometrically spaced grid from $10$ to $1000$. In both settings, we recorded the per-observation update time averaged over a block of observations to suppress measurement noise. When varying $t$, the update time was taken as the total time to process the observations since the previous value of $t$, divided by the number of observations in the block. When varying $p$, the update time was taken as the total time to process the first $500$ observations, divided by $500$.

In contrast to the update time, the memory consumption was not averaged over observations, but recorded as the detector's size at the last processed observation. For every detector except MdFOCuS, this size was measured directly as the size of the R object. Owing to the implementation of MdFOCuS, we recorded only the size of the cumulative sums MdFOCuS stored at its active candidate changepoint locations. This serves as a lower bound on its memory consumption and is thus to this detector's advantage. 

All methods were run separately on $20$ independent data sets, and the final update times and memory consumptions were taken as medians over the $20$ runs. For direct comparability, all detectors were run using implementations in native R.\footnote{We remark that MdFOCuS, although implemented in R, calls a C++ routine using the package \textbf{geometry} \citep{geometry_cran} to compute convex hulls. The remaining methods run natively in R.} We nonetheless emphasise that empirical run times depend on the implementation and, in particular, its level of optimisation. We further remark that both CHAD and MdFOCuS additionally admit Rcpp-accelerated implementations. In Section~\ref{sec:runtimesabsolute} of the supplementary material we repeat the study with these optimised variants included, showing qualitatively similar results as below.

\begin{figure}
    \centering
    \includegraphics[width=\textwidth]{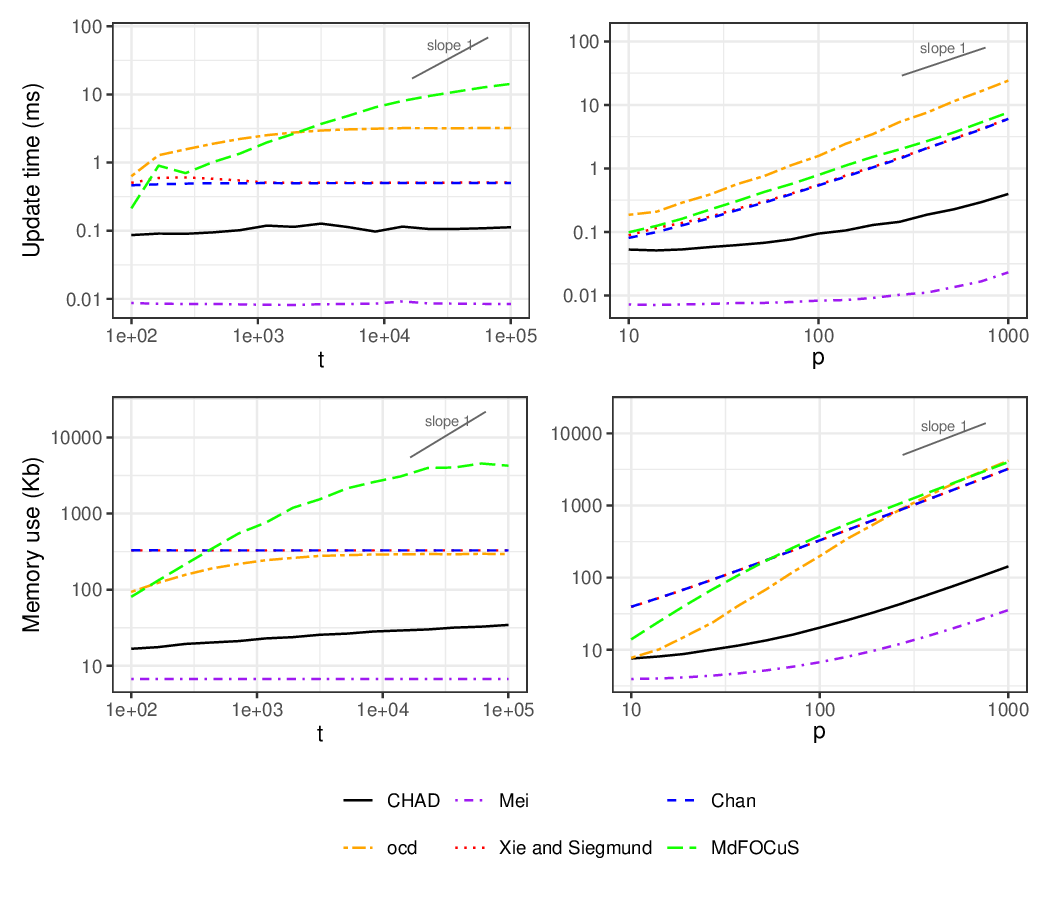}
    \caption{Update time in milliseconds (top) and memory consumption in kilobytes (bottom) of the detectors, shown as a function of $t$ (left) and $p$ (right).}
    \label{fig:combined_runtime}
\end{figure}

Figure~\ref{fig:combined_runtime} displays the detectors' average update times in milliseconds (top) and memory consumption in kilobytes (bottom) as a function of $t$ (left) and $p$ (right). As the outputs span several orders of magnitude, both axes are on a logarithmic scale, which renders the scaling exponents of the curves as slopes. In particular, a slope of $1$ (indicated in each panel) corresponds to linear growth, a slope of $0$ (a flat line) to constant-order growth, and a slope between the two to sub-linear growth. 

In theory, the per-observation update times of all detectors except CHAD and MdFOCuS are of constant order in $t$, whereas those of CHAD and MdFOCuS grow logarithmically. This is broadly reflected in Figure \ref{fig:combined_runtime}. Indeed, the update times of Mei, Xie--Siegmund, and Chan are essentially flat in $t$. Meanwhile, the update time of ocd initially increases with $t$, but converges for sufficiently large $t$. This initial increase in the update time of ocd is likely a transient warm-up effect, since its implementation maintains an internal state whose size grows with $t$ before converging. Turning to CHAD and MdFOCuS, the update time of MdFOCuS increases visibly with $t$, whereas that of CHAD appears essentially flat. We stress that CHAD's flat appearance should not be taken to mean that it is of constant order in $t$, which by construction it is not. Rather, its logarithmic growth is likely masked by constant-order overhead dominating at these scales. A separate simulation for CHAD alone, extended to larger sample sizes, is reported in the supplementary material (Section~\ref{sec:runtimeonlychad}). It illustrates that CHAD's update time indeed grows with the sample size, albeit very slowly. We further remark that the apparent dent in the update time of MdFOCuS is likely due to its pruning routine being invoked only periodically. 

As with the update times, the memory consumption is in theory of constant order in $t$ for all detectors except CHAD and MdFOCuS, for which it grows logarithmically. This is again reflected in the figure, with the exception of ocd, whose memory consumption grows initially before converging, mirroring the warm-up effect noted above. In contrast with the update time, the memory consumption of CHAD is seen to clearly grow with $t$. 

Lastly, with respect to $p$, the update times of all detectors grow approximately linearly, as do the memory consumptions. For Mei and CHAD this linear dependence, in both the update time and the memory, becomes apparent only at larger $p$, because at smaller values of $p$ a constant-order overhead may dominate and mask the linear term.

\notocsection{Real data example}\label{sec:realdata}
The online changepoint detector from Section \ref{sec:covariance} was applied to a historical data set of exchange rates to detect covariance changes,\footnote{The data and source code for the real data example is found in the subdirectory \textit{inst} of the R package CHAD \citep{CHAD}} which can indicate market instability, shifts in volatility, and structural changes in economic relationships. 
The data set consists of daily exchange rates for the ten most traded currencies\footnote{See ``Triennial Central Bank Survey Foreign exchange turnover in April 2022''. Bank for International Settlements. p. 13. 2022. \href{https://www.bis.org/statistics/rpfx22_fx.pdf}{https://www.bis.org/statistics/rpfx22\_fx.pdf}.} (excluding the US dollar) from 3 January 2000 to 16 January 2026, sourced from the US Federal Reserve\footnote{\href{https://www.federalreserve.gov/datadownload/default.htm}{https://www.federalreserve.gov/datadownload/default.htm}}, where each row reflects the value of one US Dollar expressed in terms of the selected currencies at a given day. To standardise the exchange rates to a comparable level, each series was normalised using its value on 3 January 2000 as the baseline. Given the autoregressive nature of exchange rates and the potential presence of evolving means, first-order differencing was applied to each time series, yielding a $10$-dimensional vector $Y_i$ of normalised and differenced exchange rates for $i=1, 2,\ldots, 6529$. The nominal noise level $\sigma^2 = \normmop{\mathrm{Cov}(Y_i)}$  was estimated using the first year of data, and the real-time estimator $\widehat{\sigma}^{(t)}_g$ used in \eqref{tvardef2} was replaced by this estimate. The detector from Section \ref{sec:covariance} was then calibrated to have false alarm probability of approximately $5\%$, by choosing the leading constant $\lambda$ of the time-dependent critical value in \eqref{xicovar} via Monte Carlo simulation, drawing $K=1000$ sequences of length $N = 1000$ consisting of independent Gaussian variables with covariance matrix $I$. The detector was then applied (sequentially) to the exchange rate data, resetting each time after a changepoint was detected. 

\begin{figure}
    \centering
    \includegraphics[width=\textwidth]{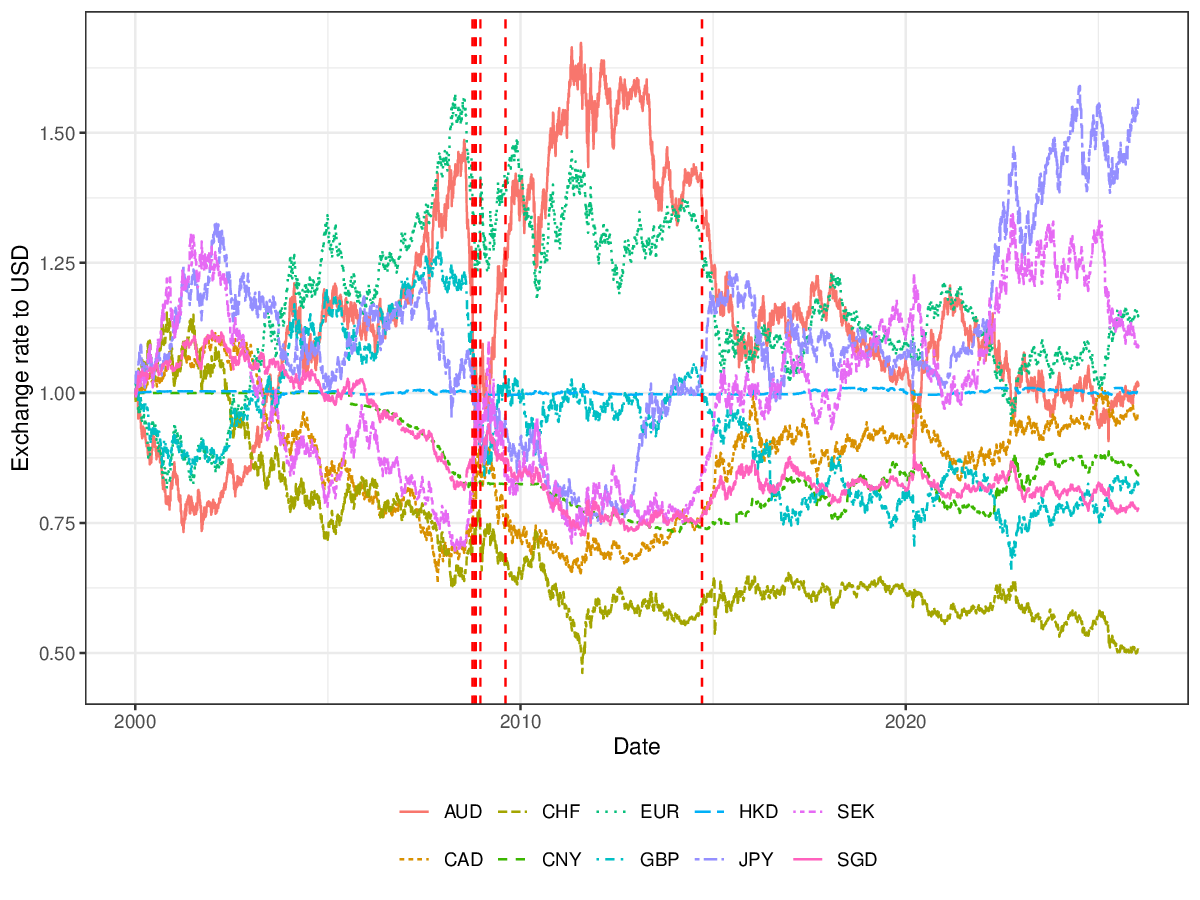}
    \caption{Normalised values of one US Dollar in terms of ten selected currencies from 3 January 2000 to 16 January 2026. Times at which changepoints were detected are indicated by dashed red vertical lines. }
    \label{fig:currencies}
\end{figure}

The detector declared six changepoints during the observed period. Plotted against the normalised (but not differenced) exchange rates,\footnote{The currencies are abbreviated as follows. AUD: Australian Dollar, CHF: Swiss Franc, EUR: Euro, HKD: Hong Kong Dollar, SEK: Swedish Krona, CAD: Canadian Dollar, CNY: Chinese Yuan, GBP: British Pound, JPY: Japanese Yen, SGD: Singapore Dollar.} the dates at which these changepoints were detected are indicated by vertical dashed lines in Figure \ref{fig:currencies}. The detector raised an alarm at the following dates: 3 October 2008, 20 October 2008, 3 November 2008, 16 December 2008, 12 August 2009, and 17 September 2014. We emphasise that the labelled dates are the times at which the null of no change was rejected, and may not necessarily be good estimates of when changes occurred. Moreover, since the changepoint detector was reset after each detection, no single changepoint was detected multiple times. The detected changepoints largely appear to coincide with the global financial crisis that began in 2008. For example, the first changepoint was declared on 3 October 2008, roughly two and a half weeks after Lehman Brothers Inc.~filed for bankruptcy.\footnote{\href{https://www.reuters.com/article/business/lehman-to-file-for-bankruptcy-plans-to-sell-units-idUSN15469897/}{https://www.reuters.com/article/business/lehman-to-file-for-bankruptcy-plans-to-sell-units-idUSN15469897/}} The three subsequent changepoint declarations occurred during the height of the financial crisis, during which major fiscal and monetary interventions took place, such as the coordinated interest rate cuts on 17 October 2008\footnote{\href{https://www.federalreserve.gov/newsevents/pressreleases/monetary20081008a.htm}{https://www.federalreserve.gov/newsevents/pressreleases/monetary20081008a.htm}} and the American government's bailout of Citigroup on 24 October 2008.\footnote{\href{https://www.reuters.com/article/world/citigroup-gets-massive-government-bailout-idUSTRE4AJ45G/}{https://www.reuters.com/article/world/citigroup-gets-massive-government-bailout-idUSTRE4AJ45G/}} Interestingly, no changes were detected during the COVID-19 pandemic nor the aftermath of the ``Liberation Day'' tariffs introduced by the United States on 2 April 2025.\footnote{\href{https://www.reuters.com/world/us/us-businesses-brace-more-pain-liberation-day-tariffs-loom-2025-04-02/}{https://www.reuters.com/world/us/us-businesses-brace-more-pain-liberation-day-tariffs-loom-2025-04-02/}}

Finally, we remark that the detector was calibrated assuming Gaussian data, which may be optimistic. It was also calibrated using heavy-tailed simulated data ($t$-distributed), covered in the supplementary material (Section \ref{heavytaildataexample}), which resulted in only two changepoint declarations.

\notocsection{Discussion}\label{sec:discussion}
In this paper, we have proposed a grid-based methodology for turning offline changepoint tests into online changepoint detectors. Under explicit conditions on the test statistic (Assumptions \ref{asscomp} and \ref{assstat}), the resulting detectors have update and storage costs growing at most logarithmically in the sample size, control of the false alarm probability, and non-asymptotic upper bounds on the detection delay depending on the power of the test. The computational savings are obtained by restricting the search for changepoints to a dynamic geometric grid of candidate locations. As discussed in Sections \ref{sec:unimean} and \ref{sec:statproperties}, this comes at a statistical price, since some detection power is likely lost relative to a full scan. Moreover, the bounds on the detection delay require the offline test to be robust to constant-factor misspecification of the changepoint location.

For some well-studied models, efficient online changepoint detectors already exist, and the grid-based methodology in this paper is not intended to replace these. For instance, FOCuS and MdFOCuS compute full-scan exact likelihood-ratio statistics online for exponential-family models \citep{focus, computationalgeometry}, and ocd has update and storage costs that are constant in the sample size for detecting high-dimensional Gaussian mean changes from a known pre-change mean \citep{chen_high-dimensional_2022}. Rather, the appeal of the methodology arguably lies in its flexibility: whenever an offline test with cheaply updatable summary statistics is available, one obtains an online detector with low update and storage costs and, under Assumption \ref{assstat}, explicit false alarm and detection delay guarantees, without designing a new sequential procedure from scratch. This is illustrated in Sections \ref{sec:multimean} and \ref{sec:covariance}, where we embedded existing offline changepoint tests into the methodology to obtain computationally efficient online detectors with near-optimal detection delay rates. More broadly, the methodology is also implemented in the open-source Python package \texttt{gridcp} \citep{gridcp}, which provides online detectors for changes in the mean, variance, covariance and regression coefficients, as well as nonparametric and exponential-family tests.

Several directions remain for further work. In the supplementary material, we use the methodology to construct online detectors of changes in regression coefficients, with logarithmic update and storage costs and finite-sample false alarm control (Section \ref{sec:regression}), and we outline computationally efficient online adaptations of offline mean-change tests for sub-Weibull noise and for dependent data (Section \ref{sec:meandiscussion}). Establishing detection delay guarantees for these procedures is an interesting avenue for future research. The methodology may also be applied to existing offline tests for other types of change, such as a change in slope or in a piecewise-linear trend \citep{NOT}, which we leave for future work. On the theoretical side, it would be interesting to characterise more precisely when the robustness requirement in Assumption \ref{assstat-b} holds. As noted in Section \ref{sec:statproperties}, tests whose signal is spread across the sample appear well suited to this requirement, and making this intuition precise, for instance by identifying a general property of the test statistic that guarantees robustness, is a natural direction for future research. Finally, extending the framework to provide post-detection inference on the changepoint location \citep[see, e.g.,][]{onlineinference} would likely be of practical interest.

\notocsection{Data availability}
The currency exchange rate data in the real data example were downloaded from the US Federal Reserve, \href{https://www.federalreserve.gov/datadownload/default.htm}{https://www.federalreserve.gov/datadownload/default.htm}. The data are also available in the R package CHAD \citep{CHAD}, as described in Section \ref{sec:overviewcode} in the supplementary material.

\notocsection{Acknowledgements}
The author gratefully acknowledges Ingrid Kristine Glad and Martin Tveten for constructive feedback and insightful discussions, and Gaetano Romano and Guillem Rigaill for their help with the implementation of MdFOCuS in the simulation studies.

\notocsection{Funding}
This work was supported by the research project SODA (Norwegian Research Council, project number
356322) and the research centre Integreat (Norwegian Research Council, project number 332645). 
\noeqref{rdef}
\noeqref{kappadefcovar}

\appendix
\renewcommand{\theequation}{S\arabic{equation}}
\renewcommand{\thefigure}{S\arabic{figure}}
\renewcommand{\thesection}{S\arabic{section}}
\renewcommand{\thesubsection}{S\arabic{section}.\arabic{subsection}}
\renewcommand{\thesubsubsection}{S\arabic{section}.\arabic{subsection}.\arabic{subsubsection}}
\renewcommand{\N}{\mathrm{N}}
\clearpage
\begin{center}
    \textbf{\LARGE Supplementary Material}
\end{center}

\noindent This is the online supplementary material for ``A grid-based methodology for fast online changepoint detection'', which is hereby referred to as the main text. References to equations in this supplementary material carry the prefix S, to differentiate them from those in the main text, which are referenced as (1), (2), etc. The numbering of Theorems and Propositions in the supplementary material is a continuation of their numbering in the main text. As such, Theorems and Propositions cited in this supplementary material may refer to either the main text or the supplementary material itself. The supplementary material has the following contents. 

\tableofcontents

\section{Online changepoint detection using the static grid}\label{sec:staticgrid}
In the main text we introduced and analysed online changepoint detectors utilizing the dynamic geometric grid $G^{(t)}$ in \eqref{thegrid}. Here, we instead focus on the static geometric grid, which we recall is defined by
\begin{align}
    G_{\mathrm{stat}}^{(t)} = \{1,2,\ldots, 2^{\lfloor \log_2 (t-1)\rfloor}\}.\label{naivegrid2}
\end{align}
Given a test statistic $T^{(t)}_g$, an online changepoint detector based on the static geometric grid is given by
\begin{align}
    \tauhat_{\mathrm{stat}} = \inf \ \left \{t \geq 2 \ : \ \underset{g \in G^{(t)}_{\mathrm{stat}}}{\max} \ T^{(t)}_g >0\right\}.\label{taustatdef}
\end{align}

In the following, we will briefly analyse the properties of $\tauhat_{\mathrm{stat}}$ in a similar fashion as for the dynamic geometric grid in Section \ref{sec:general} in the main text. 
\subsection{Computational properties}
As alluded to in Section \ref{sec:unimean}, the static geometric grid will in general lead to storage costs growing at least linearly with the sample size $t$.  However, as we will now see, logarithmic update costs may still be achieved under certain classes of test statistics. To show this, we make the following assumption about the test $T^{(t)}_g$. 

\begin{assumption}\label{asscompstatic}
There exists a collection $\{S^{(t)} = S^{(t)}(Y_1, \ldots, Y_t) \}_{t\in \NN}$ of sets of real numbers such that for each $t\in \NN$: 
\begin{enumerate}[label={{\Alph*: }},
  ref={\theassumption.\Alph*}]
    \item \label{asscompstatic-a}
        For each $g \in [t-1]$, the test $T^{(t)}_g$ can be computed from $S^{(1)},S^{(2)},\ldots,S^{(t)}$ using at most $C_T(p)$ unit‑cost operations; 
    \item \label{asscompstatic-b}
        The set $S^{(t)}$ 
has cardinality at most $M_S(p)$ and can be computed from $Y_t$ and $S^{(1)},S^{(2)},\ldots,S^{(t-1)}$ using at most $C_S(p)\log t$ unit‑cost operations. %
\end{enumerate}
\end{assumption}

We remark that Assumption \ref{asscompstatic} differs slightly from Assumption \ref{asscomp} in the main text. This discrepancy is mainly due to the recycling property of the dynamic geometric grid, which the static geometric grid lacks, and which Assumption \ref{asscomp} implicitly hinges on. Indeed, the recycling property (as defined in Lemma \ref{gridlemma}) of the dynamic geometric grid ensures that only $\mathcal{O}(\log t)$ summary sets $S_g^{(t)}$ (as defined in Assumption \ref{asscomp}) need to be stored at time $t$ under Assumption \ref{asscomp}. Since the summary sets may be reused from the previous iteration by updating them using a constant-order number of unit-cost operations, the recycling property thus enables logarithmic update times.  In contrast, under Assumption \ref{asscomp}, the static geometric grid would require storing all of the summary sets $S_1^{(t)}, S_2^{(t)}, \ldots, S_{t-1}^{(t)}$ at time $t$. In the worst case, under Assumption \ref{asscomp}, the number of unit-cost operations required to update the summary sets using the static geometric grid would then be linear in $t$. 

Instead of relying on the recycling property, Assumption \ref{asscompstatic} requires the existence of a (constant)  sequence $S^{(1)}, \ldots, S^{(t)}$ of sets from which $T^{(t)}_g$ can be easily computed. Moreover, given $S^{(1)}, \ldots, S^{(t)}$, Assumption \ref{asscompstatic} requires $S^{(t)}$ to be easily computable from the previous sets and $Y_t$, specifically using a logarithmic number of unit-cost operations (so that the update cost of the summaries does not exceed the cost of evaluating the test). We remark that all the test statistics considered in the main text (including the examples) can be shown to satisfy Assumption \ref{asscompstatic}. For the CUSUM test from Section \ref{sec:unimean}, for example, Assumption \ref{asscompstatic} is satisfied with $S^{(t)} = \{\sum_{i=1}^t Y_i\}$ for each $t \in \NN$.

The following straightforward result shows that $\tauhat_{\mathrm{stat}}$ achieves a  logarithmic update cost and linear storage cost under Assumption \ref{asscompstatic}. 

\begin{proposition}\label{compstatthm}
    Let $\widehat{\tau}_{\mathrm{stat}}$ be as in \eqref{taustatdef}, with $G_{\mathrm{stat}}^{(t)}$ as in \eqref{naivegrid2}. If $T^{(t)}_g$ satisfies Assumption \ref{asscompstatic}, it holds that $\mathrm{UC}(\widehat{\tau}_{\mathrm{stat}}, t) = \mathcal{O}[\{C_T(p) + C_S(p)\}\log t ]$ and $\mathrm{SC}(\widehat{\tau}_{\mathrm{stat}}, t) = \mathcal{O}\{p + M_S(p)  t\}$. 
\end{proposition}

\subsection{Statistical properties}
While the computational properties of changepoint detectors based on the dynamic geometric grid outperform those based on the static geometric grid, this is not the case when it comes to statistical performance. Indeed, as alluded to in Section \ref{sec:statproperties}, one may obtain bounds on the detection delay of $\widehat{\tau}_{\mathrm{stat}}$ under slightly weaker assumptions than for $\tauhat$ in \eqref{taudef} in the main text. 

We make the following relaxation of Assumption \ref{assstat} in the main text. 

\begin{assumption}\label{assstatstatic} \phantom{linebreak} 
\begin{enumerate}[label={{\Alph*: }},
  ref={\theassumption.\Alph*}]
    \item \label{assstatstatic-a}
        For any $t,g$, we have $\PP_{\infty}( T_g^{(t)} = 1) \leq \delta /\{t^2 (1+\log_2 t)\}$.
    \item \label{assstatstatic-b}
        For some non-negative functional $\rho \, : \, \mathcal{P}\times \mathcal{P} \rightarrow [0,\infty)$ such that $\rho(P_1, P_2)>0$ whenever $P_1\neq P_2$, and some function $r(t,\delta)>0$ non-decreasing in $t$, the following holds. 
        If $\tau<\infty$,  $\tau < t \leq 2\tau$, and $(t-\tau) \rho(P_1, P_2) \geq r(t, \delta)$,   we have that
        $$
        \PP_{\tau}(T^{(t)}_{t-\tau} = 1) \geq 1-\delta.
        $$
\end{enumerate}
\end{assumption}
As with Assumption \ref{assstat} in the main text, Assumption \ref{assstatstatic-a} is simply a requirement of uniform Type I error control, so that a union bound can guarantee $\mathrm{FA}(\tauhat_{\mathrm{stat}})\leq \delta$. Therefore, in Assumption \ref{assstatstatic-a}, the term  $\delta / \{t^2(1+\log_2 t)\}$ may be replaced by any sequence $\epsilon_{t,g}$ such that $\sum_{t\geq2}\sum_{g \in G^{(t)}_{\mathrm{stat}}} \epsilon_{t,g} \leq \delta$. Assumption \ref{assstatstatic-b}, however, is weaker than Assumption \ref{assstat-b}. Indeed, the latter requires the test to have high power whenever $g$ is sufficiently close to $t-\tau$, so that power is guaranteed even when the candidate changepoint location $t-g$ is slightly misspecified. In contrast, Assumption \ref{assstatstatic-b} only requires the test to have power when the candidate changepoint location $t-g =t-\tau$ is correctly specified. 

Under Assumption \ref{assstatstatic}, we have the following result. 

\begin{proposition}\label{generalddpropstatic}
    Let $\tauhat_{\mathrm{stat}}$ be chosen as in \eqref{taustatdef}, with $T_g^{(t)}$ satisfying Assumption \ref{assstatstatic} and $G^{(t)}_{\mathrm{stat}}$ given in \eqref{naivegrid2}.
    Then $\mathrm{FA}(\tauhat_{\mathrm{stat}})\leq \delta$ and, for any $\tau \in \NN$, if $\tau \rho(P_1, P_2)\geq 2 r(2\tau,\delta)$, then
    \begin{equation}
    \PP_{\tau} \left\{ \widehat{\tau}_{\mathrm{stat}}  \leq \tau + 2 \left \lceil \frac{r(2\tau, \delta)}{\rho(P_1, P_2)} \right\rceil \right\}\geq 1- \delta.\label{ddelaystatic}
    \end{equation}
\end{proposition}

We remark that the upper bound on the detection delay in Proposition \ref{generalddpropstatic} involves slightly larger constants than that in Proposition \ref{generalddprop} in the main text. These increased constants are a direct result of Assumption \ref{assstatstatic-b} being weaker than Assumption \ref{assstat-b}. Heuristically, the detection delay in \eqref{ddelaystatic} involves larger constants because one might have to observe more samples before the reversed grid $t - G^{(t)}_{\mathrm{stat}}$ overlaps the true change location $\tau$. That is, to observe samples until there exists some $g \in G^{(t)}_{\mathrm{stat}}$ such that the candidate change location is correctly specified, i.e., $t-g = \tau$. If Assumption \ref{assstatstatic-b} is replaced by Assumption \ref{assstat-b}, one can show that $\tauhat_{\mathrm{stat}}$ achieves precisely the same upper bound on the detection delay as $\tauhat$ in the main text, using the same arguments as in the proof of Proposition \ref{generalddprop}. 

\section{Some more online detectors for changes in covariance matrices}\label{sec:covariancesupp}
\subsection{Sparse changes in covariance}\label{sec:sparsecovariance}
The online covariance changepoint detector in Section \ref{sec:covariance} in the main text attains near minimax rate optimal performance for dense changepoints. Still, the signal strength requirement and detection delay in Theorem \ref{theorem5} grow linearly with $p$, which may be unacceptable in high-dimensional settings, where possibly few entries of the $Y_i$ are affected by the covariance change. To account for sparsity, we here return to the setup in Section \ref{sec:covariance} in the main text. We now consider online application of the test statistic proposed by \citet{moen2024minimax}, which is adaptive to a certain notion of sparsity, defined shortly. This test uses approximate sparse eigenvalues to measure discrepancies between pre- and post-change covariance matrices. Specifically, for any $s \in [p]$, we define the largest $s$-sparse eigenvalue of a symmetric matrix $A$ as
\begin{align}
    \lambda_{\max}^s(A) &= \underset{v \in \mathbb{S}^{p-1}_s}{\sup} \ | v^\top A v |, \label{firsteigsparse}
\end{align}
where $\mathbb{S}_s^{p-1}$ denotes the subspace of the $p$-dimensional Euclidean unit sphere $\mathbb{S}^{p-1}$ containing only vectors with at most $s$ non-zero entries. Recalling that $\Sigma_1$ and $\Sigma_2$ respectively denote the pre- and post-change covariances of the data, the $s$-sparse eigenvalue $\lambda^s_{\max}(\Sigma_1 - \Sigma_2)$ measures the largest change in variance along any space spanned by an $s$-sparse unit vector. Note that the $s$-sparse eigenvalue recovers the operator norm when $s=p$. 

Since the $s$-sparse eigenvalue is NP-hard to compute, the $s$-sparse eigenvalue can be approximated by the convex relaxation of the implicit optimisation problem (see e.g.~\citealt{berthet_optimal}), given by
\begin{align}
    \widehat{\lambda}_{\max}^s(A) &= \underset{\substack{Z \in N(p,s) }}{\sup} \  | \mathrm{Tr}(AZ) | \label{conveigdef}, 
\end{align}
where $N(p,s) = \{ Z \in \RR^{p \times p} \ ; \ Z \succcurlyeq 0, \ \Tr(Z) = 1, \ \normm{Z}_1 \leq s\}$. Being a semidefinite program and thus a convex optimisation problem, it can be solved efficiently using for instance first order methods \citep[see][]{bach2010convex}, whose computational cost scales polynomially with $p$. 

An online version of the test in \citet{moen2024minimax} for a sparse covariance change is given by 
    \begin{align}
        T^{(t)}_{g} &=  \ind \left \{  \underset{s \in \mathcal{S}}{\max} \ \frac{\widehat{\lambda}_{\max}^{s} (\widetilde{\Sigma}_{1,g}^{(t)} - \widetilde{\Sigma}_{2,g}^{(t)} )}{(\widehat{\sigma}^{(t)}_{g})^{2}\xi_{g,s}^{(t)}}  > 1  \right\} \label{tvardef3},
    \end{align}
    where 
    \begin{align}
    &\widetilde{\Sigma}_{1,g}^{(t)} = 2^{- \lfloor \log_2 g \rfloor} \sum_{i=1}^{2^{ \lfloor \log_2 g \rfloor}} Y_i Y_i^\top, &\widetilde{\Sigma}_{2,g}^{(t)} = g^{-1} \sum_{i=t-g+1}^t Y_{i} Y_i^\top, \label{sigmatilde}
\end{align}
and
    \begin{align}
        \widehat{\sigma}^{(t)}_{g} &= \widehat{\lambda}_{\max}^1(\widetilde{\Sigma}_{1,g}^{(t)})^{1/2}. \label{varestsparse}
    \end{align}
The test $T^{(t)}_g$ in \eqref{tvardef3} rejects the null hypothesis of no change whenever the approximate $s$-sparse eigenvalue of $\widetilde{\Sigma}_{1,g}^{(t)} - \widetilde{\Sigma}_{2,g}^{(t)}$, normalised by an estimated nominal noise level, exceeds a critical value $\xi_{g,s}^{(t)}$ for some $s \in \mathcal{S}$, where the grid $\mathcal{S}$ of sparsities is given by 
    \begin{align}
    \mathcal{S} &= \left\{ 2^0, 2^1, \ldots, 2^{\lfloor \log_2 p \rfloor}\right\}\cup \left\{ p\right\}.\label{mathcals}
\end{align}

We remark that the pre-change covariance $\Sigma_1$ is estimated differently in \eqref{sigmatilde} than in Section \ref{sec:covariance} in the main text. Specifically, $\widetilde{\Sigma}_{1,g}^{(t)}$  estimates $\Sigma_1$ using only the first $2^{\lfloor \log_2 g\rfloor}$ samples, as opposed to $\widehat{\Sigma}_{1,g}^{(t)}$ in \eqref{sigmahatsdef} in the main text, which uses the first $t-g$ samples. This choice is primarily for mathematical convenience. Indeed, by using only the first $2^{\lfloor \log_2 g\rfloor}\leq g$ samples to estimate $\Sigma_1$, one can  ensure that $\Sigma_1$ and $\Sigma_2$ are estimated using only pre- and post-change samples. This avoids contaminating $\widetilde{\Sigma}_{1,g}^{(t)}$ with post-change samples, and thus simplifies the theoretical analysis. Moreover, from the theoretical perspective, this choice does not affect statistical performance, since the noise level of $\widehat{\lambda}_{\max}^{s} (\widetilde{\Sigma}_{1,g}^{(t)} - \widetilde{\Sigma}_{2,g}^{(t)} )$ is dominated by the noise from the covariance estimate in \eqref{sigmatilde} with the smallest sample size. Additionally, rounding $g$ down to a power of $2$ in the pre-change estimate in \eqref{tvardef3} ensures that Assumption \ref{asscomp} holds, while preserving the signal (up to constants). However, we remark that the mathematical results below could also have been proven with the covariance estimates in \eqref{sigmatilde} replaced by those in \eqref{sigmahatsdef} in the main text.  

Let the critical value $\xi_{g,s}^{(t)}$ be given by
    \begin{align}
        \xi_{g,s}^{(t)} = \lambda s \left\{   \sqrt{\frac{\log(p \vee t)}{g}} \vee \frac{\log(p \vee t)}{g}  \right\} \label{xisparsecovar},
    \end{align}
where $\lambda>0$ is a tuning parameter. The resulting online changepoint detector is given by
    \begin{align}
    \widehat{\tau} = \inf \ \left\{ t \geq 2 \, : \,  \max_{ g \in G^{(t)}, \ g\leq t/2} \  T_g^{(t)} >0\right\}\label{taudefcovarsparse},
\end{align}
where $G^{(t)}$ is defined in \eqref{thegrid}, and $T^{(t)}_g$ is defined in \eqref{tvardef3}.  For the purpose of theoretical analysis, for any $s \in [p]$, define
\begin{align}
    \omega_s = \frac{\lambda_{\max}^s({\Sigma_1 - \Sigma_2})}{\lambda_{\max}^s({\Sigma_1})\vee \lambda_{\max}^s ({\Sigma_2})},\label{kappakdefcovar}
\end{align}
which measures the relative magnitude of the covariance change in terms of the $s$-sparse eigenvalue. The changepoint detector $\tauhat$ has the following theoretical performance.

\begin{theorem}\label{theorem6}
     Let $\tauhat$ be defined as in \eqref{taudefcovarsparse}. It then holds that $\mathrm{UC}(\widehat{\tau},t) = \mathcal{O}\left( p^{a} \log t\right)$ for some $a\geq 2$ and $\mathrm{SC}(\widehat{\tau},t) = \mathcal{O}\left( p^2 \log t\right)$. 
    Moreover, if Assumption \ref{assmultivariate2} is satisfied for some $w,u>0$, then for any $\delta \in (0,1)$, there exist a constant $C_1>0$ depending only on $\delta,w,u$, and a constant $C_2>0$ depending only on $\delta,w,u$ and $\lambda$, such that if $\lambda \geq C_1$, then $\mathrm{FA}(\widehat{\tau}) \leq \delta$, and if $\tau < \infty$ and for some $k \in [p]$, we have $\tau \omega_k^2 \geq C_2 k^2 \log(p \vee 2\tau)$, then 
$$
\PP_{\tau} \left(  \tauhat \leq \tau + \left \lceil C_2  k^2 \frac{ \log(p \vee 2\tau) }{\omega_k^2} \right \rceil \right) \geq 1-\delta.
$$
\end{theorem}

Theorem \ref{theorem6} implies that the update and storage costs of $\tauhat$ in \eqref{taudefcovarsparse} grow logarithmically with the sample size $t$, and respectively polynomially and quadratically with the dimension $p$. 
Moreover, the detection delay of $\tauhat$ is of order $\omega_k^{-2} k^2 \log(p \vee \tau)$, where $k$ is any value for which $\tau\omega_k^2 \geq C_2 k^2 \log(p \vee 2\tau)$. In Section \ref{sec:optimality} below, we show that this is minimax rate optimal up to a factor of at most $\log(2\tau)$ for small to moderately sized relative changes in covariance when $k=1$, but not necessarily so for larger values of $k$. 

\subsection{A likelihood-ratio detector for a change in mean or covariance in Gaussian data}\label{sec:lrtprocedure}
We now consider an online changepoint detector for a change in either the mean or covariance, based on a likelihood ratio test for Gaussian data. For any dimension $p \in \NN$, let $P_1 = \N_p(\mu_1, \Sigma_1)$ and $P_2 = \N_p(\mu_2, \Sigma_2)$ be $p$-dimensional Gaussian distributions with (unknown) respective mean vectors $\mu_1, \mu_2 \in \RR^p$ and (unknown) covariance matrices $\Sigma_1, \Sigma_2 \in \RR^{p \times p}$. For this model, the pre- and post-change distributions are members of a common exponential family of distributions, and thus the likelihood-ratio statistic for a change in distribution satisfies the conditions of Assumption \ref{asscomp}, yielding logarithmic-in-$t$ update and storage costs. Here, we will show that the likelihood-ratio test can be expressed in closed form. 

In the offline setting, the likelihood-ratio test for a change in mean or covariance was studied by \citet{chen_gupta}, yielding a closed-form expression for the likelihood-ratio test. Upon observing the first $t$ data points $Y_1, \ldots, Y_t$ for $t\geq 2$, the likelihood-ratio test for a change in mean or covariance occurring $g$ time steps before the last observation,  for $p \leq g\leq t-p$, is given by
    \begin{align}
        T^{(t)}_{g} &=  \ind \left \{ \mathrm{LR}_g^{(t)}  > \xi^{(t)} \right\},\label{lrmeanvar}
    \end{align}
        where $\xi^{(t)}>0$ is some critical value and the likelihood ratio $\mathrm{LR}_g^{(t)}$ is given by
    \begin{align}
        \mathrm{LR}_g^{(t)} =  \begin{cases} t \log \det  \left(    \widehat{\Sigma}^{(t)}  \right) - (t-g) \log \det \left( \widehat{\Sigma}_{g,1}^{(t)}  \right) - g \log \det \left(   \widehat{\Sigma}_{g,2}^{(t)}  \right),& \text{ if } p\leq g \leq t-p\\
        0, & \text{ otherwise.} \end{cases}\label{likratio}
    \end{align}
    where
    \begin{align}
    \widehat{\Sigma}^{(t)} &= \frac{1}{t} \left\{  V_t - \frac{1}{t} S_t S_t^{\top} \right\},\\
        \widehat{\Sigma}_{g,1}^{(t)} &= \frac{1}{t-g} \left\{ V_{t-g}  - \frac{1}{t-g} S_{t-g} S_{t-g}^{\top}   \right\},\\
        \widehat{\Sigma}_{g,2}^{(t)} &= \frac{1}{g} \left\{ V_t - V_{t-g} - \frac{1}{g} \left(S_t - S_{t-g}\right) \left(S_t - S_{t-g}\right)^\top   \right\}, 
    \end{align}
    and $S_j = \sum_{i=1}^j Y_i$, $V_j = \sum_{i=1}^j Y_i Y_i^{\top}$ for any $j \in \NN$. Here, the $S_j$ and the $V_j$ together form the sufficient statistics. We remark that the likelihood ratio is set to zero in the case where one of the covariance estimates is rank deficient. 

    Now let 
    \begin{align}
        \tauhat = \inf \ \left\{ t \geq 2 \, : \, \underset{g \in G^{(t)}}{\max} \ T_g^{(t)}>0\right\},  
    \end{align}
    where $G^{(t)}$ is given in \eqref{thegrid}. Then Proposition \ref{enkelprop} implies that $\mathrm{UC}(\tauhat,t) = \mathcal{O}(p^3 \log t)$ and $\mathrm{SC}(\tauhat,t)=\mathcal{O}(p^2\log t)$. Indeed, in the notation of Proposition \ref{enkelprop}, we may take $h(y) = (y, y y^\top)$ (i.e., the sufficient statistics), which can be computed using $C_h(p) = \mathcal{O}(p^2)$ unit-cost operations and represented as a vector with $v(p) = p^2+p$ entries. Since the determinant of a $p\times p$ matrix, as well as products and sums of pairs of $p\times p$ matrices, can be computed using $\mathcal{O}(p^3)$ unit-cost operations, the likelihood-ratio test in \eqref{lrmeanvar} can be computed from the $S_j$ and $V_j$ using at most $C_f(p)= \mathcal{O}(p^3)$ unit-cost operations. 

    As for the critical value $\xi^{(t)}$, no closed-form finite-sample distribution of $\mathrm{LR}_g^{(t)}$ is known, as far as the author is aware. However, an approximate choice can be made using the asymptotic distribution of the likelihood ratio. Due to Wilks' Theorem, whenever $t$, $g$ and $t-g$ are large, the likelihood ratio will have the approximate distribution
    \begin{align}
        \mathrm{LR}_g^{(t)} \mathrel{\dot\sim} \chi^2_{p(p+3)/2}, \label{meanvarlrapprox}
    \end{align}
    which motivates the choice $\xi^{(t)} = \chi^2_{p(p+3)/2, 1 - \delta t^{-2} |G^{(t)}|^{-1}}$, i.e. the upper $\delta t^{-2} |G^{(t)}|^{-1}$ quantile of the Chi-squared distribution with $p(p+3)/2$ degrees of freedom, to obtain $\mathrm{FA}(\tauhat) \lessapprox \delta$. Since the approximation in \eqref{meanvarlrapprox} may be poor when $g$ is close to $p$ or $t-p$, the target false alarm probability may be closer to the desired $\delta$ when constraining the $g$'s in the grid $G^{(t)}$ to be further away from $p$ and $t-p$. 

\section{Some online detectors for changes in regression coefficients}\label{sec:regression}
\subsection{A direct approach}\label{simple}
We now briefly demonstrate how the methodology from Section \ref{sec:general} may be used to detect changes in regression coefficients. Let $\epsilon_i \iid \mathrm{N}(0,\sigma^2)$ for $i \in \NN$, where we for simplicity assume that $\sigma^2$ is known. Assume that the sequence of regression model responses $Y_i$ has a change in regression coefficient at time $\tau \in \NN \cup \{\infty\}$: 
\begin{align}
    Y_i &= x_i^\top \beta_{1} + \epsilon_i, \text{ if } i \leq \tau\\
    Y_i &= x_i^\top \beta_{2} + \epsilon_i,  \text{ if } i > \tau, 
\end{align}
where the $x_i \in \RR^q$ are (fixed) covariate vectors and $\beta_1, \beta_2 \in \RR^q$ are unknown pre- and post-change regression coefficients. Upon observing the response variables $Y_1, \ldots, Y_t$ with corresponding covariate vectors $x_1, \ldots, x_t$, and suspecting a changepoint having occurred $g$ steps before the last observation, the pre- and post-change regression coefficients can be estimated by 
\begin{align}
    \widehat{\beta}_{1,g}^{(t)} &= \left( M^{(t)}_{1,g}\right)^{-1} \sum_{i=1}^{t-g} x_i Y_i, & \widehat{\beta}_{2,g}^{(t)} &= \left( M^{(t)}_{2,g}\right)^{-1} \sum_{i=t-g+1}^{t} x_i Y_i, \label{regcoef}
    \intertext{assuming that the matrices}
    M^{(t)}_{1,g} &=  \sum_{i=1}^{t-g} x_i x_i^\top, & M^{(t)}_{2,g} &=  \sum_{i=t-g+1}^{t} x_i x_i^\top \label{mdef}
\end{align}
are invertible, in particular constraining $g \geq q$ and $g \leq t-q$.  If no changepoint is present ($\tau \geq t$) and the estimators in \eqref{regcoef} are defined, then $\widehat{\beta}_{1,g}^{(t)}$ and $\widehat{\beta}_{2,g}^{(t)}$ are independent with distributions 
\begin{align}
    \widehat{\beta}_{i,g}^{(t)} &\sim \N\left ( \beta_1, \sigma^2 \left( M^{(t)}_{i,g}\right)^{-1}\right),
\end{align}
for $i=1,2$.
To test for a change in regression coefficients occurring $g$ time steps before the last observation, a direct and simple approach is then to take
\begin{align}
    T_{g}^{(t)} &= \ind  \left \{ D_g^{(t)} > \xi^{(t)}\right\},\label{regressiontest}\end{align}
    where \begin{align}
    D_{g}^{(t)} &=  \frac{\left( \widehat{\beta}_{1,g}^{(t)} - \widehat{\beta}_{2,g}^{(t)}\right)^\top \left\{ (M^{(t)}_{1,g})^{-1} + (M^{(t)}_{2,g})^{-1}\right\}^{-1} \left( \widehat{\beta}_{1,g}^{(t)} - \widehat{\beta}_{2,g}^{(t)}\right)}{\sigma^2},\label{ddefregression}
\end{align}
whenever $M_{i,g}^{(t)}$ is invertible for $i=1,2$ and zero otherwise. 
Here,  $T_g^{(t)}$ rejects the null of no changepoint whenever a Wald-type quadratic form in the difference $\widehat{\beta}_{1,g}^{(t)} - \widehat{\beta}_{2,g}^{(t)}$ (standardised by its exact null covariance) exceeds the critical value $\xi^{(t)}$. Under the null of no changepoint, $D_g^{(t)}$ follows a Chi-squared distribution with $q$ degrees of freedom, as long as the $M_{i,g}^{(t)}$ are invertible. 

Now let $G^{(t)}$ be as in \eqref{thegrid}, $\delta \in (0,1)$, and set 
$$\xi^{(t)} = \chi^2_{q, 1-\delta t^{-2} |G^{(t)}|^{-1}},$$
i.e. the upper $\delta t^{-2} |G^{(t)}|^{-1}$ quantile of the Chi-squared distribution with $q$ degrees of freedom. With this choice of grid and critical value, let $\tauhat$ be defined as in \eqref{taudef} using the test in \eqref{regressiontest}. Then we have the following control over false alarms and update and storage costs.  
\begin{proposition}\label{propreg}
    Let $\tauhat$ be defined as above. Then it holds that
    \begin{enumerate}
\item $\mathrm{FA}(\widehat{\tau}) \leq \delta$, and
\item $\mathrm{UC}(\widehat{\tau}, t) = \mathcal{O}(q^3\log t)$, and $\mathrm{SC}(\widehat{\tau},t) = \mathcal{O}(q^2 \log t)$.
\end{enumerate}
\end{proposition}
We remark that the computational properties in Proposition \ref{propreg} are proved by appealing to Proposition \ref{enkelprop}. 
For a closed-form critical value $\xi^{(t)}$, since $|G^{(t)}| = \mathcal{O}(\log t)$, one could alternatively take
\begin{align}
    \xi^{(t)} =  q + \lambda (\sqrt{q\log t} \vee \log t),
\end{align}
which for a sufficiently large $\lambda>0$, depending only on $\delta$, also ensures $\mathrm{FA}(\tauhat)\leq \delta$ due to Lemma \ref{birgelemma}. We leave it for future research to upper bound the detection delay of the above changepoint detector. 

\subsection{Discussion: High-dimensional data, non-Gaussianity and temporal dependence}
Although simple and intuitive, the detector in Section \ref{simple} for changes in regression coefficients is restrictive. By estimating the pre- and post-change regression coefficients $\widehat{\beta}_{i,g}^{(t)}$ directly for $i=1,2$ using least squares, these estimates are only well defined when the matrices $M_{i,g}^{(t)}$ are invertible for $i=1,2$. In particular, the estimation of the pre- and post-change regression coefficients requires a minimum sample size of $q$ (the dimension of the covariate vector) for both the pre- and post-change sample. In high-dimensional settings, where $q$ may be exceedingly large, this mechanically results in a large detection delay. One option is to replace the least squares estimates of the pre- and post-change regression coefficients by $\ell_1$-penalised LASSO estimates, as is commonly done in the offline changepoint literature, for instance in \citet{leonardi2016computationallyefficientchangepoint} and \citet{xu2024changepointinferenceinhighdimensional}. However, it is unclear whether these LASSO estimates can be computed with computational cost independent of the sample size, and thus unclear whether Assumption \ref{asscomp} or the conditions of Proposition \ref{enkelprop} are satisfied. Indeed, the computational cost when computing LASSO estimates using e.g.~the LARS algorithm \citep{lars} depends linearly on the number of samples, resulting in an update cost that is prohibitively large in an online setting. 

Instead, a solution is to detect regression coefficient changes by monitoring the covariance $\mathrm{Cov}(Y_i, x_i)$ between the response $Y_i$ and the covariate $x_i$. This is the approach taken in \citet{cho2024detectioninferencechangeshighdimensional}, who propose an offline method, McScan, for multiple changepoint detection and localisation. Reformulated within our notation, the test statistic they use for detecting a (possibly sparse) change in regression coefficient is of the form
\begin{align}
    T_g^{(t)}= \left \{ \sqrt{\frac{g(t-g)}{t}} \left\lVert (t-g)^{-1} \sum_{i=1}^{t-g} x_iY_i - g^{-1}\sum_{i=t-g+1}^t x_i Y_i \right\rVert_{\infty} > \xi^{(t)} \right\}.\label{regcovar}
\end{align}

Used in combination with the grid $G^{(t)}$ in \eqref{thegrid}, Proposition \ref{enkelprop} implies that the online changepoint detector $\widehat{\tau}$ based on $T_g^{(t)}$ in \eqref{regcovar} has update cost $\mathrm{UC}(\widehat{\tau}, t) = \mathcal{O}(q \log t)$ and storage cost $\mathrm{SC}(\widehat{\tau}, t) = \mathcal{O}(q \log t)$, i.e. both are linear in the covariate dimension $q$. To establish detection delay guarantees for $\widehat{\tau}$, one could appeal to Proposition \ref{generalddprop}, provided an appropriate time-dependent critical value $\xi^{(t)}$ is chosen to satisfy Assumption \ref{assstat-a}, together with a local power bound as in Assumption \ref{assstat-b}. For a range of models allowing for non-Gaussian and temporally dependent errors, \citet{cho2024detectioninferencechangeshighdimensional} already derive explicit thresholds $\xi^{(t)}$ that control the Type I error in finite samples, as well as implicit power guarantees. These results suggest that, under their assumptions, it should be possible to modify their critical values and power guarantees to match the conditions of Proposition \ref{generalddprop}, and thus obtain non-asymptotic detection delay bounds for the corresponding online procedure. Given the technical nature of their analysis and the variety of dependence structures considered, we do not develop this extension here, and leave a full transfer of their guarantees to the online setting for future work.

\section{Comparison of rates for the multivariate mean-change problem}\label{sec:ratecomp}
In the following, we compare the detection delay rates reported in Table \ref{tab:ratetable}. While we do not attempt a full characterisation, we show that the detection delay rate of the changepoint detector in Section \ref{sec:multimean} is never greater than that of either ocd \citep{chen_high-dimensional_2022} or MdFOCuS \citep{computationalgeometry}, and is strictly smaller in certain regimes.

We begin by comparing the rate of the online changepoint detector from Section \ref{sec:multimean} with that of ocd. We recall that the guaranteed detection delay of the method from Section \ref{sec:multimean} (denoted CHAD) is given by 
\begin{align}
     \mathrm{DD}(\mathrm{CHAD}) &= \frac{1}{\phi^2} 
    \begin{cases}
        \sqrt{p\log (2\tau)}, & \text{if } k > \sqrt{p\log (2\tau)},\\
        k\log \left\{ \frac{ep \log (2\tau)}{k^2} \right\} \vee \log (2\tau), & \text{otherwise,}
    \end{cases} \label{detectiondelaymeansupp}
\end{align}
ignoring constants and assuming that $\sigma=1$, where $k=\normm{\mu_1-\mu_2}_0$ denotes the sparsity and $\phi=\normm{\mu_1-\mu_2}$ denotes the magnitude of the mean change under the alternative. 
\subsection{Comparison with ocd}
Recall from Section \ref{sec:multimean} that the detection delay of the ocd method is bounded from above by at least
\begin{align}
    \mathrm{DD}(\mathrm{ocd}) = \begin{cases}
         \frac{\sqrt{p}\log(ep\tau)}{\phi^2}, &\text{if } k \geq \sqrt{p} \log^{-1}(ep)\\
          \frac{k \log(ep\tau)\log(ep)}{\beta^2}, &\text{otherwise,}
    \end{cases}\label{ocdddsupp}
\end{align}
whenever the average run length is at least $\tau$ and all non-zero entries of $\mu_1 - \mu_2$ have the same magnitude, and $\beta\leq \phi$ is some known lower bound on $\phi$. 

We now compare the rate in \eqref{ocdddsupp} with that in \eqref{detectiondelaymeansupp}. We consider three different cases. 

\paragraph{Case 1: } $k\geq \sqrt{p}\log^{-1}(ep)$ and $k >\sqrt{p\log (2\tau)}$. \newline\newline
In this case, 
\begin{align}
    \frac{\mathrm{DD}(\mathrm{ocd})}{\mathrm{DD}(\mathrm{CHAD})} &= \frac{\log(ep\tau)}{\sqrt{\log (2\tau)}}\geq 1,
\end{align}
which diverges as $p\rightarrow \infty$ or $\tau\rightarrow \infty$. 

\paragraph{Case 2: } $k\geq \sqrt{p}\log^{-1}(ep)$ and $k \leq\sqrt{p\log (2\tau)}$. \newline\newline
In this case, note that
\begin{align}
    \mathrm{DD}(\mathrm{CHAD}) &= \phi^{-2}k\log \left\{ \frac{ep \log (2\tau)}{k^2} \right\} \vee \log (2\tau) \\
    &\leq \phi^{-2}\frac{2}{\sqrt{e}}\sqrt{p\log (2\tau)} \vee \log(2\tau),
\end{align}
which follows by maximising $k\log \left\{ ep \log (2\tau){k^{-2}} \right\} \vee \log (2\tau)$ with respect to $k$. Ignoring constant factors, it follows that
\begin{align}
    \frac{\mathrm{DD}(\mathrm{ocd})}{\mathrm{DD}(\mathrm{CHAD})} &\geq  \frac{\log(ep\tau)}{\sqrt{\log( 2\tau)}} \wedge \frac{\sqrt{p}\log(ep\tau)}{\log(2\tau)}\geq 1,
\end{align}
which diverges as $p\rightarrow \infty$.

\paragraph{Case 3: } $k<\sqrt{p}\log^{-1}(ep)$. \newline\newline
In this case, we have that
\begin{align}
    \mathrm{DD}(\mathrm{ocd})&\geq \frac{k\log(ep\tau)\log(ep)}{\beta^2}\\
&\geq \frac{k\log(ep\tau)\log(ep)}{\phi^2}. 
\end{align}
Moreover, since we must have that $k\leq \sqrt{p\log (2\tau)}$, it holds that
\begin{align}
    \frac{\mathrm{DD}(\mathrm{ocd})}{\mathrm{DD}(\mathrm{CHAD})} &\geq  \frac{\log(ep\tau)\log(ep)}{\log(ep) + \log\log (2\tau) - 2\log k} \wedge \frac{k\log(ep\tau)\log(ep)}{\log(2\tau)}\geq 1,
\end{align}
for $\tau\geq 2$ (ignoring constant factors) which diverges as $p\rightarrow \infty$. 

\subsection{Comparison with MdFOCuS}
Recall from Section \ref{sec:multimean} that the detection delay of MdFOCuS is bounded from above by at least
\begin{align}
    \mathrm{DD}(\mathrm{MdFOCuS}) = \frac{1}{\phi^2} \left( \log \tau + \sqrt{k\log\tau} + \ind\{k<p\} k\log p\right)\label{mdfocusddsupp}
\end{align}
whenever the average run length is at least $\tau$. 

We now compare the rate in \eqref{mdfocusddsupp} with that in \eqref{detectiondelaymeansupp}. We consider four different cases. 

\paragraph{Case 1:} $k = p$ and $k\geq \sqrt{p\log (2\tau)}$. \newline\newline
In this case, 
\begin{align}
    \mathrm{DD}(\mathrm{MdFOCuS}) \geq \phi^{-2}\sqrt{p\log \tau},
\end{align}
and
\begin{align}
    \mathrm{DD}(\mathrm{CHAD}) = \phi^{-2}\sqrt{p\log (2\tau)},
\end{align}
and so the rates are identical in this case. 

\paragraph{Case 2:} $k = p$ and $k< \sqrt{p\log (2\tau)}$.\newline\newline
In this case, we have
\begin{align}
    \mathrm{DD}(\mathrm{MdFOCuS}) \geq \phi^{-2}\left\{ \sqrt{p\log (2\tau)} \vee \log (2\tau)\right\},
\end{align}
ignoring constant factors and constant terms (in particular replacing $\tau$ by $2\tau$). Due to the assumption that $k = p <\sqrt{p\log (2\tau)}$, we have that $p < \log(2\tau)$ and thus $\sqrt{p\log(2\tau)} < \log(2\tau)$, and so 
\begin{align}
    \mathrm{DD}(\mathrm{MdFOCuS}) \geq \phi^{-2}\log (2\tau).
\end{align}

We also have that
\begin{align}
    \mathrm{DD}(\mathrm{CHAD}) &= \phi^{-2}\left[ k \log \left\{ \frac{ep\log (2\tau)}{k^2} \right\} \vee \log (2\tau)\right]\\
    &=\phi^{-2}\left[ p \log \left\{ \frac{e\log (2\tau)}{p} \right\} \vee \log (2\tau)\right]
\end{align}
 Now, let $L = \log(2\tau)$ and $f(p) = p\log(eL/p)$. Then $f$ is strictly increasing on $(0,L)$, and $f(L) = L$, so that
\begin{align}
p \log \left\{ \frac{e\log (2\tau)}{p} \right\} &< \log(2\tau). 
\end{align}
Therefore, $\mathrm{DD}(\mathrm{CHAD}) = \phi^{-2}\log(2\tau)$, and thus $\mathrm{DD}(\mathrm{CHAD})$ and $\mathrm{DD}(\mathrm{MdFOCuS})$ have the same rate in this case. 

\paragraph{Case 3:} $k <p$ and $k\geq \sqrt{p\log (2\tau)}$.\newline\newline
In this case, we have
\begin{align}
    \mathrm{DD}(\mathrm{MdFOCuS}) \geq \phi^{-2}\left\{ \sqrt{p\log(2\tau)}\log (ep) + \log(2\tau) \right\},
\end{align}
for $p>1$, ignoring constant factors and constant terms (in particular replacing $\tau$ by $2\tau$ and $p$ by $ep$). Moreover, 
\begin{align}
    \mathrm{DD}(\mathrm{CHAD}) &= \phi^{-2}\sqrt{p\log (2\tau)},
\end{align}
and so 
\begin{align}
    \frac{\mathrm{DD}(\mathrm{MdFOCuS})}{\mathrm{DD}(\mathrm{CHAD})} &\geq \log(ep) \geq 1,
\end{align}
which diverges as $p\rightarrow \infty$. 

\paragraph{Case 4:} $k <p$ and $k<  \sqrt{p\log (2\tau)}$.\newline\newline
In this case, we have
\begin{align}
    \mathrm{DD}(\mathrm{MdFOCuS}) \geq \phi^{-2}\left\{ k\log (ep) + \log(2\tau) + \sqrt{k\log(2\tau)}\right\},
\end{align}
for $p>1$, ignoring constant factors and constant terms (in particular replacing $\tau$ by $2\tau$ and $p$ by $ep$). Moreover,  
\begin{align}
    \mathrm{DD}(\mathrm{CHAD})&\leq \phi^{-2}\left[ k\log(ep) + \log (2\tau) + k\log \left\{ \frac{\log(2\tau)}{k^2}\right\}\right]. 
\end{align}
Therefore, 
\begin{align}
    \frac{\mathrm{DD}(\mathrm{MdFOCuS})}{\mathrm{DD}(\mathrm{CHAD})} &\geq \frac{k\log (ep) + \log(2\tau) + \sqrt{k\log(2\tau)}}{k\log(ep) + \log (2\tau) + k\log \left\{ \frac{\log(2\tau)}{k^2}\right\}}. \label{tmpcomp1}
\end{align}
Now, if $k> \sqrt{\log (2\tau)}$, then $k\log \left\{ {\log(2\tau)} k^{-2}\right\}<0$. In this case, the right-hand side of \eqref{tmpcomp1} is no less than $1$. If, on the other hand, $k\leq \sqrt{\log (2\tau)}$, then $k\log \left\{ {\log(2\tau)} k^{-2}\right\}\leq \sqrt{k} \log^{1/4}(2\tau) \log\log(2\tau)$, which is of smaller order than $\sqrt{k\log(2\tau)}$ with respect to $\tau$, so that the right-hand side of \eqref{tmpcomp1} is bounded away from zero in this case as well. 

Thus, in the case where $k<p$ and $k<\sqrt{p\log(2\tau)}$, we conclude that the rate of $\mathrm{DD}(\mathrm{MdFOCuS})$ is at least that of $\mathrm{DD}(\mathrm{CHAD})$. 

\section{Discussion: weakening of assumptions in the multivariate mean-change problem}\label{sec:meandiscussion}
The online changepoint detector presented in Section \ref{sec:multimean} in the main text attains near-optimal performance with independent and isotropic Gaussian noise (see Section \ref{sec:optimalitymean}). However, this performance is not guaranteed under weaker distributional assumptions. Here, we discuss how these assumptions may be weakened by using other offline test statistics from the literature that are known to have optimal performance, yet under weaker assumptions. First, we discuss a possible relaxation of the Gaussianity assumption. Then, we discuss a possible relaxation of the independence assumption, to either spatial or temporal Gaussian noise. For the sake of brevity, we only discuss how the alternative test statistics can be used within the general framework in Section \ref{sec:general} in the main text to attain logarithmic storage and update costs with respect to the sample size $t$. The choice of critical values and resulting statistical guarantees is left for future research. 

\subsection{Sub-Weibull noise}\label{sec:meandiscussionweibull}
Recently, \cite{li2023robust} proposed a modified variant of the test of \cite{liu_minimax_2021} with near-optimal performance under a relaxed assumption of independent sub-Weibull error terms. Note that the $Y_i$ are still assumed to be independent, with independent entries. We now outline how this test can be adapted into our online framework with similar update and storage costs as in Section \ref{sec:multimean} in the main text. Specifically, we will demonstrate how the test of \cite{li2023robust} can be refitted with small modifications to satisfy Assumption \ref{asscomp}, yielding logarithmic update and storage costs with respect to the sample size $t$. %

In the offline case, when testing for a dense changepoint, \cite{li2023robust} use precisely the same test statistic as in \cite{liu_minimax_2021}. Within our notation (and without any modifications), the test for a change in mean occurring $g$ time steps before the last observation, with $t$ observations, is given by:
\begin{align}
    T^{(t)}_{g,\text{dense}} &= \ind \left\{  A^{(t)}_g > \xi^{(t)}_{\text{dense}} \right\}\label{testdenseunmodified},
\intertext{where $\xi^{(t)}_{\text{dense}}$ is a critical value, and}
    A^{(t)}_g &= \sum_{j=1}^p \left\{ Z_g^{(t)}(j)^2 -1 \right\},
    \intertext{is the result of aggregating squared and centered CUSUM-like quantities given by}
    Z_g^{(t)} &= \frac{\sum_{i=1}^{g} Y_i - \sum_{i=t-g+1}^{t}Y_i}{\sqrt{2g}}.\label{firstz}
\end{align}

For the test $T^{(t)}_{g,\mathrm{dense}}$ in \eqref{testdenseunmodified} to satisfy Assumption \ref{asscomp}, one may replace $Z_g^{(t)}$ by 
\begin{align}
        \widetilde{Z}_g^{(t)} &=  \frac{\sum_{i=1}^{t-g} Y_i }{\sqrt{2(t-g)}}- \frac{\sum_{i=t-g+1}^{t}Y_i}{\sqrt{2g}}. \label{ztilde}
\end{align}
Clearly, with this minor modification, the test $T^{(t)}_{g,\mathrm{dense}}$ will satisfy Assumption \ref{asscomp} due to Proposition \ref{theorem2}. Moreover, the signal is preserved, in the sense that $\EE\{(\widetilde{Z}^{(t)}_g)^2\}$ is equal to $\EE\{(Z_g^{(t)})^2\}$ under the null, and of the same order under the alternative. 

For sparse changepoints, the theoretical properties of the test in \cite{liu_minimax_2021} (and the changepoint detector from Section \ref{sec:multimean}) rely on the tractability of a truncated chi-squared distribution, which is not guaranteed for non-parametric classes of distribution. As a remedy, the modification of \cite{li2023robust} introduces sample splitting. Here we outline how sample splitting may be performed online in a fashion that adheres to Assumption \ref{asscomp}. For the data sequence $Y_1, Y_2, \ldots$, define
\begin{align}
    \widetilde{Y}_i &= \begin{cases} Y_i, &\text{ if } i \text{ is odd,}\\
    0, &\text{otherwise,}\end{cases}
\intertext{and,}
\widehat{Y}_i &= \begin{cases} Y_i, &\text{ if } i \text{ is even,}\\
    0, &\text{otherwise.}\end{cases}
\end{align}
As sample-split variants of $\widetilde{Z}_g^{(t)}$ in \eqref{ztilde} for $2\leq g \leq t-2$, define 
\begin{align}
    \widetilde{Z}_{g,1}^{(t)} &= \frac{\sum_{i=1}^{t-g} \widetilde{Y}_i}{\sqrt{2\sum_{i=1}^{t-g} \ind\{\widetilde{Y}_i\neq 0\}}} -
    \frac{\sum_{i=t-g+1}^{t}\widetilde{Y}_i}{\sqrt{2\sum_{i=t-g+1}^{t} \ind\{\widetilde{Y}_i\neq0\}}},\label{samplesplit1}
    \intertext{and}
    \widetilde{Z}_{g,2}^{(t)} &= \frac{\sum_{i=1}^{t-g} \widehat{Y}_i}{\sqrt{2\sum_{i=1}^{t-g} \ind\{\widehat{Y}_i\neq 0\}}} -
    \frac{\sum_{i=t-g+1}^{t}\widehat{Y}_i}{\sqrt{2\sum_{i=t-g+1}^{t} \ind\{\widehat{Y}_i\neq0\}}}.\label{samplesplit2}
\end{align}
We remark that this form of sample splitting preserves the signal in a similar way as before, in the sense that $\EE\{(\widetilde{Z}_{g,1}^{(t)})^2\}$ and $\EE\{(\widetilde{Z}_{g,2}^{(t)})^2\}$ are equal to $\EE\{(Z_g^{(t)})^2\}$ in \eqref{firstz} under the null ($\tau=\infty$), and of the same order under the alternative. 

Given any $t\geq 2$, $g\in [t-1]$, a set $\mathcal{S} = \mathcal{S}^{(t)}$ of sparsity levels, and sparsity-dependent critical value $\xi_s^{(t)}$, a refitted testing procedure of \cite{li2023robust} is given by 
\begin{align}
    T^{(t)}_{g, \text{sparse}} &=  \ind \left\{ \underset{s \in \mathcal{S}}{\max} \ A^{(t)}_{g,s} > \xi^{(t)}_{s,\mathrm{sparse}} \right\},
    \end{align}
where $ A^{(t)}_{g,s}$ is the result of aggregating and thresholding the $\widetilde{Z}_g^{(t)}$ using sample splitting, 
\begin{align}
    A^{(t)}_{g,s} &= \begin{cases}
    \sum_{j=1}^p \left\{ \widetilde{Z}_{g,1}^{(t)}(j)^2 -1 \right\} \ind\{ |\widetilde{Z}_{g,2}^{(t)}(j)|>a_s\}, &\text{ if } t\geq 4,\\
    \sum_{j=1}^p \left\{ \widetilde{Z}_{g}^{(t)}(j)^2 -1 \right\} \ind\{ |\widetilde{Z}_{g}^{(t)}(j)|>a_s\}, &\text{ otherwise,}\end{cases}
\end{align}
where the $a_s$ are sparsity-specific choices of thresholding values. As long as $a_s$ and $\xi^{(t)}_{s,\mathrm{sparse}}$ require a constant-order number of unit-cost operations to compute, the test $T^{(t)}_{g,\mathrm{sparse}}$ satisfies Assumption \ref{asscomp} when taking the set 
\begin{align}S_g^{(t)} = \Bigg \{&\sum_{i=1}^{t-g} \widetilde{Y}_i, \sum_{i=1}^{t-g} \widehat{Y}_i,  \sum_{i=1}^{t} \widetilde{Y}_i, \sum_{i=1}^{t} \widehat{Y}_i,\sum_{i=1}^{t-g} \ind\{\widetilde{Y}_i\neq 0\}, \sum_{i=1}^{t-g} \ind\{\widehat{Y}_i\neq 0\}, \\&\sum_{i=1}^{t} \ind\{\widetilde{Y}_i\neq 0\}, \sum_{i=1}^{t} \ind\{\widehat{Y}_i\neq 0\}\Bigg\}\end{align}
as the summary set required in the Assumption. 

\subsection{Temporal and spatial dependence}
For dense changepoints, \cite{liu_minimax_2021} proved that the offline variant of the test in \eqref{testdenseunmodified} attains minimax rate performance even with spatial or temporal dependence, although under the assumption that the $Y_i$ are jointly Gaussian. This optimal performance is achieved by adjusting the critical value $\xi^{(t)}_{\mathrm{dense}}$. For spatial dependence, the optimal critical value depends on the trace, Frobenius norm and operator norm of $\mathrm{Cov}(Y_i)$, assumed to be constant. For temporal dependence, the critical value depends on a mixing constant $B$, which for a fixed sample size $t$ is the smallest $b$ for which $\sum_{i \in [t]\setminus\{j\}} \normmop{\mathrm{Cov}(Y_i, Y_j)} \leq b$ for all $j=1,2,\ldots,t$. Naturally, this can be extended to an online case by instead defining $B$ to be the smallest $b$ for which $\sum_{i \in [t]\setminus\{j\}} \normmop{\mathrm{Cov}(Y_i, Y_j)} \leq b$ for all $j\in \NN$. 

Since the test in \eqref{testdenseunmodified} can be refitted to both temporal and spatial dependence by adjusting only the critical value, it can be used online within the grid-based methodology from  Section \ref{sec:general} in the main text with logarithmic update and storage costs, as discussed in Section \ref{sec:meandiscussionweibull}. We leave it for future research to formalise and prove performance guarantees for this detector, as well as how to estimate $\mathrm{Cov}(Y_i)$ and the mixing constant $B$ in a computationally efficient online fashion. 

\section{Optimality of the detectors in Section \ref{sec:theoryspec} in the main text}\label{sec:optimality}
We now assess the optimality of the two online changepoint detectors presented in Section \ref{sec:theoryspec} in the main text, by establishing minimax lower bounds.

\subsection{Change in mean}\label{sec:optimalitymean}
Consider first the problem of online detection of a change in the mean in a sequence $Y = (Y_i)_{i \in \NN}$ of independent $p$-dimensional Gaussian vectors,
\begin{align}
    Y_i \overset{\mathrm{i.i.d.}}{\sim} \mathrm{N}_p(\theta_i, \sigma^2 I), \ i \in \NN, \label{modeldefminimax}
\end{align}
for some fixed dimension $p \in \NN$ and fixed variance $\sigma^2>0$.  We will show that the signal strength requirement and the detection delay in Theorem \ref{theorem3} are rate optimal up to at most a logarithmic factor. 

Let $\theta = (\theta_i)_{i \in \NN}$ denote the sequence of means of the $Y_i$, and for such $\theta$, let $\PP_{\theta}$ denote the probability measure under the data generating mechanism in \eqref{modeldefminimax}. Let $(\mathcal{F}_t)_{t \in \NN}$ denote the natural filtration of the $Y_i$,\footnote{That is, for each $t \in \NN$, $\mathcal{F}_t$ is the $\sigma$-algebra generated by the first $t$ observations of the $Y_i$, i.e. $\mathcal{F}_t = \sigma(Y_1, Y_2, \ldots, Y_t) \ \forall \ t \in \NN$.} and for any $\delta>0$, let 
\begin{align}
    \mathcal{T}(\delta) = \Big\{ \widehat{\tau} : \ &\tauhat \text{ is an extended stopping time with respect to } (\mathcal{F}_t)_{t \in \NN}, \\ &\PP_{\theta} ( \widehat{\tau} < \infty) \leq \delta \text{ for any } \theta \in \Theta_0(p)\Big\}, \label{mathcaltmean}
\end{align}
which is the set of all extended stopping times with respect to $(\mathcal{F}_t)_{t\in\NN}$ for which the false alarm probability is no larger than $\delta$ for any $\theta \in \Theta_0(p)$, where
\begin{align}
\Theta_0(p) 
= &\left\{ \theta = (\theta_i)_{i \in \NN} : \theta_i = \mu \text{ for all } i \in \NN \text{ and some } \mu\in\mathbb{R}^p \right\}\label{Theta0}
\end{align}
is the parameter space for all mean sequences $\theta$ that contain no changepoints. 

For any $\tau \in \NN$, $\phi>0$ and $k\in [p]$, define
\begin{align}
\Theta(k, p , \tau, \phi) &= \bigg\{ \theta =(\theta_i)_{i \in \NN} : \exists \ \mu_1, \mu_2 \in \RR^p \text{ such that }
\\ & \quad\quad \theta_i=\mu_1 \text{ for all } 1\leq i\leq \tau,\ \theta_i=\mu_2 \text{ for all } i\geq \tau+1, \\
& \quad\quad \normm{\mu_1 - \mu_2} = \phi, \ \normm{\mu_1 - \mu_2}_0 \leq k  \bigg\}, \label{Theta}
\end{align}
which is the parameter space of all mean sequences $\theta$ containing a single change in the mean at time $\tau$, magnitude $\phi$ and sparsity $k$.  Finally, define the function
\begin{align}
    v(k, p) =
    \begin{cases}
        \sqrt{p}, & \text{if } k > \sqrt{p},\\
        k\log \left( \frac{ep}{k^2} \right) , & \text{otherwise.}
    \end{cases} \label{vdef}
\end{align}
We then have the following result, which can be seen as a Corollary of Proposition 3 in \cite{liu_minimax_2021}.
\begin{proposition}\label{prop:meanminimax}
    For any $\delta \in (0,1)$ and $\epsilon \in (0, 1-\delta)$ , there exist a constant $c>0$ depending only on $\epsilon$, such that the following holds for any $\tau \in \NN$, $p\in \NN$, $k\in [p]$, $\phi>0$ and $\sigma>0$:
    \begin{enumerate}
    \item If $(\phi^2 / \sigma^2) \tau \leq c v(k,p)$, then for any $n\geq \tau$ we have
    \begin{align}
        \underset{\tauhat \in \mathcal{T}(\delta)}{\inf} \ \underset{\theta \in \Theta(k,p,\tau, \phi)}{\sup} \PP_{\theta} \left\{ \tauhat -\tau  >  n 
        \right\} \geq 1- \delta - \epsilon. \label{meanminimaxeq}
    \end{align}
    \item Conversely, if $(\phi^2 / \sigma^2) \tau > c v(k,p)$, we have
    \begin{align}
        \underset{\tauhat \in \mathcal{T}(\delta)}{\inf} \ \underset{\theta \in \Theta(k,p,\tau, \phi)}{\sup} \PP_{\theta}\left\{ \tauhat -\tau  >  c \frac{\sigma^2}{\phi^2} v(k,p) 
        \right\} \geq 1- \delta - \epsilon. \label{meanminimaxeq2}
    \end{align}
    \end{enumerate}
\end{proposition}

Proposition \ref{prop:meanminimax} implies that the signal strength requirement and the detection delay in Theorem \ref{theorem3} are rate-optimal up to a factor at most logarithmic in the changepoint location $\tau$. Indeed, consider first the signal strength requirement in Theorem \ref{theorem3}, which for any fixed $\delta \in (0,1)$ requires that the signal strength $\tau \phi^2 \sigma^{-2}$ satisfies $\tau \phi^2 \sigma^{-2} \geq C z(k,p, 2\tau)$ for the changepoint detector in Section \ref{sec:multimean} in the main text to be guaranteed detection of a changepoint occurring at time $\tau$ within $2\tau$ observations with probability at least $1-\delta$, for some $C>0$, where the function $z$ is defined in \eqref{rdef}. In comparison, for any $\epsilon>0$, Proposition \ref{prop:meanminimax} guarantees that, by choosing $\tau \phi^2 \sigma^{-2}$ to be a sufficiently small non-zero factor of $v(k,p)$, then for any extended stopping time $\tauhat \in \mathcal{T}(\delta)$ (such as the detector from Section \ref{sec:multimean} in the main text), the worst-case probability of observing $\tauhat > \tau + n$ for any $n\geq \tau$ is
at least $1- \delta - \epsilon$. 
The signal strength requirement in Theorem \ref{theorem3} is therefore rate optimal up to a factor of at most $z(k,p, 2\tau) / v(k,p) \leq \log(2\tau)$. Moreover, the detection delay in Theorem \ref{theorem3} is for any $\delta \in (0,1)$ of order at most $\phi^{-2} \sigma^{2} z(k,p, 2\tau)$ with probability at least $1-\delta$. In comparison, for any $\epsilon>0$, Proposition \ref{prop:meanminimax} guarantees a detection delay of order at least $\phi^{-2} \sigma^{2} v(k,p)$ with probability at least $1 - \delta - \epsilon$ for any $\tauhat \in \mathcal{T}(\delta)$. Thus, the detection delay in Theorem \ref{theorem3} is also rate optimal up to a factor of at most $z(k,p, 2\tau) / v(k,p) \leq \log(2\tau)$ (ignoring constant factors).

The factor $\log(2\tau)$, which constitutes the maximum  gap between the rate in Theorem \ref{theorem3} and the minimax lower bound in Proposition \ref{prop:meanminimax}, may likely be the result of a loose lower bound, at least in part. In Theorem \ref{theorem3}, the factor $\log(2\tau)$ in $z(k,p,2\tau)$ ensures uniformly controlled Type I errors of the test $T^{(t)}_g$, specifically ensuring that $\PP_{\infty}(\max_{g \in G^{(t)}} T_g^{(t)})\leq \delta/t^2$ for any $t\geq 2$, and thus a false alarm probability bounded from above by $\delta$. In particular, the changepoint detector is agnostic to the location $\tau$ of the changepoint, and Theorem \ref{theorem3} applies regardless of the value of $\tau$. In comparison, the minimax lower bound in Proposition \ref{prop:meanminimax} takes the changepoint location $\tau$ as fixed. In fact, a core argument used to prove Proposition \ref{prop:meanminimax} is to show a weakened variant of Proposition 3 in \citet{liu_minimax_2021}, which is an (offline) minimax lower bound for data with fixed sample size. Thus, the additional hardness of the problem induced by sequential testing and not knowing the value of $\tau$ is not taken into account in Proposition \ref{prop:meanminimax}. 

A comparison between Proposition \ref{prop:meanminimax} and Proposition 4.1 in \citet{anote} reveals that the minimax lower bound in Proposition \ref{prop:meanminimax} for the detection delay is provably loose by a factor $\log(\tau)$ in the univariate case, i.e.~when $p = 1$. Moreover, in the univariate case, the rate in Theorem \ref{theorem3} is optimal. 
We leave it for future research to pinpoint the exact minimax detection delay rate in the multivariate case, where $p>1$. 

\subsection{Change in covariance}\label{sec:optimalitycovariance}
Consider the problem of online detection of a change in the covariance in a sequence $Y = (Y_i)_{i \in \NN}$ of independent, mean-zero $p$-dimensional sub-Gaussian vectors. It suffices to assume the $Y_i$ to be independent Gaussian vectors with positive definite covariance matrices, as Assumption \ref{assmultivariate2} can then be shown to hold for some $u,w>0$. Assume that 
\begin{align}
    Y_i \overset{\mathrm{i.i.d.}}{\sim} \mathrm{N}(0, \gamma_i), \ i \in \NN, \label{modeldefminimaxcovariance}
\end{align}
for some fixed dimension $p \in \NN$.  We will show that both the signal strength requirement and the detection delay in Theorem \ref{theorem5} are rate optimal up to at most a logarithmic factor for small to moderately sized changes. 

Let $\gamma = (\gamma_i)_{i \in \NN}$ denote the sequence of covariance matrices of the $Y_i$, and for such $\gamma$, let $\PP_{\gamma}$ denote the probability measure under the data generating mechanism in \eqref{modeldefminimaxcovariance}. Let $(\mathcal{F}_t)_{t \in \NN}$ denote the natural filtration of the $Y_i$, and similar to before, for any $\delta>0$, let 
\begin{align}
    \mathcal{T}(\delta) = \Big\{ \widehat{\tau} : \ &\tauhat \text{ is an extended stopping time with respect to } (\mathcal{F}_t)_{t \in \NN}, \\ &\PP_{\gamma} ( \widehat{\tau} < \infty) \leq \delta \text{ for any } \gamma \in \Gamma_0(p)\Big\}, \label{mathcaltcovariance}
\end{align}
which is the set of all extended stopping times with respect to $(\mathcal{F}_t)_{t\in\NN}$ for which the false alarm probability is no larger than $\delta$ for any $\gamma\in \Gamma_0(p)$, where
\begin{align}
\Gamma_0(p) 
= &\left\{ \gamma = (\gamma_i)_{i \in \NN} : \gamma_i = \Sigma \text{ for all } i \in \NN \text{ and some } \Sigma \in\mathbb{R}^{p\times p} \ \Sigma \succ 0 \right\},\label{Gamma0}
\end{align}
is the parameter space for all covariance sequences $\gamma$ that contain no changepoints. 

For any $\tau \in \NN$, $\omega \in (0,1/2]$ and $k\in [p]$, define
\begin{align}
\Gamma(k, p , \tau, \omega) &= \bigg\{ \gamma =(\gamma_i)_{i \in \NN} : \exists \ \Sigma_1, \Sigma_2 \in \RR^{p\times p}, \Sigma_1 \succ 0, \Sigma_2 \succ 0, \text{ such that }
\\ & \quad\quad \gamma_i=\Sigma_1 \text{ for all } 1\leq i\leq \tau,\ \gamma_i=\Sigma_2 \text{ for all } i\geq \tau+1, \\
& \quad\quad \frac{\lambda_{\max}^k(\Sigma_1 - \Sigma_2)}{\lambda_{\max}^k(\Sigma_1) \vee \lambda_{\max}^k(\Sigma_2)} = \omega \bigg\}, \label{Gamma}
\end{align}
which is the parameter space of all covariance matrix sequences $\gamma$ containing a single change at location $\tau$ with relative covariance change magnitude, measured in terms of the $k$-sparse eigenvalue from \eqref{firsteigsparse}, given by $\lambda_{\max}^k(\Sigma_1 - \Sigma_2) \{\lambda_{\max}^k(\Sigma_1) \vee \lambda_{\max}^k(\Sigma_2)\}^{-1} = \omega$. Note that when $k=p$, we have $\lambda_{\max}^k(\Sigma_1 - \Sigma_2) \{\lambda_{\max}^k(\Sigma_1) \vee \lambda_{\max}^k(\Sigma_2)\}^{-1} = \normmop{\Sigma_1 - \Sigma_2} (\normmop{\Sigma_1} \vee \normmop{\Sigma_2})^{-1}$.
We then have the following result, which can be seen as a Corollary of Proposition 3 in \cite{moen2024minimax} or Theorem 5.1 in \cite{berthet_optimal}.

\begin{proposition}\label{prop:covarianceminimax}
    For any $\delta \in (0,1)$ and $\epsilon \in (0, 1-\delta)$ , there exist a constant $c>0$ depending only on $\epsilon$, such that the following holds for any $\tau \in \NN$, $p\in \NN$, $k\in [p]$, and $\omega \in (0, 1/2]$:
    \begin{enumerate}
    \item If $\omega^2 \tau \leq c k \log \left(\frac{ep}{k}\right)$, then for any $n\geq \tau$ we have
    \begin{align}
        \underset{\tauhat \in \mathcal{T}(\delta)}{\inf} \ \underset{\gamma \in \Gamma(k,p,\tau, \omega)}{\sup} \PP_{\gamma} \left\{ \tauhat -\tau  >  n 
        \right\} \geq 1- \delta - \epsilon. \label{covarianceminimaxeq1}
    \end{align}
    \item Conversely, if $\omega^2 \tau > c k \log \left(\frac{ep}{k}\right)$, we have
    \begin{align}
        \underset{\tauhat \in \mathcal{T}(\delta)}{\inf} \ \underset{\gamma \in \Gamma(k,p,\tau, \omega)}{\sup} \PP_{\gamma}\left\{ \tauhat -\tau  >  c \frac{k \log\left(\frac{ep}{k}\right)}{\omega^2} 
        \right\} \geq 1- \delta - \epsilon. \label{covarianceminimaxeq2}
    \end{align}
    \end{enumerate}
\end{proposition}

Proposition \ref{prop:covarianceminimax} implies that the signal strength requirement and the detection delay in Theorem \ref{theorem5} are rate-optimal up to a factor at most logarithmic in the changepoint location $\tau$ when the change is dense. Indeed, consider first the signal strength requirement in Theorem \ref{theorem5}, which for any fixed $\delta \in (0,1)$ requires that $\tau \omega^2 \geq C \{p \vee \log(2\tau)\}$ in order for the changepoint detector in Section \ref{sec:covariance} in the main text to be guaranteed detection of a changepoint occurring at time $\tau$ within $2\tau$ observations with probability at least $1-\delta$, for some $C>0$, where $\omega = \normmop{\Sigma_1 - \Sigma_2} ( \normmop{\Sigma_1}\vee \normmop{\Sigma_2})^{-1}$. In comparison, when $\omega \leq 1/2$, then for any $\epsilon>0$, Proposition \ref{prop:covarianceminimax} guarantees that, by choosing $\tau \omega^2$ to be a sufficiently small non-zero factor of $p$, then for any extended stopping time $\tauhat \in \mathcal{T}(\delta)$ (such as the changepoint detector from Section \ref{sec:covariance} in the main text), the worst-case probability of observing $\tauhat > \tau + n$ for any $n\geq \tau$ is
at least $1- \delta - \epsilon$. 
The signal strength requirement in Theorem \ref{theorem5} is therefore rate optimal up to a factor of at most $1 \vee p^{-1}\log(2\tau)$. Moreover, the detection delay in Theorem \ref{theorem5} is for any $\delta \in (0,1)$ of order $\omega^{-2} \{ p \vee \log(2\tau)\}$ with probability at least $1-\delta$. In comparison, for any $\epsilon>0$, Proposition \ref{prop:covarianceminimax} guarantees a detection delay of order at least $\omega^{-2} p$ with probability at least $1 - \delta - \epsilon$ for any $\tauhat \in \mathcal{T}(\delta)$. Thus, the detection delay in Theorem \ref{theorem5} is also rate optimal up to a factor of at most $1 \vee p^{-1}\log(2\tau)$.

Next, consider the sparsity adaptive changepoint detector in Section \ref{sec:sparsecovariance}. For any sparsity $k\in[p]$, the signal strength requirement in Theorem \ref{theorem6} is that $\tau \omega_k^2 \geq C k^2 \log(p \vee 2\tau)$, under which the detection delay is of order $\omega_k^{-2} k^2 \log(p\vee 2\tau)$ with high probability, for some $C>0$, where $\omega_k = \lambda_{\max}^k(\Sigma_1 - \Sigma_2) ( \lambda_{\max}^k(\Sigma_1)\vee \lambda_{\max}^k(\Sigma_2))^{-1}$, and $\lambda_{\max}^k(\cdot)$ is defined in \eqref{firsteigsparse}. In comparison, Proposition \ref{prop:covarianceminimax} guarantees that all extended stopping times with bounded false alarm probability will take arbitrarily long to detect the changepoint (with high probability) whenever $\tau \omega_k^2 \leq c k\log(ep/k)$. Moreover, Proposition \ref{prop:covarianceminimax} guarantees any extended stopping time with bounded false alarm probability has a detection delay of order $\omega_k^{-2} k \log(ep/k)$ with high probability, as long as $\omega_k \leq 1/2$. Thus, the signal strength requirement and detection delay in Theorem \ref{theorem6} are only close to minimax rate optimal for very small values of $k$ and when the relative magnitude of the covariance change is small to moderate. When $k=1$, the signal strength requirement and detection delay in Theorem \ref{theorem6} are rate optimal up to a factor of at most $\log(2\tau)$.

Similar to the minimax lower bound for the mean-change problem, the proof of Proposition \ref{prop:covarianceminimax} builds on an existing (offline) minimax lower bound for data with fixed sample size. Thus, the additional hardness of the problem induced by sequential testing is not taken into account in Proposition \ref{prop:covarianceminimax}, which may result in a loose lower bound. We leave it for future research to pinpoint the exact minimax detection delay rate in the multivariate case when $p>1$. 

\section{Additional results and details from the simulation study}\label{sec:simdetails}
\subsection{Details on how the changepoint detectors were calibrated in Section \ref{sec:simulations_statistics}}\label{sec:tuningdetails}
We now provide more details on how the changepoint detectors in the simulation study from Section \ref{sec:simulations_statistics} in the main text were calibrated to attain approximately $5\%$ false alarm probability after processing the first $N$ observations. We first consider the changepoint detector from Section \ref{sec:multimean}, which is designed to control false alarm over infinite data sequences, and thus needed some minor modifications. In particular, both the threshold values $a(s,t)$ and mean-centring terms $\nu_{a(s,t)}$ in \eqref{adef}, and the grid $\mathcal{S}^{(t)}$ of sparsities, grow with $t$. This also applies to the critical value $\xi_s^{(t)}$ suggested in Theorem \ref{theorem3}, causing the detector to be potentially overly conservative in finite-sample settings. For the simulation study, the detector from Section \ref{sec:multimean} was therefore modified by replacing $a(s,t)$ by $a(s,2)$, $\nu_{a(s,t)}$ by $\nu_{a(s,2)}$,  and $\mathcal{S}^{(t)}$ by $\mathcal{S}^{(2)}$, ensuring these to be constant with respect to $t$. The critical value $\xi_s^{(t)} = \xi_s$ was also chosen independently of $t$, as 
\begin{align}
    {\xi}_s &=  \begin{cases}
        \widehat{\lambda}_1 \widetilde{z}(s,p,2)  & \text{ for } s > \sqrt{p \log 2} \\
        \widehat{\lambda}_2 \widetilde{z}(s,p,2)  & \text{ for } s \leq \sqrt{p \log 2}
    \end{cases}, \label{xiempirical}
    \intertext{where}
    \widetilde{z}(s,p,t) &= s \log \left(1+\frac{\sqrt{p \log t}}{s}\right) + \log t,
\end{align}
and the leading constants $\widehat{\lambda}_1$ and $\widehat{\lambda}_2$ were chosen by Monte Carlo simulation, described shortly. First, two remarks are in order. As the first remark, note that the function $z(s,p,t)$, present in critical value suggested by Theorem \ref{theorem3}, is in \eqref{xiempirical} replaced by a variant $\widetilde{z}(s,p,t)$, where $t$ is set to $2$ (to ensure that $\xi_s$ in \eqref{xiempirical} to be constant in $t$). Moreover, as opposed to $z(s,p,t)$, the function $\widetilde{z}(s,p,t)$ is monotonically increasing in $s$, which is preferable in practice, seeing as the statistic $A^{(t)}_{s,g}$ in \eqref{adef} is (deterministically) non-decreasing in $s$. The function $\widetilde{z}(s,p,t)$ can therefore be seen as a monotonically increasing variant of the function ${z}(s,p,t)$, as there can be shown to exist constants $c,C>0$, independent of $s,p$ and $t$, such that $c z(s,p,t) < \widetilde{z}(s,p,t) < C z(s,p,t)$. As the second remark,  note that the critical value $\xi_s$ in \eqref{xiempirical} has two distinct leading constants, one for the sparse regime ($s \leq \sqrt{p \log 2}$) and one for the dense regime ($s > \sqrt{p \log 2}$). This was done for practical purposes, after experiencing that a single leading constant for all sparsity regimes resulted in a slightly conservative choice of critical value.  

The leading constants $\widehat{\lambda}_1$ and $\widehat{\lambda}_2$ were chosen by Monte Carlo simulations as follows. For $k = 1, 2, \ldots, K$, we sampled a data set $Y^{(k)} = (Y_1^{k}, \ldots, Y_N^{(k)})$ from the model described in Section \ref{sec:multimean} in the main text, with no changepoint, $\mu_1 = 0$ and noise level $\sigma=1$. For each data set $Y^{(k)}$, the statistic $A_{s,g}^{(t)} / \xi_s$ was computed for each $s \in \mathcal{S}^{(2)}$, $g \in G^{(t)}$ from \eqref{thegrid}, and $t=2, 3, \ldots, N$, with $\xi_s$ given as in \eqref{xiempirical} and $A_{s,g}^{(t)}$ defined in \eqref{adef}. By applying a Bonferroni correction, $\lambda_1$ was chosen to be the upper $2.5\%$ empirical quantile of 
$$\underset{s \in \mathcal{S}^{(2)},\  s > \sqrt{p\log 2}}{\max} \quad  \underset{t =2,3, \ldots, N}{\max} \ \underset{g \in G^{(t)}}{\max}  \ \frac{A_{s,g}^{(t)}}{\xi_s},$$
computed from the $K$ Monte Carlo sampled data sets. Similarly, $\lambda_2$ was chosen as the upper $2.5\%$ empirical quantile of 
$$\underset{s \in \mathcal{S}^{(2)}, \ s \leq \sqrt{p\log 2}}{\max} \quad \underset{t =2,3, \ldots, N}{\max} \ \underset{g \in G^{(t)}}{\max}  \ \frac{A_{s,g}^{(t)}}{\xi_s},$$
also computed from the $K$ data sets. 

The ocd detector of \cite{chen_high-dimensional_2022} was calibrated as follows. For $k = 1, 2, \ldots, K$, we sampled a data set $Y^{(k)} = (Y_1^{k}, \ldots, Y_N^{(k)})$ from the model described in Section \ref{sec:multimean} in the main text, with no changepoint, $\mu_1 = 0$ and noise level $\sigma=1$. For each data set $Y^{(k)}$, the statistics $S^{\text{diag}}_t$, $S^{\text{off}, \text{s}}_t$ and $S^{\text{off}, \text{d}}_t$ (defined in \citealt{chen_high-dimensional_2022}) where computed for $t=2,3, \ldots, N$. The corresponding critical values $T^{\text{diag}}$, $T^{\text{off}, \text{s}}$ and $T^{\text{off}, \text{d}}$ were then respectively taken as the empirical upper $(5/3)\%$ quantiles of $\max_{t=2,3,\ldots, N} \ S^{\text{diag}}_t$, $\max_{t=2,3,\ldots, N} \ S^{\text{off}, \text{s}}_t$ and $\max_{t=2,3,\ldots, N} \ S^{\text{off}, \text{d}}_t$ from the $K$ simulated data sets. 

The remaining changepoint detectors were calibrated similarly as the ocd detector, where Bonferroni corrections were applied to the detectors using more than one test statistic to test for a changepoint. 

\newpage
\subsection{Detection delays (original scale) from Section \ref{sec:simulations_statistics}}\label{sec:mainsim_nonlog}
Figure \ref{fig:combined_p=100} below displays the average detection delays from the simulation study in Section \ref{sec:simulations_statistics} on a linear scale. 
\begin{figure}[h!]
    \centering
    \includegraphics[width=\textwidth]{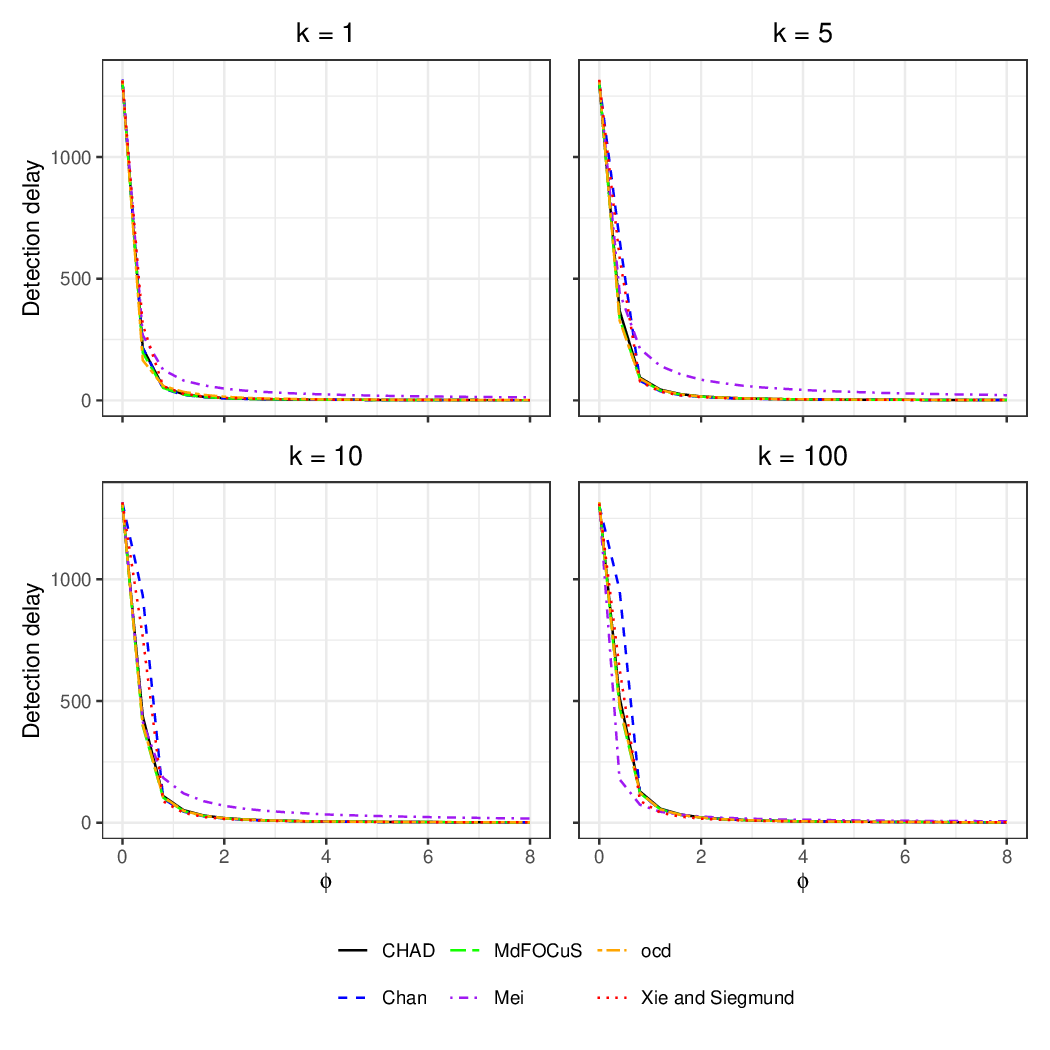}
    \caption{Average detection delay of the detectors for varying change magnitudes ($\phi$) and changepoint sparsities $k = 1,5,10,100$.}
    \label{fig:combined_p=100}
\end{figure}
\newpage
\subsection{Run time and memory consumption of CHAD alone, on linear axes}\label{sec:runtimeonlychad}
The simulation study in Section~\ref{sec:simulationcomp} was re-run using only the
native R implementation of CHAD, with the maximum sample size increased to $10^6$.
The resulting update time and memory consumption, as functions of $t$, are shown in
Figures~\ref{fig:chadisolatedruntime1} and~\ref{fig:chadisolatedruntime2},
respectively. In contrast with the corresponding figure in
Section~\ref{sec:simulationcomp}, these are plotted on linear axes. In all other
respects the simulation setup is identical to that of Section~\ref{sec:simulationcomp}.

Figure~\ref{fig:chadisolatedruntime1} illustrates that the update time of CHAD grows
with $t$, albeit very slowly. The initial zig-zag pattern is likely caused by noisy
timing estimates at small sample sizes. The memory consumption in
Figure~\ref{fig:chadisolatedruntime2} also grows with $t$, at a decreasing rate, in
line with CHAD's memory consumption growing logarithmically in $t$.

\begin{figure}[h!]
    \centering
    \includegraphics[width=\textwidth]{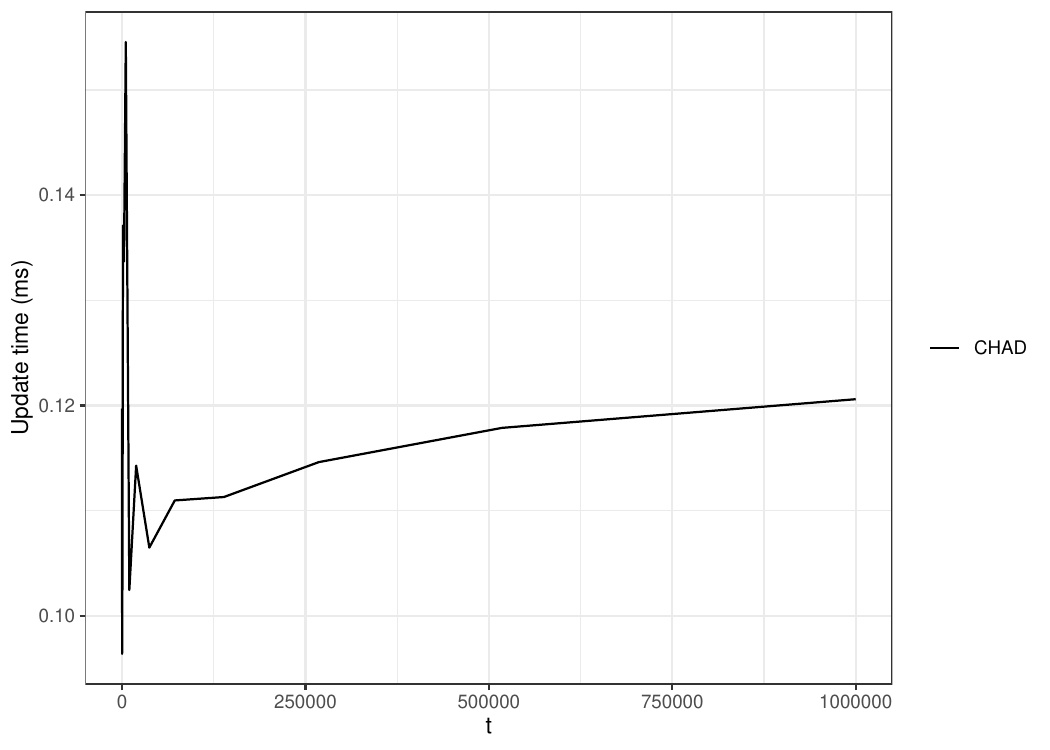}
    \caption{Update time (in milliseconds) of CHAD shown as a function of $t$.}
    \label{fig:chadisolatedruntime1}
\end{figure}

\begin{figure}[h!]
    \centering
    \includegraphics[width=\textwidth]{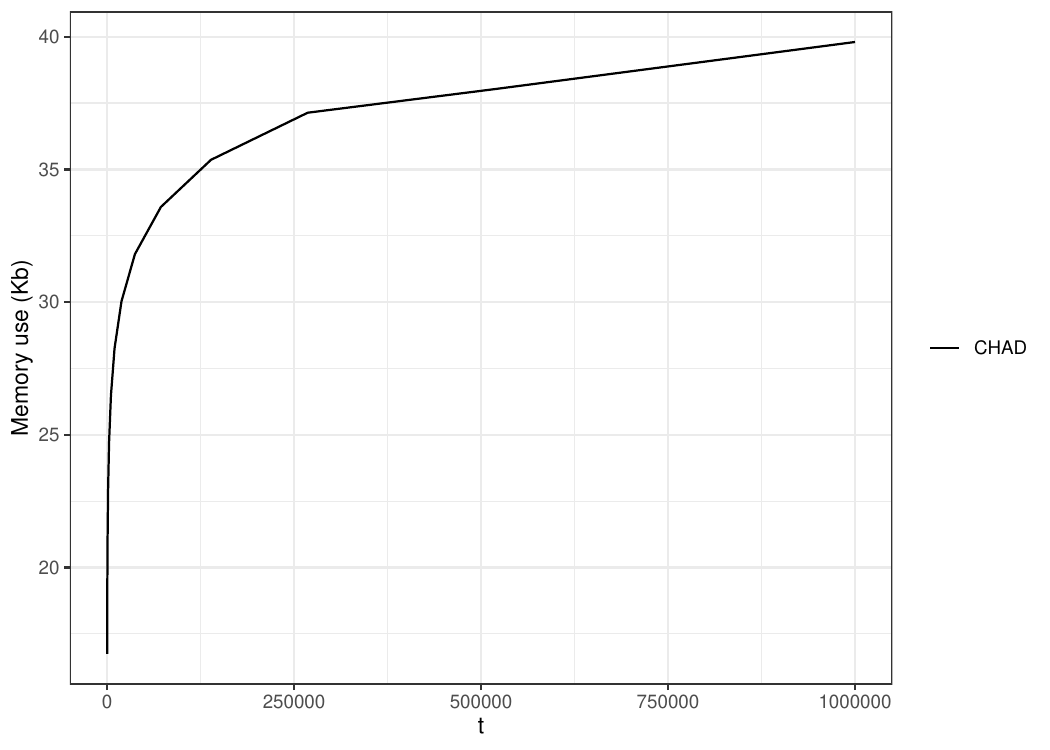}
    \caption{Memory consumption (in kilobytes) of CHAD shown as a function of $t$.}
    \label{fig:chadisolatedruntime2}
\end{figure}
\newpage
\subsection{Run times and memory consumptions from Section \ref{sec:simulationcomp} including Rcpp-accelerated variants}\label{sec:runtimesabsolute}
Figure \ref{fig:combined_runtime_all} displays the result of including the Rcpp-accelerated variants of CHAD and MdFOCuS in the simulation study of Section \ref{sec:simulationcomp}. The accelerated version of CHAD performs its CUSUM and threshold computations through a C function, but is otherwise implemented identically to the pure R version. The accelerated version of MdFOCuS is taken as the implementation in the package focus \cite{focuspackage}, written in Rcpp. It relies on the same core convex-hull pruning methodology as the pure R version (accounting for the bulk of the computation), but omits the additional sparsity-adaptive aggregation that the R version performs on top. It is thus an optimised implementation of a simpler, non-sparsity-adaptive variant, and is therefore not directly comparable to either CHAD or the R version of MdFOCuS. While both accelerated versions exhibit a clear speed-up, their scaling in both the sample size $t$ and the dimension $p$ appears qualitatively similar to that of the pure R versions. We further remark that the apparent zig-zag pattern of MdFOCuS is likely caused by its pruning routine being invoked only periodically.

\begin{figure}[h!]
    \centering
    \includegraphics[width=\textwidth]{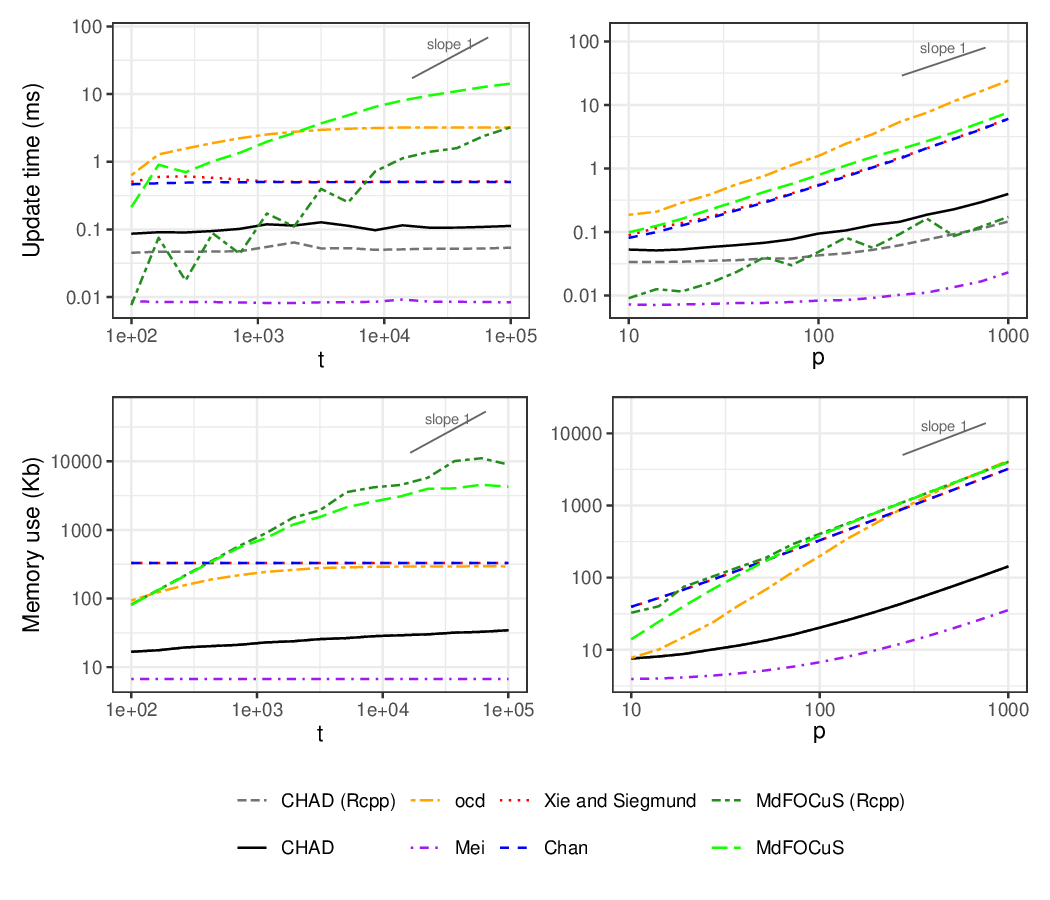}
    \caption{Update time in milliseconds (top) and memory consumption in kilobytes (bottom) of the detectors, including the Rcpp-accelerated variants of CHAD and MdFOCuS, shown as a function of $t$ (left) and $p$ (right).}
    \label{fig:combined_runtime_all}
\end{figure}
\newpage\subsection{Smaller $p$, larger $N$}\label{sec:smallplargeN}
The simulation study from Section \ref{sec:simulations_statistics} was also performed with a larger maximum stream length ($N=5000$). To make the simulation study computationally feasible, the dimension was reduced to $p=10$, and the sparsity levels were set to $k\in \{1,2,5,10\}$. The parameters were otherwise the same as in Section \ref{sec:simulationcomp}, although the methods were re-calibrated due to the change in stream length and dimension. The resulting detection delays are shown in Figure \ref{fig:combined_p=10_log} (on a log scale) and in Figure \ref{fig:combined_p=10} (linear scale). Here, we see that CHAD, MdFOCuS and the detector of \cite{xs2013} have the lowest detection delays across all sparsity regimes. However, all the detectors appear largely comparable, with the exception of \cite{mei2010}, which has a larger detection delay for strong signals than the competitors. For $\phi=0$, corresponding to no changepoint, the rate of false alarms (i.e.~the frequency of $\tauhat \leq N$) were $4.1\%$ for CHAD, $2.7\%$ for MdFOCuS, $4.2\%$ for ocd, $4.6\%$ for the detector of \cite{mei2010}, $5.7\%$ for the detector of \cite{xs2013}, and $5.7\%$ for the detector of \cite{chan2017}. 

\begin{figure}[h!]
    \centering
    \includegraphics[width=\textwidth]{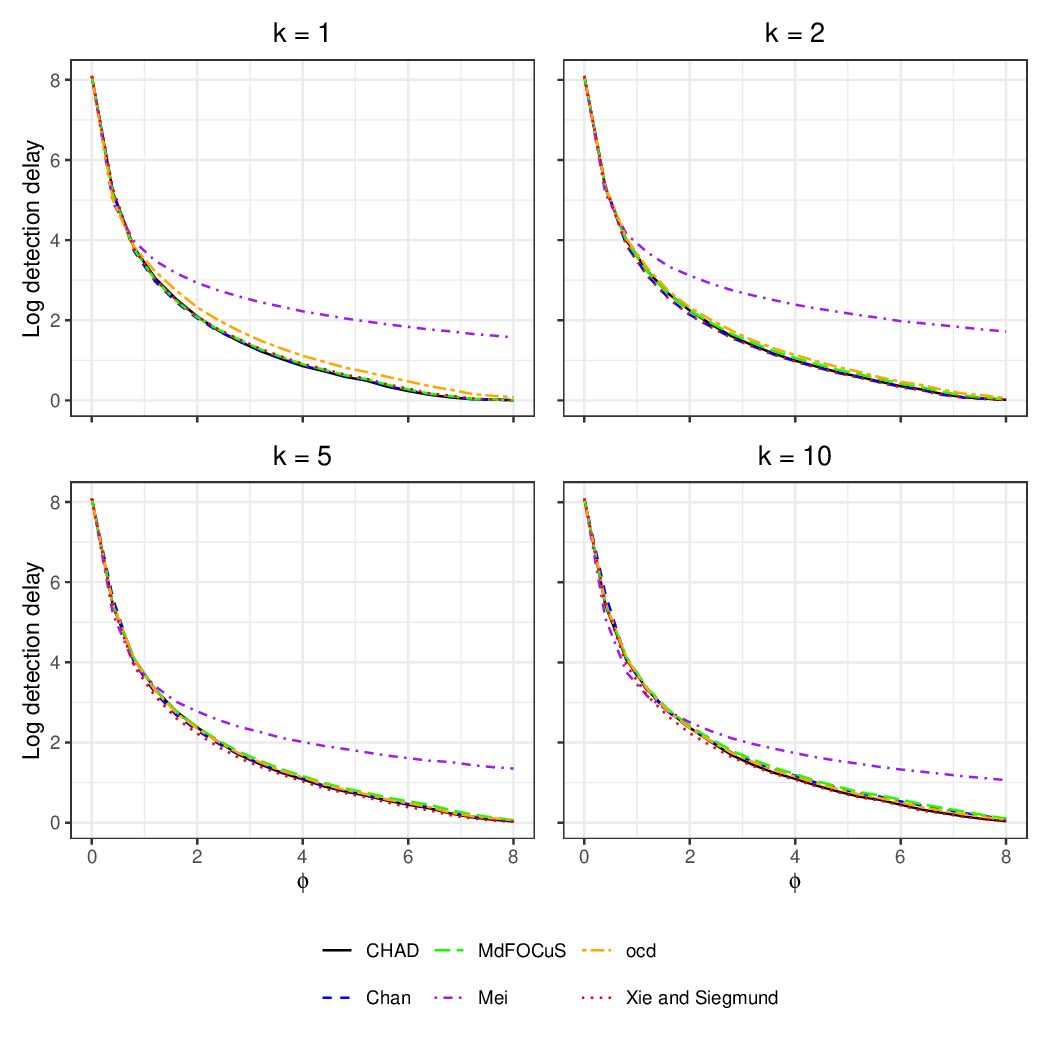}
    \caption{Average detection delay of the detectors (on log scale) for $p=10$ and varying change magnitudes ($\phi$) and changepoint sparsities $k = 1,2,5,10$.}
    \label{fig:combined_p=10_log}
\end{figure}

\begin{figure}[h!]
    \centering
    \includegraphics[width=\textwidth]{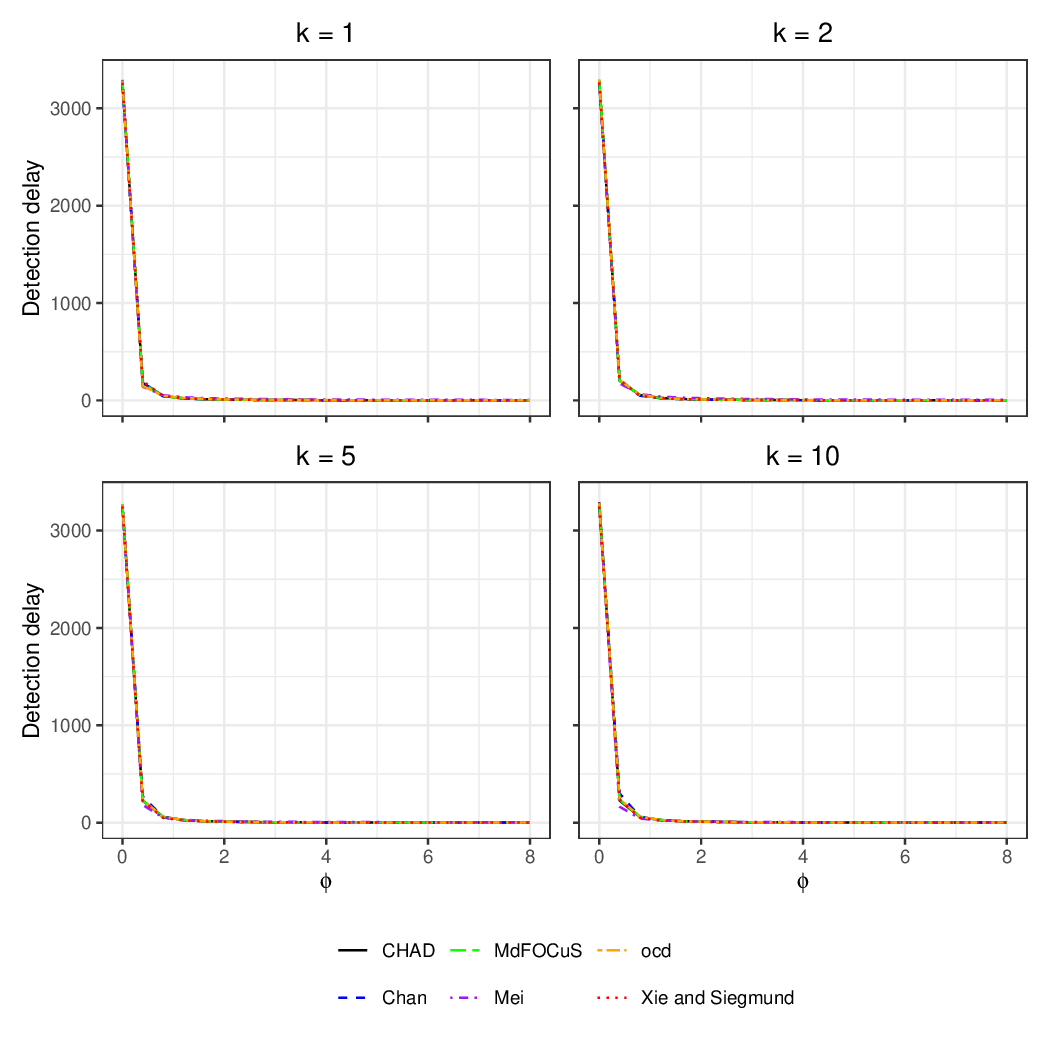}
    \caption{Average detection delay of the detectors for $p = 10$ and varying change magnitudes ($\phi$) and changepoint sparsities $k = 1,2,5,10$.}
    \label{fig:combined_p=10}
\end{figure}
\leavevmode\newpage

\subsection{Larger $p$, smaller $N$}\label{sec:largepsmallN}
The simulation study from Section \ref{sec:simulations_statistics} was also performed with a larger dimension ($p=1000$). To make the simulation study computationally feasible, the maximum stream length was reduced to $N = 200$. The sparsity levels were set to $k\in \{1,10,100,1000\}$. The parameters were otherwise the same as in Section \ref{sec:simulationcomp}, although the methods were re-calibrated due to the change in stream length and dimension. The resulting detection delays are shown in Figure \ref{fig:combined_p=1000_log} (on a log scale) and in Figure \ref{fig:combined_p=1000} (linear scale). Here, we see that CHAD, MdFOCuS and the detectors of \cite{xs2013} and \cite{chan2017} outperform the competitors for sparse changes ($k=1,10,100$). For dense changes ($k =p=1000$), the detectors of \cite{chan2017} and \cite{xs2013} outperform the competitors for both strong and weak signals. For $\phi=0$, corresponding to no changepoint, the rate of false alarms (i.e.~the frequency of $\tauhat \leq N$) were $5.8\%$ for CHAD, $3.1\%$ for MdFOCuS, $5.1\%$ for ocd, $4.1\%$ for the detector of \cite{mei2010}, $5.1\%$ for the detector of \cite{xs2013}, and $3.9\%$ for the detector of \cite{chan2017}. 

\begin{figure}[h!]
    \centering
    \includegraphics[width=\textwidth]{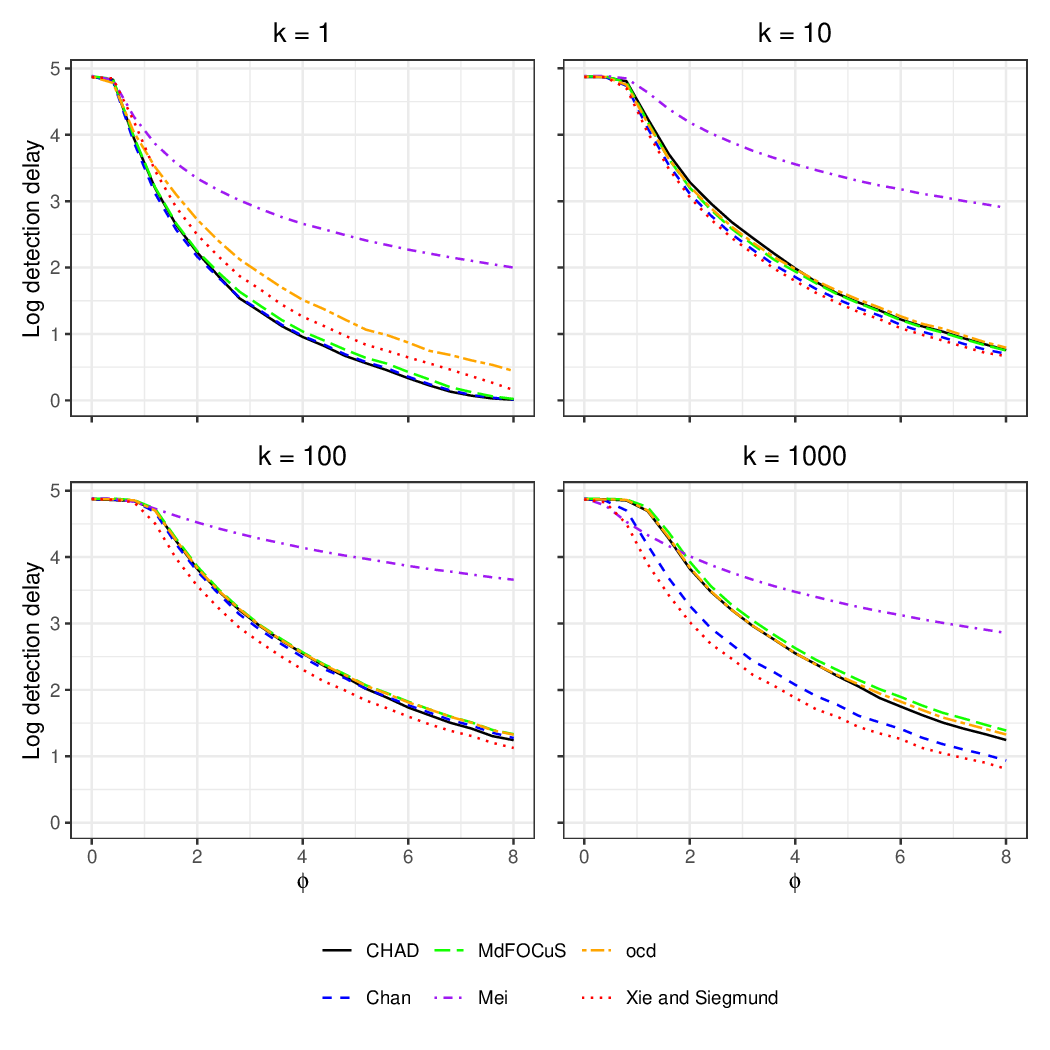}
    \caption{Average detection delay of the detectors (on log scale) for $p=1000$ and varying change magnitudes ($\phi$) and changepoint sparsities $k = 1, 10,100,1000$.}
    \label{fig:combined_p=1000_log}
\end{figure}
\begin{figure}[h!]
    \centering
    \includegraphics[width=\textwidth]{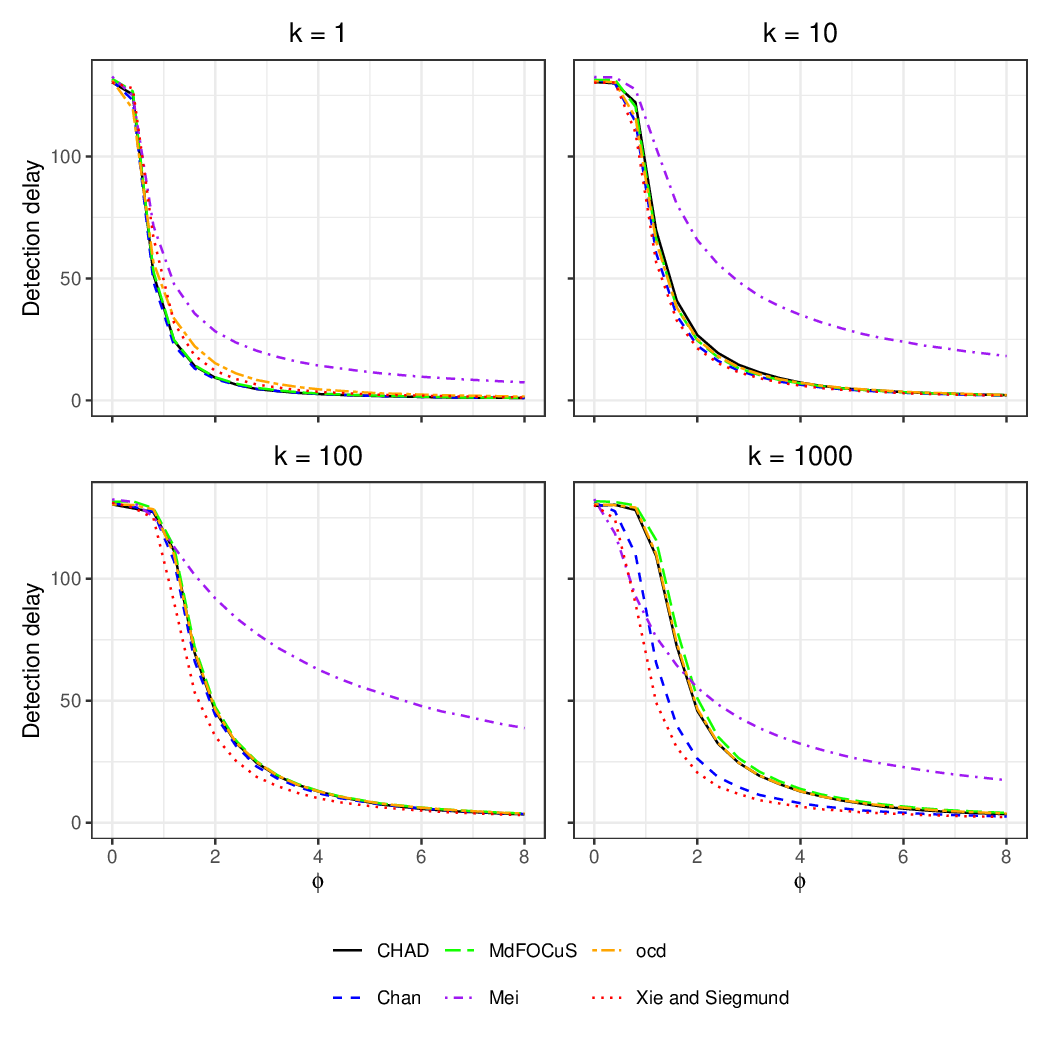}
    \caption{Average detection delay of the detectors for $p=1000$ and varying change magnitudes ($\phi$) and changepoint sparsities $k = 1,10,100,1000$.}
    \label{fig:combined_p=1000}
\end{figure}
\leavevmode\newpage

\subsection{Unknown pre-change mean}\label{simnonzeromean}
The simulation study from Section \ref{sec:simulations_statistics} was also performed for the setting where the pre-change mean $\mu_1$ was taken as unknown. This is a setting that the changepoint detector from Section \ref{sec:multimean} and MdFOCuS \citep{computationalgeometry} are natively designed for, while the remaining methods are not. To adjust the detectors to the unknown-pre-change-mean setting, the changepoint detector from Section \ref{sec:multimean} (CHAD) was equipped with the standard CUSUM statistic defined in \eqref{ydef}, while MdFOCuS was run using a routine tailored to the setting with $\mu_1$ unknown. For the remaining changepoint detectors, i.e., those of \cite{chen_high-dimensional_2022}, \cite{mei2010}, \cite{xs2013} and \cite{chan2017}, there is no obvious way to estimate the pre-change mean online. One option would be to give these detectors access to an additional sample of the pre-change observations, but this would give the detectors access to more data and thus an implicit advantage. Instead, we opted for giving these detectors a probation (training) period, letting them use the first $\lceil \tau/2\rceil = 334$ observations of the data to estimate the pre-change mean without scanning for changepoints. In contrast, CHAD and MdFOCuS scanned for changepoints at all sample sizes, thus being given a slight disadvantage. 

Apart from the modifications described above, the simulation study was performed in a similar fashion as in Section \ref{sec:simulations_statistics}, with the same stream length ($N=2000$) and dimension ($p=100$). The critical values were re-computed using the modified detectors. The resulting detection delays are shown in Figure \ref{fig:plot_p=100_combined_unknownmean}, and shown on log scale in Figure \ref{fig:plot_p=100_combined_log_unknownmean}.

An inspection of the Figures \ref{fig:plot_p=100_combined_log_unknownmean} and \ref{fig:plot_p=100_combined_unknownmean} reveals that the changepoint detectors CHAD, MdFOCuS and the detector of \cite{xs2013} appear to be largely comparable and out-perform the remaining detectors. However, the detector of \cite{xs2013} appears to have a slight edge in the dense case ($k=100$) when the signal strength is very weak, but appears to have slightly weaker performance for large signal strengths for all values of the sparsity $k$. Interestingly, the detectors CHAD and MdFOCuS appear to have very similar performance across the sparsity regimes, and their performance appears to be relatively more competitive than in the case with a known pre-change mean (i.e., Section \ref{sec:simulations_statistics}). For $\phi=0$, corresponding to no changepoint, the rate of false alarms (i.e.~the frequency of $\tauhat \leq N$) were $5.7\%$ for CHAD, $3.6\%$ for MdFOCuS, $3.6\%$ for ocd, $5.1\%$ for the detector of \cite{mei2010}, $5.5\%$ for the detector of \cite{xs2013}, and $4.0\%$ for the detector of \cite{chan2017}. 

\begin{figure}
    \centering
    \includegraphics[width=\textwidth]{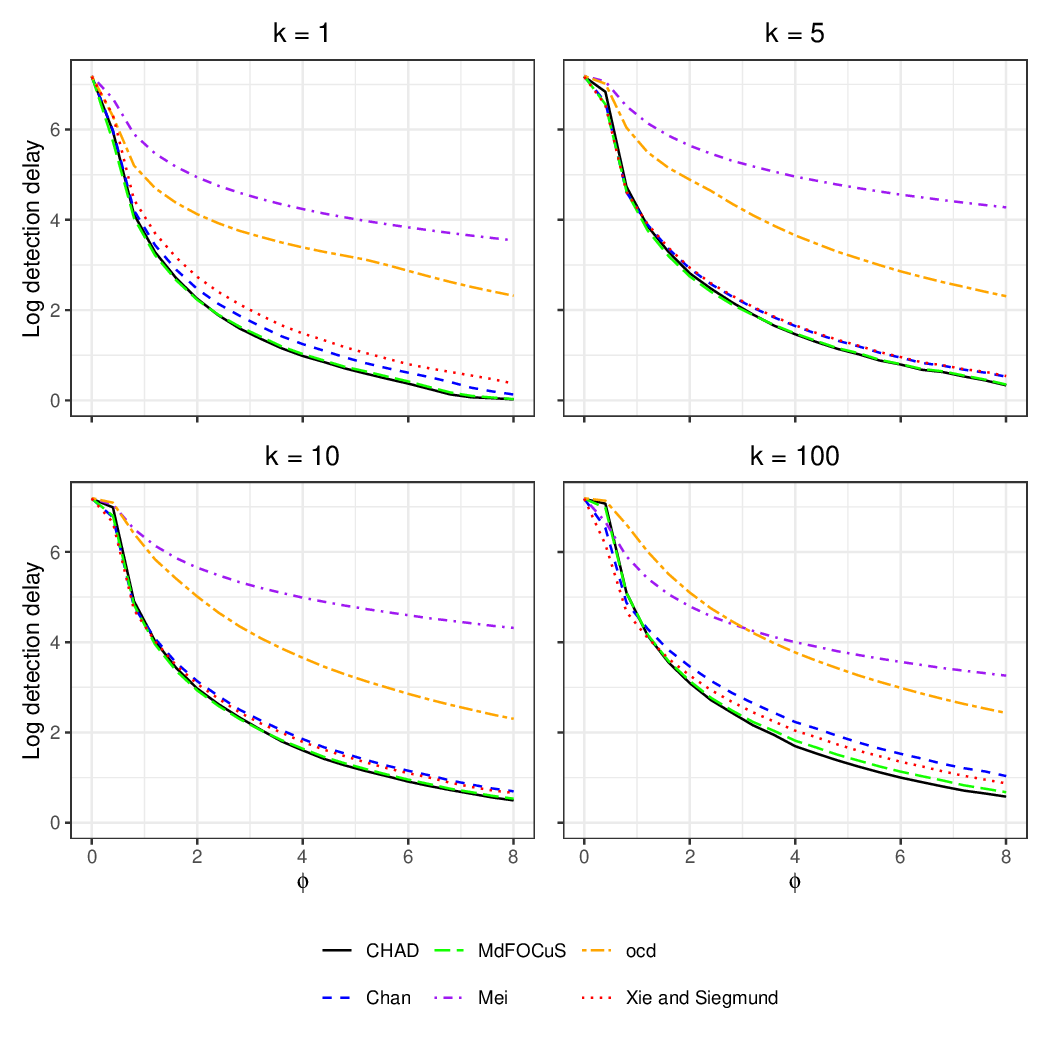}
    \caption{Average detection delays (on log scale) for varying change magnitudes ($\phi$) and changepoint sparsities $k = 1,5,10,100$ over $K=1000$ independent data sets.}
    \label{fig:plot_p=100_combined_log_unknownmean}
\end{figure}

\begin{figure}
    \centering
    \includegraphics[width=\textwidth]{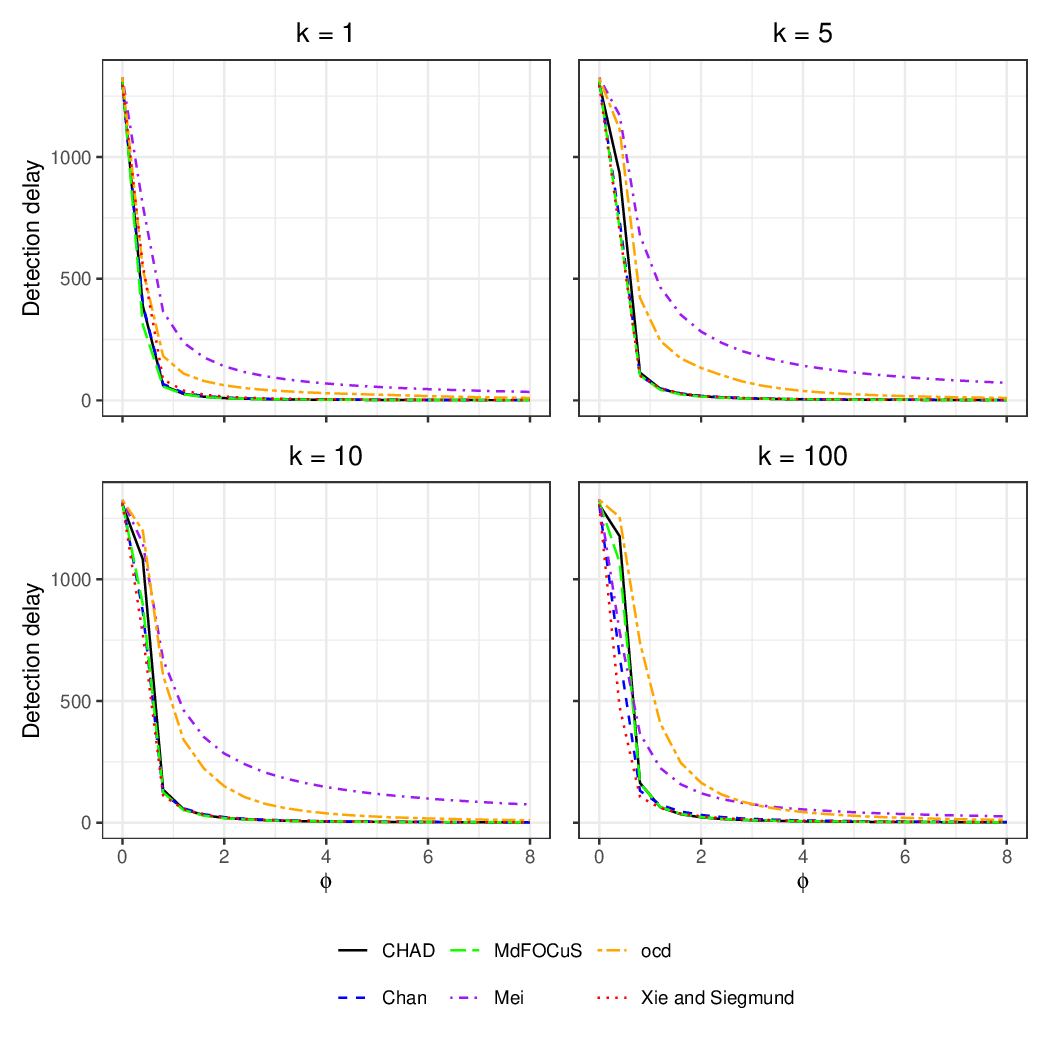}
    \caption{Average detection delays for varying change magnitudes ($\phi$) and changepoint sparsities $k = 1,5,10,100$ over $K=1000$ independent data sets.}
    \label{fig:plot_p=100_combined_unknownmean}
\end{figure}

\newpage

\subsection{The univariate mean-change model with long stretches under the null}\label{sec:univariatesimulation}

In this section, we cover a simulation study that compares the detection delay of the CUSUM-based univariate changepoint detector in \eqref{stoppingtime} with $G^{(t)}$ chosen in \eqref{thegrid} to that of the ``full scan'' $G^{(t)} = [t-1]$. We remark that the latter case corresponds to applying the likelihood-ratio test sequentially with a time-varying critical value. The main goal of the simulation study is to investigate to what extent the dynamic geometric grid reduces detection power or increases detection delay, compared to a full scan. 

Let $\delta = 0.05$ denote the desired false alarm probability, and assume Gaussian data with unit variance, and unknown pre- and post-change means. We consider a stream length of $N = 20 \, 000$ data points. 

In the following, we let 
\begin{align}
    \tauhat_{\mathrm{grid}} &= {\inf} \ \left \{   t\geq 2 \ : \   \max_{g \in G^{(t)}}\ \left(C_g^{(t)} \right)^2 > 1 + \lambda_{\mathrm{grid}}\left [\log\left(\frac{t}{\delta}\right) + \sqrt{\log\left(\frac{t}{\delta}\right)} \right] \right\}, \label{unisimtau1} \\
    \tauhat_{\mathrm{full}} &= {\inf} \ \left \{  t\geq 2 \ :  \ \max_{g \in [t-1]}\ \left(C_g^{(t)} \right)^2 > 1 + \lambda_{\mathrm{full}}\left [\log\left(\frac{t}{\delta}\right) + \sqrt{\log\left(\frac{t}{\delta}\right)} \right]  \right\}, \label{unisimtau2}  
\end{align}
where the CUSUM statistic $C_g^{(t)}$ is defined as in \eqref{cusum} and $G^{(t)}$ is the dynamic geometric grid from \eqref{thegrid}. In \eqref{unisimtau1} and \eqref{unisimtau2}, the critical values have a slightly different functional form than suggested by Theorem \ref{theorem1}. Although still being of the same order, the functional form in \eqref{unisimtau1} and \eqref{unisimtau2} are chosen in accordance with the finite-sample Chi-squared tail bound in \cite{laurentmassart}, which provides more accurate finite-sample upper-tail calibration.

In the simulation study, in order to make the results as comparable as possible, the critical values $\lambda_{\mathrm{grid}}$ and $\lambda_{\mathrm{full}}$ were chosen to be the upper $5\%$ quantiles of
\begin{align}
    &\underset{t=2,3,\ldots, N}{\max} \  \frac{\max_{g \in G^{(t)}}\ \left(C_g^{(t)} \right)^2-1}{\log\left(\frac{t}{\delta}\right) + \sqrt{\log\left(\frac{t}{\delta}\right)}}
    \intertext{and}
    &\underset{t=2,3,\ldots, N}{\max} \  \frac{\max_{g \in [t-1]}\ \left(C_g^{(t)} \right)^2-1}{\log\left(\frac{t}{\delta}\right) + \sqrt{\log\left(\frac{t}{\delta}\right)}}
\end{align}
respectively, using $K = 10\, 000$ Monte Carlo samples under the null of no change with unit variance, with a maximum stream length of $N= 20 \, 000$. For maximal computational efficiency, the changepoint detectors were implemented in Python using the Numba library to compile the code to high-performance code. 

The average detection delay under the alternative was then estimated for the two calibrated changepoint detectors $\tauhat_{\mathrm{grid}}$ and $\tauhat_{\mathrm{full}}$ for a range of values of $\phi= |\mu_1 - \mu_2|$, using $M = 1000$ independent Monte Carlo samples for each value of $\phi$, with changepoint location $\tau$ sampled uniformly from $[5000, 10\ 000]$, and pre-change mean $\mu_1=0$ (taken as unknown). Specifically, for each value of $\phi$ in a piece-wise linearly spaced grid over $[0, 7/2]$, the detection delay $\EE (\ \tauhat \wedge N - \tau \ | \ \tauhat > \tau)$ was estimated for $\tauhat_{\mathrm{grid}}$ and $\tauhat_{\mathrm{full}}$. The resulting detection delays for $\phi\in [0,1]$ are shown in Figure \ref{fig:unimeanzoomed}, as functions of $\phi$. In the figure, the blue curve is the detection delay of $\tauhat_{\mathrm{grid}}$ and the orange curve is that of $\tauhat_{\mathrm{full}}$. A similar plot, with the detection delay on log scale for $\phi\in [0,7/2]$ is shown in Figure \ref{fig:unimeanlog}. 

\begin{figure}
    \centering
    \includegraphics[width=\textwidth]{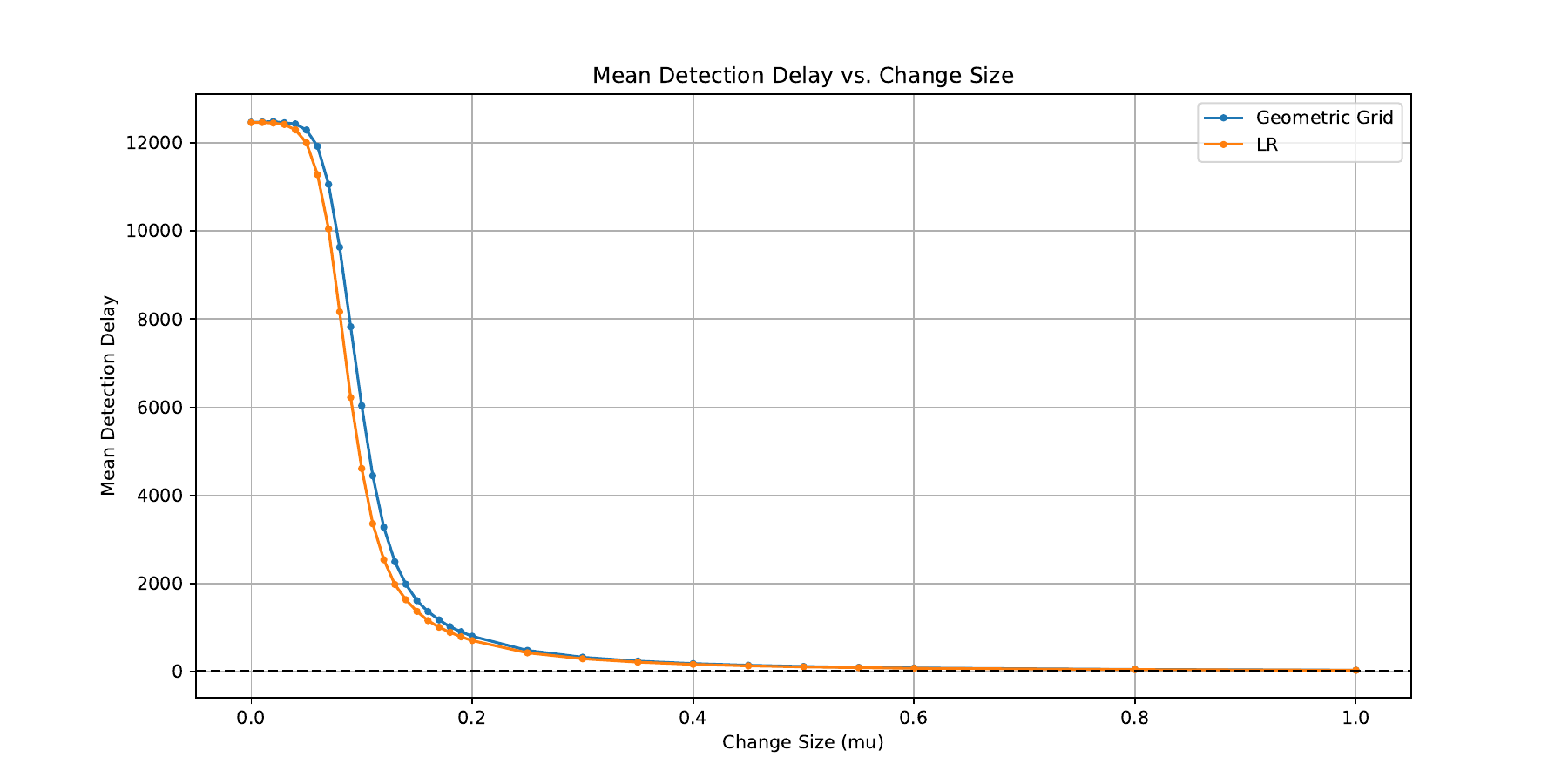}
    \caption{Detection delay of $\tauhat_{\mathrm{grid}}$ (blue) and $\tauhat_{\mathrm{full}}$ (orange) as functions of $\phi=|\mu_1 - \mu_2|$. }
    \label{fig:unimeanzoomed}
\end{figure}
\begin{figure}
    \centering
    \includegraphics[width=\textwidth]{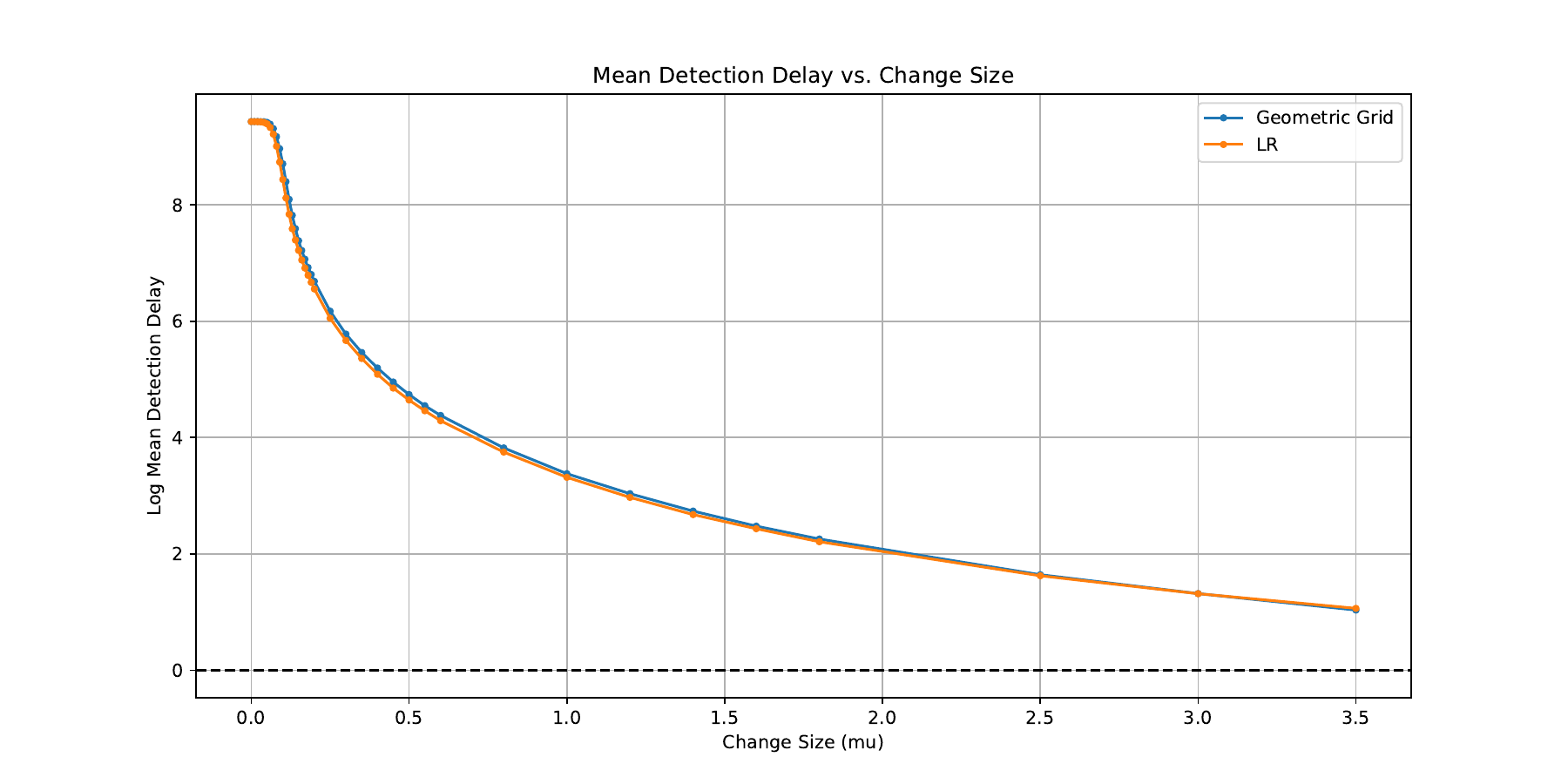}
    \caption{Log detection delay of $\tauhat_{\mathrm{grid}}$ (blue) and $\tauhat_{\mathrm{full}}$ (orange) as functions of $\phi=|\mu_1 - \mu_2|$.}
    \label{fig:unimeanlog}
\end{figure}

Figure \ref{fig:unimeanzoomed} and \ref{fig:unimeanlog} indicate that the sequential likelihood-ratio detector $\tauhat_{\mathrm{full}}$ has a slight edge over the grid-based changepoint detector $\tauhat_{\mathrm{grid}}$ in terms of detection delay, most apparent for small changes. This is perhaps not surprising, as the dynamic geometric grid in \eqref{thegrid} is much sparser than $[t-1]$. Thus, one might expect the sequential likelihood-ratio test to have somewhat smaller detection delay than the grid-based variant. However, the discrepancy between the two detectors is small, and their performance appears largely comparable. Moreover, as illustrated in Figure \ref{fig:unimeanlog}, the detection delays are of the same order as functions of $\phi$. In particular, this indicates that the grid-based changepoint detector does not have any comparative disadvantage for small signal strengths (when many post-change samples are required for detection), in spite of the dynamic geometric grid in \eqref{thegrid} having large spacing between large successive elements (in absolute terms). 

As an aside, we remark that the output of $\widehat{\tau}_{\mathrm{full}}$ could have been computed exactly, with logarithmic update and storage costs, by using the FOCuS algorithm \citep{focus} or MdFOCuS \citep{computationalgeometry}. Since computational performance was not the main interest in this simulation study, however, $\tauhat_{\mathrm{full}}$ was for simplicity computed using the definition in \eqref{unisimtau2} directly.

\subsection{Univariate Poisson data}\label{simulationpoisson}
In this section, we cover a simulation study that compares the detection delay of the grid-based methodology proposed in the main text to a full scan for Poisson data. The model is that $P_1 = \mathrm{Poisson}(\lambda_1)$ and that $P_2 = \mathrm{Poisson}(\lambda_2)$ for some unknown $\lambda_1, \lambda_2>0$, so that
\begin{align}
    Y_i \overset{\mathrm{i.i.d.}}{\sim}\begin{cases}
        \mathrm{Poisson}(\lambda_1), &\text{if} \ i \leq \tau\\
        \mathrm{Poisson}(\lambda_2), &\text{if} \ i > \tau,
    \end{cases}
\end{align}
where we recall that $\tau \in \NN\cup\{\infty\}$ is the unknown changepoint location. Thus, we are interested in detecting a change in the rate of the data. 

Given a sample size $t$ and some $g \in [t-1]$, the likelihood-ratio test for a changepoint at $t-g$ versus no change is given by 
\begin{align}\mathrm{LR}^{(t)}_g = S_{\mathrm{L}}^{(t)}(g)\log\Big(\frac{S_{\mathrm{L}}^{(t)}(g)}{t-g}\Big) + S_{\mathrm{R}}^{(t)}(g)\log\Big(\frac{S_{\mathrm{R}}^{(t)}(g)}{g}\Big) - S^{(t)} \log\Big(\frac{S^{(t)}}{t}\Big),\label{lrpoisson}\end{align}
where
\begin{align}S_{\mathrm{L}}^{(t)}(g) = \sum_{i=1}^{t-g} Y_i,\qquad S_{\mathrm{R}}^{(t)}(g) = \sum_{i=t-g+1}^{t} Y_i,\qquad S^{(t)} = S_{\mathrm{L}}^{(t)}(g) + S_{\mathrm{R}}^{(t)}(g), \end{align}
using the convention that $0\log 0 := 0$. 

Now define the detectors
\begin{align}
    \tauhat_{\mathrm{grid}} &= {\inf} \ \left \{   t\geq 2 \ : \   \max_{g \in G^{(t)}}\ \mathrm{LR}_g^{(t)} >  \lambda_{\mathrm{grid}} \right\}, \label{unipoisstau1} \\
    \tauhat_{\mathrm{full}} &= {\inf} \ \left \{  t\geq 2 \ :  \ \max_{g \in [t-1]}\ \mathrm{LR}_g^{(t)} >  \lambda_{\mathrm{full}}  \right\}. \label{unipoisstau2}  
\end{align}
Here, $\tauhat_{\mathrm{grid}}$ is a changepoint detector of the form as considered in Example \ref{examplecomp2}, which tests for changes using the dynamic geometric grid $G^{(t)}$ in \eqref{thegrid}. As shown in the Example \ref{examplecomp2}, $\tauhat_{\mathrm{grid}}$ can be implemented with logarithmic update and storage costs with respect to the sample size $t$. In contrast, the detector $\tauhat_{\mathrm{full}}$ applies the likelihood-ratio test over all possible change locations. We remark that $\tauhat_{\mathrm{full}}$ can be implemented, with logarithmic storage and update costs, using the methodology of \cite{focus} or \cite{computationalgeometry} \citep[see also][who focus exclusively on the Poisson model]{poissonfocus}.

Under the null ($\tau=\infty$), the distribution of $\mathrm{LR}_g^{(t)}$ does not admit a closed-form distribution, although the asymptotic distribution is Chi-squared. Thus, for simplicity, the detectors $\tauhat_{\mathrm{grid}}$ and $\tauhat_{\mathrm{full}}$ are equipped with constant critical values, which we will shortly choose via Monte Carlo simulations. Moreover, since the finite-sample distribution may depend on the pre-change rate $\lambda_1$, we will for simplicity assume throughout the simulations that $\lambda_1 = 1$, which simplifies the choice of the critical values. We remark that the detectors treat $\lambda_1$ as unknown, although the assumption that $\lambda_1 = 1$ could easily have been incorporated into the likelihood ratio $\mathrm{LR}^{(t)}_g$ in \eqref{lrpoisson}, which would likely result in slightly smaller detection delays. 

Fixing the stream length $N = 10\, 000$,  critical values $\lambda_{\mathrm{grid}}$ and $\lambda_{\mathrm{full}}$ were chosen to be the upper $5\%$ quantiles of
\begin{align}
    &\underset{t=2,3,\ldots, N}{\max} \  \max_{g \in G^{(t)}}\ \mathrm{LR}_g^{(t)}
    \intertext{and}
    &\underset{t=2,3,\ldots, N}{\max} \ \max_{g \in [t-1]}\ \mathrm{LR}_g^{(t)}
\end{align}
respectively, using $K = 10\, 000$ Monte Carlo samples under the null of no change. For maximal computational efficiency, the changepoint detectors were implemented in Python using the Numba library to compile the code to high-performance code. 

The average detection delay under the alternative was then estimated for the two changepoint detectors $\tauhat_{\mathrm{grid}}$ and $\tauhat_{\mathrm{full}}$ for a range of post-change rates $\lambda_2$, using $M = 1000$ independent Monte Carlo samples for each value of $\lambda_2$, with changepoint location $\tau$ sampled uniformly from $[2500, 5000]$, and pre-change rate $\lambda_1=1$. Specifically, for each value of $\lambda_2$ in a piece-wise linearly spaced grid over $[0, 5]$, the detection delay $\EE (\ \tauhat \wedge N - \tau \ | \ \tauhat > \tau)$ was estimated for $\tauhat_{\mathrm{grid}}$ and $\tauhat_{\mathrm{full}}$. The resulting detection delays are shown in Figure \ref{fig:unipoiss}, as functions of $\lambda_2$. In the figure, the blue curve is the detection delay of $\tauhat_{\mathrm{grid}}$ and the orange curve is that of $\tauhat_{\mathrm{full}}$. A similar plot, with the detection delay presented on a log scale, is shown in Figure \ref{fig:unipoisslog}. 

\begin{figure}[h!]
    \centering
    \includegraphics[width=\textwidth]{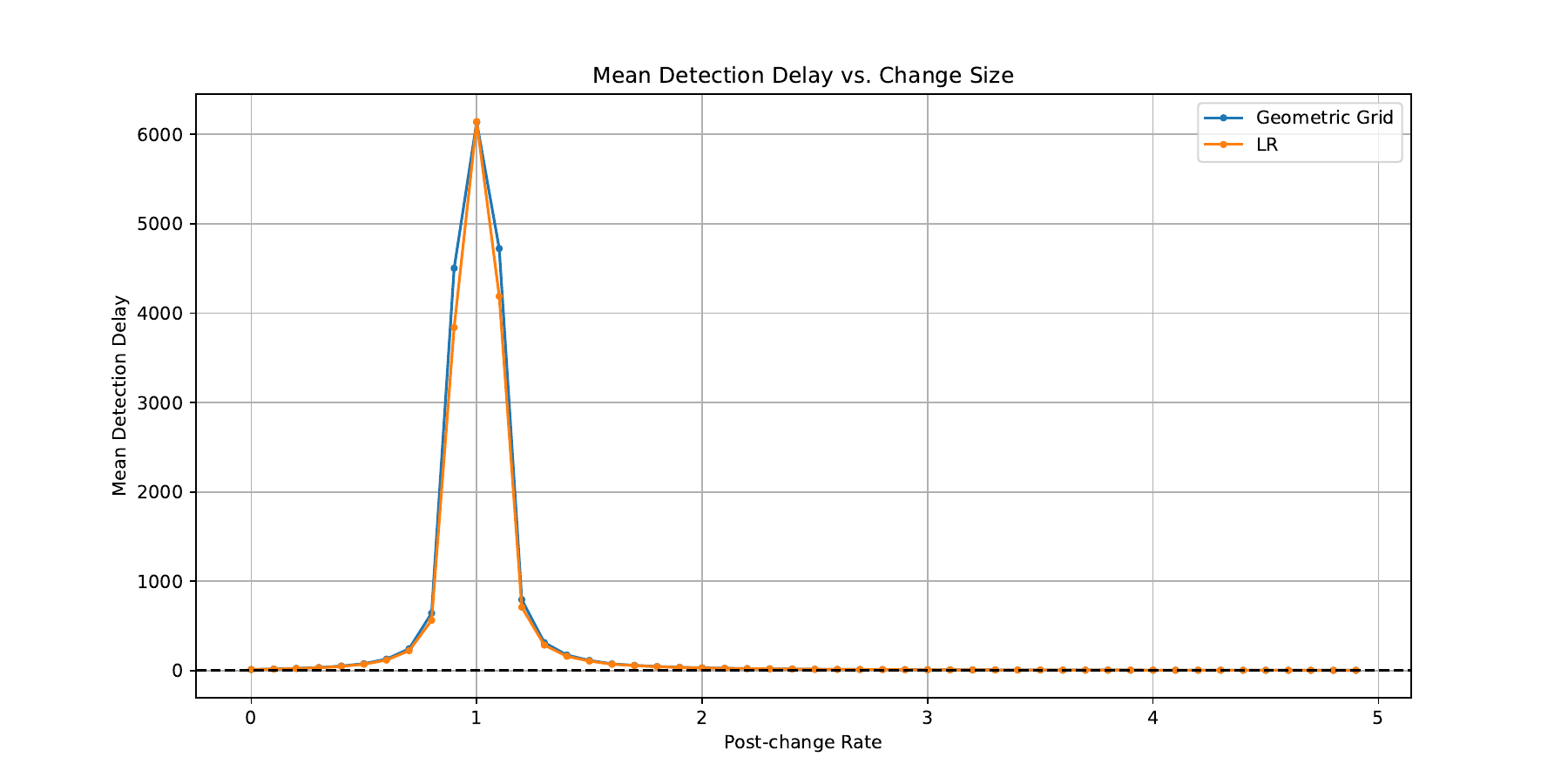}
    \caption{Detection delay of $\tauhat_{\mathrm{grid}}$ (blue) and $\tauhat_{\mathrm{full}}$ (orange) as functions of the post-change rate $\lambda_2$. }
    \label{fig:unipoiss}
\end{figure}
\begin{figure}[h!]
    \centering
    \includegraphics[width=\textwidth]{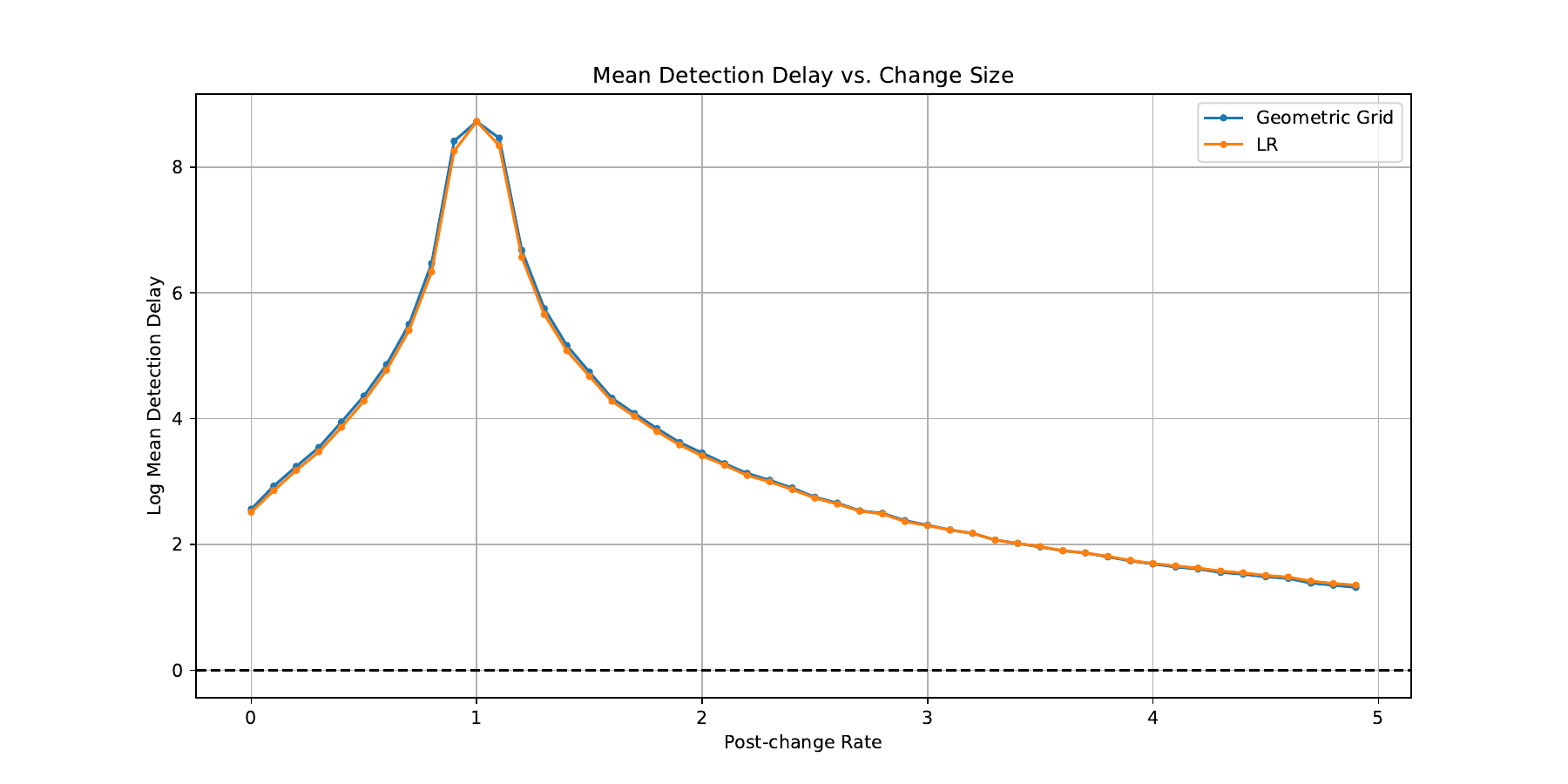}
    \caption{Log detection delay of $\tauhat_{\mathrm{grid}}$ (blue) and $\tauhat_{\mathrm{full}}$ (orange) as functions of the post-change rate $\lambda_2$.}
    \label{fig:unipoisslog}
\end{figure}

Figure \ref{fig:unipoiss} and \ref{fig:unipoisslog} indicate that the sequential likelihood-ratio detector $\tauhat_{\mathrm{full}}$ has a slight edge over the grid-based changepoint detector $\tauhat_{\mathrm{grid}}$ in terms of detection delay, but only for weak signals. However, the discrepancy between the two detectors is close to negligible, and their performances appear largely comparable. In particular, Figure \ref{fig:unipoisslog} suggests that the detection delays are of the same order for both weak and strong signals, which indicates that the grid-based changepoint detector does not have any comparative disadvantage for small signal strengths. 
\leavevmode\newpage
\leavevmode\newpage
\subsection{Real data example with heavy-tail calibration}\label{heavytaildataexample}
In the real data example in Section \ref{sec:realdata}, the changepoint detector from Section \ref{sec:covariance} was calibrated via Monte Carlo simulations under the assumption of Gaussian data. Although there were long periods at which the detector did not raise any alarms (e.g., between the beginning of 2000 until October 2008), this assumption may still be optimistic. The real data example was therefore also performed using a more conservative critical value, obtained via similar Monte Carlo simulations under the assumption that the $Y_i$ were drawn independently from the $t$ distribution with $5$ degrees of freedom. With this more conservative choice of critical value, the changepoint detector only declared two changepoints, namely at 23 October 2008 and 29 December 2009. These two dates are indicated by vertical dashed lines in Figure \ref{fig:currencies_heavytail}. Again, we emphasise that the dates at which a changepoint was detected are not necessarily good estimates of when the changes may have occurred. 

\begin{figure}[h!]
    \centering
    \includegraphics[width=\textwidth]{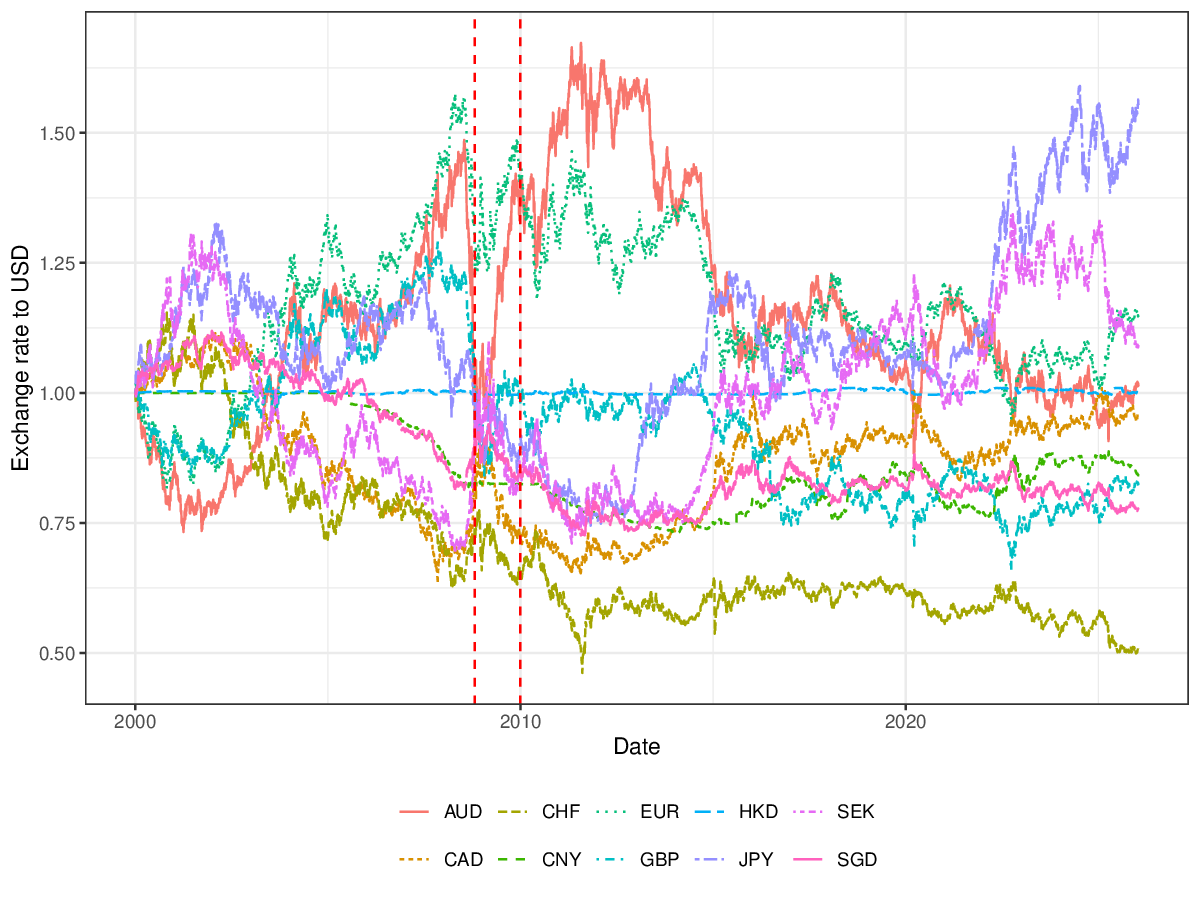}
    \caption{Normalised values of one US Dollar in terms of ten selected currencies from 3 January 2000 to 16 January 2026. Times at which changepoints were detected are indicated by dashed red vertical lines. }
    \label{fig:currencies_heavytail}
\end{figure}

\leavevmode\newpage

\subsection{Overview of code}\label{sec:overviewcode}
The source code for all simulation studies and the real data example can be found in the \textit{inst/} folder in the R package CHAD \cite{CHAD}, available at:\newline \href{https://github.com/peraugustmoen/CHAD}{https://github.com/peraugustmoen/CHAD}. In the following, we give a brief description of the contents. 
\begin{itemize}
    \item \textit{simulation\_study\_main.R}: Contains the source code for the simulation study covered in Section \ref{sec:simulations_statistics}. 
    \item \textit{simulation\_study\_runtime.R}: Contains the source code for the simulation study covered in Section \ref{sec:simulationcomp}.
        \item \textit{simulation\_study\_runtime\_only\_CHAD.R}: Contains the source code for the isolated CHAD run-time study covered in Section \ref{sec:runtimeonlychad}.
    \item \textit{data\_example.R}: Contains the source code for the real data example covered in Section \ref{sec:realdata} and Section \ref{heavytaildataexample}.
        \item \textit{extdata/exchange\_rates.csv}: The exchange-rate data set analysed in the real data example (Section \ref{sec:realdata}).
    \item \textit{simulation\_study\_p=10\_N=5000.R}: Contains the source code for the simulation study covered in Section \ref{sec:smallplargeN}. 
    \item \textit{simulation\_study\_p=1000\_N=200.R}: Contains the source code for the simulation study covered in Section \ref{sec:largepsmallN}. 
        \item \textit{simulation\_study\_nonzero.R}: Contains the source code for the simulation study covered in Section \ref{simnonzeromean}. 
    \item \textit{pythoncode/univariate.ipynb}: Contains the source code for the simulation study covered in Section \ref{sec:univariatesimulation}. Note that this is a Jupyter Notebook file. 
    \item \textit{pythoncode/poisson\_rates.ipynb}: Contains the source code for the simulation study covered in Section \ref{simulationpoisson}. Note that this is a Jupyter Notebook file. 
\end{itemize}

\section{Proofs of results in the main text}\label{sec:proofs}

\subsection{Proof of Lemma \ref{gridlemma}}
\begin{proof}[Proof.]

    We begin by showing Claim \ref{gridlemmaclaim1}. If $t \leq 3$, the claim holds trivially, so we can assume that $t\geq 4$. Note that $g_{\text{L},1} / 1 = 2$, and 
    \begin{align}
        1 < \frac{g_{\text{L},j}^{(t)}}{g_{\text{R},j-1}^{(t)}} &= \frac{2^j + \left\{(t-1)\bmod 2^{j-1}\right\}}{3\cdot 2^{j-2} + \left\{(t-1)\bmod2^{j-2}\right\}} \leq \frac{2^j + 2^{j-1}}{3 \cdot 2^j /4}=2
        \intertext{for $2\leq j \leq \left \lfloor \log_2\left\{ (t-1)/3\right\} \right\rfloor+1$, and }
        1 < \frac{g_{\text{R},j}^{(t)}}{g_{\text{L},j}^{(t)}} &= \frac{3 \cdot 2^{j-1} + (t-1)\bmod 2^{j-1}}{2^{j} + (t-1) \bmod 2^{j-1}} \leq \frac{3\cdot 2^{j-1} + 2^{j-1}}{2^j} = 2,
    \end{align}
    for $j = 1, \ldots, \left \lfloor \log_2 (t-1)\right\rfloor -1$. Hence, writing $G^{(t)} = g_1, \ldots, g_m$ as an increasing enumeration of $G^{(t)}$, we have $g_{i+1}/g_i \leq 2$ for all $i = 1, \ldots, m-1$. To show that Claim \ref{gridlemmaclaim1} holds, it therefore suffices to show that $\max \ G^{(t)} \geq t/2$ for all $t$. To this end, we will for each $j \in \NN$ show that $\max \ G^{(t)} \geq t/2$ for all $t \in [t_j+1, t_{j+1}]$, where $t_j = 3 \cdot 2^{j-1}$, which then implies \ref{gridlemmaclaim1}. 

    We partition the closed integer interval $\left[ t_j+1, t_{j+1} \right]$ into three disjoint sub-intervals, namely 
    \begin{align}
        \left[t_j+1, (4/3)t_j \right],
        \left[(4/3)t_j+1, (5/3)t_j \right ],
        \left[(5/3)t_j +1, t_{j+1} \right ].
    \end{align}
    Since $\max \ G^{(t)} - t/2$ is increasing in $t$ on each of these intervals, it suffices to show that $\max G^{(t)} \geq t/2$ for $t = t_j+1, (4/3)t_j +1, (5/3)t_j +1$. For $t = t_j +1= 3 \cdot 2^{j-1} + 1$, the largest element of $G^{(t)}$ is at least
    \begin{align}
        \max \ G^{(t)} &\geq g_{\text{L}, j}^{(t)} \\
        &\geq \frac{2}{3}(t-1) \\
        &\geq t/2,
    \end{align}
    where we in the last line used that $t \geq 4$ since $j\geq 1$. For $t = (4/3)t_j+1 =  2^{j+1} + 1$, the largest element of $G^{(t)}$ is at least
    \begin{align}
        \max \ G^{(t)} &\geq  g_{\text{R}, j}^{(t)}\\
        &= \frac{3}{4} (t-1) \\
        &\geq t/2,
    \end{align}
    where we in the last line also used that $t\geq 4$.  Lastly, for $t = (5/3) t_j+1 = 5 \cdot 2^{j-1}+1$, the largest element of $G^{(t)}$ is at least
    \begin{align}
        \max \ G^{(t)} \geq g_{\text{R}, j}^{(t)} &= 3\cdot 2^{j-1}\\
        &= \frac{6}{2} 2^{j-1}\\
        &\geq \frac{5}{2} 2^{j-1} + \frac{1}{2}\\
        &= t/2.
    \end{align}
    Hence $\max \ G^{(t)} \geq t/2$ for each $t \in \left[ t_j+1, t_{j+1}\right]$, and Claim \ref{gridlemmaclaim1} is proved.  

    To see that Claim \ref{gridlemmalclaim2} holds, note first that $|G^{(2)}|=1<3\log2$, $|G^{(3)}|=1<3\log 3$ and $|G^{(4)}|=2< 3\log4$, so the claim holds for $t=2,3,4$. Next, for $t\geq 5$, we have that
    \begin{align}
        |G^{(t)}| &= 1 + \left \lfloor \log_2\{(t-1)/3\} \right \rfloor + \left\lfloor \log_2(t-1)\right\rfloor\\
        &\leq 1 + \log_2\{(t-1)/3\}+  \log_2(t-1) \\
        &= 1 - \log_2(3) + 2\log_2(t-1)\\
        &< 2\log_2(t-1)\\
        &< \frac{2}{\log 2} \log  t\\
        &< 3\log t. 
    \end{align}
    
    To prove Claim \ref{gridlemmaclaim3}, we will show that
    \begin{align}
        G^{(t)}\setminus\{1\} -1 \subseteq G^{(t-1)},
    \end{align}
    for $t \geq 5$, as the claim can easily be verified for $t=3,4$. We have
    \begin{align}
        g_{\text{L},j}^{(t)} -1 &= \begin{cases}
            g_{\text{L},j}^{(t-1)}, &\text{if } (t-1) \bmod 2^{j-1} >0\\
            g_{\text{R},j-1}^{(t-1)}, &\text{otherwise, }
        \end{cases}
         \intertext{for $2\leq j \leq \left \lfloor \log_2\left\{ (t-1)/3\right\} \right\rfloor+1$, and }
        g_{\text{R},j}^{(t)} -1 &= \begin{cases}
            g_{\text{R},j}^{(t-1)}, &\text{if } (t-1) \bmod 2^{j-1} >0\\
            g_{\text{L},j}^{(t-1)}, &\text{otherwise, }
        \end{cases}
    \end{align}
    for $j = 1, \ldots, \left \lfloor \log_2 (t-1)\right\rfloor -1$. Since also $g_{\mathrm{L},1}^{(t)}-1 = 1$, it follows that $g_{\text{L},j}^{(t)}-1, g_{\text{R},j}^{(t)}-1 \in G^{(t-1)}$ for all $t \geq 4$ and all relevant $j$, and Claim \ref{gridlemmaclaim3} is proved. 
\end{proof}
\subsection{Proof of Theorem \ref{theorem1}}
\begin{proof}[Proof.]
We begin by showing that $\mathrm{FA}(\tauhat) \leq \delta$ whenever $\lambda\geq C_1$, for some suitable absolute constant $C_1>0$. Assume that $\tau = \infty$. For any $t\geq2$ and $g \in [t-1]$, we have $C^{(t)}_g = \sum_{i=1}^t a_i Y_i$, where $a_i = g^{1/2} \left\{ t(t-g)\right\}^{-1/2}$ for $i \leq t-g$ and $a_i = -(t-g)^{1/2} (tg)^{-1/2}$ for $i> t-g$. Since $\sum_{i=1}^t a_i = 0$, we may equivalently write $C^{(t)}_g = \sum_{i=1}^t a_i Z_i$, where the $Z_i = Y_i - \mu_1$ are mean-zero sub-Gaussian random variables satisfying $\normmo{Z_i} \leq \sigma$. Theorem 2.6.3 in \citet{Vershynin_2018} in combination with a union bound then implies that
\begin{align}
\PP_{\infty}\left\{  \underset{1 \leq g \leq t-1}{\max} \ \left (C_g^{(t)}  \right)^2  > \lambda \sigma^2 \log(t/\delta) \right\} \leq \delta / t^2,\label{firstbound}
\end{align}
for all $\lambda\geq C_1$, where $C_1>0$ is some absolute constant. Since $\xi^{(t)} = \lambda \sigma^2 \log(t/\delta)$, a union bound over the $t$'s then yields
\begin{align}
     \PP_{\infty}\left( \widehat{\tau} < \infty \right)
     &= \PP_{\infty} \left( \bigcup_{t=2}^{\infty} \left\{ \widehat{\tau} = t \right\} \right) \\
     &\leq  \sum_{t=2}^{\infty} \PP_{\infty} \left\{  \underset{1 \leq g \leq t-1}{\max} \ \left (C_g^{(t)}  \right)^2 > \xi^{(t)} \right\}\\
    &\leq \delta \sum_{t=2}^{\infty} \frac{1}{t^2} < \delta,
\end{align}
which shows that $\mathrm{FA}(\tauhat)\leq \delta$. 

Now consider the detection delay, and assume that $\tau < \infty$. We shall prove that $\PP_{\tau}\left(\widehat{\tau} \leq \tau+d\right)\geq 1-\delta$, where
\begin{align}
    d = \left \lceil   C_2 \frac{\sigma^2 \log (2\tau/\delta)  }{\phi^2} \right \rceil,
\end{align}
and $C_2 = 48 \lambda$, under the assumption that $\tau \phi^2 / \sigma^2 \geq 2C_2 \log(2\tau/\delta)$. Set $t = \tau + d$, and note that $d \leq \tau$. By Lemma \ref{gridlemma} there exists some $g_0 \in G^{(t)}$ such that $d/2 \vee 1 \leq g_0 \leq d$. By the linearity of the CUSUM transformation, we can write
\begin{align}
    \left(C_{g_0}^{(t)}\right)^2 = \left( \theta_{g_0}^{(t)} + \sum_{i=1}^t a_i Z_i \right)^2, 
\end{align}
where $Z_i$ are mean zero sub-Gaussian variables with $\normmo{Z_i}\leq \sigma$ and $\theta^{(t)}_{g_0}$ is the CUSUM transformation \eqref{cusum} applied to the sequence of the first $t$ true means, evaluated at ${g_0}$. By a similar bound as in \eqref{firstbound}, using $\xi^{(t)} = \lambda \sigma^2 \log(t/\delta)$, we have that the event
\begin{align}
    \mathcal{E}_1 = \left\{  \left(C_{g_0}^{(t)}  \right)^2 - \xi^{(t)} \geq \left(\theta_{g_0}^{(t)}\right)^2 - 2\sqrt{\xi^{(t)}} \left| \theta_{g_0}^{(t)}\right| - \xi^{(t)} \right\}
\end{align}
has probability at least $\PP_{\tau}\left( \mathcal{E}_1\right) \geq 1-\delta$. On the event $\mathcal{E}_1$, by solving a quadratic inequality, we will have that $T^{(t)}_{g_0}$ defined in \eqref{cusumtest} satisfies $T^{(t)}_{g_0} > 0$ whenever 
\begin{align}
    \left(\theta_{g_0}^{(t)}\right)^2 \geq 6 \xi^{(t)}.
\end{align}
Due to the assumption $\tau \phi^2 /\sigma^2 \geq 2C_2 \log(2\tau/\delta)$, we have $d \leq \tau$ and $t = \tau + d \leq 2\tau$. By Lemma \ref{cusumlemma}, we have that 
\begin{align}
    \left(\theta^{(t)}_{g_0} \right)^2  &= \frac{{g_0}}{t(t-{g_0})}\tau^2 \phi^2\\
    &\geq \frac{{g_0}\phi^2}{4}.
\end{align}
Since ${g_0}$ was chosen so that ${g_0}\geq d/2$, we will have $T^{(t)}_{g_0} > 0$ on the event $\mathcal{E}_1$ as long as 
\begin{align}
    d &\geq \frac{48 \xi^{(t)}  }{\phi^2}\\
    &= 48 \lambda \frac{\sigma^2 \log (t/\delta)  }{\phi^2}.
\end{align}
In particular, since $t \leq 2\tau$, we have $\log(t/\delta)\leq \log(2\tau /\delta)$, and so $T^{(t)}_{g_0} > 0$ holds on $\mathcal{E}_1$ whenever $d \geq 48 \lambda \sigma^2 \log (2\tau/\delta)   / \phi^2$, which holds true by the definition of $d$. Hence, $\PP_{\tau}\left(\widehat{\tau} \leq \tau+d\right)\geq \PP_{\tau}(\mathcal{E}_1) \geq 1-\delta$.

Lastly, to show that $\mathrm{SC}(\tauhat, t) = \mathrm{UC}(\tauhat, t)= \mathcal{O}(\log t)$, note that
\begin{align}
    C_g^{(t)} =  g^{1/2}\left\{t(t-g)\right\}^{-1/2} \ \sum_{i=1}^{t-g} Y_i - (t-g)^{1/2}\left(tg\right)^{-1/2} \ \left( \sum_{i=1}^{t} Y_i -\sum_{i=1}^{t-g} Y_i \right)\label{cgeq}.
\end{align}
Thus, the test $T^{(t)}_g$ in \eqref{cusumtest} may be written as 
\begin{align}
    T^{(t)}_g = f^{(t)}_g \left( \sum_{i=1}^{t-g} h(Y_i), \sum_{i=t-g+1}^t h(Y_i) \right), 
\end{align}
for any $t\geq 2$ and $g \in G^{(t)}$, where the number of unit‑cost operations required to compute $f_g^{(t)}$ is of constant order with respect to $t$ and $g$.  Proposition \ref{enkelprop} then implies that $\mathrm{SC}(\tauhat, t)$ and $\mathrm{UC}(\tauhat, t)$ are of order $\mathcal{O}(\log t)$ for any $t\geq 2$.
\end{proof}

\subsection{Proof of Proposition \ref{theorem2}}
\begin{proof}[Proof.]
To prove Proposition \ref{theorem2}, note first that we may store data in memory as follows: 
\begin{itemize}
    \item At time $t=1$, we store only $S_1^{(1)}$;
    \item At time $t= 2$, we store $\{S_g^{(2)}\}_{g \in G^{(2)}}$, $S_1^{(1)}$ and $Y_2$. 
    \item At time $t\geq 3$, we store $\{S_g^{(t)}\}_{g \in G^{(t)}}$, $\{S_{g}^{(t-1)}\}_{g \in G^{(t-1)}}$ and $Y_t$. 
\end{itemize}
Due to Assumption \ref{asscomp-b}, the collection $\{S_g^{(t)}\}_{g \in G^{(t)}}$ can be stored at machine precision using at most $\mathcal{O}\{M_S(p) |G^{(t)}|\} = \mathcal{O}\{M_S(p) \log t\}$ scalar values due to Lemma \ref{gridlemma}. Moreover, a single $Y_t$ may be stored at machine precision using at most $\mathcal{O}(p)$ scalar values. With the above data-storage regime, it follows immediately that  $\mathrm{SC}(\tauhat,t) = \mathcal{O}\{ p + M_S(p)\log t \}$. 

Next, we shall show that $\mathrm{UC}(\tauhat, t) =\mathcal{O}[\{C_T(p) + C_S(p)\}\log t ]$. To show this, consider first the case where $t=2$. Since $G^{(2)} = \{1\}$, we may thus compute $S_1^{(2)}$ from $S_1^{(1)}$ and $Y_2$ using at most $C_S(p)$ unit-cost operations due to Assumption \ref{asscomp-b}. Note also that $S_1^{(1)}$ and $Y_2$ are already stored in memory. Moreover, we may compute $T^{(2)}_1$ from $S_1^{(2)}$ using at most $C_T(p)$ unit cost operations due to Assumption \ref{asscomp-a}. Therefore, $\mathrm{UC}(\tauhat,2)=\mathcal{O}\{C_S(p)+C_T(p)\}$. Next, at time $t>2$, we may compute $S_g^{(t)}$ from $S_{(g-1)\vee 1}^{(t-1)}$ using at most $C_S(p)$ unit-cost operations, for each $g \in G^{(t)}$, due to Assumption \ref{asscomp-b}. Note that we always have $\{S_{(g-1)\vee 1}^{(t-1)}\}_{g \in G^{(t)}} \subseteq \{S_g^{(t-1)}\}_{g \in G^{(t-1)}}$ due to Lemma \ref{gridlemma}, and thus all necessary quantities to compute the $S_g^{(t)}$ are present in memory. Due to Assumption \ref{asscomp-a}, we may then compute $T^{(t)}_g$ using at most $C_T(p)$ unit-cost operations, for each $g \in G^{(t)}$. Since $|G^{(t)}| = \mathcal{O}(\log t)$ due to Lemma \ref{gridlemma}, and since $G^{(t)}$ may be computed using at most $\mathcal{O}(|G^{(t)}|) = \mathcal{O}(\log t)$ unit-cost operations, we obtain $\mathrm{UC}(\tauhat,t) = \mathcal{O}[\{C_S(p)+C_T(p)\}\log t]$. 

\end{proof}

\subsection{Proof of Proposition \ref{enkelprop}}
\begin{proof}[Proof.]
To prove Proposition \ref{enkelprop}, we will show that Assumption \ref{asscomp} holds and argue in a similar fashion as demonstrated in Example \ref{compexample1}. 

We begin by defining the summary statistic $S_g^{(t)}$ similarly as in Example \ref{compexample1}. Define
\begin{align}
    S_1^{(1)} = \{h(Y_1)\}, 
    \intertext{and, for $t\geq 2$, define}
    S_g^{(t)} = \left \{\sum_{i=1}^{t-g}h(Y_i), \,\sum_{i=1}^t h(Y_i)\right\},
\end{align}
for $g \in G^{(t)}$. 
We first claim that, for any $t\geq 2$ and $g\in G^{(t)}$, the test $T^{(t)}_g$ may be computed from $S_g^{(t)}$ using at most $\mathcal{O}(v(p)) + C_f(p)$ unit-cost operations. To see this, note that the first argument of $T^{(t)}_g$ is contained in $S_g^{(t)}$, and that the second argument can be written as
    $$
    \sum_{i=t-g+1}^t h(Y_i) = \sum_{i=1}^t h(Y_i) - \sum_{i=1}^{t-g}h(Y_i),
    $$
    and can thus be computed from $S_g^{(t)}$ using $\mathcal{O}\{v(p)\}$ unit-cost operations. Therefore,  $T^{(t)}_g$ can be computed from $S_g^{(t)}$ using at most $\mathcal{O}\{v(p)\} + C_f(p)$ unit-cost operations. Thus, Assumption \ref{asscomp-a} is satisfied with $C_T(p) = \mathcal{O}(v(p)) + C_f(p)$. 

    Next, we show that Assumption \ref{asscomp-b} holds. To this end, since we have that $h(Y_i) \in \RR^{v}$, storing each $S_g^{(t)}$ entails storing an order of $\mathcal{O}\{v(p)\}$ scalars. Moreover, for $t \geq 2$, we have that 
\begin{align}
\sum_{i=1}^{t-g}h(Y_i) \in S^{(t-1)}_{(g-1)\vee 1},
\end{align}
so that the first element in $S_g^{(t)}$ may simply be retrieved from $S_{(g-1)\vee 1}^{(t-1)}$. Moreover, we have that
\begin{align}
\sum_{i=1}^t h(Y_i) = \sum_{i=1}^{t-1}h(Y_i) + h(Y_t),\end{align} where $\sum_{i=1}^{t-1}h(Y_i) \in S_{(g-1)\vee 1}^{(t-1)}$ and $h(Y_t)$ can be computed from $Y_t$ using at most $C_h(p)$ unit-cost operations. Therefore, $S_g^{(t)}$ may be computed from $S_{(g-1)\vee 1}^{(t-1)}$ and $Y_t$ using at most $\mathcal{O}\{v(p)\} + C_h(p)$ unit-cost operations. Assumption \ref{asscomp-b} is therefore satisfied with $M_S(p) = \mathcal{O}\{v(p)\}$ and $C_S(p) = \mathcal{O}\{v(p)\} + C_h(p)$. Proposition \ref{theorem2} then yields that $\mathrm{UC}(\tauhat, t) = \mathcal{O}[\{v(p) + C_h(p) + C_f(p)\}\log t]$ and $\mathrm{SC}(\tauhat, t) = \mathcal{O}\{p + v(p)\log t\}$.

\end{proof}
\subsection{Proof of Proposition \ref{generalddprop}}
\begin{proof}[Proof.]
    Note first that
    \begin{align}
        \mathrm{FA}(\tauhat) &\leq \sum_{t=2}^{\infty}\sum_{g \in G^{(t)}} \PP_{\infty}(T^{(t)}_g = 1)\\
        &\leq \delta \sum_{t=2}^{\infty} \frac{|G^{(t)}|}{3t^2 \log t}\\
        &\leq \delta \sum_{t=2}^{\infty} t^{-2}\\
        &<\delta,
    \end{align}
since $|G^{(t)}|<3\log t$, due to Lemma \ref{gridlemma}.
Next, suppose there is a changepoint at time $\tau$ that satisfies
\begin{align}
    \tau \rho(P_1, P_2) \geq r (2\tau, \delta). \label{gencond1}
\end{align}
Define
\begin{align}
    d = \left \lceil \frac{r(2\tau, \delta)}{\rho(P_1, P_2)}\right\rceil. 
\end{align}
Then, due to \eqref{gencond1}, we must have that $d\leq \tau$ since
\begin{align}
    \left \lceil \frac{ r(2\tau,\delta)}{\rho(P_1, P_2)} \right\rceil \leq \tau.
\end{align}
Thus, setting $t = \tau +d$, we have
\begin{align}
    (t-\tau) \rho(P_1, P_2) &= d \rho(P_1, P_2)\\
    &\geq r(2\tau, \delta)\\
    &\geq r(t, \delta),
\end{align}
where the last inequality follows from $t = \tau+d\leq 2\tau$. 
Due to Lemma \ref{gridlemma}, we may choose a $g \in G^{(t)}$ such that $(t-\tau)/2\leq g \leq t-\tau$. For this $g$, due to Assumption \ref{assstat-b}, we therefore have
that
\begin{align}
    \PP_{\tau}(T^{(t)}_g=1)\geq 1-\delta,
\end{align}
and thus
\begin{align}
\PP_{\tau}\left\{\tauhat \leq \tau+  \left \lceil \frac{r(2\tau, \delta)}{\rho(P_1, P_2)}\right\rceil\right\}
&= \PP_{\tau}(\tauhat\leq \tau+d) \\&= \PP_{\tau}(\tauhat \leq t)\\
&\geq \PP_{\tau}(T^{(t)}_g=1)\\
    &\geq 1-\delta,
\end{align}
and we are done.
\end{proof}
\subsection{Proof of Proposition \ref{example1statprop}}
\begin{proof}[Proof.] We will prove Proposition \ref{example1statprop} by appealing to Proposition \ref{generalddprop}. 
To this end, we only need to show that $T_g^{(t)}$ in \eqref{testexample1} satisfies Assumption \ref{assstat}. Consider first Assumption \ref{assstat-a}, and assume that $\tau = \infty$, and fix any $t\geq 2$. For any $g \in [t-1]$, we have
\begin{align}
    \normm{\widehat{\theta}_{1,g}^{(t)} - \widehat{\theta}_{2,g}^{(t)}} \leq \normm{\frac{1}{t-g}\sum_{i=1}^{t-g}\{h(Y_i) - \theta(P_1)\}} + \normm{\frac{1}{g}\sum_{i=t-g+1}^{t}\{h(Y_i) - \theta(P_1)\}}. 
\end{align}
Applying \eqref{concentrationass} twice, with $n = t-g$ and $n = g$ respectively (both of which are at most $t$), a union bound implies that
\begin{align}
    \PP_{\infty}\left( \normm{\widehat{\theta}_{1,g}^{(t)} - \widehat{\theta}_{2,g}^{(t)}}  > 2 \frac{\xi(t,\delta)}{\sqrt{g\wedge(t-g)}} \right) \leq \frac{2\delta}{6 t^2\log t} = \frac{\delta}{3 t^2\log t},
\end{align}
since $\xi(t,\delta)/\sqrt{t-g} + \xi(t,\delta)/\sqrt{g} \leq 2\xi(t,\delta)/\sqrt{g\wedge(t-g)}$.
Therefore, Assumption \ref{assstat-a} holds.

    Next we will show that \ref{assstat-b} holds with $\rho(P_1, P_2) = \normm{\theta(P_1)-\theta(P_2)}^2$ and $r(t,\delta) = 88\xi^2(t,\delta)$. To this end, assume that $\tau<\infty$ and let $\tau < t\leq 2\tau$ be  such that $(t-\tau)\rho(P_1, P_2)\geq r(t,\delta)$. Let $g$ be such that $(t-\tau)/2 \leq g \leq t-\tau$. We shall show that $\PP_{\tau}(T^{(t)}_g=1)\geq 1-\delta$. 

     Note first that
    \begin{align}
        \widehat{\theta}_{1,g}^{(t)} - \widehat{\theta}_{2,g}^{(t)} &= \frac{1}{t-g}\sum_{i=1}^{t-g}h(Y_i) - \frac{1}{g} \sum_{i=t-g+1}^t h(Y_i)\\
        &= \frac{1}{t-g} \sum_{i=1}^{\tau}h(Y_i) + \frac{1}{t-g} \sum_{i=\tau+1}^{t-g}h(Y_i) - \frac{1}{g}\sum_{i=t-g+1}^{t} h(Y_i)\\
        &= \frac{1}{t-g}  \sum_{i=1}^{\tau}\{h(Y_i) - \theta(P_1) \}+  \frac{\tau}{t-g}\theta(P_1)\\
        &+ \frac{1}{t-g}  \sum_{i=\tau+1}^{t-g}\{h(Y_i) - \theta(P_2)\} + \frac{t-g-\tau}{t-g}\theta(P_2)\\
        &- \frac{1}{g}\sum_{i=t-g+1}^{t} \{h(Y_i) - \theta(P_2)\} - \theta(P_2).
    \end{align}
Due to the reverse triangle inequality, we therefore have that 
\begin{align}
    \normm{\widehat{\theta}_{1,g}^{(t)} - \widehat{\theta}_{2,g}^{(t)}}&\geq \frac{\tau}{t-g}\normm{\theta(P_1)-\theta(P_2)} \\
    &- \normm{\frac{1}{t-g}  \sum_{i=1}^{\tau}\{h(Y_i) - \theta(P_1) \}} \\
    &- \normm{\frac{1}{t-g}  \sum_{i=\tau+1}^{t-g}\{h(Y_i) - \theta(P_2)\}} \\
    & - \normm{\frac{1}{g}\sum_{i=t-g+1}^{t} \{h(Y_i) - \theta(P_2)\}}.
\end{align}
Now, the three sums above run over blocks of $\tau$, $t-g-\tau$ and $g$ observations, respectively. Due to \eqref{concentrationass}, applied at the current time $t$ with $n$ equal to each of the three block lengths, we have that
\begin{align}
    \PP_{\tau} \left[\normm{\frac{1}{t-g} \sum_{i=1}^{\tau}\{h(Y_i) - \theta(P_1) \} }> \frac{\sqrt{\tau}}{t-g}\xi(t,\delta) \right] &\leq \frac{\delta}{6t^2 \log t}, \\
    \PP_{\tau} \left[\normm{\frac{1}{t-g} \sum_{i=\tau+1}^{t-g}\{h(Y_i) - \theta(P_2) \}} > \frac{\sqrt{t-g-\tau}}{t-g}{\xi(t,\delta)} \right] &\leq \frac{\delta}{6t^2 \log t}, \\
    \PP_{\tau} \left[\normm{\frac{1}{g} \sum_{i=t-g+1}^{t}\{h(Y_i) - \theta(P_2) \}} > \frac{\xi(t,\delta)}{\sqrt{g}} \right] &\leq \frac{\delta}{6t^2 \log t},
\end{align}
where the second probability is zero whenever $t-g-\tau = 0$, as the corresponding sum is then empty.
Noting that $\tau \leq t-g$ and $t-g-\tau < t-g$, a union bound then implies that, with probability at least $1-\delta/(2t^2\log t) \geq 1-\delta$, we will have that
\begin{align}
       \normm{\widehat{\theta}_{1,g}^{(t)} - \widehat{\theta}_{2,g}^{(t)}}&\geq \frac{\tau}{t-g}\normm{\theta(P_1)-\theta(P_2)} - \frac{2}{\sqrt{t-g}}
       \xi(t,\delta) - \frac{1}{\sqrt{g}} \xi(t,\delta).
\end{align}
Since also $g\geq (t-\tau)/2$ and $t \leq 2\tau$, it must hold with probability at least $1-\delta$ that
\begin{align}
       \normm{\widehat{\theta}_{1,g}^{(t)} - \widehat{\theta}_{2,g}^{(t)}}&\geq \frac{2}{3} \normm{\theta(P_1)-\theta(P_2)} - \frac{2}{\sqrt{t-g}}
       \xi(t,\delta) - \frac{1}{\sqrt{g}} \xi(t,\delta)\\
       &\geq \frac{2}{3} \normm{\theta(P_1) - \theta(P_2)} - \frac{2}{\sqrt{t-\tau}} \xi(t, \delta) - \frac{1}{\sqrt{g}} \xi(t, \delta),
\end{align}
where we used that $t-g\geq t- \tau$.
Thus, with probability at least $ 1-\delta$, we have
\begin{align}
    \normm{\widehat{\theta}_{1,g}^{(t)} - \widehat{\theta}_{2,g}^{(t)}} - \frac{2 \xi(t,\delta)}{\sqrt{g\wedge(t-g)}}&= \normm{\widehat{\theta}_{1,g}^{(t)} - \widehat{\theta}_{2,g}^{(t)}} - \frac{2 \xi(t,\delta)}{\sqrt{g}} \\
    &\geq \frac{2}{3} \normm{\theta(P_1)-\theta(P_2)} - \frac{3\xi(t,\delta)}{\sqrt{g}}  - \frac{2 \xi(t,\delta)}{\sqrt{t-\tau}}\\
    &\geq \frac{2}{3} \normm{\theta(P_1)-\theta(P_2)} - \frac{3\sqrt{2}+2}{\sqrt{t-\tau}} \xi(t,\delta)\\
    &>0\label{rrhs}
\end{align}
whenever $(t-\tau)\normm{\theta(P_1)-\theta(P_2)}^2\geq 88 \xi^2(t,\delta)$,
where we used that $g\geq (t-\tau)/2$. %
We therefore conclude that
\begin{align}
    \PP_{\tau}(T_g^{(t)}=1)\geq 1-\delta, 
\end{align}
and the proof is complete. 

\end{proof}

\subsection{Proof of Theorem \ref{theorem3}}
\begin{proof}[Proof.]
    We begin the proof by showing that $\mathrm{UC}(\tauhat, t) = \mathcal{O}(p\log p \log t)$ and $\mathrm{SC}(\tauhat, t) = \mathcal{O}(p\log t)$. To this end, we will appeal to Proposition \ref{enkelprop}. Note first that $T_g^{(t)}$ in \eqref{tdef} may be written as 
    \begin{align}
        T_g^{(t)} = f_g^{(t)}\left( \sum_{i=1}^{t-g} Y_i, \sum_{i=t-g+1}^t Y_i\right), \label{reptmp}
    \end{align}
    so that we may take $h(\cdot)$ to be the identity with computational cost $C_h(p) = 0$. In \eqref{reptmp}, $f_g^{(t)}$ is the function that computes $C_g^{(t)}$ from the two arguments in \eqref{reptmp}, computes $A_{s,g}^{(t)}$ from $C_g^{(t)}$ for each $s \in \mathcal{S}^{(t)}$, and then computes $\max_{s \in \mathcal{S}} \{ A_{s,g}^{(t)} / \xi_s^{(t)}\}$ and returns $1$ if the maximum is greater than $1$ (and zero otherwise). Since computing $A_{s,g}^{(t)}$ requires $\mathcal{O}( p )$ unit-cost operations, and $|\mathcal{S}| = \mathcal{O}(\log p)$, we get that $f_g^{(t)}$ can be computed in $C_f(p) = \mathcal{O}(p\log p)$ unit-cost operations. Due to Proposition \ref{enkelprop}, it therefore follows that $\mathrm{UC}(\tauhat, t) = \mathcal{O}(p\log p \log t)$ and $\mathrm{SC}(\tauhat, t) = \mathcal{O}(p\log t)$, as claimed. 

To show the remaining claims in Theorem \ref{theorem3}, we only need to show that Assumption \ref{assstat} holds, so that Proposition \ref{generalddprop} implies the claims. To this end, let $k= \normm{\mu_1 - \mu_2}_0$ denote the true sparsity of the change. Let $C_1$ equal the constant in Lemma \ref{lhihghprobupper} (depending only on $\delta$), and given $\lambda\geq C_1$, set $C_2 =  16 \left(C_0^2 + C_0\sqrt{C_0 + 2 \lambda} + C_0 + 2\lambda \right)$, where $C_0$ is the constant from Lemma \ref{lhihghproblower} (depending only on $\delta$). Finally, let $\rho(P_1, P_2) = \phi^2\sigma^{-2} = \normm{\mu_1 - \mu_2}_2^2\sigma^{-2}$ and $r(t,\delta) = C_2 z(k,p,t)$.

We first show that Assumption \ref{assstat-a} holds. To this end, note first that, if $\tau = \infty$, then for any $t \geq 2$ and $g \in [t-1]$, the rescaled vector $C^{(t)}_g/\sigma$ of variance-rescaled CUSUMs from \eqref{ydef} then has independent entries with standard normal distributions. Since $C_1$ is chosen as the constant from Lemma \ref{lhihghprobupper}, and $\xi_s^{(t)} = \lambda z(s,p,t) \geq C_1 z(s,p,t)$, the Lemma implies that 
\begin{align}
    \PP_{\infty}(T_g^{(t)}=1) &= \PP_{\infty}\left( \underset{s \in \mathcal{S}}{\max}\ \frac{A_{s,g}^{(t)}}{\xi_s^{(t)}}  >1\right)
    \\
    &\leq \PP_{\infty} \left[  \bigcup_{s \in \mathcal{S}^{(t)}} \left\{    A^{(t)}_{s,g} >  C_1 z(s,p,t)     \right\} \right]\\
    &\leq \frac{\delta}{3t^2\log t}, 
\end{align}
where $T^{(t)}_g$ is defined in \eqref{tdef}, and thus Assumption \ref{assstat-a} is satisfied. 

Next, we show that Assumption \ref{assstat-b} holds. Assume that $\tau<\infty$, let $t \in \NN$ be such that $\tau < t \leq 2\tau$, and assume that $(t-\tau) \rho(P_1, P_2) \geq r(t, \delta)$, so that
\begin{align}
    (t-\tau)\frac{\phi^2}{\sigma^2} \geq C_2 z(k,p,t). \label{tmpsnr}
\end{align}
Given any $g$ such that $(t - \tau)/2 \leq g \leq t-\tau$, we need to show that $\PP_{\tau}(T^{(t)}_g = 1) \geq 1-\delta$. 

By Lemma \ref{lhihghproblower}, for some $s_0 \in \mathcal{S}^{(t)}$ such that $k/2 \leq s_0\leq k$ whenever $k < \sqrt{p\log t}$ and $s_0 = p$ whenever $k\geq \sqrt{p\log t}$, the event 
\begin{align}
    \mathcal{E}_2 &= \left\{     A_{g,s_0}^{(t)} - \lambda z(s_0,p,t) \geq \psi - (C_0+2\lambda)  z(k,p,t) - C_0\psi^{1/2}  \right\},
\end{align}
has probability at least $\PP_{\tau} \left( \mathcal{E}_2\right) \geq 1-\delta$, where $C_0>0$ is the constant from Lemma \ref{lhihghproblower} and $\psi = g\tau^2 \{ t(t-g)\}^{-1} \phi^2\sigma^{-2}$.

On the event $\mathcal{E}_2$, by solving a quadratic inequality and exploiting that $z(s_0, p, t)\leq 2z(k, p,t)$ since $k/2\leq s_0\leq k$, the test $T_g^{(t)}$ in \eqref{tdef} will satisfy $T_{g}^{(t)}>0$
whenever
\begin{align}
    \psi &> 2^{-1} \left (C_0 \sqrt{ C_0^2 + 4 (C_0 + 2 \lambda)  z(k,p,t)} + C_0^2 + 2 (C_0 + 2\lambda) z(k,p,t)\right).\label{solvequad}
\end{align}
By a crude upper bound of the right hand side of \eqref{solvequad}, using also that $z(k,p,t) \geq \sqrt{\log 2} \geq 1/2$, we will thus have $T_{g}^{(t)}=1$ on $\mathcal{E}_2$ whenever
\begin{align}
    \psi &\geq 2\left(C_0^2 + C_0 \sqrt{C_0 + 2\lambda} + C_0 + 2\lambda \right) z(k,p,t)\label{solvequad2}.
\end{align}

Now, since $t\leq 2\tau$ and $(t-\tau)/2\leq g\leq t-\tau$, we have that
\begin{align}
    \psi &= \frac{g\tau^2}{t(t-g)} \frac{\phi^2}{\sigma^2}\\
    &\geq \frac{t-\tau}{2} \frac{\tau^2}{4\tau^2} \frac{\phi^2}{\sigma^2}\\
    &\geq (t-\tau)\frac{\phi^2}{\sigma^2} /8 \\
    &\geq \frac{C_2}{8} z(k,p,t), 
\end{align}
where the last inequality follows from \eqref{tmpsnr}. Due to the definition of $C_2$ and \eqref{solvequad2}, we therefore have that $T^{(t)}_g=1$ on $\mathcal{E}_2$. Thus, $\PP_{\tau}(T_g^{(t)}=1)\geq 1-\delta$, and Assumption \ref{assstat-b} is satisfied, and the claims in Theorem \ref{theorem3} are implied by Proposition \ref{generalddprop}. 

\end{proof}

\subsection{Proof of Theorem \ref{theorem5}}
\begin{proof}[Proof.]

We begin the proof by showing that $\mathrm{UC}(\tauhat, t) = \mathcal{O}(p^3 \log t)$ and $\mathrm{SC}(\tauhat, t) = \mathcal{O}(p^2\log t)$. To this end, we will appeal to Proposition \ref{enkelprop}. Note first that $T_g^{(t)}$ in \eqref{tvardef2} may be written as 
    \begin{align}
        T_g^{(t)} = f_g^{(t)}\left( \sum_{i=1}^{t-g} Y_i Y_i^\top, \sum_{i=t-g+1}^t Y_iY_i^\top\right), \label{reptmp2}
    \end{align}
    so that we may take $h(y) = y y^\top$, which can be computed using $C_h(p) = \mathcal{O}(p^2)$ unit-cost operations and may be represented as a vector with dimension $v(p) = p^2$ . In \eqref{reptmp2}, $f_g^{(t)}$ is the function that normalises the input arguments by $(t-g)$ and $g$, respectively, computes the operator norm of the first argument and the difference between the two arguments, computes the threshold $\xi^{(t)}_g$, and evaluates whether the threshold is exceeded. 
    Since computing the operator norm of a $p\times p$ matrix can be done using $\mathcal{O}(p^3)$ unit‑cost operations using a standard numerical method to compute the operator norm, such as QR decomposition or Jacobi's method \citep{complinear_algebra}, the computational cost of evaluating $f_g^{(t)}$ is at most $C_f(p) = \mathcal{O}(p^3)$ unit-cost operations. Due to Proposition \ref{enkelprop}, it therefore follows that $\mathrm{UC}(\tauhat, t) = \mathcal{O}(p^3 \log t)$ and $\mathrm{SC}(\tauhat, t) = \mathcal{O}(p^2\log t)$, as claimed. 

To show the remaining claims in Theorem \ref{theorem5}, we only need to show that Assumption \ref{assstat} holds, so that Proposition \ref{generalddprop} implies the claims. However, as a technical simplification, we will show directly that $\mathrm{FA}(\tauhat)\leq \delta$, as opposed to the stricter requirement in Assumption \ref{assstat-a} which also implies $\mathrm{FA}(\tauhat)\leq \delta$. 
To this end, let $C_1 = c_2/c_1$ (where $c_1, c_2$ are specified below), which depends only on $\delta,w,u$, and given $\lambda\geq C_1$, set 
\begin{align}
    C_2 = \frac{9}{4}\left\{ \frac{(2+\sqrt{2})c_2}{3} + \lambda  \sqrt{2}\left(\frac{2c_2}{3}+1\right) \right\}^2+2, \label{c2defhihi}
\end{align}
which only depends on $\delta, w,u,\lambda$.  Moreover, let 
\begin{align}
\rho(P_1, P_2) &= \left(\frac{\normm{\Sigma_1- \Sigma_2}_{\mathrm{op}}}{\normmop{\Sigma_1}\vee \normmop{\Sigma_2}} \right)^2,
\end{align}
and define $r(t,\delta) = C_2 (p\vee \log t)$.

We first show that $\mathrm{FA}(\tauhat)\leq \delta$. Assume that $\tau = \infty$, so that $\Sigma_1$ (the covariance associated with the pre-change distribution $P_1$) is the common covariance matrix of all the $Y_i$. Define the events 
\begin{align}
    \mathcal{E}_3 &= \bigcap_{t=2}^{\infty} \bigcap_{ g \in[t-1] } \left\{ \widehat{\sigma}_{g}^{(t)} \geq \left( \normmop{\Sigma_1}  c_1\right)^{1/2}  \right\},\\
    \mathcal{E}_4 &= \bigcap_{t=2}^{\infty} \bigcap_{ g\in [t-1]} \left\{ \normmop{\widehat{\Sigma}_{1,g}^{(t)} - \widehat{\Sigma}_{2,g}^{(t)}}  \leq c_2 \normmop{\Sigma_1} \left( \frac{p \vee \log t}{g\wedge(t-g) } \vee \sqrt{\frac{p \vee \log t}{g \wedge(t-g)}}   \right) \right\},
\end{align}
where $c_1 = (2e\pi w^2)^{-1} \delta^2 (\delta+2)^{-2}$, $c_2 = 4c_0 \{3 + \log(4/\delta) / \log(2)\}$, $c_0$ is the constant from Lemma \ref{moenminimaxlemma2} depending only on $u>0$, $\widehat{\Sigma}_{1,g}^{(t)}$ and $\widehat{\Sigma}_{2,g}^{(t)}$ are defined in \eqref{sigmahatsdef}, and $\widehat{\sigma}^{(t)}_g= \lVert \widehat{\Sigma}_{1,g}^{(t)}\rVert_{\mathrm{op}}^{1/2}$. Define $\mathcal{E} = \mathcal{E}_3 \cap \mathcal{E}_4$. Then Lemma \ref{covarianceeventbound1} implies that $\PP_{\tau}(\mathcal{E}) \geq 1 - \delta$.  On $\mathcal{E}$, for any $t \geq 2$ and $g \in G^{(t)}$, we have
\begin{align}
    \frac{\lVert \widehat{\Sigma}_{1,g}^{(t)} - \widehat{\Sigma}_{2,g}^{(t)} \rVert }{(\widehat{\sigma}_g^{(t)})^{2}} &\leq \frac{c_2}{c_1} \left( \frac{p \vee \log t}{g \wedge(t-g)} \vee \sqrt{\frac{p \vee \log t}{g \wedge(t-g)}}   \right).
\end{align}
Hence, choosing $C_1 = c_2/c_1$, which only depends on $\delta, w,u$, and also choosing $\lambda \geq C_1$ and $\xi_g^{(t)}$ as in \eqref{xicovar}, we have that $T^{(t)}_{g} =  \ind  \{ \lVert \widehat{\Sigma}_{1,g}^{(t)} - \widehat{\Sigma}_{2,g}^{(t)} \rVert_{\mathrm{op}} \  (\widehat{\sigma}^{(t)}_g)^{-2} > \xi_g^{(t)} \}$ in \eqref{tvardef2} satisfies $T^{(t)}_g= 0$ for all $t \geq 2$ and all $g \in G^{(t)}$. Thus, $\mathrm{FA}(\tauhat)\leq \delta$.

Next, we show that Assumption \ref{assstat-b} holds. Assume that $\tau<\infty$, let $t \in \NN$ be such that $\tau < t \leq 2\tau$, and assume that $(t-\tau) \rho(P_1, P_2) \geq r(t, \delta)$, so that
\begin{align}
    (t-\tau)\left(\frac{\normm{\Sigma_1- \Sigma_2}_{\mathrm{op}}}{\normmop{\Sigma_1}\vee \normmop{\Sigma_2}} \right)^2\geq C_2 (p\vee \log t), \label{tmpsnr2}
\end{align}
where we recall that $\Sigma_1$ and $\Sigma_2$ respectively denote the pre- and post-change covariances of the $Y_i$
Given any $g$ such that $(t - \tau)/2 \leq g \leq t-\tau$, we need to show that $\PP_{\tau}(T^{(t)}_g = 1) \geq 1-\delta$. 

To this end, using arguments similar as in the proof of Proposition \ref{example1statprop}, note first that
    \begin{align}
        \widehat{\Sigma}_{1,g}^{(t)} - \widehat{\Sigma}_{2,g}^{(t)} &= \frac{1}{t-g}\sum_{i=1}^{t-g}Y_i Y_i^\top - \frac{1}{g} \sum_{i=t-g+1}^t Y_i Y_i^\top\\
        &= \frac{1}{t-g} \sum_{i=1}^{\tau}Y_i Y_i^\top + \frac{1}{t-g} \sum_{i=\tau+1}^{t-g}Y_i Y_i^\top - \frac{1}{g}\sum_{i=t-g+1}^{t} Y_i Y_i^\top\\
        &= \frac{1}{t-g}  \sum_{i=1}^{\tau}\{Y_i Y_i^\top - \Sigma_1 \}+  \frac{\tau}{t-g}\Sigma_1\\
        &+ \frac{1}{t-g}  \sum_{i=\tau+1}^{t-g}\{Y_i Y_i^\top - \Sigma_2\} + \frac{t-g-\tau}{t-g}\Sigma_2\\
        &- \frac{1}{g}\sum_{i=t-g+1}^{t} \{Y_i Y_i^\top - \Sigma_2\} - \Sigma_2.
    \end{align}
Due to the reverse triangle inequality, we therefore have that 
\begin{align}
    \normmop{\widehat{\Sigma}_{1,g}^{(t)} - \widehat{\Sigma}_{2,g}^{(t)}}&\geq \frac{\tau}{t-g}\normmop{\Sigma_1-\Sigma_2} \\
    &- \normmop{\frac{1}{t-g}  \sum_{i=1}^{\tau}\{Y_i Y_i^\top - \Sigma_1 \}} \\
    &- \normmop{\frac{1}{t-g}  \sum_{i=\tau+1}^{t-g}\{Y_i Y_i^\top - \Sigma_2\}} \\
    & - \normmop{\frac{1}{g}\sum_{i=t-g+1}^{t} \{Y_i Y_i^\top - \Sigma_2\}}.\label{tmphehe}
\end{align}
Now, define the event
\begin{align}
    \mathcal{E}_5 =&  \left\{ \normmop{\frac{1}{\tau}\sum_{i=1}^{\tau} Y_i Y_i^\top - \Sigma_1} \geq \frac{c_2}{3} \normmop{\Sigma_1} \left( \frac{p\vee \log t}{\tau } \vee \sqrt{\frac{p\vee\log t}{\tau}}\right)\right\} \\
    &\bigcup \left\{ \normmop{\frac{1}{t-g-\tau}\sum_{i=\tau+1}^{t-g} Y_i Y_i^\top - \Sigma_2} \geq \frac{c_2}{3} \normmop{\Sigma_2} \left( \frac{p\vee \log t}{t-g-\tau } \vee \sqrt{\frac{p\vee\log t}{t-g-\tau}}\right)\right\}\\
    &\bigcup \left\{ \normmop{\frac{1}{g}\sum_{i=t-g+1}^{t} Y_i Y_i^\top - \Sigma_2} \geq \frac{c_2}{3} \normmop{\Sigma_2} \left( \frac{p\vee \log t}{g} \vee \sqrt{\frac{p\vee\log t}{g}}\right)\right\},
    \end{align}
where $c_2$ and $c_0$ are as before, and the second set is taken as the empty set if $t-g=\tau$. Then Lemma \ref{covarianceeventbound2} implies that $\PP(\mathcal{E}_5)\leq \delta$, so that $\PP(\mathcal{E}_5^{\complement})\geq 1-\delta$. 

On $\mathcal{E}_5^{\complement}$, due to \eqref{tmphehe}, we have
\begin{align}
    \normmop{\widehat{\Sigma}_{1,g}^{(t)} - \widehat{\Sigma}_{2,g}^{(t)}}&\geq \frac{\tau}{t-g}\normmop{\Sigma_1-\Sigma_2} \\
    &- \frac{c_2}{3} \normmop{\Sigma_1} \frac{\tau}{t-g}\left( \frac{p\vee \log t}{\tau } \vee \sqrt{\frac{p\vee\log t}{\tau}}\right)\\
    &- \frac{c_2}{3} \normmop{\Sigma_2}\ind\{t-g>\tau \}\frac{t-g-\tau}{t-g} \left( \frac{p\vee \log t}{t-g-\tau } \vee \sqrt{\frac{p\vee\log t}{t-g-\tau}}\right)\\
    & - \frac{c_2}{3} \normmop{\Sigma_2} \left( \frac{p\vee \log t}{g} \vee \sqrt{\frac{p\vee\log t}{g}}\right)\\
    &\geq \frac{\tau}{t-g}\normmop{\Sigma_1-\Sigma_2} \\
    &- \frac{c_2}{3} \normmop{\Sigma_1} \left( \frac{p\vee \log t}{t-g } \vee \sqrt{\frac{p\vee\log t}{t-g}}\right)\\
    &- \frac{c_2}{3} \normmop{\Sigma_2} \left( \frac{p\vee \log t}{t-g } \vee \sqrt{\frac{p\vee\log t}{t-g}}\right)\\
    & - \frac{c_2}{3} \normmop{\Sigma_2} \left( \frac{p\vee \log t}{g} \vee \sqrt{\frac{p\vee\log t}{g}}\right).
\end{align}
Now, since $g\geq (t-\tau)/2$, it holds that $\tau/(t-g)\geq 2/3$. Moreover, due to the signal-strength condition in \eqref{tmpsnr2}, and the fact that $g\geq (t-\tau)/2$,  Lemma \ref{tullelemma} implies that 
\begin{align}
    t-g &\geq g\\
    &\geq \frac{(t-\tau)}{2}\\
    &\geq (C_2 /2) (p\vee \log t)\left(\frac{\normm{\Sigma_1- \Sigma_2}_{\mathrm{op}}}{\normmop{\Sigma_1}\vee \normmop{\Sigma_2}} \right)^{-2} \\
    &\geq (C_2/2) (p\vee \log t)\\
    &\geq p\vee \log t,\label{tmphehe3}
\end{align}
using the definition of $C_2$ in \eqref{c2defhihi}. 
Since also $t-g\geq t-\tau$, it follows that 
\begin{align}
    \normmop{\widehat{\Sigma}_{1,g}^{(t)} - \widehat{\Sigma}_{2,g}^{(t)}}
    &\geq \frac{2}{3}\normmop{\Sigma_1-\Sigma_2} \\
    &- \frac{2c_2}{3} (\normmop{\Sigma_1} \vee \normmop{\Sigma_2}) \sqrt{\frac{p\vee\log t}{t-\tau}}\\
    & - \frac{c_2}{3}(\normmop{\Sigma_1} \vee \normmop{\Sigma_2}) \sqrt{\frac{p\vee\log t}{g}},\label{tmphehe5}\end{align}
on $\mathcal{E}_5^{\complement}$. 
Next, using similar arguments as above, we have that
\begin{align}
    \widehat{\Sigma}_{1,g}^{(t)} &= \frac{1}{t-g}\sum_{i=1}^{\tau} Y_i Y_i^{\top} + \frac{1}{t-g}\sum_{i=\tau+1}^{t-g}Y_iY_i^{\top}\\
    &= \frac{1}{t-g}\sum_{i=1}^{\tau} \{Y_i Y_i^{\top}-\Sigma_1\} + \frac{1}{t-g}\sum_{i=\tau+1}^{t-g}\{Y_iY_i^{\top} - \Sigma_2\} \\
    &+ \frac{\tau}{t-g}\Sigma_1 + \frac{t-g-\tau}{t-g}\Sigma_2.
\end{align}
On $\mathcal{E}_5^{\complement}$, also using similar arguments as above, we thus have that
\begin{align}
    (\widehat{\sigma}_g^{(t)})^2 &= \normmop{\widehat{\Sigma}_{1,g}^{(t)}}\\
    &\leq \frac{2c_2}{3}(\normmop{\Sigma_1}\vee \normmop{\Sigma_2}) \sqrt{\frac{p\vee \log t}{t-\tau}} + \frac{\tau}{t-g}\normmop{\Sigma_1} + \frac{t-g-\tau}{t-g}\normmop{\Sigma_2}\\
    &\leq \left(\frac{2c_2}{3}+1\right)(\normmop{\Sigma_1}\vee \normmop{\Sigma_2}),\label{tmphehe4}
\end{align}
where we in the last inequality used that $t-\tau\geq p\vee\log t$ due to \eqref{tmpsnr2}. Using \eqref{tmphehe3}, \eqref{tmphehe5}, \eqref{tmphehe4}, we have on $\mathcal{E}_5^{\complement}$ that
\begin{align}
    \lVert \widehat{\Sigma}_{1,g}^{(t)} - \widehat{\Sigma}_{2,g}^{(t)} \rVert_{\mathrm{op}} - (\widehat{\sigma}^{(t)}_g)^2 \xi_g^{(t)} &\geq \lVert \widehat{\Sigma}_{1,g}^{(t)} - \widehat{\Sigma}_{2,g}^{(t)} \rVert_{\mathrm{op}} - \lambda  \left(\frac{2c_2}{3}+1\right)(\normmop{\Sigma_1}\vee \normmop{\Sigma_2}) \sqrt{\frac{p\vee\log t}{g}} \\
    &\geq \frac{2}{3}\normmop{\Sigma_1 - \Sigma_2} - \frac{2c_2}{3}(\normmop{\Sigma_1}\vee \normmop{\Sigma_2}) \sqrt{\frac{p\vee \log t}{t-\tau}}\\
    &- \left\{ \frac{c_2}{3} + \lambda  \left(\frac{2c_2}{3}+1\right) \right\}(\normmop{\Sigma_1}\vee \normmop{\Sigma_2}) \sqrt{\frac{p\vee \log t}{g}}\\
    &\geq \frac{2}{3}\normmop{\Sigma_1 - \Sigma_2} \\
    &- \left\{ \frac{(2+\sqrt{2})c_2}{3} + \lambda  \sqrt{2}\left(\frac{2c_2}{3}+1\right) \right\}(\normmop{\Sigma_1}\vee \normmop{\Sigma_2}) \sqrt{\frac{p\vee \log t}{t-\tau}}, \label{lasthihi}
\end{align}
where the last inequality follows from $g\geq (t-\tau)/2$. Due to \eqref{tmpsnr2} and the definition of $C_2$ in \eqref{c2defhihi}, the right-hand side of \eqref{lasthihi} is strictly positive. On $\mathcal{E}_5^{\complement}$, we will therefore have that
\begin{align}
    \lVert \widehat{\Sigma}_{1,g}^{(t)} - \widehat{\Sigma}_{2,g}^{(t)} \rVert_{\mathrm{op}} \  (\widehat{\sigma}^{(t)}_g)^{-2} > \xi_g^{(t)}, 
\end{align}
and it therefore follows that $\PP_{\tau}(T_g^{(t)} = 1)\geq 1-\delta$, and we are done.

\end{proof}

\section{Proof of results in the Supplementary Material}
\subsection{Proof of Proposition \ref{compstatthm}}
\begin{proof}[Proof.]
To prove Proposition \ref{compstatthm}, assume that we at time $t$ store $S^{(1)},S^{(2)},\ldots,S^{(t)}$ in memory. 
Due to Assumption \ref{asscompstatic-b}, 
it follows immediately that $\mathrm{SC}(\tauhat_{\mathrm{stat}},t) = \mathcal{O}\{ p + M_S(p) t \}$. Moreover, since $|G^{(t)}_{\mathrm{stat}}|\leq 1 + \log_2(t-1) = \mathcal{O}(\log t)$, it follows directly from Assumption \ref{asscompstatic-a} and \ref{asscompstatic-b} that  $\mathrm{UC}(\tauhat_{\mathrm{stat}}, t) =\mathcal{O}[\{C_T(p) + C_S(p)\}\log t ]$. 
\end{proof}

\subsection{Proof of Proposition \ref{generalddpropstatic}}
\begin{proof}[Proof.] 
    For the first claim, note that
    \begin{align}
        \mathrm{FA}(\tauhat_{\mathrm{stat}}) &\leq \sum_{t=2}^{\infty}\sum_{g \in G^{(t)}_{\mathrm{stat}}} \PP_{\infty}(T^{(t)}_g =1) \\
         &\leq \sum_{t=2}^{\infty}|G^{(t)}_{\mathrm{stat}}| \ \frac{\delta}{(1 + \log_2 t)t^2} \\
         &\leq \sum_{t=2}^{\infty} \ \frac{\delta}{t^2} \\
         &<\delta.
    \end{align}
    Next, assume that $\tau< \infty$ and that $\tau \rho(P_1, P_2)\geq  2r(2\tau,\delta)$. Define
    \begin{align}
        d = \left \lceil \frac{r(2\tau, \delta)}{\rho(P_1, P_2)}\right\rceil.
    \end{align}
    Then $d\leq \lceil \tau/2 \rceil$. Writing $k = \lfloor \log_2 \tau \rfloor$, we have $2^{k} \leq \tau \leq 2^{k+1}-1$, and thus $\lceil \tau/2 \rceil \leq 2^{k}$, from which we obtain $d \leq 2^{\lfloor \log_2 \tau \rfloor}$. Now let $g = 2^{\lceil \log_2 d \rceil}$, so that $d \leq g \leq 2d$. Moreover, since $2^{\lfloor \log_2 \tau \rfloor}$ is a power of two no smaller than $d$, we must have $g \leq 2^{\lfloor \log_2 \tau \rfloor} \leq \tau$. Now set $t = \tau + g$. Then $\tau < t \leq 2\tau$. In addition, since $g$ is a power of two satisfying $g \leq \tau + g - 1 = t-1$, we have $g \in G^{(t)}_{\mathrm{stat}}$, and also $T^{(t)}_{t - \tau} = T^{(t)}_g$. Due to Assumption \ref{assstatstatic-b}, we have $\PP_{\tau}(T^{(t)}_g = 1)\geq 1-\delta$, since $\tau < t \leq 2\tau$ and $(t-\tau)\rho(P_1, P_2) = g\rho(P_1, P_2) \geq d \rho(P_1, P_2) \geq r(2\tau,\delta) \geq r(t, \delta)$, where the last inequality uses that $r(t,\delta)$ is non-decreasing in $t$. Noting that
    \begin{align}
        t &= \tau + g\\
        &\leq \tau+ 2 \left \lceil \frac{r(2\tau, \delta)}{\rho(P_1, P_2)}\right\rceil,
    \end{align}
    we therefore have that
    \begin{align}
        \PP_{\tau}\left\{ \tauhat_{\mathrm{stat}} \leq \tau + 2\left\lceil \frac{r(2\tau,\delta)}{\rho(P_1,P_2)} \right\rceil\right\} &= \PP_{\tau}(\tauhat_{\mathrm{stat}}\leq \tau+2d)\\
        &\geq \PP_{\tau}(\widehat{\tau}_{\mathrm{stat}} \leq t)\\
        &\geq \PP_{\tau}(T_{g}^{(t)} = 1)\\
        &\geq 1-\delta,
    \end{align}
    and we are done.

\end{proof}

\subsection{Proof of Theorem \ref{theorem6}}
\begin{proof}[Proof.]
We begin by showing that $\mathrm{UC}(\tauhat, t) = \mathcal{O}(p^{a} \log t)$ for some $a\geq 2$ and $\mathrm{SC}(\tauhat, t) = \mathcal{O}(p^2\log t)$. For simplicity, we will show this directly. At time $t \in \NN$, store the set $S^{(t)}$ in memory, where $S^{(t)}$ is given by 
\begin{align}
     S^{(1)}= \{Y_1 Y_1^{\top}\},
\end{align}
and 
\begin{align}
    S^{(t)} = \left(\bigcup_{g \in G^{(t)}} \left\{ \sum_{i=1}^{t-g} Y_iY_i^\top\right\} \right)\bigcup\left(\bigcup_{l=0}^{\lfloor \log_2(t)\rfloor} \left\{ \sum_{i=1}^{2^l} Y_iY_i^{\top}\right\} \right) \bigcup \left\{\sum_{i=1}^{t} Y_iY_i^{\top} \right\},
\end{align}
for $t\geq 2$. 
Then, clearly, for each $t\geq 2$ and $g\in G^{(t)}$, the matrices $\widetilde{\Sigma}_{1,g}^{(t)}$,   $\widetilde{\Sigma}_{2,g}^{(t)}$ from \eqref{sigmatilde}, and their difference, can be computed from $S^{(t)}$ using at most $\mathcal{O}(p^2)$ unit-cost operations. Moreover, for any $s$, $\widehat{\lambda}_{\max}^s(\widetilde{\Sigma}_{1,g}^{(t)})$ and $\widehat{\lambda}_{\max}^s (\widetilde{\Sigma}_{1,g}^{(t)} - \widetilde{\Sigma}_{1,g}^{(t)})$ can be computed using, e.g., first order methods \citep[see][]{bach2010convex}, which requires at most $\mathcal{O}(p^{a_0})$ unit‑cost operations for some suitably large $a_0 >0$. Computing $T^{(t)}$ in \eqref{tvardef3} thus requires at most $\mathcal{O}(p^{a_0 \vee 2} |\mathcal{S}| |G^{(t)}|) = \mathcal{O}(p^a\log t)$ unit-cost operations for some suitably large $a\geq 2$. Moreover, $S^{(t+1)}$ can be computed from $S^{(t)}$ and $Y_{t+1}$ using at most $\mathcal{O}(p^2)$ unit-cost operations due to the recycling property of $G^{(t)}$ in \eqref{thegrid} (see Lemma \ref{gridlemma}) and using similar arguments as in the Proof of Proposition \ref{theorem2}. Thus, $\mathrm{UC}(\tauhat,t) = \mathcal{O}(p^a \log t)$ for some $a\geq 2$. Next, since $S^{(t)}$ consist of at most $\mathcal{O}(\log t)$ matrices, each of which can be represented using $\mathcal{O}(p^2)$ scalars, it holds that $\mathrm{SC}(\tauhat,t)= \mathcal{O}(p^2\log t)$. 

To prove the remaining claims, we will show that Assumption \ref{assstat} holds and appeal to Proposition \ref{generalddprop}. As a technical simplification, we will show directly that $\mathrm{FA}(\tauhat)\leq \delta$, as opposed to the stricter requirement in Assumption \ref{assstat-a} (similar to the proof of Theorem \ref{theorem5}).
To this end, let $C_1 = c_4/c_3+1$ (where $c_3$ and $c_4$ are specified below), which depends only on $\delta,w,u$, and given $\lambda\geq C_1$, set 
\begin{align}
    C_2 = 8 \{4c_4 + \lambda(c_4/2+1)\}^2 \label{c2defhihi2}
\end{align}
which only depends on $\delta, w,u,\lambda$, where
\begin{align}
    h(p,t,g,s) = s \left\{    \frac{\log(p \vee t)}{g} \vee  \sqrt{\frac{\log(p \vee t)}{g}} \right\}.\label{hdef}
\end{align}  
Moreover, let 
\begin{align}
\rho(P_1, P_2) &= \underset{k\in [p]}{\max} \ \frac{\omega_k^2}{k^2}
\end{align}
where $\omega_k$ is defined in \eqref{kappakdefcovar}, 
and define $r(t,\delta) = C_2 \log (p \vee t)$.

We first show that $\mathrm{FA}(\tauhat)\leq \delta$. Assume that $\tau = \infty$, so that $\Sigma_1$ (the covariance associated with the pre-change distribution $P_1$) is the common covariance matrix of all the $Y_i$. Define the events 
\begin{align}
    \mathcal{E}_6 &= \bigcap_{t=2}^{\infty}  \bigcap_{ g = 1}^{ \lfloor t/2 \rfloor} \left\{ \left(\widehat{\sigma}_{g}^{(t)}\right)^2 \geq c_3 \lambda_{\max}^1({\Sigma_1})   \right\},\\
    \mathcal{E}_7 &= \bigcap_{t=2}^{\infty} \bigcap_{s \in \mathcal{S}} \bigcap_{ g = 1}^{ \lfloor t/2 \rfloor} \left\{ \widehat{\lambda}_{\max}^s ({\widetilde{\Sigma}_{1,g}^{(t)} - \widetilde{\Sigma}_{2,g}^{(t)}})  \leq c_4 \lambda_{\max}^1({\Sigma_1}) h(p,t,g,s) \right\},
\end{align}
where $c_3 = (2e\pi w^2)^{-1}\delta^2 (\delta + 2)^{-2}$, $c_4 = (c_0 / 2)\{ 3 + \log_2(8/\delta) / \log 2\}$, $c_0$ is the constant from Lemma \ref{moenminimaxlemma2} depending only on $u>0$, $\widetilde{\Sigma}_{1,g}^{(t)}$ and $\widetilde{\Sigma}_{2,g}^{(t)}$ are defined in \eqref{sigmatilde}, $\widehat{\sigma}_{g,s}^{(t)}$ is defined in \eqref{varestsparse}, $\mathcal{S}$ is defined in \eqref{mathcals}, $\lambda^{s}_{\max}(\cdot)$ is defined in \eqref{firsteigsparse}, and $\widehat{\lambda}^{s}_{\max}(\cdot)$ is defined in \eqref{conveigdef}. Define $\mathcal{E} = \mathcal{E}_6 \cap \mathcal{E}_7$. Then Lemma \ref{covarianceeventbound3} implies that $\PP_{\infty}(\mathcal{E}) \geq 1 - \delta$. 

On $\mathcal{E}$, for any $t \geq 2$ and $g \in G^{(t)}$ and $s \in \mathcal{S}$, we have
\begin{align}
    \frac{\widehat{\lambda}_{\max}^s(\widetilde{\Sigma}_{1,g}^{(t)} - \widetilde{\Sigma}_{2,g}^{(t)})}{\left(\widehat{\sigma}_{g}^{(t)}\right)^{2}} &\leq \frac{c_4}{c_3} h(p,t,g,s).
\end{align}
Hence, since $C_1 = c_4/c_3+1$ and $\lambda \geq C_1$, the test $T^{(t)}_g$ in \eqref{tvardef3} will satisfy $T^{(t)}_g = 0$ for all $t \geq 2$ and $g \in G^{(t)}$ on $\mathcal{E}$. It follows that $\tauhat=\infty$  on $\mathcal{E}$, and thus $\mathrm{FA}(\tauhat)\leq \delta$. 

Next, we show that Assumption \ref{assstat-b} holds. Assume that $\tau<\infty$, let $t \in \NN$ be such that $\tau < t \leq 2\tau$, and assume that $(t-\tau) \rho(P_1, P_2) \geq r(t, \delta)$, so that
\begin{align}
    (t-\tau) \frac{\omega_k^2}{k^2} \geq C_2 \log(p\vee t) \label{tmpsnrsparse}
\end{align}
for some fixed $k \in [p]$, where we recall the definition of $\omega_k$ from \eqref{kappakdefcovar}. 
Given any $g$ such that $(t - \tau)/2 \leq g \leq t-\tau$, we need to show that $\PP_{\tau}(T^{(t)}_g = 1) \geq 1-\delta$. 

Note that \eqref{tmpsnrsparse} implies
\begin{align}
    {g} \geq \frac{t-\tau}{2} \geq \frac{C_2}{2} \frac{k^2 \log(p \vee t)}{\omega_k^2} \geq \frac{C_2}{2} k^2 \log (p \vee t), \label{hihi22}
\end{align}
where the third inequality follows from the fact that \begin{align}
    \omega_k = \lambda_{\max}^k (\Sigma_1 - \Sigma_2)\{\lambda_{\max}^k(\Sigma_1) \vee \lambda_{\max}^k(\Sigma_2)\}^{-1}\leq 1, \label{hihi33}
\end{align} due to Lemma \ref{tullelemma}. In particular, $g\geq \log(p\vee t)$ since $C_2\geq 2$. 

Now, there exists an $s_0 \in \mathcal{S}$ such that $s_0/2 \leq k \leq s_0$, due to the definition of $\mathcal{S}$ in \eqref{mathcals}. For this $s_0$, define the events
\begin{align}
    \mathcal{E}_8 &=  \bigcap_{i=1,2} \left\{ \widehat{\lambda}_{\max}^{s_0} (\widetilde{\Sigma}_{i,g}^{(t)} - \Sigma_i)  \leq \frac{c_4}{2} \lambda_{\max}^{1}({\Sigma_i}) s_0 \sqrt{\frac{\log(p\vee t)}{g}} \right\},\\
    \mathcal{E}_9 &=  \left\{ \widehat{\lambda}_{\max}^{1} (\widetilde{\Sigma}_{1,g}^{(t)} - \Sigma_1)  \leq \frac{c_4}{2} \lambda_{\max}^{1}({\Sigma_1}) \sqrt{\frac{\log (p\vee t)}{g}} \right\},
    \end{align}
where $c_4$ is as before. Let $\mathcal{E} = \mathcal{E}_8 \cap \mathcal{E}_9$. Then Lemma \ref{covarianceeventbound4} implies that $\PP_{\tau}(\mathcal{E})\geq 1 - \delta$. 

On $ \mathcal{E}$, we claim that $\widehat{\sigma}^{(t)}_{g}\leq \sqrt{(c_4/2+1) \lambda_{\max}^{k}(\Sigma_1)}$. Indeed, since the triangle inequality holds for $\widehat{\lambda}_{\max}^{1}(\cdot)$, we have
\begin{align}
    \left(\widehat{\sigma}_{g}^{(t)}\right)^2 &= \widehat{\lambda}_{\max}^{1}(\widetilde{\Sigma}_{1,g}) \\
    &\leq \widehat{\lambda}_{\max}^{1}(\Sigma_1) + \widehat{\lambda}_{\max}^{1}(\widetilde{\Sigma}_{1,g}^{(t)} - \Sigma_1)\\
    &\leq \widehat{\lambda}_{\max}^{1}(\Sigma_1) + \frac{c_4}{2} \lambda_{\max}^{1}(\Sigma_1),\label{hihi34}
    \end{align}
where the last inequality follows from $g\geq \log(p\vee t)$. Moreover, due to Lemma \ref{moenminimaxlemma11}, we have $\widehat{\lambda}_{\max}^{1}(\Sigma_1) \leq \normm{\Sigma_1}_{\infty}$, where $\normm{\Sigma_1}_{\infty}$ denotes the largest absolute entry of $\Sigma_1$. Since this largest entry is contained on the diagonal, and $\lambda_{\max}^1(\Sigma_1)$ is the largest diagonal entry of $\Sigma_1$, we thus have that $\widehat{\lambda}_{\max}^{1}(\Sigma_1) = \lambda_{\max}^1(\Sigma_1)$. Due to \eqref{hihi34}, we thus have that
\begin{align}
    \left(\widehat{\sigma}_{g}^{(t)}\right)^2 &\leq (c_4/2 + 1)\lambda_{\max}^1(\Sigma_1)
    \\
    &\leq (c_4/2+1)\lambda_{\max}^k(\Sigma_1)\label{hihi11},
\end{align}
on $\mathcal{E}$, as claimed. 

Next, since the reverse triangle inequality holds for $\widehat{\lambda}_{\max}^{s_0}(\cdot)$, we will also have on $\mathcal{E}$ that
\begin{align}
{\widehat{\lambda}_{\max}^{s_0} (\widetilde{\Sigma}_{1,g}^{(t)} - \widetilde{\Sigma}_{2,g}^{(t)} ) }  &\geq {\widehat{\lambda}_{\max}^{s_0} ({\Sigma}_1 - {\Sigma}_2 ) - \widehat{\lambda}_{\max}^{s_0} (\widetilde{\Sigma}_{1,g}^{(t)} - \Sigma_1 ) - \widehat{\lambda}_{\max}^{s_0} (\widetilde{\Sigma}_{2,g}^{(t)} - \Sigma_2 )} \\
&\geq {\lambda}_{\max}^{s_0} ({\Sigma}_1 - {\Sigma}_2 )  - c_4 \{\lambda_{\max}^{s_0}(\Sigma_1) \vee \lambda_{\max}^{s_0}(\Sigma_2)\} s_0 \sqrt{\frac{\log(p\vee t)}{g}}\\
&\geq {\lambda}_{\max}^{k } ({\Sigma}_1 - {\Sigma}_2 )  - 4 c_4 \{\lambda_{\max}^{k}(\Sigma_1) \vee \lambda_{\max}^{k}(\Sigma_2)\} s_0 \sqrt{\frac{\log(p\vee t)}{g}}\end{align}
on $\mathcal{E}$, where we in the second inequality used that $\widehat{\lambda}_{\max}^{s_0}(A) \geq \lambda_{\max}^{s_0}(A)$ for any matrix $A$, since $\widehat{\lambda}_{\max}^{s_0}(\cdot)$ is a relaxation of the implicit optimisation problem defining $\lambda_{\max}^{s_0}(\cdot)$, and moreover that $\lambda_{\max}^{s_0}(\cdot)\geq \lambda_{\max}^k(\cdot)$ since $k\leq s_0$. Consequently, we have that
\begin{align}
    &{\widehat{\lambda}_{\max}^{s_0} (\widetilde{\Sigma}_{1,g}^{(t)}  - \widetilde{\Sigma}_{2,g}^{(t)} ) } - \left(\widehat{\sigma}_{g}^{(t)}\right)^2 \xi_{g}^{(t)} \\
    \geq & {\lambda}_{\max}^{k} ({\Sigma}_1 - {\Sigma}_2 ) - s_0  \left\{ 4c_4 + \lambda (c_4/2+1) \right\}\{\lambda_{\max}^k(\Sigma_1) \vee \lambda_{\max}^k(\Sigma_2)\}  \sqrt{\frac{\log (p \vee t)}{g}} \\
    \geq & {\lambda}_{\max}^{k} ({\Sigma}_1 - {\Sigma}_2 ) - 2k   \left\{ 4c_4 + \lambda (c_4/2+1) \right\}\{\lambda_{\max}^k(\Sigma_1) \vee \lambda_{\max}^k(\Sigma_2)\}  \sqrt{\frac{\log (p \vee t)}{g}} \\
    \geq & {\lambda}_{\max}^{k} ({\Sigma}_1 - {\Sigma}_2 ) - 2\sqrt{2}k   \left\{ 4c_4 + \lambda (c_4/2+1) \right\}\{\lambda_{\max}^k(\Sigma_1) \vee \lambda_{\max}^k(\Sigma_2)\}  \sqrt{\frac{\log (p \vee t)}{t-\tau}} \\
    \geq & \left[ \frac{\omega_k\sqrt{t-\tau}}{k} - 2\sqrt{2}   \left\{ 4c_4 + \lambda (c_4/2+1) \right\}  \sqrt{\log (p \vee t)}  \right ] \left \{\lambda_{\max}^k(\Sigma_1) \vee \lambda_{\max}^k(\Sigma_2)\right \} k (t-\tau)^{-1/2}\\
    >&0
\end{align}
on $\mathcal{E}$, where we in the second inequality used that $s_0 \leq 2k$, in the third inequality used that $g\geq (t-\tau)/2$, and the definition of $C_2$ in the last. It follows that $T^{(t)}_{g,s_0}=1$ on $\mathcal{E}$, and thus $\PP_{\tau}(T_g^{(t)}=1)\geq 1-\delta$, and we are done.

\end{proof}

\subsection{Proof of Proposition \ref{propreg}}
\begin{proof}[Proof.]
    For the first claim, we have that
    \begin{align}
        \mathrm{FA}(\tauhat) &\leq \PP_{\infty}(\tauhat < \infty) \\
        &\leq \sum_{t=2}^{\infty} \sum_{g \in G^{(t)}} \PP_{\infty}\left(D_g^{(t)}>\xi^{(t)}\right)\\
        &\leq \sum_{t=2}^{\infty} |G^{(t)}| \delta t^{-2}|G^{(t)}|^{-1}\\
        &\leq \delta,
    \end{align}
    where we in the third inequality used that $\widehat{\beta}_{1,g}^{(t)} - \widehat{\beta}_{2,g}^{(t)} \sim \N\left(0, \sigma^2\{ (M^{(t)}_{1,g})^{-1} + (M^{(t)}_{2,g})^{-1}\}\right)$ and hence $D_g^{(t)} \sim \chi_q^2$ under the null, in addition to the fact that $\xi^{(t)}$ was chosen as the upper $\delta t^{-2}|G^{(t)}|^{-1}$ quantile for this distribution.
 
    For the second claim, note that the quantity $D_g^{(t)}$ in \eqref{ddefregression} satisfies
\begin{align}
        D_g^{(t)} &= \frac{1}{\sigma^2}\left( (M^{(t)}_{1,g})^{-1}\sum_{i=1}^{t-g} x_i Y_i - (M^{(t)}_{2,g})^{-1}\sum_{i=t-g+1}^{t} x_i Y_i\right)^\top \left\{ (M^{(t)}_{1,g})^{-1} + (M^{(t)}_{2,g})^{-1}\right\}^{-1} \nonumber\\
        &\quad \times \left( (M^{(t)}_{1,g})^{-1}\sum_{i=1}^{t-g} x_i Y_i - (M^{(t)}_{2,g})^{-1}\sum_{i=t-g+1}^{t} x_i Y_i\right),
    \end{align}
    whenever $\sum_{i=1}^{t-g}x_i x_i^\top$ and $\sum_{i=t-g+1}^{t}x_i x_i^\top$ are invertible, and zero otherwise. Thus, the test $T^{(t)}_g$ in \eqref{regressiontest} has the form of \eqref{enkelpropt}, and Proposition \ref{enkelprop} applies. Indeed, as the function $h(\cdot)$, we may take 
    \begin{align}
        h(Y_i, x_i) = (Y_i x_i^\top, \mathrm{vec}(x_i x_i^\top))^\top,  
    \end{align}
    where $\mathrm{vec}(\cdot)$ denotes the operation that flattens a matrix into a vector. Then $h(\cdot)$ maps vectors from $\RR^{q+1}$ to $\RR^{q(q+1)}$, and can be computed using $C_h = \mathcal{O}\{q(q+1)\}$ unit-cost operations. Moreover, given $h(Y_i, x_i)$, computing $D_g^{(t)}$ in \eqref{ddefregression} costs $C_f = \mathcal{O}(q^3 + q^2 + q)$ unit-cost operations. Here, the $q^3$ term stems from computing the inverses $(M^{(t)}_{1,g})^{-1}$ and $(M^{(t)}_{2,g})^{-1}$ of $\sum_{i=1}^{t-g}x_i x_i^\top$ and $\sum_{i=t-g+1}^{t}x_i x_i^\top$ and from inverting their sum $(M^{(t)}_{1,g})^{-1} + (M^{(t)}_{2,g})^{-1}$, while the $q^2$ term stems from the matrix-vector multiplications in the quadratic form, and the $q$ term stems from vector subtraction. Due to Proposition \ref{enkelprop}, we obtain that $\mathrm{UC}(\tauhat,t) = \mathcal{O}(q^3 \log t)$ and $\mathrm{SC}(\tauhat, t) = \mathcal{O}(q^2\log t)$.
\end{proof}

\subsection{Proof of Proposition \ref{prop:meanminimax}}
\begin{proof}[Proof.]
To begin, note first that for any $\theta \in \Theta(k,p,\tau, \phi)$ from \eqref{Theta}, $\tauhat \in \mathcal{T}(\delta)$ from \eqref{mathcaltmean} and $x>0$, we have
\begin{align}
    \PP_{\theta}(\tauhat - \tau  > x) &= \PP_{\theta}(\tauhat > \tau + \lfloor x \rfloor),
\end{align}
since $\tau$ and $\tauhat$ are integer valued.
Moreover, for any $\theta_0 \in \Theta_0(p)$ from \eqref{Theta0}, we have
\begin{align}
     \PP_{\theta}(\tauhat > \tau + \lfloor x \rfloor) &= \PP_{\theta}(\tauhat > \tau + \lfloor x \rfloor) + \PP_{\theta_0}(\tauhat \leq \tau + \lfloor x \rfloor) - \PP_{\theta_0}(\tauhat \leq  \tau + \lfloor x \rfloor).
\end{align}
Since $\tauhat \in \mathcal{T}(\delta)$, we have $\PP_{\theta_0}(\tauhat \leq \tau + \lfloor x \rfloor) \leq \PP_{\theta_0}(\tauhat < \infty) \leq \delta$ for any such $\theta_0$, and thus
\begin{align}
     \PP_{\theta}(\tauhat > \tau + \lfloor x \rfloor) &\geq \PP_{\theta}(\tauhat > \tau + \lfloor x \rfloor) + \underset{\theta_0 \in \Theta_0(p)}{\sup} \PP_{\theta_0}(\tauhat \leq \tau + \lfloor x \rfloor) -\delta.
\end{align}
Write $l = \tau + \lfloor x \rfloor$.  Since $\tauhat$ is an extended stopping time with respect to the filtration $(\mathcal{F}_t)_{t \in \NN}$ generated by the $Y_i$, there exists a measurable function $\psi : \RR^{p \times l} \mapsto \{0,1\}$ such that we may write $\ind\{\tauhat \leq l\} = \psi(Y_1, \ldots, Y_l)$. Since $\tauhat \in \mathcal{T}(\delta)$ was arbitrary, it therefore follows that   
\begin{align}
     &\underset{\tauhat \in \mathcal{T}(\delta)}{\inf} \ \underset{\theta \in \Theta(k,p,\tau, \phi)}{\sup}  \PP_{\theta}(\tauhat - \tau  > x) \\
     \geq & \underset{\psi \in \Psi(p,l)}{\inf} \ \left[ \underset{\theta \in \Theta(k,p,\tau, \phi)}{\sup}  \PP_{\theta}\{\psi(Y^{(l)}) = 0\} + \underset{\theta \in \Theta_0(p)}{\sup}  \PP_{\theta}\{\psi(Y^{(l)}) = 1\} \right] - \delta,
\end{align}
where $Y^{(l)} = (Y_1, \ldots, Y_l)$ and $\Psi(p,l)$ is the set of all measurable functions $\psi : \RR^{p\times l} \mapsto \{0,1\}$. 
Now, for any $n\geq \tau$ and some $c>0$ to be chosen sufficiently small, set 
\begin{align}
    x = \begin{cases} n, &\text{if } \frac{\phi^2\tau}{\sigma^2} \leq c v(k,p),\\
    c \frac{\sigma^2}{\phi^2} v(k,p), &\text{if }  \frac{\phi^2\tau}{\sigma^2} > c v(k,p)
    \end{cases}
\end{align}
so that $l = \tau +n$ if $\phi^2\tau{\sigma^{-2}} \leq c v(k,p)$ and $l = \tau + \lfloor c \sigma^{2}\phi^{-2} v(k,p) \rfloor$ otherwise. 
To prove Proposition \ref{prop:meanminimax}, it suffices to choose a sufficiently small value of $c>0$ (depending only on $\epsilon$), such that 
\begin{align}
    \underset{\psi \in \Psi(p,l)}{\inf} \ \left[ \underset{\theta \in \Theta(k,p,\tau, \phi)}{\sup}  \PP_{\theta}\{\psi(Y^{(l)}) = 0\} + \underset{\theta \in \Theta_0(p)}{\sup}  \PP_{\theta}\{\psi(Y^{(l)}) = 1\} \right] \geq 1 - \epsilon, \label{minimaxproofclaim}
\end{align}
for any $\phi>0$, when $l$ is chosen as above.  %

To prove \eqref{minimaxproofclaim},  %
we will argue in a very similar fashion as in the proof of Proposition 3 in \cite{liu_minimax_2021}.  Due to Lemmas 8 and 10 in \cite{liu_minimax_2021}, there exists a sufficiently small $\alpha>0$, depending only on $\epsilon$, such that \eqref{minimaxproofclaim} holds whenever we can find a value of $c$ depending only on $\alpha$ and a prior distribution $\nu$ with support on $\Theta(k,p, \tau, \phi)$ such that
    \begin{equation}
    \EE_{(\theta^{(1)}, \theta^{(2)}) \sim \nu\otimes\nu} \exp\left( \frac{1}{\sigma^2} \sum_{i \in [l]} \sum_{j\in [p]}  \theta_{i}^{(1)}(j) \theta_{i}^{(2)}(j) \right) \leq 1+ \alpha.\label{minimaxtihi}
    \end{equation}
    
Define the prior distribution $\nu$ to be the distribution of $\theta = (\theta_i)_{i \in \NN} \in \Theta(k,p, \tau, \phi)$ generated according to the following process:
    \begin{enumerate}
        \item Sample a subset $S \subseteq [p]$ of cardinality $k$ uniformly at random;
        \item Independently of $S$, sample $u = (u(1), \ldots, u(p)) \in \RR^p$, where $u(j) \overset{\mathrm{i.i.d.}}{\sim} \mathrm{Unif}(\{-1, 1\})$ for all $j \in [p]$;
        \item Given the tuple $(S,u)$, if $\tau \phi^2 \sigma^{-2} \leq c v(k,p)$, set $\theta_i(j) = u(j) \phi / \sqrt{k}$ whenever  $(i,j) \in [\tau] \times S$  and $\theta_i(j) = 0$ otherwise, and if $\tau \phi^2 \sigma^{-2} > c v(k,p)$, set $\theta_i(j) = u(j) \phi / \sqrt{k}$ whenever $i > \tau$ and $j \in S$, and  $\theta_i(j) = 0$ otherwise.
    \end{enumerate}
    Now, let $(S, u)$ and $(T,v)$ denote two independent tuples sampled according to steps 1 and 2 above, and let $\theta^{(1)}$ and $\theta^{(2)}$ be the result of generating $\theta$ according to these respective tuples, as in step 3. We first claim that, for any value of $\phi$,
    \begin{align}
         &\EE_{(\theta^{(1)}, \theta^{(2)}) \sim \nu\otimes\nu} \exp\left( \frac{1}{\sigma^2 } \sum_{i \in [l]} \sum_{j\in [p]}  \theta_{i}^{(1)}(j) \theta_{i}^{(2)}(j) \right)  \\
         \leq & \EE \left\{   \left(  \frac{1}{2} e^{c v(k,p) k^{-1}} + \frac{1}{2}e^{-c v(k,p)k^{-1}}    \right)^{|S\cap T|}      \right\}, \label{minimaxclaim2}
    \end{align}
    where the expectation on the right-hand side is taken with respect to the joint distribution of $S$ and $T$. To see this, note that if $\tau \phi^2 \sigma^{-2} \leq c v(k,p)$, then
    \begin{align}
        \sigma^{-2} \sum_{i \in [l]} \sum_{j\in [p]}  \theta_{i}^{(1)}(j) \theta_{i}^{(2)}(j) &= \frac{\tau \phi^2}{\sigma^2 k} \sum_{j \in S\cap T} u(j) v(j),
    \end{align}
    and if $\tau \phi^2 \sigma^{-2} > c v(k,p)$ then
    \begin{align}
        \sigma^{-2} \sum_{i \in [l]} \sum_{j\in [p]}  \theta_{i}^{(1)}(j) \theta_{i}^{(2)}(j) &= \frac{(l- \tau) \phi^2}{\sigma^2 k} \sum_{j \in S\cap T} u(j) v(j)\\
        &= \left \lfloor c v(k,p) \frac{\sigma^{2}}{\phi^{2}}  \right \rfloor\frac{ \phi^2}{\sigma^2 k} \sum_{j \in S\cap T} u(j) v(j).
    \end{align}
    In both cases, we may therefore write $\sigma^{-2} \sum_{i \in [l]} \sum_{j\in [p]}  \theta_{i}^{(1)}(j) \theta_{i}^{(2)}(j) = \gamma \sum_{j \in S\cap T} u(j) v(j)$, where $\gamma$ is non-random and satisfies $0 \leq \gamma \leq c v(k,p) k^{-1}$, using that $\tau \phi^2 \sigma^{-2} \leq c v(k,p)$ in the first case, and that $\lfloor c v(k,p) \sigma^{2} \phi^{-2} \rfloor \phi^{2} \sigma^{-2} \leq c v(k,p)$ in the second case. In particular, the left-hand side of \eqref{minimaxclaim2} equals $\EE \exp \{ \gamma \sum_{j \in S\cap T} u(j) v(j) \}$. Since $u(j) v(j) \overset{\mathrm{i.i.d.}}{\sim} \mathrm{Unif}(\{-1,1\})$ for $j \in [p]$, independently of $(S,T)$, conditioning on $(S,T)$ gives
    \begin{align}
         \EE \exp \left( \gamma \sum_{j \in S\cap T} u(j) v(j) \right)
         &= \EE \left\{   \left(  \frac{1}{2} e^{\gamma} + \frac{1}{2}e^{-\gamma}    \right)^{|S\cap T|}      \right\} \\
         &\leq \EE \left\{   \left(  \frac{1}{2} e^{c v(k,p) k^{-1}} + \frac{1}{2}e^{-c v(k,p)k^{-1}}    \right)^{|S\cap T|}      \right\},
    \end{align}
    where the inequality uses that $x \mapsto (e^x + e^{-x})/2$ is non-decreasing on $[0, \infty)$, combined with $0 \leq \gamma \leq c v(k,p) k^{-1}$. This establishes the claim \eqref{minimaxclaim2}.

    Consider first the case when $k \geq \sqrt{p}$, so that $v(k,p) = \sqrt{p}$. Since $(e^x + e^{-x})/2 \leq e^{x^2/2}$ for any $x \in \RR$, the right hand side of \eqref{minimaxclaim2} satisfies
    \begin{align}
        \EE \left\{   \left(  \frac{1}{2} e^{c v(k,p) k^{-1}} + \frac{1}{2}e^{-c v(k,p)k^{-1}}    \right)^{|S\cap T|}      \right\}  &=\EE \left\{   \left(  \frac{1}{2} e^{c \sqrt{p/k^2}} + \frac{1}{2}e^{-c \sqrt{p/k^2}}    \right)^{|S\cap T|}      \right\} \\ 
         &\leq \EE \exp \left(  | S\cap T| \frac{c^2 p}{2k^2} \right).
    \end{align} Now, $|S\cap T|$ is distributed following the Hypergeometric distribution Hyp$(p,k,k)$, which is dominated by Bin$(k, k/p)$ in the convex ordering \citep[see the proof of Proposition 3][]{liu_minimax_2021}, which implies that 
    \begin{align}
        \EE \exp \left(  | S\cap T| \frac{c^2 p}{2k^2} \right)
         &\leq   \left \{ 1 - \frac{k}{p} + \frac{k}{p} \exp\left (\frac{c^2 p}{2k^2}\right) \right\}^k\\
         &\leq    \left \{ 1+ \frac{1}{k} \frac{c^2 }{2}\exp\left (\frac{c^2 p}{2k^2}\right) \right\}^k,
    \end{align}
    using that $e^x - 1\leq xe^x$ for $x \geq 0$. 
    Now, if $c\in (0, 1]$, we have $c^2 p /(2k^2) \leq 1/2$
    due to $k\geq \sqrt{p}$, and it then follows that
    \begin{align}
       \left \{ 1+ \frac{1}{k} \frac{c^2 }{2}\exp\left (\frac{c^2 p}{2k^2}\right) \right\}^k
         &\leq   \left ( 1+ \frac{1}{k} c^2  \right)^k\\
         &\leq e^{c^2},\label{minimaxtihi2}
    \end{align}
    where we in the second inequality used that $(1+x/y)^y \leq e^x$ for all $x,y \in \RR$ satisfying $x \leq y$ and $y\geq 1$. By choosing $c\leq \log^{1/2}(1+\alpha)\wedge 1$, the inequality in \eqref{minimaxtihi2} then implies that \eqref{minimaxtihi} holds when $k \geq \sqrt{p}$.
    
    Next, consider the case where $k < \sqrt{p}$. In this case, $k^{-1}v(k,p) = \log(ep k^{-2})$, and using that $( e^x + e^{-x})/2 \leq e^{x}$ for all $x \geq 0$, the right hand side of \eqref{minimaxclaim2} satisfies
    \begin{align}
         \EE \left\{   \left(  \frac{1}{2} e^{c v(k,p) k^{-1}} + \frac{1}{2}e^{-c v(k,p)k^{-1}}    \right)^{|S\cap T|}      \right\}
         &\leq  \EE  \exp \left\{ |S\cap T| c \log(ep k^{-2})\right\}.
    \end{align}
    
    Again using that $|S \cap T|$ follows a Hypergeometric distribution with parameters $p,k,k$, we have that
\begin{align}
         \EE  \exp \left\{ |S\cap T| c \log(ep k^{-2})\right\} &\leq  \left [ 1 - \frac{k}{p} + \frac{k}{p}  \exp\left\{  c \log \left (\frac{ep}{k^2}\right)  \right\}  \right]^k \\
         &\leq    \left [ 1+ \frac{k}{p}  c \log\left (\frac{ep}{k^2}\right)  \exp\left\{  c \log\left (\frac{ep}{k^2}\right)  \right\}  \right]^k,
    \end{align}
    again using that $e^x - 1\leq xe^x$ for $x \geq 0$. Now, if $c \in (0, 1/4]$, we have
    \begin{align}
        \left [ 1+ \frac{k}{p}  c \log\left (\frac{ep}{k^2}\right)  \exp\left\{  c \log\left (\frac{ep}{k^2}\right)  \right\}  \right]^k &= 
        \left \{ 1+ \frac{ec}{k} \left(\frac{ep}{k^2}\right)^{c - 1} \log \left( \frac{ep}{k^2}\right)  \right\}^k\\
        &\leq  
        \left \{ 1+ \frac{ec}{k} \left(\frac{ep}{k^2}\right)^{-1/2} \log \left( \frac{ep}{k^2}\right)  \right\}^k\\
        &\leq \left ( 1+ \frac{ec}{k}   \right)^k\\
        &\leq \exp(ec),\label{minimaxtihi40}
    \end{align}
    where we in the first inequality used that $x^{-1/2}\log x <1$ for $x\geq 1$, and in the second inequality used that $ec \leq 1\leq k$. By choosing $c \leq e^{-1}\log(1+\alpha) \wedge 2^{-2}$, the inequality in \eqref{minimaxtihi40} then implies that \eqref{minimaxtihi} holds also when $k< \sqrt{p}$. 

    Thus, we may take $c = \log^{1/2}(1 + \alpha) \wedge e^{-1}\log(1+\alpha) \wedge 1/4$, which depends only on $\alpha$ and thus only on $\epsilon$, and the proof is complete. 
\end{proof}

\subsection{Proof of Proposition \ref{prop:covarianceminimax}}
\begin{proof}[Proof.]
Arguing in a similar fashion as in the proof of Proposition \ref{prop:meanminimax}, we have that
\begin{align}
     &\underset{\tauhat \in \mathcal{T}(\delta)}{\inf} \ \underset{\gamma \in \Gamma(k,p,\tau, \omega)}{\sup}  \PP_{\gamma}(\tauhat - \tau  > x) \\
     \geq & \underset{\psi \in \Psi(p,l)}{\inf} \ \left[ \underset{\gamma \in \Gamma(k,p,\tau, \omega)}{\sup}  \PP_{\gamma}\{\psi(Y^{(l)}) = 0\} + \underset{\theta \in \Gamma_0(p)}{\sup}  \PP_{\gamma}\{\psi(Y^{(l)}) = 1\} \right] - \delta,
\end{align}
for any $x>0$, $\tau \in \NN$, $p\in [p]$, $k\in[p]$, $\delta \in (0,1)$ and $\omega \in (0,1/2]$, where $\mathcal{T}(\delta)$ is given in \eqref{mathcaltcovariance}, $\Gamma(k,p,\tau, \omega)$ is given in \eqref{Gamma},  $l = \tau + \lfloor x \rfloor$, $Y^{(l)} = (Y_1, \ldots, Y_l)$ and $\Psi(p,l)$ is the set of measurable functions $\psi : \RR^{p \times l} \mapsto \{0,1\}$. 

Now, for any $n \geq \tau$ and some $c>0$ to be chosen sufficiently small, set 
\begin{align}
    x = \begin{cases} n, &\text{if } \tau \omega^2 \leq c k\log \left(\frac{ep}{k}\right),\\
    c \frac{k\log \left(\frac{ep}{k}\right)}{\omega^2}, &\text{if }  \tau \omega^2 > c k\log \left(\frac{ep}{k}\right)
    \end{cases}
\end{align}
so that $l = \tau +n$ if $\tau \omega^2 \leq c v(k,p)$ and $l = \tau + \lfloor c \omega^{-2} k\log \left(ep/k\right) \rfloor$ otherwise. 
To prove Proposition \ref{prop:covarianceminimax}, it suffices to choose a sufficiently small value of $c>0$ (depending only on $\epsilon$), such that 
\begin{align}
    \underset{\psi \in \Psi(p,l)}{\inf} \ \left[ \underset{\gamma \in \Gamma(k,p,\tau, \omega)}{\sup}  \PP_{\gamma}\{\psi(Y^{(l)}) = 0\} + \underset{\theta \in \Gamma_0(p)}{\sup}  \PP_{\gamma}\{\psi(Y^{(l)}) = 1\} \right] \geq 1 - \epsilon, \label{minimaxproofclaimcovariance}
\end{align}
for any $\omega \in (0,1/2]$.

To prove \eqref{minimaxproofclaim}, %
we will argue in a very similar fashion as in the proof of Proposition 3 in \cite{moen2024minimax}.  Due to Lemma 8 in \cite{liu_minimax_2021}, there exists a sufficiently small $\alpha>0$, depending only on $\epsilon$, such that \eqref{minimaxproofclaim} holds whenever we can find a value of $c$ depending only on $\alpha$ and a prior distribution $\nu$ with support on $\Gamma(k,p, \tau, \omega)$ such that
\begin{align}
    \chi^2(f_1, f_0) +1 = \int \frac{f_1^2}{f_0} \leq 1 + \alpha,
\end{align}
where $f_0$ denotes the joint density of $Y_i \overset{\mathrm{i.i.d.}}{\sim}\mathrm{N}(0, \sigma^2 I)$ for $i \in [l]$ and some fixed $\sigma>0$, and $f_1$ denotes the joint  density of $Y_i \mid \gamma \overset{\mathrm{i.i.d.}}{\sim}\mathrm{N}(0, \gamma_i)$ for $i \in [l]$ marginalised over $\gamma = (\gamma_j)_{j \in \NN} \sim \nu$. 

Define the prior distribution $\nu$ to be the distribution of $\gamma = (\gamma_i)_{i \in \NN} \in \Gamma(k,p, \tau, \omega)$ generated according to the following 
process:
\begin{enumerate}
        \item Sample a subset $S \subseteq [p]$ of cardinality $k$;
        \item Given $S$, sample $u = (u(1), \ldots, u(p)) \in S^{p-1}_k$, where $u(j) \overset{\mathrm{i.i.d.}}{\sim} \mathrm{Unif}(\{-k^{-1/2}, k^{-1/2}\})$ for all $j \in S$ and $u(j) = 0$ for all $j \in [p] \setminus S$;
        \item Given the tuple $(S,u)$, if $\tau \omega^2 \leq c k \log(ep/k)$, set $\gamma_i = \sigma^2 I - \sigma^2 \omega u u^\top  $ whenever  $(i,j) \in [\tau] \times S$  and $\gamma_i(j) = \sigma^2 I$ otherwise, and if $\tau \omega^2 > c k\log(ep/k)$, set $\gamma_i = \sigma^2 I - \sigma^2 \omega u u^\top$ whenever $i > \tau$, and $\gamma_i = \sigma^2 I$ otherwise.
\end{enumerate}
Note $\nu$ does indeed have support on $\Gamma(k,p, \tau, \omega)$. Indeed, when $\gamma$ is sampled according to the above steps, we have $\lambda_{\max}^k(\gamma_1) \vee \lambda_{\max}^k(\gamma_{\tau+1}) = \sigma^2$, and $\lambda_{\max}^k(\gamma_1 - \gamma_{\tau+1}) = \sigma^2 \omega$, so that $\lambda_{\max}^k(\gamma_1 - \gamma_{\tau+1}) \left\{   \lambda_{\max}^k(\gamma_1) \vee \lambda_{\max}^k(\gamma_{\tau+1}) \right\}^{-1} = \omega$. 

 Now, let $(S, u)$ and $(T,v)$ denote two independent tuples sampled according to steps 1 and 2 above, and let $\gamma^{(1)} = (\gamma_i^{(1)})_{i \in \NN}$ and $\gamma^{(2)} = (\gamma_i^{(2)})_{i \in \NN}$ be the result of generating $\gamma$ according to these respective tuples, as in step 3. We first claim that
\begin{align}
    \chi^2(f_1, f_0) +1 \leq \EE \exp \left\{  2 \langle u, v\rangle^2 c k \log\left(\frac{ep}{k}\right)   \right\}, \label{minimaxcovariancetihi20}
\end{align}
where the expectation on the right-hand side is taken with respect to the joint distribution of $S,u,T,v$. To see this, note first that
\begin{align}
    \chi^2(f_1, f_0) +1 = \EE_{(\gamma^{(1)}, \gamma^{(2)})\sim \nu \otimes \nu} \left[ \EE_{X \sim \mathrm{N}_{lp}(0, \sigma^2I)} \left\{  \frac{\phi_{V_1}(X)\phi_{V_2}(X)}{\phi^2_{\sigma^2 I}(X)}    \right\} \right] ,
\end{align}
due to the definitions of $f_0$ and $f_1$, where $\phi_{V}(\cdot)$ denotes the density of any $X \sim  \mathrm{N}_{lp}(0, V)$, and $V_1 = \mathrm{Diag}(\gamma^{(1)}_1, \ldots, \gamma^{(1)}_l)$, $V_2 = \mathrm{Diag}(\gamma^{(2)}_1, \ldots, \gamma^{(2)}_l)$ denote the $(lp)\times (lp)$ block diagonal matrices formed from the first $l$ elements in the sequences $\gamma^{(1)}$ and $\gamma^{(2)}$, respectively. If $\tau \omega^2 \leq c k \log(ep/k)$, then $\gamma^{(1)}_i = \sigma^2 I - \mathbbm{1}\{i \leq \tau\}\sigma^2 \omega u u^\top$ and $\gamma^{(2)}_i = \sigma^2 I - \mathbbm{1}\{i \leq \tau\}\sigma^2 \omega v v^\top$ for all $i$, and due to Lemma 9 in \cite{moen2024minimax}, we have
\begin{align}
     \EE_{X \sim \mathrm{N}_{lp}(0, \sigma^2I)} \left\{  \frac{\phi_{V_1}(X)\phi_{V_2}(X)}{\phi^2_{\sigma^2 I}(X)}    \right\} &\leq \exp \left\{\frac{1}{2} \langle u, v \rangle^2 \tau \left( \frac{\sigma^2 \omega}{\sigma^2 - \sigma^2 \omega}\right)^2 \right\}\\
     &\leq \exp \left( 2 \langle u, v \rangle^2 \tau \omega^2 \right)\\
     &\leq \exp \left\{ 2 \langle u, v \rangle^2 c k \log\left(\frac{ep}{k}\right) \right\},
\end{align}
where we in the second inequality used that $\omega \leq 1/2$. Conversely, if $\tau \omega^2 > c k \log(ep/k)$, then $l = \tau + \lfloor c \omega^{-2} k\log \left(ep/k\right) \rfloor$, $\gamma^{(1)}_i = \sigma^2 I - \mathbbm{1}\{i > \tau\}\sigma^2 \omega u u^\top$ and $\gamma^{(2)}_i = \sigma^2 I - \mathbbm{1}\{i > \tau\}\sigma^2 \omega v v^\top$ for all $i$. Due to symmetry, Lemma 9 in \cite{moen2024minimax} implies in this case that
\begin{align}
     \EE_{X \sim \mathrm{N}_{lp}(0, \sigma^2I)} \left\{  \frac{\phi_{V_1}(X)\phi_{V_2}(X)}{\phi^2_{\sigma^2 I}(X)}    \right\} &\leq \exp \left\{\frac{1}{2} \langle u, v \rangle^2 (l - \tau) \left( \frac{\sigma^2 \omega}{\sigma^2 - \sigma^2 \omega}\right)^2 \right\}\\
     &\leq \exp \left\{ 2 \langle u, v \rangle^2 (l - \tau) \omega^2 \right\}\\
     &\leq \exp \left\{ 2 \langle u, v \rangle^2 \left \lfloor c \omega^{-2} k\log \left(\frac{ep}{k}\right) \right \rfloor \omega^2 \right\}\\
     &\leq \exp \left\{ 2 \langle u, v \rangle^2 c k \log\left(\frac{ep}{k}\right) \right\}.
\end{align}
Thus, \eqref{minimaxcovariancetihi20} holds, and it thus suffices to choose $c$ sufficiently small so that the right hand side of \eqref{minimaxcovariancetihi20} is bounded above by $1+\alpha$. To this end, notice that the distribution of $\sqrt{k}\langle u,v\rangle$ equals that of $\sum_{i=1}^{H}R_i $, where the $R_i$ are independent Rademacher random variables and $H \sim \mathrm{Hyp}(p,k,k)$, follows a Hypergeometric distribution with parameters $p,k,k$, and so 
\begin{align}
\chi^2(f_1, f_0) +1 &\leq 
    \EE \exp \left\{ 2 \langle u, v \rangle^2 c k \log\left(\frac{ep}{k}\right) \right\} \\
     &= \exp \left\{ \frac{2  c}{k}  \log\left(\frac{ep}{k}\right) \left(\sum_{i=1}^{H} R_i \right)^2 \right\},\label{covarianceminimaxtihi21}
\end{align}
where the expectation on the right hand side is taken with respect to $H$ and the $R_i$. 
Due to Lemma 1 in \cite{cai_optimal_2015}, we may choose $c>0$ sufficiently small so that the right hand side of \eqref{covarianceminimaxtihi21} is bounded above by $1 + \alpha$. The proof is complete. 
\end{proof}

\section{Auxiliary lemmas} 

\begin{lemma}\label{cusumlemma}
For any $t \geq 2$, $1\leq \tau \leq t-1$ and $\mu_1,\mu_2 \in \RR$, define $\mu = (\mu(1), \mu(2), \ldots, \mu(t))^\top \in \RR^{t}$ by $\mu(i) = \mu_1$ for $i\leq \tau$ and $\mu(i) = \mu_2$ for $i > \tau$. Let $\theta_g$ denote the CUSUM transformation applied to $\mu$, i.e., 
\begin{align}
    \theta_g  &= \left\{\frac{g}{t(t-g)}\right\}^{1/2} \ \sum_{i=1}^{t-g} \mu(i) - \left(\frac{t-g}{tg}\right)^{1/2} \ \left( \sum_{i=1}^t \mu(i) - \sum_{i=1}^{t-g} \mu(i) \right).
\end{align}
Then if $1\leq g \leq t - \tau$, we have
\begin{align}
    \theta_g^2 =  \frac{g \tau^2}{t(t-g)} (\mu_1 - \mu_2)^2.
\end{align}
\end{lemma}
\begin{proof}[Proof.]
Using that $t-g \geq \tau$ and inserting for $\mu$, we have
\begin{align}
    \theta_g &= \left\{\frac{g}{t(t-g)}\right\}^{1/2} \tau \mu_1 + \left\{\frac{g}{t(t-g)}\right\}^{1/2} \left(t-g-\tau\right)\mu_2 -  \left(\frac{t-g}{tg}\right)^{1/2} g \mu_2\\
    &= \left\{\frac{g}{t(t-g)}\right\}^{1/2} \left\{  \tau \mu_1 - \left(  \tau + g - t  \right) \mu_2 - (t-g) \mu_2   \right\}\\
    &= \tau \left\{\frac{g}{t(t-g)}\right\}^{1/2} (\mu_1 - \mu_2),
\end{align}
and we are done.
\end{proof}

\begin{lemma}\label{lhihghprobupper}
Let $p \in \NN$ and let $Y_1, Y_2, \ldots$ be i.i.d., each following a $p$-dimensional multivariate normal distribution with mean vector $\mu\in \RR^p$ and covariance matrix $\sigma^2 I$ for some $\sigma>0$. Let $C_g^{(t)}$ and $A_{s,g}^{(t)}$ be defined as in equations \eqref{ydef} and \eqref{adef}, respectively, where $a^2(s,t) = 4 \log \left(ep \log (t) s^{-2}\right) \ind\left\{ s \leq \sqrt{p \log t} \right\}$ and $\nu_{a(s,t)} = \EE \left\{Z^2 \ | \ |Z| > a(s,t)\right\}$ with $Z \sim \N(0,1)$. Let $\mathcal{S}^{(t)} = \{ 1,2,4, \ldots, 2^{\log_2\left( \sqrt{p\log t} \ \wedge \ p \right)} \}\cup\{ p\}$, and let $z(s,p,t)$ be as in \eqref{rdef}.
Then for any $\delta \in (0,1)$, there exist a constant $C>0$ depending only on $\delta$, so that for any $t\geq 2$ and $g \in [t-1]$, the event
\begin{align}
    \mathcal{E}_1 &= \bigcup_{s \in \mathcal{S}^{(t)}} \left\{    A^{(t)}_{s,g} >  C z(s,p,t)     \right\}.
    \end{align}
has probability at most $\PP_{\infty}\left(\mathcal{E}_1 \right) \leq \delta / (3t^2\log t)$.
\end{lemma}
\begin{proof}[Proof.]
    Set $C =  18 \{6 + 2\log(16/\delta)\log^{-1}(2)\} \vee  \{ 19 + 3\log^{1/2}(1/\delta) + 3\log(1/\delta) \}$. Fix any $t, g$ and $s \in \mathcal{S}^{(t)}$. Noting that $C_g^{(t)}(j)/\sigma \sim \N(0,1)$ independently for all $j \in [p]$, we fix $x_{s,t}>0$ (to be specified shortly), so that Lemma \ref{lemmaliu} implies that
    \begin{align}
        \PP_{\infty} \left( A_{s,g}^{(t)} \geq 9 \left[ \left\{pe^{-a^2(s,t)/2}x_{s,t} \right\}^{1/2} + x_{s,t} \right] \right) &\leq e^{-x_{s,t}}.
        \end{align}
        By a union bound, it follows that
        \begin{align}
        &\quad \PP_{\infty} \left( \exists  s \in \mathcal{S}^{(t)}\setminus\{p\} \ : \ A_{s,g}^{(t)} \geq 9 \left[ \left\{pe^{-a^2(s,t)/2}x_{s,t} \right\}^{1/2} + x_{s,t} \right] \right) \\
        &\leq \sum_{s \in \mathcal{S}^{(t)} \setminus \{p\}} e^{-x_{s,t}} \label{unionbound}.
    \end{align}
    Now set $x_{s,t} = c \left\{ \frac{p\log^2(t)}{s^2} \wedge z(s,p,t) \right\}$, for some $c>1$ to be specified later. Then the right hand side of \eqref{unionbound} is bounded above by
    \begin{align}
        \sum_{s \in \mathcal{S}^{(t)} \setminus \{p\}} e^{-x_{s,t}} &\leq \sum_{s \in \mathcal{S}^{(t)} \setminus \{p\}} \exp \left\{ -\frac{c p \log^2(t)}{s^2} \right\} + \sum_{s \in \mathcal{S}^{(t)} \setminus \{p\}} \exp\left \{ - c z(s,p,t)\right\}.
    \end{align}
    For the first term, since $s \in \mathcal{S}^{(t)}\setminus \{p\}$ satisfies $s \leq \sqrt{p\log t}$ and the ordered elements of $\mathcal{S}^{(t)}$ are increasing by a factor of $2$, we have
    \begin{align}
        \sum_{s \in \mathcal{S}^{(t)} \setminus \{p\}} \exp \left\{ -\frac{c p \log^2(t)}{s^2} \right\} &\leq \sum_{k=0}^{\infty} \exp\left\{ -c \log(t) 4^k\right\}\\
        &\leq t^{-c} \left( 1 + \sum_{k=1}^{\infty} t^{-3ck}\right)\\
        &\leq 2t^{-c},
    \end{align}
    using that $c>1$ and $t\geq 2$. For the second term, noting that $c z(s,p,t) \geq (c/2)s\log \left( \frac{ep \log t}{s^2} \right) + (c/2) \log t$, we have
    \begin{align}
        \sum_{s \in \mathcal{S}^{(t)} \setminus \{p\}} \exp\left \{ - c z(s,p,t)\right\} &\leq t^{-c/2} \sum_{s \in \mathcal{S}^{(t)}\setminus\{p\}} \left( \frac{s^2}{ep \log t}\right)^{cs/2} \\
        &\leq t^{-c/2} \left( 1 + \sum_{k=1}^{\infty} 4^{-ck/2}\right) \\
        &\leq 2 t^{-c/2}. 
    \end{align}
    We conclude that $\sum_{s \in \mathcal{S}^{(t)} \setminus \{p\}} e^{-x_{s,t}} \leq 4t^{-c/2}$ with this choice of $x_{s,t}$. Moreover, we have that
    \begin{align}
        9 \left[ \left\{ pe^{-a^2(s,t)/2} x_{s,t}\right\}^{1/2} + x_{s,t}\right] &< 18 cz(s,p,t). 
    \end{align}
    Now set $c = 6 + 2\log(16/\delta) \log^{-1}(2)$. Then, since $C\geq  18 c$ we have
    \begin{align}
         \PP_{\infty} \left\{ \exists  s \in \mathcal{S}^{(t)}\setminus\{p\} \ : \ A_{s,g}^{(t)} \geq Cz(s,p,t) \right\} &\leq 4t^{-c/2}\\
         &\leq\frac{\delta}{4t^3}.
    \end{align}
    
    Now consider the case where $s = p$. If $s< \sqrt{p\log t}$, then $a(p,t) >0$, and thus 
    \begin{align}
    \PP_{\infty} \left\{A_{p,g}^{(t)}  \geq  C z(s,p,t) \right\} < \frac{\delta}{4t^3},
    \end{align}
    using similar arguments as above. If instead $s\geq \sqrt{p\log t}$, then $a(p,t) = 0$ and by Lemma \ref{birgelemma} we have
    \begin{align}
        \PP_{\infty} \left[A_{p,g}^{(t)} \geq 2\left\{p\log(4t^3/\delta)\right\}^{1/2} + 2 \log\left( 4t^3/\delta\right) \right] &\leq \frac{\delta}{4t^3}.
    \end{align}
    Since $2\{p\log(4t^3/\delta)\}^{1/2} + 2 \log( 4t^3/\delta) \leq \{ 19 + 3\log^{1/2}(1/\delta) + 3\log(1/\delta) \} z(p,p,t) \leq  C z(p,p,t)$ for $t\geq 2$, a union bound over all $s \in \mathcal{S}^{(t)}$ gives
    \begin{align}
        \PP_{\infty} \left( \exists  s \in \mathcal{S}^{(t)} \ : \ A_{s,g}^{(t)} \geq C z(s,p,t) \right) &\leq \frac{\delta}{4t^3} + \frac{\delta}{4t^3} = \frac{\delta}{2t^3}.
    \end{align}
    Noting that $2t > 3\log t$ for all $t\geq 2$, we get that $\PP_{\infty}(\mathcal{E}_1) \leq \delta /(3t^2\log t)$.
\end{proof}

\begin{lemma}\label{lhihghproblower}
Let $p\in \NN$ and let $Y_1, Y_2, \ldots$ be independent $p$-dimensional random vectors following the model given in Section \ref{sec:intro} in the main text, assuming that $\tau <\infty$. Let $P_1 = \N_p(\mu_1, \sigma^2 I)$ and $P_2 = \N_p(\mu_2, \sigma^2 I)$ be $p$-dimensional Gaussian pre- and post-change distributions with respective mean vectors $\mu_1, \mu_2 \in \RR^p$ and variance $\sigma^2>0$. Let $C_g^{(t)}$ and $A_{s,g}^{(t)}$ be defined as in equations \eqref{ydef} and \eqref{adef}, respectively, where $a^2(s,t) = 4 \log \left(ep \log (t) s^{-2}\right) \ind\left\{ s \leq \sqrt{p \log t} \right\}$ and $\nu_{a(s,t)} = \EE \left\{Z^2 \ | \ |Z| > a(s,t)\right\}$ with $Z \sim \N(0,1)$. Let $\mathcal{S}^{(t)} =\{1,2,\ldots, 2^{\log_2  (\sqrt{p\log t} \wedge p) } \}\cup \{p\}$, and let $z(s,p,t)$ be as in \eqref{rdef}. Let $\phi = \normm{\mu_1 - \mu_2}_2$, $k = \normm{\mu_1 - \mu_2}_0$. 
Then for any $\lambda >0$ and $\delta \in (0,1)$, the exist an $s \in \mathcal{S}^{(t)}$ such that $k/2\leq s \leq k$ whenever $k <\sqrt{p\log t}$ and $s = p$ whenever $k\geq \sqrt{p\log t}$, and a constant $C>0$ depending only on $\delta$, such that for any $t> \tau$ and $1\leq g \leq t-\tau$, the event
\begin{align}
    \mathcal{E}_2 &= \left\{     A_{g,s}^{(t)} - \lambda z(s,p,t) \geq \psi - (C+ 2\lambda)  z(k,p,t) - C\psi^{1/2}  \right\},
\end{align}
has probability at least $\PP_{\tau} \left( \mathcal{E}_2\right) \geq 1-\delta$, where $\psi = g\tau^2 \{ t(t-g)\}^{-1} \phi^2\sigma^{-2}$.
\end{lemma}
\begin{proof}[Proof.]
    We shall show that $\PP_{\tau}\left(\mathcal{E}_2^{\text{c}}\right) \leq \delta$. Set $C = 18  + 7 \sqrt{\log(2/\delta)} +5 \log (2/\delta)$, and choose 
    \begin{align}
    s = \begin{cases}
        p, & \text{if } k \geq \sqrt{p\log t}\\
        \max \left \{z \in \mathcal{S}^{(t)} \ : \ z \leq k\right\}, & \text{otherwise,}
    \end{cases}
    \end{align}
    which satisfies $k/2\leq s\leq k$ whenever $k< \sqrt{p\log t}$. We treat the cases $k < \sqrt{p\log t}$ and $k\geq \sqrt{p\log t}$ separately. 
    
    \textbf{Step 1.} Assume first that $k \geq \sqrt{p\log t}$. Then $s = p$, $a(s,t) = 0$ and $\nu_{a(s,t)} = 1$, so that 
    \begin{align}
        A_{g,s}^{(t)} = \sum_{j=1}^p \left\{ C_g^{(t)}(j)^2/\sigma^2 -1\right\}.
    \end{align}
    By the linearity of the CUSUM transformation, we may for any $j \in [p]$ write
    \begin{align}
        C_g^{(t)}(j)^2 = \left(\theta_g^{(t)}(j) + Z(j) \right)^2,
    \end{align}
    where the $Z(j)/\sigma$ are i.i.d. standard normals and $\theta_g^{(t)}(j)$ is the CUSUM transformation of $\EE ( Y_1(j), \ldots, Y_t(j))^\top$ given by 
    \begin{align}
        \theta_g^{(t)}(j)^2 &= \frac{g\tau^2}{t(t-g)}(\mu_1(j) - \mu_2(j))^2,
    \end{align}
    due to Lemma \ref{cusumlemma} since $g \leq t - \tau$. It follows that
    \begin{align}
        A_{g,s}^{(t)} + p \sim \chi_p^2(\psi),
    \end{align}
    where $\psi = \sum_{j=1}^p \theta_g^{(t)}(j)^2/\sigma^2 = g\tau^2 \{ t(t-g)\}^{-1} \phi^2\sigma^{-2}$. 
    By Lemma \ref{birgelemma}, we then have 
    \begin{align}
    \PP_{\tau} \left[  A_{g,s}^{(t)} \leq \psi -2 \left\{\log(1/\delta)(p + 2\psi)\right\}^{1/2}    \right] &\leq \delta.
    \end{align}
    Since $z(s,p,t) = z(k,p,t) = \sqrt{p\log t}$, it follows that
    \begin{align}
    \PP_{\tau} \left[  A_{g,s}^{(t)} - \lambda z(s,p,t) \leq \psi - (C + \lambda) z(k,p,t) - C \psi^{1/2}  \right] &\leq \delta,
    \end{align}
    using that $C \geq 2\sqrt{2\log (1/\delta)}$. 

    \textbf{Step 2.} 
    Now suppose that $k < \sqrt{p\log t}$. Without loss of generality, we may assume that only the first $k$ components of the mean vector undergo a change. By a deterministic lower bound, we have 
    \begin{align}
        \quad A^{(t)}_{g,s} &\geq \sum_{j=1}^k \left\{ C_g^{(t)}(j)^2/\sigma^2 - \nu_{a(s,t)} \right\} \\
        &+ \sum_{j=k+1}^p\left\{ C_g^{(t)}(j)^2/\sigma^2 - \nu_{a(s,t)} \right\} \ind\left\{  |C_g^{(t)}(j)/\sigma| > a(s,t)\right\}\label{sumsplit}.
    \end{align}
    We lower bound each term separately, beginning with the first. Similar to above, we have $C_g^{(t)}(j) = \theta_g^{(t)}(j) + Z(j)$ for each $j \in [p]$, where the $Z(j)/\sigma$ are i.i.d. standard normals. It follows that $\sum_{j=1}^k C_g^{(t)}(j)^2/\sigma^2 \sim \chi_k^2(\psi)$, where $\psi = \sum_{j=1}^p \theta_g^{(t)}(j)^2 /\sigma^2 = g\tau^2 \{ t(t-g)\}^{-1} \phi^2\sigma^{-2}$ since only the first $k$ coordinates undergo a change in the mean.  
    Using Lemma \ref{palemma1} and Lemma \ref{birgelemma}, we obtain 
    \begin{align}
        \PP_{\tau}\left[ \sum_{j=1}^k \left\{ \frac{C_g^{(t)}(j)^2}{\sigma^2} - \nu_{a(s,t)} \right\} \leq \psi - 2\left\{ \log(2/\delta)(k+2\psi)\right\}^{1/2} - k\left\{ a^2(s,t)+2 \right\}\right] &\leq \frac{\delta}{2}.
    \end{align}
    By the definition of $s$ and $\mathcal{S}^{(t)}$, we have that $k/2 \leq s \leq k$. Since $k < \sqrt{p \log t}$, it follows from some simple algebra that $ka^2(s,t)\leq k a^2(k/2,t) = k a^2(k,t) + 4 k \log (4) \leq 4 k a^2(k,t) \leq  16 z(k,p,t)$, and $k \leq z(k,p,t)$. Hence, 
    \begin{align}
        &\PP_{\tau}\Bigg [ \sum_{j=1}^k \left\{ \frac{C_g^{(t)}(j)^2}{\sigma^2} - \nu_{a(s,t)} \right\} \leq \psi - \left\{ 2 \sqrt{\log\left(\frac{2}{\delta}\right)} + 18 \right\} z(k,p,t) \\
        &- 2\sqrt{2 \log\left (\frac{2}{\delta}\right)}\psi^{1/2} \Bigg] \leq  \frac{\delta}{2}\label{uboundpart1},
    \end{align}
    using that $z(k,p,t)\geq 1$. For the second term in \eqref{sumsplit}, define 
    $$W(j) = \left\{ C_g^{(t)}(j)^2/\sigma^2 - \nu_{a(s,t)} \right\} \ind\left\{  \left|C_g^{(t)}(j)/\sigma\right| > a(s,t)\right\},$$ for $j = k+1,\ldots, p$. Then Lemma \ref{palemma2} implies that
    \begin{align}
        \PP_{\tau}\left[ \sum_{j=k+1}^p W(j) \leq -5 \log(2/\delta) - 5 \left\{ pe^{-a^2(s,t)/2} \log(2/\delta) \right\}^{1/2} \right] \leq \frac{\delta}{2}.
    \end{align}
    Inserting for $a(s,t)$ and using that $s \leq k \leq z(k,p,t) \wedge \sqrt{p\log t}$ for $k<\sqrt{p\log t}$, we obtain 
    \begin{align}
        \PP_{\tau}\left[ \sum_{j=k+1}^p W(j) \leq  - 5 z(k,p,t) \sqrt{\log\left(\frac{2}{\delta}\right)} - 5 \log\left (\frac{2}{\delta}\right) \right] \leq \frac{\delta}{2}\label{uboundpart2}.
    \end{align}
    Combining the lower bounds in equations \eqref{uboundpart1} and \eqref{uboundpart2} by a union bound, we obtain
    \begin{align}
        &\PP_{\tau} \Bigg[ A_{g,s}^{(t)} \leq \psi - \left\{ (2+5) \sqrt{\log\left(\frac{2}{\delta}\right)} + 18 + 5 \log\left(\frac{2}{\delta}\right) \right\} z(k,p,t)  \\
        & - 2\sqrt{2 \log\left(\frac{2}{\delta}\right)}\psi^{1/2}  \Bigg] \leq \delta.
    \end{align}
    By the definition of $C$, %
    it then follows that 
    \begin{align}
    \PP_{\tau} \left[  A_{g,s}^{(t)} - \lambda z(s,p,t)\leq \psi - (C+2\lambda) z(k,p,t)  - C \psi^{1/2}    \right] &\leq \delta,
    \end{align}
    where we used that $z(s,p,t) \leq 2 z(k,p,t)$ since $k/2 \leq s \leq k$. The proof is complete. 
\end{proof}

\begin{lemma}\label{covarianceeventbound1}
Let $p \in \NN$ and let $Y_1, Y_2, \ldots$ be $p$-dimensional random variables satisfying Assumption \ref{assmultivariate2} for some $w,u>0$. Assume that $\EE Y_i Y_i^\top = \Sigma_1$ for all $i \in \NN$ and some positive definite matrix $\Sigma_1 \in \RR^{p\times p}$. Given any $t \geq 2$ and $g \in [t-1]$, define
\begin{align}
    &\widehat{\Sigma}_{1,g}^{(t)} = (t-g)^{-1} \sum_{i=1}^{t-g} Y_i Y_i^\top, &\widehat{\Sigma}_{2,g}^{(t)} = g^{-1} \sum_{i=t-g+1}^t Y_{i} Y_i^\top,
\end{align}
as in \eqref{sigmahatsdef}, and define $\widehat{\sigma}_{g}^{(t)} = \lVert \widehat{\Sigma}_{1,g}^{(t)}\rVert_{\mathrm{op}}^{1/2}$.
For any $\delta\in (0,1)$, define the events
\begin{align}
    \mathcal{E}_3 &= \bigcap_{t=2}^{\infty} \bigcap_{ g \in[t-1] } \left\{ \widehat{\sigma}_{g}^{(t)} \geq \left( \normmop{\Sigma_1}  c_1\right)^{1/2}  \right\},\\
    \mathcal{E}_4 &= \bigcap_{t=2}^{\infty} \bigcap_{ g\in [t-1]} \left\{ \normmop{\widehat{\Sigma}_{1,g}^{(t)} - \widehat{\Sigma}_{2,g}^{(t)}}  \leq c_2 \normmop{\Sigma_1} \left( \frac{p \vee \log t}{g \wedge(t-g) } \vee \sqrt{\frac{p \vee \log t}{g \wedge(t-g)}}   \right) \right\},
\end{align}
where $c_1 = (2e\pi w^2)^{-1} \delta^2 (\delta+2)^{-2}$, $c_2 = 4c_0 \{3 + \log(4/\delta) / \log(2)\}$, and $c_0$ is the constant from Lemma \ref{moenminimaxlemma2} depending only on $u>0$. Then $\PP(\mathcal{E}_3 \cap \mathcal{E}_4) \geq 1 - \delta$. 

\end{lemma}
\begin{proof}[Proof.]
 We first show that  $\PP(\mathcal{E}_3^\complement)\leq \delta/2$. Due to the definition of $\widehat{\sigma}_g^{(t)}$, we have
 \begin{align}
     \PP\left(\mathcal{E}_3^\complement\right) &\leq \PP \left( \bigcup_{t=2}^{\infty} \bigcup_{g = 1, \ldots, \lfloor t/2 \rfloor} \left\{   \normmop{\widehat{\Sigma}_{1,g}^{(t)}} \leq c_1 \normmop{\Sigma_1} \right\}\right) \\
     &\leq \PP \left( \bigcup_{t=1}^{\infty} \left\{   \normmop{t^{-1} \sum_{i=1}^{t} Y_i Y_i^\top} \leq c_1 \normmop{\Sigma_1} \right\}\right)\\
     &\leq \sum_{t=1}^{\infty}  \PP \left(   \normmop{t^{-1} \sum_{i=1}^{t} Y_i Y_i^\top} \leq c_1 \normmop{\Sigma_1} \right)
 \end{align}

Now, let $v_1 \in \mathbb{S}^{p-1}$ be a unit vector satisfying $v_1^\top \Sigma_1 v_1 = \normmop{\Sigma_1}$\footnote{Note that such a $v$ always exists, since the unit sphere in $\RR^p$ is compact.}. For any $t \geq 2$ we have
\begin{align}
    \normmop{t^{-1} \sum_{i=1}^{t} Y_i Y_i^\top} &= t^{-1} \underset{v \in \mathbb{S}^{p-1}}{\sup} \ \sum_{i=1}^t (v_1^\top Y_i)^2\\
    &\geq  t^{-1} \sum_{i=1}^t (v_1^\top Y_i)^2. 
\end{align}
Due to Assumption \ref{assmultivariate2}, Lemma \ref{moenminimaxlemma1} yields that 
\begin{align}
     \PP \left(\mathcal{E}_3^\complement\right) &\leq  \sum_{t = 1}^{\infty} \PP \left(   \sum_{i=1}^t (v_1^\top Y_i)^2    \leq  t c_1 v^\top \Sigma_1 v \right)\\
    &\leq   \sum_{t=1}^{\infty} \exp \left\{ \frac{t}{2}\log(2e\pi w^2 c_1)  \right\}\\
    &= \frac{(2e\pi w^2 c_1)^{1/2}}{ 1 - (2e\pi w^2 c_1)^{1/2}}\\
    &= \frac{\delta}{2},
\end{align}
where the two last equalities used the definition of $c_1$. Next we will show that $\PP(\mathcal{E}_4^{\complement} ) \leq \delta/2$. Since 
$\lVert \widehat{\Sigma}_{1,g}^{(t)} - \widehat{\Sigma}_{2,g}^{(t)}\rVert \leq \lVert \widehat{\Sigma}_{1,g}^{(t)} - \Sigma_1 \rVert + \lVert \widehat{\Sigma}_{2,g}^{(t)} - \Sigma_1 \rVert$, by symmetry we have
\begin{align}
    \PP\left(\mathcal{E}_4^{\complement} \right) &\leq 2 \PP \left[ \bigcup_{t=2}^{\infty} \bigcup_{ g \in [t-1]} \left \{ \normmop{\widehat{\Sigma}_{2,g}^{(t)} - \Sigma_1}  \geq \frac{c_2}{2} \normmop{\Sigma_1} \left( \frac{p \vee \log t}{g} \vee \sqrt{\frac{p \vee \log t}{g}}   \right) \right\}\right]\\
    &\leq 2\sum_{t=2}^{\infty} \sum_{g \in [t-1]} \PP \left\{ \normmop{g^{-1}\sum_{i=t-g+1}^t Y_i Y_i^\top - \Sigma_1}  \geq \frac{c_2}{2} \normmop{\Sigma_1} \left( \frac{p \vee \log t}{g} \vee \sqrt{\frac{p \vee \log t}{g}}   \right) \right\}
\end{align}
Due to Lemma \ref{moenminimaxlemma2}, we have\footnote{\label{note1}Lemma \ref{moenminimaxlemma2} can be applied here since $\normmop{A} = \lambda_{\max}^p(A)$ for any symmetric matrix $A$, where $\lambda_{\max}^s(A) = \sup_{v \in \mathbb{S}^{p-1}_s} \ | v^\top A v |$ for any $s\in [p]$, as in \eqref{firsteigsparse}.}
\begin{align}
    \PP\left(\mathcal{E}_4^{\complement} \right) &\leq 2 \sum_{t=2}^{\infty} \sum_{g = 1}^{t-1} \exp\left\{ - \frac{c_2}{2c_0} \left( p \vee \log t\right)\right\}\\
    &\leq 2 \sum_{t=2}^{\infty} \sum_{g = 1}^{t-1} \exp\left\{ - \frac{c_2}{4c_0} \left( p \vee \log t\right)\right\}\\
    &\leq 2 \sum_{t=2}^{\infty} t^{1 - c_2/(4c_0)}, 
\end{align}
where $c_0$ is the constant from that Lemma. Now, since $1 - c_2/(4c_0) = -2 - \log (4/\delta) /\log 2 \leq -2 - \log (4/\delta) /\log t$, we have that $t^{1 - c_2/(4c_0)} \leq \delta/(4t^2)$ for all $t \geq 2$. Summing the series, we obtain that $\PP (\mathcal{E}_4^{\complement}) \leq \delta/2$, and we are done. 
\end{proof}

\begin{lemma}\label{covarianceeventbound2}
Let $p\in \NN$, $\tau \in \NN$, $\tau < t\leq 2\tau$ and let $Y_1, \ldots, Y_{t}$ be $p$-dimensional random variables satisfying Assumption \ref{assmultivariate2} for some $w,u>0$. Assume further that $\EE Y_i Y_i^\top = \Sigma_1$ for all $i \leq \tau$ and $\EE Y_i Y_i^\top = \Sigma_2$ for $i>\tau$, for some positive definite matrices $\Sigma_1, \Sigma_2 \in \RR^{p\times p}$. Given any fixed $g \leq t - \tau$ and $\delta \in (0,1)$, define the event
\begin{align}
    \mathcal{E}_5 =&  \left\{ \normmop{\frac{1}{\tau}\sum_{i=1}^{\tau} Y_i Y_i^\top - \Sigma_1} \geq \frac{c_2}{3} \normmop{\Sigma_1} \left( \frac{p\vee \log t}{\tau } \vee \sqrt{\frac{p\vee\log t}{\tau}}\right)\right\} \\
    &\bigcup \left\{ \normmop{\frac{1}{t-g-\tau}\sum_{i=\tau+1}^{t-g} Y_i Y_i^\top - \Sigma_2} \geq \frac{c_2}{3} \normmop{\Sigma_2} \left( \frac{p\vee \log t}{t-g-\tau } \vee \sqrt{\frac{p\vee\log t}{t-g-\tau}}\right)\right\}\\
    &\bigcup \left\{ \normmop{\frac{1}{g}\sum_{i=t-g+1}^{t} Y_i Y_i^\top - \Sigma_2} \geq \frac{c_2}{3} \normmop{\Sigma_2} \left( \frac{p\vee \log t}{g} \vee \sqrt{\frac{p\vee\log t}{g}}\right)\right\},
    \end{align}
where $c_2 = 4c_0 \{3 + \log(4/\delta) / \log(2)\}$, and $c_0$ is the constant from Lemma \ref{moenminimaxlemma2} depending only on $u>0$, and the second set is taken as the empty set if $t-g=\tau$. Then we have $\PP(\mathcal{E}_5) \leq\delta$. 

\end{lemma}
\begin{proof}[Proof.]
The proof is similar to that of Lemma \ref{covarianceeventbound1}. By symmetry, and a union bound, we have that
\begin{align}
    \PP\left(\mathcal{E}_5 \right) &\leq 3 \PP  \left (\normmop{\frac{1}{\tau}\sum_{i=1}^{\tau}Y_i Y_i^\top - \Sigma_1}  \geq \frac{c_2}{3} \normmop{\Sigma_1} \left( \frac{p \vee \log t}{\tau} \vee \sqrt{\frac{p \vee \log t}{\tau}}   \right) \right).
\end{align}
Due to Lemma \ref{moenminimaxlemma2}, we have\textsuperscript{\ref{note1}}
\begin{align}
    \PP\left(\mathcal{E}_5 \right) &<  3  \exp\left\{ - \frac{c_2}{4c_0} \left( p \vee \log t\right)\right\}\\
    &\leq \frac{3}{4} \delta \\
    &< \delta, 
\end{align}
by the definition of $c_2$, using a similar argument as in the proof of Lemma \ref{covarianceeventbound1}.
\end{proof}

\begin{lemma}\label{covarianceeventbound3}
Let $p\in \NN$ and let $Y_1, Y_2, \ldots$ be $p$-dimensional random variables satisfying Assumption \ref{assmultivariate2} for some $w,u>0$. Assume that $\EE Y_i Y_i^\top = \Sigma_1$ for all $i \in \NN$ and some positive definite matrix $\Sigma_1 \in \RR^{p\times p}$. For any $t\geq 2$ and $g \in [t-1]$, define
\begin{align}
    &\widetilde{\Sigma}_{1,g}^{(t)} = 2^{- \lfloor \log_2 g \rfloor} \sum_{i=1}^{2^{ \lfloor \log_2 g \rfloor}} Y_i Y_i^\top, &\widetilde{\Sigma}_{2,g}^{(t)} = g^{-1} \sum_{i=t-g+1}^t Y_{i} Y_i^\top,
\end{align}
as in \eqref{sigmatilde}, and define $\widehat{\sigma}^{(t)}_{g} = \widehat{\lambda}_{\max}^1(\widetilde{\Sigma}_{1,g}^{(t)})^{1/2}$, where $\widehat{\lambda}_{\max}^s(\cdot)$ is defined in \eqref{conveigdef}.
Let $h(p,t,g,s)$ be defined as in \eqref{hdef}, and for any $\delta\in (0,1)$, define the events
\begin{align}
    \mathcal{E}_6 &= \bigcap_{t=2}^{\infty}  \bigcap_{ g = 1}^{ \lfloor t/2 \rfloor} \left\{ \left(\widehat{\sigma}_{g}^{(t)}\right)^2 \geq c_3 \lambda_{\max}^1({\Sigma_1})   \right\},\\
    \mathcal{E}_7 &= \bigcap_{t=2}^{\infty} \bigcap_{s \in \mathcal{S}} \bigcap_{ g = 1}^{ \lfloor t/2 \rfloor} \left\{ \widehat{\lambda}_{\max}^s ({\widetilde{\Sigma}_{1,g}^{(t)} - \widetilde{\Sigma}_{2,g}^{(t)}})  \leq c_4 \lambda_{\max}^1({\Sigma_1}) h(p,t,g,s) \right\},
\end{align}
where $\lambda_{\max}^s(\cdot)$ is defined in \eqref{firsteigsparse}, $\mathcal{S}$ is given in \eqref{mathcals}, $c_3 = (2e\pi w^2)^{-1}\delta^2 (\delta + 2)^{-2}$, $c_4 = (c_0 / 2)\{ 3 + \log_2(8/\delta) / \log 2\}$, and $c_0$ is the constant from Lemma \ref{moenminimaxlemma2} depending only on $u>0$. Then $\PP(\mathcal{E}_6 \cap \mathcal{E}_7) \geq 1 - \delta$. 
\end{lemma}
\begin{proof}[Proof.]
 We first show that  $\PP(\mathcal{E}_6^\complement)\leq \delta/2$. Since $\widehat{\lambda}_{\max}^s(\cdot)$ is a convex relaxation of the implicit optimisation problem in the definition of $\lambda_{\max}^s(\cdot)$, we have
 \begin{align}
     \widehat{\sigma}_{g}^{(t)} &= \widehat{\lambda}_{\max}^1(\widetilde{\Sigma}_{1,g}^{(t)})^{1/2}\\
     &\geq \lambda_{\max}^1(\widetilde{\Sigma}_{1,g}^{(t)})^{1/2}.
 \end{align}
From a union bound, it follows that 
 \begin{align}
     \PP\left(\mathcal{E}_6^\complement\right) &\leq \PP \left[ \bigcup_{t=1}^{\infty} \bigcup_{g=1}^{\lfloor t/2\rfloor} \left\{    \lambda_{\max}^1(\widetilde{\Sigma}_{1,g}^{(t)}) \leq c_3  \lambda_{\max}^1(\Sigma_1) \right\}\right] \\
     &= \PP \left[ \bigcup_{t=1}^{\infty} \left\{   \lambda_{\max}^1 \left( {t^{-1} \sum_{i=1}^{t} Y_i Y_i^\top} \right) \leq c_3 \lambda_{\max}^1(\Sigma_1) \right\}\right]\\
     &\leq \sum_{t=1}^{\infty} \PP \left\{  \lambda_{\max}^1 \left( {t^{-1} \sum_{i=1}^{t} Y_i Y_i^\top} \right) \leq c_3 \lambda_{\max}^1(\Sigma_1) \right\}.
 \end{align}

Now let $v_1 \in \mathbb{S}^{p-1}_1$ be a vector in the unit sphere with one non-zero entry satisfying $v_1^\top \Sigma_1 v_1 = \lambda_{\max}^1(\Sigma_1)$. For any fixed $t\in \NN$, we then have
\begin{align}
    \lambda_{\max}^1 \left( {t^{-1} \sum_{i=1}^{t} Y_i Y_i^\top}\right) &= \underset{v \in \mathbb{S}^{p-1}_1}{\sup} \ \sum_{i=1}^t (v^\top Y_i)^2\\
    &\geq  \sum_{i=1}^t (v_1^\top Y_i)^2. 
\end{align}
Due to Assumption \ref{assmultivariate2}, Lemma \ref{moenminimaxlemma1} implies that 
\begin{align}
    \PP \left\{  \lambda_{\max}^1 \left( {t^{-1} \sum_{i=1}^{t} Y_i Y_i^\top} \right) \leq c_3 \lambda_{\max}^1(\Sigma_1) \right\} &\leq \PP \left(   \sum_{i=1}^t (v_1^\top Y_i)^2    \leq  t c_3 v_1^\top \Sigma v_1 \right) \\
    &\leq \exp \left\{ - \frac{t}{2}\log\left(\frac{1}{2e\pi w^2 c_3}\right)  \right\}.
\end{align}
We thus have that
\begin{align}
     \PP \left(\mathcal{E}_6^\complement\right)  &\leq \sum_{t=1}^{\infty} \exp \left\{ - \frac{t}{2}\log\left(\frac{1}{2e\pi w^2 c_3}\right)  \right\} \\
     &= \sum_{t=1}^{\infty} (\sqrt{2e\pi w^2 c_3})^t. 
\end{align}
Now, since $2e\pi w^2 c_3<1$ we thus have
\begin{align}
    \PP \left(\mathcal{E}_6^\complement\right)  &\leq \frac{\sqrt{2e\pi w^2 c_3}}{1 - \sqrt{2e\pi w^2 c_3}},
\end{align}
and due to the definition of $c_3$ we have $\PP(\mathcal{E}_6^\complement)\leq \delta/2$.

Next we will show that $\PP(\mathcal{E}_7^{\complement} ) \leq \delta/2$. For any matrix $A = (a_{i,j})_{i\in [k_1], j \in [k_2]}$, let $\normm{A}_{\infty} = \max_{i \in [k_1], j \in [k_2]} |a_{i,j}|$ denote the largest absolute entry of $A$. 
Then for any $t \geq 2$, $g \leq t/2$ and $s \in \mathcal{S}$,  
Lemma \ref{moenminimaxlemma11} implies that\footnote{In the Lemma, set $A = {\widetilde{\Sigma}_{1,g}^{(t)} - \widetilde{\Sigma}_{2,g}^{(t)}}$ and $Y = {\widetilde{\Sigma}_{2,g}^{(t)} - \widetilde{\Sigma}_{1,g}^{(t)}}$.}
\begin{align}
    \widehat{\lambda}_{\max}^s ( {\widetilde{\Sigma}_{1,g}^{(t)} - \widetilde{\Sigma}_{2,g}^{(t)}} ) &\leq s\normm{{\widetilde{\Sigma}_{1,g}^{(t)} - \widetilde{\Sigma}_{2,g}^{(t)}}}_{\infty}\\
    &\leq  s\normm{{\widetilde{\Sigma}_{1,g}^{(t)} - \Sigma_1}}_{\infty} + s\normm{{\widetilde{\Sigma}_{2,g}^{(t)} - \Sigma_1}}_{\infty}.
\end{align}
Now, let $Y_{i,t,g,j,k}$ denote the $(j,k)$-th element of $\widetilde{\Sigma}_{i,g}^{(t)} - \Sigma_1$, for $i = 1,2$ and $j,k \in [p]$, so that  
\begin{align}
    Y_{1,t,g, j,k} &= {2^{-\lfloor \log_2 g  \rfloor}} \sum_{l=1}^{2^{\lfloor \log_2 g  \rfloor}} \left\{ Y_l(j) Y_l(k) - \EE Y_l(j)Y_l(k)\right\},\\
    Y_{2,t,g, j,k} &= \frac{1}{g} \sum_{l=t-g+1}^{t} \left\{ Y_l(j) Y_l(k) - \EE Y_l(j)Y_l(k)\right\},\\
\end{align}
By symmetry and a union bound, we have that
\begin{align}
    \PP(\mathcal{E}_{7}^\complement) &\leq \sum_{t=2}^{\infty} \sum_{s \in \mathcal{S}} \sum_{g = 1}^{ \lfloor t/2 \rfloor}\PP \left\{ \lambda_{\max}^s ({\widetilde{\Sigma}_{1,g}^{(t)} - \widetilde{\Sigma}_{2,g}^{(t)}})  > c_4 \lambda_{\max}^1({\Sigma_1}) h(p,t,g,s) \right\} \\
    \leq & 2 \sum_{t=2}^{\infty}  \sum_{s \in \mathcal{S}} \sum_{g =1}^{\lfloor t/2 \rfloor} \sum_{j,k \in [p]}  \PP \left\{ | Y_{1,t,g,j,k}| > \frac{c_4}{2s} \lambda_{\max}^1(\Sigma_1) h(p,t,g,s)    \right\}. \label{tihi123}
\end{align}
Due to Lemma 2.7.7 in \citet{Vershynin_2018} and Assumption \ref{assmultivariate2}, we have that
\begin{align}
    \normm{Y_l(j) Y_l(k)}_{\Psi_1} &\leq \normmo{Y_l(j)} \normmo{Y_l(k)}\\
    &\leq u \lambda_{\max}^1(\Sigma_1),
\end{align}
for all $i \in \NN$, $s \in [p]$ and $(j,k) \in [p]\times [p]$, where we recall that $\normm{\cdot}_{\Psi_1}$ denotes the sub-exponential norm of a univariate random variable, defined by $\normm{X}_{\Psi_1} = \inf\{x > 0: \EE \exp(|X|/x) \leq 2\}$.  

For any $x>0$, Bernstein's Inequality \citep[Theorem 2.8.1 in][]{Vershynin_2018} therefore implies that
    \begin{align}
        &\PP \left[ | Y_{1,t,g,j,k}| \geq 2 c u \lambda_{\max}^1 (\Sigma) \left( \sqrt{\frac{x}{g}} \vee \frac{x}{g}\right) \right]\\
        &\leq \PP \left[ | Y_{1,t,g,j,k}| \geq  c u \lambda_{\max}^1 (\Sigma) \left( \sqrt{\frac{x}{2^{\lfloor \log_2 g \rfloor}}} \vee \frac{x}{2^{\lfloor \log_2 g \rfloor}}\right) \right] \\
        &\leq 2e^{-x},
    \end{align}
for any $x>0$ and some absolute constant $c\geq 1$, where we used that $2^{\lfloor \log_2 g \rfloor} \geq g/2$.  Taking $c_0 = 16cu$ (the same constant as in Lemma \ref{moenminimaxlemma2}, for convenience) and  $x = 4 c_4 c_0^{-1} \log (p \vee t)$, we obtain
    \begin{align}
         &\PP \left\{ | Y_{1,t,j,k}| > \frac{c_4}{2s} \lambda_{\max}^1(\Sigma) h(p,t,g,s)    \right\}  \\ =& \PP \left[| Y_{1,t,j,k}| > \frac{c_4}{2} \lambda_{\max}^1(\Sigma) \left\{ \frac{\log (p\vee t)}{g} \vee  \sqrt{\frac{\log (p\vee t)}{g}} \right\}  \right] \\ 
         \leq & \PP \left[ | Y_{1,t,g,j,k}| \geq 2 c u \lambda_{\max}^1 (\Sigma) \left( \sqrt{\frac{x}{g}} \vee \frac{x}{g}\right) \right]\\
         \leq &2e^{-x}\\
         = & 2 \exp\left(- \frac{4c_4}{c_0}   \log(p \vee t) \right)\\
         \leq & 2 \exp\left\{- \frac{2c_4}{c_0}  \left(\log p  + \log t\right)\right\},
    \end{align}
using that $\log (p \vee t)\geq (1/2)\log p + (1/2) \log t$. 

Inserting into \eqref{tihi123}, we obtain
\begin{align}
     \PP(\mathcal{E}_{7}^\complement)  &\leq 2 \sum_{t=2}^{\infty}  \sum_{s \in \mathcal{S}} \sum_{g =1}^{\lfloor t/2 \rfloor} \sum_{j,k \in [p]}  \PP \left\{ | Y_{1,t,g,j,k}| > \frac{c_4}{2s} \lambda_{\max}^1(\Sigma_1) h(p,t,g,s)    \right\}\\
     &\leq 2 \sum_{t=2}^{\infty}  \sum_{s \in \mathcal{S}} \sum_{g =1}^{\lfloor t/2 \rfloor} \sum_{j,k \in [p]}  2 \exp\left\{- \frac{2c_4}{c_0}  \left( \log p  + \log t\right)\right\}\\
     &\leq  4 |\mathcal{S}| \sum_{t=2}^{\infty} |\mathcal{G}^{(t)}|   p^2 t^{- 2c_4 /c_0} p^{-2 c_4/c_0}\\
     &\leq 4 |\mathcal{S}|\sum_{t=2}^{\infty} t p^2 t^{- 2c_4 /c_0} p^{-2 c_4/c_0}\\
     &\leq 4 \frac{|\mathcal{S}|}{p}\sum_{t=2}^{\infty} t^{1  - 2c_4 /c_0} \\
     &\leq 4 \sum_{t=2}^{\infty} t^{-2} t^{3  - 2c_4 /c_0} \\
     &\leq 4 \cdot 2^{3  - 2c_4 /c_0} \sum_{t=2}^{\infty} t^{-2}  \\
     &\leq 4 \cdot  2^{3  - 2c_4 /c_0} \\
     &\leq \delta/2, 
\end{align}
where we in the fifth inequality used that $p^2 p ^{-2c_4 / c_0} \leq p^{-1}$ since $c_4 \geq (3/2)c_0$, in the seventh inequality used that $3 - 2c_4/c_0 <0$, and in the last inequality used the definition of $c_4$. The proof is complete. 
\end{proof}

\begin{lemma}\label{covarianceeventbound4}

Let $p \in \NN$, $\tau \in \NN$, $s \in [p]$, $\tau <t\leq 2\tau$ and let $Y_1, \ldots, Y_{t}$ be $p$-dimensional random variables satisfying Assumption \ref{assmultivariate2} for some $w,u>0$. Assume that $\EE Y_i Y_i^\top = \Sigma_1$ for all $i \leq \tau$ and $\EE Y_i Y_i^\top = \Sigma_2$ for $i>\tau$, for some positive definite matrices $\Sigma_1, \Sigma_2 \in \RR^{p\times p}$. Given any fixed $g \leq t - \tau$, define
\begin{align}
    &\widehat{\Sigma}_{1,g}^{(t)} = 2^{- \lfloor \log_2 g \rfloor} \sum_{i=1}^{2^{ \lfloor \log_2 g \rfloor}} Y_i Y_i^\top, &\widehat{\Sigma}_{2,g}^{(t)} = g^{-1} \sum_{i=t-g+1}^t Y_{i} Y_i^\top.
\end{align}
Let $h(p,t,g,s)$ be defined as in \eqref{hdef}, and for any $\delta\in (0,1)$, define the events
\begin{align}
    \mathcal{E}_8 &=  \bigcap_{i=1,2} \left\{ \widehat{\lambda}_{\max}^{s} (\widehat{\Sigma}_{i,g}^{(t)} - \Sigma_i)  \leq \frac{c_4}{2} \lambda_{\max}^{1}({\Sigma_i}) h(p,t,g,s) \right\},\\
    \mathcal{E}_9 &=  \left\{ \widehat{\lambda}_{\max}^{1} (\widehat{\Sigma}_{1,g}^{(t)} - \Sigma_1)  \leq \frac{c_4}{2} \lambda_{\max}^{1}({\Sigma_1}) h(p,t,g,1) \right\},
    \end{align}
where $\widehat{\lambda}_{\max}^s(\cdot)$ is defined in \eqref{conveigdef}, $\lambda_{\max}^s(\cdot)$ is defined in \eqref{firsteigsparse}, $c_4 = (c_0 / 2)\{ 3 + \log_2(8/\delta) / \log 2\}$, and $c_0$ is the constant from Lemma \ref{moenminimaxlemma2} depending only on $u>0$.  Then $\PP(\mathcal{E}_8\cap \mathcal{E}_9) \geq 1 - \delta$. 
\end{lemma}
\begin{proof}[Proof.]
The proof of Lemma \ref{covarianceeventbound4} is very similar to the second part of the proof of Lemma \ref{covarianceeventbound3}. For any matrix $A = (a_{i,j})_{i\in [k_1], j \in [k_2]}$, define $\normm{A}_{\infty} = \max_{i \in [k_1], j \in [k_2]} |a_{i,j}|$ to be the largest absolute entry in $A$. 
Then 
Lemma \ref{moenminimaxlemma11} implies that
\begin{align}
    \widehat{\lambda}_{\max}^s ( {\widehat{\Sigma}_{i,g}^{(t)} - \Sigma_i} ) &\leq s\normm{{\widehat{\Sigma}_{i,g}^{(t)} - \Sigma_i}}_{\infty},
\end{align}
for $i=1,2$. 
Now, let $Y_{i,t,g,j,k}$ denote the $(j,k)$-th element of $\widehat{\Sigma}_{i,g}^{(t)} - \Sigma_i$, for $i = 1,2$, and $j,k \in [p]$, so that  
\begin{align}
    Y_{1,t,g, j,k} &= {2^{-\lfloor \log_2 g  \rfloor}} \sum_{l=1}^{2^{\lfloor \log_2 g  \rfloor}} \left\{ Y_l(j) Y_l(k) - \EE Y_l(j)Y_l(k)\right\},\\
    Y_{2,t,g, j,k} &= \frac{1}{g} \sum_{l=t-g+1}^{t} \left\{ Y_l(j) Y_l(k) - \EE Y_l(j)Y_l(k)\right\},\\
\end{align}
using that $g \leq \tau \wedge(t-\tau)$. 
By symmetry and a union bound, we have that
\begin{align}
    \PP(\mathcal{E}_{8}^\complement) 
    & \leq  2  \sum_{j,k \in [p]}  \PP \left\{ | Y_{1,t,g,j,k}| > \frac{c_4}{2s} \lambda_{\max}^1(\Sigma_1) h(p,t,g,s)    \right\},
    \intertext{and}
    \PP(\mathcal{E}_{9}^\complement) 
    & \leq   \sum_{j,k \in [p]}  \PP \left\{ | Y_{1,t,g,j,k}| > \frac{c_4}{2} \lambda_{\max}^1(\Sigma_1) h(p,t,g,1)    \right\}.
\end{align}
Arguing in a precisely similar fashion as in the proof of Lemma \ref{covarianceeventbound3}, one concludes that $\PP(\mathcal{E}_8^\complement \cup \mathcal{E}_9^\complement) \leq \delta/2 < \delta$, 
and the proof is complete. 
\end{proof}

The following Lemma is due to \citet{moen2023efficient}.
\begin{lemma}[\citealt{moen2023efficient}, Lemma F.1]\label{palemma1}
For any $a\geq 0$, define $\nu_a = \EE\left( Z^2 \ \mid \ |Z|\geq a\right)$ where $Z \sim \N(0,1)$. Then
$$
a^2 + 1 \leq \nu_a \leq a^2+2.
$$
\end{lemma}
The following Lemma is due to \citet{liu_minimax_2021}. 
\begin{lemma}[\citealt{liu_minimax_2021}, Lemma 5]\label{lemmaliu}
Let $Z_i \iid \N(0,1)$ for $i \in [p]$, where $p \in \NN$. Let $a\geq 0$ and define $\nu_a = \EE\left( Z^2 \ \mid \ |Z|\geq a\right)$. Then for all $x>0$, 
$$
\PP \left[ \sum_{i=1}^p (Z_i^2 - \nu_a) \ind{\left( |Z_i| \geq a\right)} \geq 9 \left \{  \left( p e^{-a^2/2}x\right)^{1/2} + x \right\}\right] \leq e^{-x}.
$$
\end{lemma}
The following Lemma is due to \citet{moen2023efficient}
\begin{lemma}[\citealt{moen2023efficient}, Lemma F.3]\label{palemma2}
Let $Z_i \iid \N(0,1)$ for $i \in [p]$, where $p \in \NN$. Let $a\geq 1$ and define $\nu_a = \EE\left( Z_1^2 \ \mid \ |Z_1|\geq a\right)$. Then for all $x>0$, 
$$
\PP \left[  \sum_{i=1}^p (Z_i^2 - \nu_a) \ind{\left( |Z_i| \geq a\right)} \leq - 5 \left \{  \left( p e^{-a^2/2}x\right)^{1/2} + x \right\}\right] \leq e^{-x}.
$$
\end{lemma}

The following Lemma is due to \citet{birge_alternative_2001}.
\begin{lemma}[\citealt{birge_alternative_2001}, Lemma 8.1]\label{birgelemma}
Let $Y \sim \chi_p^2(\Psi)$ have a non-central Chi Square distribution with $p$ degrees of freedom and non-centrality parameter $\Psi\geq0$. Then, for any $x>0$, we have that
\begin{align}
\PP\left[     Y\geq p + \Psi  + 2 \left\{x(p + 2\Psi)\right\}^{1/2} +2x      \right] &\leq e^{-x},
\intertext{and,}
\PP\left[     Y\leq p + \Psi  - 2\left\{x(p + 2\Psi)\right\}^{1/2}       \right] &\leq e^{-x},
\end{align}
\end{lemma}
The following Lemma is due to \citet{moen2024minimax}. 
\begin{lemma}[\citealt{moen2024minimax}, Lemma 1]\label{moenminimaxlemma1}
Let $Y_1, \ldots, Y_n$ be independent random variables, and assume that each $Y_i /\sigma$ has a continuous density bounded above by $w$ for $i=1, \ldots, n\in \NN$ and some $w>0$. Let $S = \sum_{i=1}^n Y_i^2$. Then for any $x>0$ we have
\begin{align}
    \PP\left(  S \leq \sigma^2 x\right) &\leq \exp \left[ \frac{n}{2} \left\{1 + \log\left(2\pi w^2\right) - \log\left(\frac{n}{x}\right) \right\} \right].
\end{align}
\end{lemma}

In the following, we let $\lambda_{\max}^s(A) = \sup_{v \in \mathbb{S}^{p-1}_s} \ | v^\top A v |$ denote the largest $s$-sparse eigenvalue of a symmetric matrix $A\in \RR^{p\times p}$ for any $s\in[p]$, as in \eqref{firsteigsparse}.  The following Lemmas are due to \citet{moen2024minimax}. 
\begin{lemma}[\citealt{moen2024minimax}, Lemma 2]\label{moenminimaxlemma2}
    Fix any $p\in \NN$ and $s\in [p]$, and let $Y_i$ be centred and independent $p$-dimensional sub-Gaussian random variables with $\EE Y_i Y_i^\top = \Sigma$, for $i=1, \ldots, n$ and some $\Sigma \in \RR^{p \times p}$.  Assume further that $\normmo{Y_i}^2 \leq u \normmop{\Sigma}$ for all $i$ and some $u>0$. Let $\widehat{\Sigma} = n^{-1} \sum_{i=1}^n Y_i Y_i^\top$. There exists a constant $c_0>0$ depending only on $u$, such that, for all $x\geq 1$, we have
    \begin{align}
        &\PP \left[\lambda_{\max}^s (\widehat{\Sigma}- \Sigma) \geq c_0 \normm{\Sigma}_{\mathrm{op}} \left\{   \sqrt{\frac{s\log(ep/s)}{n}}\vee \frac{s\log(ep/s)}{n} \vee \sqrt{\frac{x}{n}} \vee \frac{x}{n}  \right\}  \right] \\
        \leq & e^{-x}. \label{eqberthet}
    \end{align}
    Moreover, if $\normmo{v^\top Y_i}^2 \leq u ( v^\top \Sigma v)$ for any $v \in \mathbb{S}^{p-1}$, the factor $\normm{\Sigma}_{\mathrm{op}}$ on the left-hand side of \eqref{eqberthet} can be replaced by $\lambda_{\max}^s ({\Sigma})$ as defined in \eqref{firsteigsparse}. 
\end{lemma}

\begin{lemma}[\citealt{moen2024minimax}, Lemma 11]\label{moenminimaxlemma11}
    Let $A = (a_{i,j})_{i,j \in [p]} \in \RR^{p\times p}$ be a symmetric matrix, and define $\normm{A}_{\infty} = \max_{i, j \in [p]} |a_{i,j}|$. Then we have
    \begin{align}
      \underset{\substack{Z \in N(p,s) }}{\sup} \ \mathrm{Tr}(AZ)  &\leq \underset{Y \in \mathrm{Sym}(p)}{\inf} \ \underset{\substack{Z \in \PSD(p)\\ \Tr(Z)=1}}{\sup} \Tr Z(A+Y) + s \normm{Y}_{\infty}\\
    &\leq  \underset{Y \in \mathrm{Sym}(p)}{\inf} \ \lambda_{\max}^p (A+Y) + s \normm{Y}_{\infty},
\end{align}
where $ N(p,s) = \{Z \in \PSD(p) \ ; \ \mathrm{Tr}(Z)=1, \normm{Z}_1\leq s\}$ and $\lambda_{\max}^p(\cdot)$ is defined in \eqref{firsteigsparse}. 
\end{lemma}

\begin{lemma}[\citealt{moen2024minimax}, Lemma 12]\label{tullelemma}
Let $\lambda_{\max}^s(\cdot)$ be defined as in \eqref{firsteigsparse} and let $A\in \RR^{p\times p}$ be a symmetric positive definite matrix. Then $\lambda_{\max}^s(A)$ is non-decreasing in $s$, and for any $s_0/2 \leq s\leq s_0\leq p $ we have $\lambda_{\max}^{s_0}(A) \leq 4 \lambda_{\max}^{s}(A)$, where $\lambda_{\max}^{s}(\cdot)$ is defined in \eqref{firsteigsparse}. Moreover, if $A_1, A_2 \in \RR^{p\times p}$ are two symmetric positive definite matrices, then $\lambda_{\max}^s(A_1  - A_2) \leq \lambda_{\max}^s(A_1) \vee \lambda_{\max}^s (A_2)$ for any $s \in [p]$.
\end{lemma}

\bibliography{reference}

\end{document}